%% file: main.tex
\documentclass{article}
\usepackage[utf8]{inputenc}

\input{header.tex}

\input{newcommands.tex}

\title{Hypergraph Connectivity Augmentation\\ in Strongly Polynomial Time}

\author{
Krist\'{o}f B\'{e}rczi\thanks{MTA-ELTE Matroid Optimization Research Group and HUN-REN-ELTE Egerv\'{a}ry Research Group, Department of Operations Research, ELTE E\"{o}tv\"{o}s Lor\'{a}nd
University, Budapest, Hungary, Email: \{kristof.berczi, tamas.kiraly\}@ttk.elte.hu.}
\and Karthekeyan Chandrasekaran\thanks{University of Illinois, Urbana-Champaign, USA, Email: \{karthe, smkulka2\}@illinois.edu. 
Part of this work was done while Karthekeyan and Shubhang were visiting E\"{o}tv\"{o}s Lor\'{a}nd
University.}
\and Tam\'{a}s Kir\'{a}ly\footnotemark[1]
\and Shubhang Kulkarni\footnotemark[2]
}
\date{}
\begin{document}
\maketitle
\input{abstract.tex}

\newpage
\tableofcontents
\newpage

\setcounter{page}{1}
\input{intro-2.tex}
\input{related-work}
\input{preliminaries}
\input{set-families}
\input{weak-cover-szigeti}
\input{covering-algorithm}

\input{runtime-for-covering-with-uniform-hypergraph}
\input{BK-covering-1-function-with-uniform-hypergraph}

\input{BK-covering-2-functions-with-uniform-hypergraph}

\input{applications}
\input{conclusion}

\input{acknowledgement.tex}

\bibliographystyle{abbrv}
\bibliography{references}

\appendix
\input{appendix}

\end{document}

%% file: header.tex
\usepackage[T1]{fontenc}
\usepackage{soul, comment} 
\usepackage{phfqit}
\usepackage{proba} 
\usepackage[sort,nocompress]{cite}
\usepackage{amssymb}
\usepackage{amsmath}
\usepackage{amsthm}
\usepackage{amsfonts}
\usepackage{mathtools}
\usepackage{xspace}
\usepackage{bm}
\usepackage{nicefrac}
\usepackage{commath}
\usepackage{bbm}
\usepackage{boxedminipage}
\usepackage{xparse}
\usepackage{xcolor}
\usepackage{float}
\usepackage{multirow}
\usepackage{makecell}
\usepackage{graphicx}
\usepackage{caption}
\usepackage{subcaption}
\usepackage{ifthen}
\usepackage{thm-restate}
\usepackage{algorithm}
\usepackage[noend]{algpseudocode}
\usepackage{colortbl} 
\usepackage{fancyhdr}   
\usepackage[hypertexnames=false]{hyperref}
    \hypersetup{ colorlinks=true, linkcolor=blue, citecolor=magenta, urlcolor=blue,}
\usepackage{fullpage}
\usepackage[utf8]{inputenc}
\usepackage{environ}
\usepackage{mdframed}
\usepackage{cleveref}
\usepackage[short]{optidef} 
\usepackage{enumitem}
\usepackage{tikz}
\usetikzlibrary{shapes}
\usetikzlibrary{positioning}
\usetikzlibrary{fit}
\usetikzlibrary{graphs,graphs.standard}
\usepackage{enumitem}

\newtheorem{proposition}{Proposition}[section]
\newtheorem{lemma}{Lemma}[section]
\newtheorem{remark}{Remark}[section]
\newtheorem{theorem}{Theorem}[section]
\newtheorem{corollary}{Corollary}[section]

\newtheorem{definition}{Definition}[section]

\newtheorem{observation}{Observation}[section]

\newtheorem{claim}{Claim}[section]

\usepackage{enumitem}

\newcommand{\ceil}[1]{\ensuremath{\left\lceil{#1}\right\rceil}\xspace}
\newcommand{\floor}[1]{\ensuremath{\left\lfloor{#1}\right\rfloor}\xspace}

\NewEnviron{problem}[1]{%
	\begin{center}\fbox{\parbox{6in}{%
				{\centering\scshape #1\par}%
				\parskip=1ex
				\everypar{\hangindent=1em}%
				\BODY
}}\end{center}}

\setlist[enumerate]{nosep}

%% file: newcommands.tex
\newcommand{\zeros}{\ensuremath{\mathcal{Z}}\xspace}

\newcommand{\indicator}{\ensuremath{\mathbbm{1}}\xspace}
\newcommand{\support}{\ensuremath{\mathtt{supp}\xspace}}
\newcommand{\contractionstep}{\textsc{Contraction Step}{}}
\newcommand{\createstep}{\textsc{Create Step}{}}

\newcommand{\OracleMain}{\ensuremath{p\text{-max}\mathtt{-Oracle}}\xspace}
\newcommand{\OracleMainWithArg}[1]{\ensuremath{{#1}\text{-max}\mathtt{Oracle}}\xspace}

\newcommand{\alphagood}{\ensuremath{\alpha\text{-good}}\xspace}

\newcommand{\functionMaximizationOracle}[1]{\ensuremath{{#1}\text{-}\mathtt{max}\text{-}\mathtt{wc}\text{-}\mathtt{Oracle}}\xspace}
\newcommand{\functionMaximizationOracleStrongCover}[1]{\ensuremath{{#1}\text{-}\mathtt{max}\text{-}\mathtt{sc}\text{-}\mathtt{Oracle}}\xspace}
\newcommand{\functionMaximizationEmptyOracle}[1]
{\ensuremath{{#1}\text{-}\mathtt{max}\text{-}\mathtt{Oracle}}\xspace}

\newcommand{\maximalTightSetFamily}[1]{\ensuremath{\calT_{#1}}\xspace}
\newcommand{\maximalTightSetFamilyAfterCreate}[1]{\ensuremath{\calT_{#1}'}\xspace}
\newcommand{\tightSetFamily}[1]{\ensuremath{T_{#1}}\xspace}
\newcommand{\tightSetFamilyAfterCreate}[1]{\ensuremath{T_{#1}'}\xspace}

\newcommand{\projTightSetFamilyAfterCreate}[2]{\ensuremath{T_{#1}'|_{V_{#2}}}\xspace}

\newcommand{\cumulativeMaximalTightSetFamily}[1]{\ensuremath{\calT_{\leq #1}}}

\newcommand{\projCumulativeMaximalTightSetFamily}[2]{\ensuremath{\calT_{\leq #1}|_{V_{#2}}}}

\newcommand{\cumulativeZeros}[1]{\ensuremath{\calZ_{\leq #1}}\xspace}

\newcommand{\cumulativeMinimalMaximizerFamily}[1]{\ensuremath{\calF_{\leq #1}}\xspace}
\newcommand{\minimalMaximizerFamily}[1]{\ensuremath{\calF_{#1}}\xspace}
\newcommand{\minimalMaximizerFamilyAfterCreate}[1]{\ensuremath{\calF_{#1}'}\xspace}
\newcommand{\cumulativeMinimalMaximizerFamilyAfterCreate}[1]{\ensuremath{\calF_{\leq #1}'}\xspace}
\newcommand{\maxFunctionValue}[1]{\ensuremath{K_{#1}}\xspace}

\newcommand{\cumulativeFunctionMinimalMaximizerFamily}[2]{\ensuremath{\calF_{{#1}_{\leq #2}}\xspace}}

\newcommand{\slack}{\ensuremath{\Gamma}\xspace}
\newcommand{\slackAfterCreate}{\ensuremath{\Gamma'}\xspace}
\newcommand{\slackAfterContract}{\ensuremath{\Gamma''}\xspace}
\newcommand{\pairwiseSlack}{\ensuremath{\gamma}\xspace}
\newcommand{\pairwiseSlackAfterCreate}{\ensuremath{\gamma'}\xspace}
\newcommand{\pairwiseSlackAfterContract}{\ensuremath{\gamma''}\xspace}

\newcommand{\shubhang}[1]{{\color{blue}[{\tiny Shubhang: \bf #1}]\marginpar{\color{blue}*}}}
\newcommand{\knote}[1]{{\color{red}[{\tiny Karthik: \bf #1}]\marginpar{\color{red}*}}}
\newcommand{\kristof}[1]{{\color{red}[{\tiny Kristof: \bf #1}]\marginpar{\color{red}*}}}

\newcommand{\functionContract}[2]{\ensuremath{{#1}/_{#2}}\xspace}
\newcommand{\functionRestrict}[2]{\ensuremath{{#1}\backslash_{#2}}\xspace}

\newcommand{\dsggcae}{\textsc{DS-Graph-GCA-using-E}\xspace}

\newcommand{\dsglcae}{\textsc{DS-Graph-CA-using-E}\xspace}
\newcommand{\dshlcae}{\textsc{DS-Hypergraph-CA-using-E}\xspace}
\newcommand{\dshlcah}{\textsc{DS-Hypergraph-CA-using-H}\xspace}
\newcommand{\dshlcauh}{\textsc{DS-Hypergraph-CA-using-near-uniform-H}\xspace}
\newcommand{\dsshlcah}{\textsc{DS-Simul-Hypergraph-CA-using-H}\xspace}
\newcommand{\dsshlcauh}{\textsc{DS-Simul-Hypergraph-CA-using-near-uniform-H}\xspace}

\newcommand{\dssssch}{\textsc{DS-Sym-Skew-SupMod-StrongCover-using-H}\xspace}
\newcommand{\dsswch}{\textsc{DS-Skew-SupMod-WeakCover-using-H}\xspace}
\newcommand{\dsswcuh}{\textsc{DS-Skew-SupMod-WeakCover-using-near-uniform-H}\xspace}
\newcommand{\dssswch}{\textsc{DS-Simul-Skew-SupMod-WeakCover-using-H}\xspace}
\newcommand{\dssswcuh}{\textsc{DS-Simul-Skew-SupMod-WeakCover-using-near-uniform-H}\xspace}

\newcommand{\dsssscuh}{\textsc{DS-Sym-Skew-SupMod-StrongCover-using-near-uniform-H}\xspace}
\newcommand{\dsssssch}{\textsc{DS-Simul-Sym-Skew-Skew-SupMod-StrongCover-using-H}\xspace}
\newcommand{\dssssscuh}{\textsc{DS-Simul-Sym-Skew-SupMod-StrongCover-using-near-uniform-H}\xspace}

\newcommand{\dshnacah}{\textsc{DS-HyperGraph-Node-to-Area-CA-using-H}\xspace}

\newcommand{\dcmhgcah}{\textsc{DC-Mixed-HyperGraph-Global-CA-using-H}\xspace}

\newcommand{\contrapolymatroid}{\ensuremath{\mathcal{C}}\xspace}

\newcommand{\basepolyhedron}{\ensuremath{\mathcal{B}}\xspace}
\newcommand{\gpolymatroid}{\ensuremath{\mathcal{G}}\xspace}
\newcommand{\upperTruncation}[1]{\ensuremath{{#1}^{\uparrow}}\xspace}

\newcommand{\biTruncation}[1]{\ensuremath{{#1}^{\downarrow\downarrow}}\xspace}
\newcommand{\functionMinimizationOracle}[1]{\ensuremath{{#1}\text{-min-}\mathtt{Oracle}}\xspace}
\newcommand{\QpmGeneralization}{\ensuremath{\mathcal{Q}}\xspace}

\newcommand{\mypara}[1]{\vspace{5pt} \noindent {\bf #1}}
\renewcommand{\paragraph}{\mypara}

%% file: abstract.tex
\begin{abstract}
We consider hypergraph network design problems where the goal is to construct a hypergraph that satisfies certain connectivity requirements. 
For graph network design problems where the goal is to construct a graph that satisfies certain connectivity requirements, the number of edges in every feasible solution is at most quadratic in the number of vertices. 
In contrast, for hypergraph network design problems, we might have feasible solutions in which the number of hyperedges is exponential in the number of vertices. 
This presents an additional technical challenge in hypergraph network design problems compared to graph network design problems: in order to solve the problem in polynomial time, we first need to show that there exists a feasible solution in which the number of hyperedges is polynomial in the input size. 

The central theme of this work is to show that certain hypergraph network design problems admit solutions in which the number of hyperedges is polynomial in the number of vertices and moreover, can be solved in strongly polynomial time. 
Our work improves on the previous fastest pseudo-polynomial run-time for these problems. 
In addition, we develop strongly polynomial time algorithms that return near-uniform hypergraphs as solutions (i.e., every pair of hyperedges differ in size by at most one). 
As applications of our results, we derive the first strongly polynomial time algorithms for (i) degree-specified hypergraph connectivity augmentation using hyperedges, (ii)  degree-specified hypergraph node-to-area connectivity augmentation using hyperedges, and (iii) degree-constrained mixed-hypergraph connectivity augmentation using hyperedges.  
\end{abstract}

%% file: intro-2.tex
\section{Introduction}\label{sec:intro}
In the degree-specified graph connectivity augmentation using edges problem (\dsglcae), we are given an edge-weighted undirected graph $(G=(V, E_G), c_G: E_G\rightarrow \Z_+)$, a degree-requirement function $m: V\rightarrow \Z_{\ge 0}$, and a target connectivity function $r: \binom{V}{2}\rightarrow \Z_{\ge 0}$. The goal is to verify if there exists an edge-weighted undirected graph $(H=(V, E_H), w_H: E_H\rightarrow \Z_+)$ such that the degree of each vertex $u$ in $(H, w_H)$ is $m(u)$ and for every distinct pair of vertices $u, v\in V$, the edge connectivity between $u$ and $v$ in the union of the weighted graphs $(G, c_G)$ and $(H,w_H)$ is at least $r(u,v)$; moreover, the problem asks to construct such a graph $(H, w_H)$ if it exists. 
Watanabe and Nakamura \cite{WN87} introduced \dsglcae for the case of uniform requirement function (i.e., $r(u,v)=k$ for all distinct $u, v\in V$ for some $k\in \Z_+$) and showed that this case is solvable in polynomial time in unweighted graphs. Subsequently, Frank \cite{Fra92} gave a strongly polynomial-time algorithm for \dsglcae. 
Since then, designing fast algorithms as well as parallel algorithms for \dsglcae has been an active area of research \cite{CS89, Fra94, Ben99, Gab94, BHTP03, NGM97, BK00, LY13}. The last couple of years has seen exciting progress for the uniform requirement function culminating in a near-linear time algorithm \cite{CLP22-soda, CLP22-stoc, CHLP23}. 
In addition to making progress in the algorithmic status of the problem, these works have revealed fundamental structural properties of graph cuts which are of independent interest in graph theory. 
In this work, we consider generalizations of these connectivity augmentation problems to hypergraphs and design the first strongly polynomial-time algorithms for these generalizations. 

We emphasize that \dsglcae is a feasibility problem, i.e., the goal is to verify if there exists a feasible solution and if so, then find one. There is a closely related optimization variant: the input to the optimization version is a graph $(G=(V, E_G), c_G: E_G\rightarrow \Z_+)$ and a target connectivity function $r: \binom{V}{2}\rightarrow \Z_+$ and the goal is to find a graph $(H=(V, E_H), w_H: E_H\rightarrow \Z_+)$ with minimum total weight $\sum_{e\in E_H}w_H(e)$ such that for every pair of distinct vertices $u, v\in V$, the edge connectivity between $u$ and $v$ in the union of the weighted graphs $(G, c_G)$ and $(H, w_H)$ is at least $r(u, v)$. 
This optimization version is different from the NP-hard min-cost connectivity augmentation problems (like Steiner tree and tree/cactus/forest augmentation) whose approximability have been improved recently \cite{TZ22, BGA20, TZ23,GAT22, TZ22-soda}. 
All algorithms to solve the optimization version \cite{WN87, CS89, Fra92, Fra94, Gab94, BHTP03, NGM97, Ben99, BK00, LY13} reduce it to solving the degree-specified feasibility version, i.e., \dsglcae, so we focus only on the degree-specified feasibility variant and their generalization to hypergraphs throughout this work. All our results can be extended to an optimization variant, but we avoid stating them in the interests of brevity. 

\paragraph{Hypergraphs.} 
Edges are helpful to model relationships between pairs of entities. Hyperedges are helpful to model relationships between arbitrary number of entities. 
For this reason, hypergraphs are more accurate models for a rich variety of applications in bioinformatics, statistical physics, and machine learning (e.g., see \cite{VeldtBK22,Schlagetal23,LiM17, Ornes21, Feng_Heath_Jefferson_Joslyn_Kvinge_Mitchell_Praggastis_Eisfeld_Sims_Thackray_etal_2021,Young_Petri_Peixoto_2021, Klamt_Haus_Theis_2009,Dai_Gao_2023, Feng_You_Zhang_Ji_Gao_2019, Wang_Kleinberg_2024}). These applications have in turn, renewed interests in algorithms for hypergraph optimization problems \cite{KK15, GKP17, CXY19-j, CX18, CC22, CKN20, SY19, BST19, KKTY21-FOCS, KKTY21-STOC, Lee22, JLS22, Quanrud22, GLSTW22, BCW22-enumhkcut, BCW22-enumhkpartitioning, CC23, AGL23, FPZ23}. A hypergraph $G=(V, E)$ consists of a finite set $V$ of vertices and a set $E$ of hyperedges, where every hyperedge $e\in E$ is a subset of $V$. Equivalently, a hypergraph is a set system defined over a finite set. We will denote a hypergraph $G=(V, E)$ with hyperedge weights $w: E\rightarrow \Z_+$ by the tuple $(G, w)$. 
Throughout this work, we will be interested only in hypergraphs with positive integral weights and for algorithmic problems where the input/output is a hypergraph, we will require that the weights are represented in binary. 
If all hyperedges have size at most $2$, then the hyperedges are known as edges and we call such a hypergraph as a graph. 

We emphasize a subtle but important distinction between hypergraphs and graphs: the number of hyperedges in a hypergraph could be exponential in the number of vertices. This is in sharp contrast to graphs where the number of edges is at most the square of the number of vertices. Consequently, in hypergraph network design problems where the goal is to construct a hypergraph with certain properties,  
we have to be mindful of the number of hyperedges in the solution hypergraph (to be returned by the algorithm). 
This nuanced issue adds an extra challenging layer to hypergraph network design compared to graph network design problems -- e.g., membership in NP becomes non-trivial. 
Recent works in hypergraph algorithms literature have 
focused on the number of hyperedges in the context of cut/spectral sparsification of hypergraphs \cite{KK15, CX18, SY19, BST19, CKN20, KKTY21-FOCS, KKTY21-STOC, Lee22, JLS22, Quanrud22}. We will return to the membership in NP issue after we define the relevant problems of interest to this work. 

\paragraph{Notation.} Let $(G=(V, E), w: E\rightarrow \Z_+)$ be a hypergraph. 
For $X\subseteq V$, let $\delta_G(X)\coloneqq \{e\in E: e\cap X\neq \emptyset, e\setminus X\neq \emptyset\}$ and 
$B_G(X)\coloneqq \{e\in E: e\cap X\neq \emptyset\}$.  
We define the cut function $d_{(G, w)}: 2^V\rightarrow \Z_{\ge 0}$ by $d_{(G, w)}(X) \coloneqq \sum_{e\in \delta_G(X)}w(e)$ for every $X\subseteq V$ and the coverage function $b_{(G, w)}: 2^V\rightarrow \Z_{\ge 0}$ by $b_{(G,w)}(X)\coloneqq \sum_{e\in B_G(X)}w(e)$ for every $X\subseteq V$. 
For a vertex $v\in V$, we use $d_{(G,w)}(v)$ and $b_{(G,w)}(v)$ to denote $d_{(G,w)}(\{v\})$ and $b_{(G,w)}(\{v\})$ respectively. 
We define the \emph{degree} of a vertex $v$ to be $b_{(G, w)}(v)$ -- we note that the degree of a vertex is not necessarily equal to $d_{(G, w)}(v)$ since we could have $\{v\}$ itself as a hyperedge (i.e., a singleton hyperedge that contains only the vertex $v$). 
For distinct vertices $u, v\in V$, the connectivity between $u$ and $v$ in $(G,w)$ is $\lambda_{(G, w)}(u, v)\coloneqq \min\{d_{(G, w)}(X): u\in X\subseteq V\setminus \{v\}\}$ -- i.e., $\lambda_{(G, w)}(u, v)$ is the value of a minimum $\{u, v\}$-cut in the hypergraph. 
For two hypergraphs $(G=(V, E_G), c_G:E_G\rightarrow \Z_+)$ and $(H=(V, E_H), w_H: E_H\rightarrow \Z_+)$ on the same vertex set $V$, we define the hypergraph $(G+H=(V, E_{G+H}), c_G+w_H)$ as the hypergraph with vertex set $V$ and hyperedge set $E_{G+H}\coloneqq E_G\cup E_H$ with the weight of every hyperedge $e\in E_G\cap E_H$ being $c_G(e) + w_H(e)$, the weight of every hyperedge $e\in E_G\setminus E_H$ being $c_G(e)$,  and the weight of every hyperedge $e\in E_H\setminus E_G$ being $w_H(e)$.

\paragraph{Degree-Specified Hypergraph Connectivity Augmentation.} 
We define the degree-specified hypergraph connectivity augmentation using hyperedges problem below. 

\vspace{1mm}

\begin{center}
    \fbox{
        \parbox{0.8\textwidth}{
            \small
            \textbf{Degree-specified Hypergraph Connectivity Augmentation using Hyperedges} \\
            \textbf{(\dshlcah).} 
            
            \hangindent=12mm
            \hangafter=1
            Given: A hypergraph $(G=(V, E_G, c_G: E_G\rightarrow \Z_+)$, \\
            target connectivity function $r: \binom{V}{2}\rightarrow \Z_{\ge 0}$, and \\
            degree requirement function $m: V\rightarrow \Z_{\ge 0}$. 

            \hangindent=12mm
            \hangafter=1
            Goal: Verify if there exists a hypergraph $(H=(V, E_H), w_H: E_H\rightarrow \Z_+)$ such that\\ 
            $b_{(H, w_H)}(u)=m(u)$ for every $u\in V$ and\\ 
            $\lambda_{(G+H, c_G+w_H)}(u, v)\ge r(u,v)$ for every distinct $u, v\in V$,\\ 
            and if so, then find such a hypergraph. 
        }
    }
\end{center}
\vspace{1mm}

\noindent We remark on certain aspects of the problem to illustrate the challenges. 
\begin{remark} \label{remark:hypergraph-lca-using-exponential-hyperedges}
There is a fundamental distinction between \dshlcah and \dsglcae: it is clear that \dsglcae is in NP since YES instances admit a weighted \emph{graph} $(H, w_H)$ as a feasible solution which serves as a polynomial-time verifiable certificate of YES instances; in contrast, it is not immediately clear if \dshlcah is in NP. This is because, the number of hyperedges in the desired hypergraph $(H, w_H)$ could be exponential in the number of vertices, and consequently, exponential in the size of the input. 
We give an example to illustrate this issue: suppose that the input instance is given by the empty hypergraph $(G, c_G)$ on $n$ vertices, the target connectivity function $r$ is given by $r(u, v):=2^{n-1}-1$ for every pair of distinct vertices $u, v\in V$ and the degree requirement function $m: V\rightarrow \Z_{\ge 0}$ is given by $m(u):=2^{n-1}$ for every vertex $u\in V$. We note that the input specification needs only $n^{O(1)}$ bits. Consider the hypergraph $(H=(V, E_H), w_H)$ where $E_H:=\{e\subseteq V: |e|\ge 2\}$ with all hyperedge weights being one. The hypergraph $(H, w_H)$ is a feasible solution to the input instance but the number of hyperedges in this hypergraph is $2^n-n-1$ which is exponential in the number of vertices (and hence, the input size). Thus, in order to design a polynomial-time algorithm for \dshlcah, a necessary first step is to exhibit the existence of a feasible solution in which the number of hyperedges is polynomial in the input size, i.e., prove membership in NP. For the input instance mentioned here, there is indeed a feasible solution with one hyperedge: the hypergraph $(H'=(V, E_{H'}), w_{H'})$ containing only one hyperedge, namely $E_{H'}:=\{V\}$ with the weight of that hyperedge being $2^{n-1}-1$. 
\end{remark}

\begin{remark}
Given that membership in NP is non-trivial, it is tempting to constrain the target hypergraph $(H, w_H)$ to be a graph. This leads to the degree-specified hypergraph connectivity augmentation using edges problem (\dshlcae): here, we have the same input as \dshlcah; the goal is to verify if there exists a \emph{graph} $(H=(V, E_H), w_H: E\rightarrow \Z_+)$ with the same properties as above and if so, then find one. Clearly, \dshlcae is in NP since YES instances admit a weighted \emph{graph} $(H, w_H)$ as a feasible solution which serves as a polynomial-time verifiable certificate of YES instances. However, \dshlcae is NP-complete \cite{CJK10, MI04} (see Table \ref{table:results-for-hypergraph-connectivity-augmentation}). 
\end{remark}

Showing the existence of a feasible solution hypergraph with small number of hyperedges is a technical challenge in hypergraph network design problems. During first read, we encourage the reader to focus on this issue for all problems that we define and how it is addressed by our results and techniques. The algorithmic results that we present are consequences of our techniques to address this issue (using standard algorithmic tools in submodularity). 


We now discuss the status of \dshlcah. Szigeti \cite{Szi99} showed that \dshlcah can be solved in pseudo-polynomial time: in particular, if the target connectivity function $r: \binom{V}{2}\rightarrow \Z_{\ge 0}$ is given in unary, then the problem can be solved in polynomial time. Moreover, his result implies that if the input instance is feasible, then it admits a solution hypergraph $(H, w_H)$ such that the number of hyperedges in $H$ is at most $\max\{2^{|V|}, \max\{r(u, v): \{u, v\}\in \binom{V}{2}\}\}$. In this work, we show that feasible instances admit a solution with $O(|V|)$ hyperedges and give a strongly polynomial time algorithm.

\begin{restatable}{theorem}{thmStronglyPolyDSHLCAH}\label{thm:strongly-poly-for-dshlcah}
    There exists an algorithm to solve \dshlcah that runs in time $O(n^7(n + m)^2)$, where $n$ is the number of vertices and $m$ is the number of hypergedges in the input hypergraph. Moreover, if the instance is feasible, then the algorithm returns a solution hypergraph that contains $O(n)$ hyperedges. 
\end{restatable}



Next, we consider a constrained-variant of \dshlcah. A hypergraph is \emph{uniform} if all hyperedges have the same size; a hypergraph is \emph{near-uniform} if every pair of hyperedges differ in size by at most one. We note that every graph is a near-uniform hypergraph. 
Uniformity/near-uniformity is a natural constraint in network design applications involving hypergraphs -- we might be able to create only equal-sized hyperedges in certain applications. Requiring uniform (or near-uniform) hyperedges can also be viewed as a fairness inducing constraint in certain applications. Bern\'{a}th and Kir\'{a}ly \cite{Bernath-Kiraly} showed that if an instance of \dshlcah is feasible, then it admits a near-uniform hypergraph as a solution. Moreover, their proof technique implies a pseudo-polynomial time algorithm to return a solution hypergraph that is near-uniform for feasible instances. In this work, we show that feasible instances admit a solution with $O(|V|)$ near-uniform hyperedges and give a strongly polynomial time algorithm. 
\begin{restatable}{theorem}{thmStronglyPolyDSHLCAUH}\label{thm:strongly-poly-for-dshlcauh}
    There exists a strongly polynomial time algorithm to solve \dshlcah. Moreover, if the instance is feasible, then the algorithm returns a solution hypergraph that is near-uniform and contains $O(n)$ hyperedges, where $n$ is the number of vertices in the input hypergraph. 
\end{restatable}

Next, we consider a variant of hypergraph connectivity augmentation where the goal is to simultaneously augment two input hypergraphs to achieve certain target connectivities using the same degree-specified hypergraph. 

\vspace{1mm}
\begin{center}
    \fbox{
        \parbox{0.9\textwidth}{
            \small
            \textbf{Degree-specified Simultaneous Hypergraph Connectivity Augmentation Hyperedges}\\ \textbf{(\dsshlcah).}  
            
            \hangindent=12mm
            \hangafter=1
            Given: Hypergraphs $(G_i=(V, E_i), c_i: E_i\rightarrow \Z_+)$ for $i\in \{1, 2\}$, \\
                target connectivity functions $r_i: \binom{V}{2}\rightarrow \Z_{\ge 0}$ for $i\in \{1, 2\}$ such that \\
                \hspace*{6mm}$\max\{r_1(u,v)-\lambda_{(G_1,c_1)}(u,v): u, v\in V\}=\max\{r_2(u,v)-\lambda_{(G_2,c_2)}(u,v): u, v\in V\}$, and \\
                degree requirement function $m: V\rightarrow \Z_{\ge 0}$. 

            \hangindent=12mm
            \hangafter=1
            Goal: Verify if there exists a hypergraph $(H=(V, E_H), w_H: E_H\rightarrow \Z_+)$ such that\\ 
            $b_{(H, w_H)}(u)=m(u)$ for every $u\in V$ and\\ $\lambda_{(G_i+H, c_i+w_H)}(u, v)\ge r_i(u,v)$ for every distinct $u, v\in V$ and $i\in \{1, 2\}$,\\ 
            and if so, then find such a hypergraph.  
        }
    }
\end{center}
\vspace{1mm}

\noindent Bern\'{a}th and Kir\'{a}ly \cite{Bernath-Kiraly} proposed the above problem and showed that if the assumption $\max\{r_1(u,v)-\lambda_{(G_1,c_1)}(u,v): u, v\in V\}=\max\{r_2(u,v)-\lambda_{(G_2,c_2)}(u,v): u, v\in V\}$ does not hold, then the problem is NP-complete. So, we include this assumption in the definition of the problem statement. 
In the same work, they gave a pseudo-polynomial time algorithm for \dsshlcah; moreover, they showed that if an instance of \dsshlcah is feasible, then it admits a near-uniform hypergraph as a solution. 
In this work, we show that feasible instances of \dsshlcah admit a solution with $O(|V|^2)$ near-uniform hyperedges and give a strongly polynomial time algorithm. 

\begin{restatable}{theorem}{thmStronglyPolyDSSHLCAUH}\label{thm:strongly-poly-for-dsshlcauh}
    There exists a strongly polynomial time algorithm to solve \dsshlcah. Moreover, if the instance is feasible, then the algorithm returns a solution hypergraph that is near-uniform and contains $O(n^2)$ hyperedges, where $n$ is the number of vertices in the input hypergraph. 
\end{restatable}

Our algorithm for Theorems \ref{thm:strongly-poly-for-dshlcauh} and \ref{thm:strongly-poly-for-dsshlcauh} are LP-based. The associated LPs can be solved in strongly polynomial time. However, the LP-solving time is large, so we refrain from stating the run-times explicitly. In contrast, our algorithm for Theorem \ref{thm:strongly-poly-for-dshlcah} is combinatorial and hence, we are able to provide an explicit bound on the run-time. We refer the reader to Table \ref{table:results-for-hypergraph-connectivity-augmentation} for a list of graph/hypergraph connectivity augmentation problems using edges/hyperedges, previously known results, and our results. 

\begin{table}[ht]
\centering{
\begin{tabular}{|l|c|}
\hline
\textbf{Problem} & \textbf{Complexity Status}\\
\hline
\dsglcae & Strong Poly \cite{Fra92}\\
\hline
\dshlcae & NP-comp \cite{CJK10, MI04}\\
\hline
\multirow[c]{2}{*}{\dshlcah} & Psuedo Poly \cite{Szi99}\\
& $O(n^7(n+m)^2)$ time (Thm \ref{thm:strongly-poly-for-dshlcah})\\
\hline
\multirow[c]{2}{*}{\dshlcauh} & Pseudo Poly \cite{Bernath-Kiraly}\\
& Strong Poly (Thm \ref{thm:strongly-poly-for-dshlcauh})\\
\hline
\multirow[c]{2}{*}{\dsshlcah} & Pseudo Poly \cite{Bernath-Kiraly}\\
& Strong Poly (Thm \ref{thm:strongly-poly-for-dsshlcauh})\\
\hline
\multirow[c]{2}{*}{\dsshlcauh} & Pseudo Poly \cite{Bernath-Kiraly}\\
& Strong Poly (Thm \ref{thm:strongly-poly-for-dsshlcauh})\\
\hline
\end{tabular}
}
\caption{Complexity of Graph and Hypergraph Connectivity Augmentation Problems using Edges and Hyperedges. 
Problems having ``\textsc{Near-Uniform}'' in their title are similar to the corresponding problems without ``\textsc{Near-Uniform}'' in their title but have the additional requirement that the returned solution hypergraph be \emph{near-uniform}.
Here, $n$ and $m$ denote the number of vertices and hyperedges respectively in the input hypergraph. 
}
\label{table:results-for-hypergraph-connectivity-augmentation}
\end{table}


\subsection{Degree-specified skew-supermodular cover problems}
We focus on more general function cover problems that encompass several applications in connectivity augmentation (including all hypergraph connectivity augmentation problems mentioned above). 
Our main contribution is the first strongly polynomial time algorithm for the more general function cover problems. 
In Section \ref{sec:applications}, 
we will discuss several applications of our results for the general function cover problems and in particular, will derive Theorems \ref{thm:strongly-poly-for-dshlcah}, \ref{thm:strongly-poly-for-dshlcauh}, and \ref{thm:strongly-poly-for-dsshlcauh}. 
We recall certain definitions needed to describe the general problem. 


\begin{definition}
Let $V$ be a finite set, $(H=(V, E), w:E\rightarrow \Z_+)$ be a hypergraph, and $p: 2^V\rightarrow \Z$ be a set function.  
\begin{enumerate}
    \item The hypergraph $(H, w)$ \emph{weakly covers} the function $p$ if $b_{(H, w)}(X)\ge p(X)$ for every $X\subseteq V$. 
    \item The hypergraph $(H, w)$ \emph{strongly covers} the function $p$ if $d_{(H, w)}(X)\ge p(X)$ for every $X\subseteq V$. 
    \end{enumerate}
\end{definition}

We will be interested in the problem of finding a degree-specified hypergraph that strongly/weakly covers a given function $p$. 
In all our applications (including \dshlcah), the function $p$ of interest will be skew-supermodular and/or symmetric. 

\begin{definition}
    Let $p: 2^V\rightarrow Z$ be a set function. 
    We will denote the maximum function value of $p$ by $K_p$, i.e., $K_p\coloneqq \max\{p(X): X\subseteq V\}$. 
    The set function $p$ 
    \begin{enumerate}
\item is \emph{symmetric} if $p(X)=p(V-X)$ for every $X\subseteq V$, and 
    \item is \emph{skew-supermodular} if for every $X, Y\subseteq V$, at least one of the following inequalities hold: 
    \begin{enumerate}
\item $p(X) + p(Y) \leq p(X\cap Y) + p(X\cup Y)$. If this inequality holds, then we say that $p$ is \emph{locally supermodular} at $X, Y$.
    \item $p(X) + p(Y) \leq p(X - Y) + p(Y - X)$. If this inequality holds, then we say that $p$ is \emph{locally negamodular} at $X, Y$.
    \end{enumerate}
\end{enumerate}
\end{definition}

As mentioned above, we will be interested in the problem of finding a degree-specified hypergraph that strongly/weakly covers a given skew-supermodular function $p$. 
We will assume access to the skew-supermodular function $p$ via the following oracle. 

\begin{definition}
    Let $p: 2^V\rightarrow \Z$ be a set function.    $\functionMaximizationOracleStrongCover{p}\left(\left(G_0, c_0\right), S_0, T_0, y_0\right)$ takes as input a hypergraph $(G_0=(V, E_0), c_0: E_0\rightarrow \mathbb{Z}_{+})$, disjoint sets $S_0, T_0\subseteq V$, and a vector $y_0\in \R^V$; the oracle returns a tuple $(Z, p(Z))$, where $Z$ is an optimum solution to the following problem:
        \begin{align*}
        \max & \left\{p(Z)-d_{(G_0, c_0)}(Z)+y_0(Z): S_0\subseteq Z\subseteq V-T_0\right\}. \tag{\functionMaximizationOracleStrongCover{p}}\label{tag:maximization-oracle-strong-cover}
        \end{align*}
\end{definition}

We note that \functionMaximizationOracleStrongCover{p} is strictly stronger than the function evaluation oracle\footnote{For a function $p: 2^V\rightarrow \Z$, the function evaluation oracle takes a subset $X\subseteq V$ as input and returns $p(X)$.} -- in particular, it is impossible to maximize a skew-supermodular function using polynomial number of function evaluation queries \cite{HIMT22-posimodular-function-optimization}. However, we will see later that \functionMaximizationOracleStrongCover{p} can indeed be implemented in strongly polynomial time for the functions $p$ of interest to our applications. 
In our algorithmic results, we will ensure that the size of hypergraphs $(G_0, c_0)$ used as inputs to  \functionMaximizationOracleStrongCover{p} are polynomial in the input size (in particular, the number of hyperedges in these hypergraphs will be polynomial in the size of the ground set $V$). 
We now describe the general function cover problems that will be of interest to this work. 

\paragraph{Strong Cover Problems.} In all our applications, we will be interested in obtaining a degree-specified strong cover of a function. 

\vspace{1mm}

\fbox{
    \parbox{44em}{
        \textbf{Degree-specified symmetric skew-supermodular strong cover using hyperedges problem}\\
        \textbf{(\dssssch).}
        
        \hangindent=12mm
        \hangafter=1
        Given: A degree requirement function $m: V\rightarrow \Z_{\ge 0}$ and \\ 
        a \underline{symmeric} skew-supermodular function $p:2^V\rightarrow \Z$ via \functionMaximizationOracleStrongCover{p}.

        \hangindent=12mm
        \hangafter=1
        Goal: Verify if there exists a hypergraph $(H=(V, E), w: E\rightarrow \Z_+)$ such that\\ 
        $b_{(H, w)}(u)=m(u)$ for every $u\in V$, $(H, w)$ \underline{strongly} covers the function $p$, \\ 
        and if so, then find such a hypergraph. 
}
}
\vspace{1mm}

\noindent \dssssch was introduced by Bern\'{a}th and Kir\'{a}ly \cite{Bernath-Kiraly} as a generalization of \dshlcah (and various other applications that we will discuss in Section \ref{sec:applications}). They showed that it is impossible to solve \dssssch using polynomial number of queries to the function evaluation oracle. They suggested access to \functionMaximizationOracleStrongCover{p} and we work in the same function access model as Bern\'{a}th and Kir\'{a}ly. 
We note that it is not immediately clear if feasible instances of \dssssch admit solution hypergraphs with polynomial number of hyperedges (see Remark \ref{remark:hypergraph-lca-using-exponential-hyperedges}), so membership of \dssssch in NP is not obvious. 

Bern\'{a}th and Kir\'{a}ly \cite{Bernath-Kiraly} also introduced the following generalization of \dsshlcah. 

\vspace{1mm}

\fbox{
    \parbox{44em}{
        \textbf{Degree-specified simultaneous symmetric skew-supermodular strong cover using hyperedges problem} 
        \textbf{(\dsssssch).}
        
        \hangindent=12mm
        \hangafter=1
        Given: A degree requirement function $m: V\rightarrow \Z_{\ge 0}$ and \\ 
        \underline{symmetric} skew-supermodular functions $q, r:2^V\rightarrow \Z$ via \functionMaximizationOracleStrongCover{q} and \functionMaximizationOracleStrongCover{r}, where $K_q = K_r$.

        \hangindent=12mm
        \hangafter=1
        Goal: Verify if there exists a hypergraph $(H=(V, E), w: E\rightarrow \Z_+)$ such that\\ 
        $b_{(H, w)}(u)=m(u)$ for every $u\in V$ and $(H, w)$ \underline{strongly} covers both functions $q$ and $r$, 
        and if so, then find such a hypergraph. 
}
}

\vspace{1mm}
\noindent We note that \dsssssch without the assumption that $K_q=K_r$ is known to be NP-complete \cite{Bernath-Kiraly}, so we include this assumption in the problem definition. 

\paragraph{Weak Cover Problems.} 
Although our applications will be concerned with degree-specified \emph{strong} cover, our techniques will be concerned with degree-specified \emph{weak} cover problems. 
There is indeed a close relationship between strong cover and weak cover of symmetric skew-supermodular functions that we elaborate now. If a hypergraph $(H=(V, E), w: E\rightarrow \Z_+)$ strongly covers a function $p:2^V\rightarrow \Z$, then it also weakly covers the function $p$. However, the converse is false -- i.e., a weak cover is not necessarily a strong cover\footnote{For example, consider the function $p: 2^V\rightarrow \Z$ defined by $p(X)\coloneqq 1$ for every non-empty proper subset $X\subsetneq V$ and $p(\emptyset)\coloneqq p(V)\coloneqq 0$, and the hypergraph $(H=(V, E\coloneqq \{\{u\}: u\in V\}), w: E\rightarrow \{1\})$.}.  
Bern\'{a}th and Kir\'{a}ly \cite{Bernath-Kiraly} showed that 
if a hypergraph $(H=(V, E), w: E\rightarrow \Z_+)$ with $\sum_{e\in E}w(e)=K_p$ weakly covers a symmetric skew-supermodular function $p$, then $(H,w)$ strongly covers $p$. 
Moreover, Szigeti \cite{Szi99} showed that if \dssssch is feasible, then it admits a solution hypergraph $(H=(V, E), w:E\rightarrow \Z_+)$ such that $\sum_{e\in E}w(e)=K_p$. 
These two facts together imply that in order to solve \dssssch, it suffices to 
solve a degree-specified skew-supermodular \emph{weak} cover using hyperedges problem 
for the same function $p$. We define the latter problem now. 

\vspace{2mm}
\fbox{
    \parbox{44em}{
        \textbf{Degree-specified skew-supermodular weak cover using hyperedges problem} \\
        \textbf{(\dsswch).}
        
        \hangindent=12mm
        \hangafter=1
        Given: A degree requirement function $m: V\rightarrow \Z_{\ge 0}$ and \\ 
        a skew-supermodular function $p:2^V\rightarrow \Z$ via \functionMaximizationOracleStrongCover{p}.

        \hangindent=12mm
        \hangafter=1
        Goal: Verify if there exists a hypergraph $(H=(V, E), w: E\rightarrow \Z_+)$ such that\\ 
        $b_{(H, w)}(u)=m(u)$ for every $u\in V$, $(H, w)$ \underline{weakly} covers the function $p$, and $\sum_{e\in E}w(e)=K_p$,\\ 
        and if so, then find such a hypergraph. 
}
}
\vspace{2mm}

We clarify the significance of the requirement $\sum_{e\in E}w(e)=K_p$ 
in the above problem (we note that this requirement is not present in \dssssch). In \dsswch, if we drop the requirement that $\sum_{e\in E}w(e)=K_p$ from the problem definition, then the resulting problem admits a feasible solution if and only if a trivial hypergraph is 
feasible\footnote{Consider the hypergraph $(H=(V, E), w: E\rightarrow \Z_+)$, where $E:=\{\{u\}: u\in V,\ m(u)\ge 1\}$ with $w(\{u\}):=m(u)$ for every $\{u\}\in E$.}. Imposing this requirement makes the problem non-trivial. More importantly, imposing this requirement ensures that we have a reduction from \dssssch to \dsswch. 

Next, in order to address \dsssssch, Bern\'{a}th and Kir\'{a}ly \cite{Bernath-Kiraly} showed the following two results: (1) If a hypergraph $(H=(V, E), w: E\rightarrow \Z_+)$ is a weak cover of two symmetric skew-supermodular functions $q, r$ with $\sum_{e\in E}w(e)=K_p$, then $(H, w)$ is also a strong cover of $p$. (2) If \dsssssch is feasible, then it admits a solution hypergraph $(H,w)$ such that $\sum_{e\in E}w(e)=K_p$. These two facts together imply that in order to solve \dsssssch, it suffices to solve a degree-specified skew-supermodular \emph{weak} cover using hyperedges problem for the same functions $q$ and $r$. We define the latter problem now. 


\vspace{2mm}
\fbox{
    \parbox{44em}{
        \textbf{Degree-specified simultaneous skew-supermodular weak cover using hyperedges problem} \\
        \textbf{(\dssswch).}
        
        \hangindent=12mm
        \hangafter=1
        Given: A degree requirement function $m: V\rightarrow \Z_{\ge 0}$ and \\ 
        skew-supermodular functions $q, r:2^V\rightarrow \Z$ via \functionMaximizationOracleStrongCover{q} and \functionMaximizationOracleStrongCover{r}.

        \hangindent=12mm
        \hangafter=1
        Goal: Verify if there exists a hypergraph $(H=(V, E), w: E\rightarrow \Z_+)$ such that\\ 
        $b_{(H, w)}(u)=m(u)$ for every $u\in V$, $(H, w)$ weakly covers both functions $q$ and $r$, and $\sum_{e\in E}w(e)=\max\{K_q, K_r\}$, 
        and if so, then find such a hypergraph. 
}
}
\vspace{2mm}



\subsection{Results}\label{sec:results}
As mentioned in the previous section, reductions from \dssssch to \dsswch and also from \dsssssch to \dssswch were already known. Hence, from a technical perspective, the central problems of interest to this work will be the weak cover problems, namely \dsswch and \dssswch. 
\begin{remark}\label{remark:dsswch-exponential-example}
It is not immediately clear if feasible instances of these two weak cover problems admit solution hypergraphs in which the number of hyperedges is polynomial in the number of vertices. We illustrate this issue for \dsswch with an example (that is a modification of the example in Remark \ref{remark:hypergraph-lca-using-exponential-hyperedges}): let $n:=|V|$ and consider the degree requirement function $m: V\rightarrow \Z_{\ge 0}$ given by $m(u):=2^{n-1}-1$ for every $u\in V$ and the function $p:2^V\rightarrow \Z$ given by $p(X):=2^{n-1}-1$ for every non-empty proper subset $X\subsetneq V$, $p(V):=2^n-n-1$, and $p(\emptyset):=0$. We note that this function $p$ is skew-supermodular. Consider the hypergraph $(H=(V, E_H), w_H)$, where $E_H:=\{e \subseteq V: |e|\ge 2\}$ and all hyperedge weights are one. The hypergraph $(H, w_H)$ is a feasible solution to the input instance but the number of hyperedges in this hypergraph is $2^n-n-1$ which is exponential in the number of vertices. Thus, a necessary step in designing a polynomial-time algorithm for \dsswch is to show that feasible instances admit a solution hypergraph in which the number of hyperedges is polynomial in the input size. For the instance $(m,p)$ mentioned above, the following hypergraph is feasible and has only $3$ hyperedges: pick an arbitrary vertex $u\in V$ and consider the hypergraph $(H'=(V, E'), w':E'\rightarrow \Z_+)$, where the set of hyperedges is $E':=\{\{u\}, V-\{u\}, V\}$ and their weights are given by $w'(\{u\})=2^{n-1}-n = w'(V-\{u\})$ and $w'(V)=n-1$. 
\end{remark}

Szigeti \cite{Szi99} gave a complete characterization for the existence of a feasible solution to \dsswch. His proof implies that if a given instance is feasible, then it admits a solution hypergraph $(H=(V, E), w:E\rightarrow \Z_+)$ in which the number of hyperedges is $K_p$; we note that $K_p$ need not be polynomial in $|V|$. His proof also leads to a pseudo-polynomial time algorithm to solve \dsswch (the algorithm is only pseudo-polynomial time and not polynomial time since the number of hyperedges in the returned hypergraph could be $K_p$ and hence, the run-time depends on $K_p$). In this work, we show that feasible instances admit a solution with $O(|V|)$ hyperedges and give a strongly polynomial-time algorithm to solve \dsswch. 

\begin{theorem}\label{thm:strongly-poly-for-dsswch}
    There exists an algorithm to solve \dsswch that runs in time $O(|V|^5)$ using $O(|V|^4)$ queries to \functionMaximizationOracleStrongCover{p}, where $V$ is the ground set of the input instance. Moreover, if the instance is feasible, then the algorithm returns a solution hypergraph that contains $O(|V|)$ hyperedges. 
    For each query to \functionMaximizationOracleStrongCover{p} made by the algorithm, the hypergraph $(G_0, c_0)$ used as input to the query has $O(|V|)$ vertices and $O(|V|)$ hyperedges. 
\end{theorem}


Bern\'{a}th and Kir\'{a}ly \cite{Bernath-Kiraly} strengthened Szigeti's result via an LP-based approach. They showed that if an instance of \dsswch is feasible, then it admits a solution hypergraph that is near-uniform. Their approach is also algorithmic, but the run-time of their algorithm is only pseudo-polynomial (again, owing to the dependence on $K_p$ which may not be polynomial in $|V|$). In this work, we give a strongly polynomial-time algorithm to solve \dsswch and return a near-uniform solution hypergraph for feasible instances. 

\begin{theorem}\label{thm:strongly-poly-for-dsswcuh}
    There exists an algorithm to solve \dsswch that runs in $\poly(|V|)$ time using $\poly(|V|)$ queries to \functionMaximizationOracleStrongCover{p}, where $V$ is the ground set of the input instance.  
    Moreover, if the instance is feasible, then the algorithm returns a solution hypergraph that is near-uniform and contains $O(|V|)$ hyperedges. 
    For each query to \functionMaximizationOracleStrongCover{p} made by the algorithm, the hypergraph $(G_0, c_0)$ used as input to the query has $O(|V|)$ vertices and $O(|V|)$ hyperedges. 
\end{theorem}

Bern\'{a}th and Kir\'{a}ly's \cite{Bernath-Kiraly} LP-based approach also helped in addressing \dssswch. They gave a complete characterization for the existence of a feasible solution to \dssswch. Moreover, they showed that if an instance of \dssswch is feasible, then it admits a solution hypergraph that is near-uniform. Their proof is algorithmic, but the run-time of their algorithm is only pseudo-polynomial (again, owing to the dependence on $K_p$ which may not be polynomial in $|V|$). In this work, we give a strongly polynomial-time algorithm to solve \dssswch and return a near-uniform solution hypergraph for feasible instances. 

\begin{theorem}\label{thm:strongly-poly-for-dssswcuh}
    There exists an algorithm to solve \dssswch that runs in $\poly(|V|)$ time using $\poly(|V|)$ queries to \functionMaximizationOracleStrongCover{q} and \functionMaximizationOracleStrongCover{r}, where $V$ is the ground set of the input instance. Moreover, if the instance is feasible, then the algorithm returns a solution hypergraph that is near-uniform and contains $O(|V|^2)$ hyperedges. 
    For each query to \functionMaximizationOracleStrongCover{q} and \functionMaximizationOracleStrongCover{r} made by the algorithm, the hypergraph $(G_0, c_0)$ used as input to the query has $O(|V|)$ vertices and $O(|V|^2)$ hyperedges. 
\end{theorem}

Our algorithm for Theorem \ref{thm:strongly-poly-for-dsswch} is combinatorial and hence, we are able to provide an explicit bound on the run-time. 
In contrast, our algorithms for Theorems \ref{thm:strongly-poly-for-dsswcuh} and \ref{thm:strongly-poly-for-dssswcuh} are LP-based. These LPs are solvable in strongly polynomial time, but their run-time is large, so we refrain from stating the run-times explicitly. 
We refer to Table \ref{table:results-for-function-cover-problems} for a list of degree-specified skew-supermodular cover using hyperedges problems, previously known results, and our results. The results for strong cover problems follow from our results for weak cover problems. We will later see that our results for strong cover problems imply 
Theorems \ref{thm:strongly-poly-for-dshlcah}, \ref{thm:strongly-poly-for-dshlcauh}, and \ref{thm:strongly-poly-for-dsshlcauh} 
respectively (in Section \ref{sec:applications}).

\begin{table}[ht]
\centering{
\begin{tabular}{|l|c|}
\hline
\textbf{Problem} & \textbf{Complexity Status}\\
\hline
\multirow[c]{2}{*}{\dsswch} & Pseudo Poly \cite{Szi99}\\
& Strong Poly (Thm \ref{thm:strongly-poly-for-dsswch})\\
\hline
\multirow[c]{2}{*}{\dsswcuh} & Pseudo Poly \cite{Bernath-Kiraly}\\
& Strong Poly (Thm \ref{thm:strongly-poly-for-dsswcuh})\\
\hline
\multirow[c]{2}{*}{\dssswcuh} & Pseudo Poly \cite{Bernath-Kiraly}\\
& Strong Poly (Thm \ref{thm:strongly-poly-for-dssswcuh})\\
\hline
\multirow[c]{2}{*}{\dssssch} & Pseudo Poly \cite{Szi99}\\
& Strong Poly (Coro \ref{coro:strong-poly-SzigetiStrongCoverSymSkewSupmod})\\
\hline
\multirow[c]{2}{*}{\dsssscuh} & Pseudo Poly \cite{Bernath-Kiraly}\\
& Strong Poly (Coro \ref{coro:strong-poly-StrongCoverSymSkewSupmod-near-uniform})\\
\hline
\multirow[c]{2}{*}{\dssssscuh} & Pseudo Poly \cite{Bernath-Kiraly}\\
& Strong Poly (Coro \ref{coro:strong-poly-simul-StrongCoverSymSkewSupmod-near-uniform})\\
\hline
\end{tabular}
}
\caption{Complexity of degree-specified skew-supermodular cover using hyperedges problems. Problems having ``\textsc{Near-Uniform}'' in their title are similar to the corresponding problems without ``\textsc{Near-Uniform}'' in their title but have the additional requirement that the returned solution hypergraph be \emph{near-uniform}.
}
\label{table:results-for-function-cover-problems}
\end{table}

\subsection{Techniques: Structural and Algorithmic Results}
In this section, we discuss our techniques underlying the proof of Theorems \ref{thm:strongly-poly-for-dsswch}, \ref{thm:strongly-poly-for-dsswcuh}, and \ref{thm:strongly-poly-for-dssswcuh}. For a function $m: V\rightarrow \R$, we denote $m(X)\coloneqq \sum_{u\in X}m(u)$. 

\paragraph{\dsswch.} 
Szigeti \cite{Szi99} showed that an instance $(m: V\rightarrow \Z_{\ge 0}, p: 2^V\rightarrow \Z)$ of \dsswch is feasible if and only if $m(X)\ge p(X)$ for every  $X\subseteq V$ and $m(u)\le K_p$ for every $u\in V$. 
We note that this characterization immediately implies that feasibility of a given instance of \dsswch can be verified using two calls to \functionMaximizationOracleStrongCover{p}. 
In the algorithmic problem, we are interested in finding a feasible solution (i.e., a hypergraph $(H, w)$) under the hypothesis that $m(X)\ge p(X)$ for every  $X\subseteq V$ and $m(u)\le K_p$ for every $u\in V$. 
We show that every feasible instance admits a feasible solution $(H, w)$ such that the number of hyperedges in $H$ is linear in the number of vertices and moreover, such a solution can be found in strongly polynomial time. 

\begin{restatable}{theorem}{thmSzigetiWeakCover}\label{thm:SzigetiWeakCover:main}
Let $p:2^V\rightarrow\Z$ be a skew-supermodular function and $m:V\rightarrow\Z_{\ge 0}$ be a non-negative function such that: 
\begin{enumerate}[label=$(\alph*)$, ref=(\alph*)]
    \item \label{thm:SzigetiWeakCover:main:(a)}$m(X) \ge p(X)$ for every  $X \subseteq V$ and
    \item \label{thm:SzigetiWeakCover:main:(b)}$m(u) \leq K_p$ for every  $u \in V$. 
\end{enumerate}
Then, there exists a hypergraph $\left(H = \left(V, E\right), w:E\rightarrow\Z_+\right)$ satisfying the following four properties:
\begin{enumerate}[label=$(\arabic*)$, ref=(\arabic*)]
    \item \label{thm:SzigetiWeakCover:main:(1)} $b_{(H, w)}(X) \geq p(X)$ for every  $X\subseteq V$, 
    \item \label{thm:SzigetiWeakCover:main:(2)} $b_{(H, w)}(u) = m(u)$ for every  $u \in V$, 
    \item \label{thm:SzigetiWeakCover:main:(3)} $\sum_{e\in E}w(e) = K_p$, and 
    \item \label{thm:SzigetiWeakCover:main:(4)} $|E| = O(|V|)$. 
\end{enumerate}
Furthermore, given a function $m: V\rightarrow \Z_{\ge 0}$ and access to 
\functionMaximizationOracleStrongCover{p} of a skew-supermodular function $p: 2^V\rightarrow \Z$ where $m$ and $p$ satisfy conditions (a) and (b) above, there exists an algorithm that runs in time $O(|V|^5)$ using $O(|V|^4)$ queries to \functionMaximizationOracleStrongCover{p} and returns a hypergraph satisfying properties \ref{thm:SzigetiWeakCover:main:(1)}-\ref{thm:SzigetiWeakCover:main:(4)} above. 
The run-time  includes the time to construct the hypergraphs that are
used as inputs to \functionMaximizationOracleStrongCover{p}. 
Moreover, for each query to \functionMaximizationOracleStrongCover{p}, the hypergraph $(G_0, c_0)$ used as input to the query has $O(|V|)$ vertices and $O(|V|)$ hyperedges. 
\end{restatable}

\paragraph{\dsswch with near-uniform hyperedges.}
Bern\'{a}th and Kir\'{a}ly \cite{Bernath-Kiraly} generalized Szigeti's result with an LP-based approach. They showed that every feasible instance of \dsswch admits a near-uniform hypergraph as a feasible solution. In the algorithmic problem, we are interested in finding a feasible solution under the hypothesis that $m(X)\ge p(X)$ for every $X\subseteq V$ and $m(u)\le K_p$ for every $u\in V$. We show that every feasible instance admits a feasible near-uniform hypergraph $(H, w)$ as a solution such that the number of hyperedges in $H$ is linear in the number of vertices and moreover, such a solution can be found in strongly polynomial time. 
\begin{restatable}{theorem}{thmWeakCoverViaUniformHypergraph}\label{thm:WeakCoverViaUniformHypergraph:main}
Let $p:2^V\rightarrow\Z$ be a skew-supermodular function and $m:V\rightarrow\Z_{\ge 0}$ be a non-negative function satisfying the following two conditions: 
\begin{enumerate}[label=$(\alph*)$, ref=(\alph*)]
    \item \label{thm:WeakCoverViaUniformHypergraph:main:(a)}$m(X) \ge p(X)$ for every  $X \subseteq V$ and
    \item \label{thm:WeakCoverViaUniformHypergraph:main:(b)}$m(u) \leq K_p$ for every  $u \in V$. 
\end{enumerate}
Then, there exists a hypergraph $\left(H = \left(V, E\right), w:E\rightarrow\Z_+\right)$ satisfying the following five properties:
\begin{enumerate}[label=$(\arabic*)$, ref=(\arabic*)]
    \item \label{thm:WeakCoverViaUniformHypergraph:main:(1)} $b_{(H, w)}(X) \geq p(X)$ for every  $X\subseteq V$, 
    \item \label{thm:WeakCoverViaUniformHypergraph:main:(2)} $b_{(H, w)}(u) = m(u)$ for every  $u \in V$, 
    \item \label{thm:WeakCoverViaUniformHypergraph:main:(3)} $\sum_{e\in E}w(e) = K_p$,
    \item \label{thm:WeakCoverViaUniformHypergraph:main:(4)} if $K_p > 0$, then 
    $|e|\in \{\lfloor m(V)/K_p \rfloor, \lceil m(V)/K_p \rceil\}$ 
    for every  $e \in E$, and 
    \item \label{thm:WeakCoverViaUniformHypergraph:main:(5)} $|E| = O(|V|)$. 
\end{enumerate}
Furthermore, given a function $m: V\rightarrow \Z_{\ge 0}$ and access to 
\functionMaximizationOracleStrongCover{p} 
of a skew-supermodular function $p: 2^V\rightarrow \Z$ where $m$ and $p$ satisfy conditions (a) and (b), there exists an algorithm that runs in time $\poly(|V|)$ using $\poly(|V|)$ queries to 
\functionMaximizationOracleStrongCover{p} 
to return a hypergraph satisfying the above five properties. 
The run-time  includes the time to construct the hypergraphs that are
used as inputs to \functionMaximizationOracleStrongCover{p}. 
Moreover, for each query to \functionMaximizationOracleStrongCover{p}, the hypergraph $(G_0, c_0)$ used as input to the query has $O(|V|)$ vertices and $O(|V|)$ hyperedges. 
\end{restatable}


\paragraph{\dssswch.}
Bern\'{a}th and Kir\'{a}ly \cite{Bernath-Kiraly} showed that an instance $(m: V\rightarrow \Z_{\ge 0}, q: 2^V\rightarrow \Z, r: 2^V\rightarrow \Z)$ of \dssswch is feasible if and only if $m(Z)\ge p(Z)$ for all $Z\subseteq V$ and $m(u)\le K_p$ for all $u\in V$, where $p:2^V\rightarrow \Z$ is the function defined by $p(Z):=\max\{q(Z), r(Z)\}$ for all $Z\subseteq V$. Moreover, they also showed that every feasible instance admits a near-uniform hypergraph as a feasible solution. 
We show that every feasible instance $(m: V\rightarrow \Z_{\ge 0}, q: 2^V\rightarrow \Z, r: 2^V\rightarrow \Z)$ admits a near-uniform hypergraph $(H, w)$ as a feasible solution such that the number of hyperedges in $H$ is quadratic in the number of vertices and moreover, such a solution can be found in strongly polynomial time. 


\begin{restatable}{theorem}{thmWeakCoverTwoFunctionsViaUniformHypergraph}\label{thm:WeakCoverTwoFunctionsViaUniformHypergraph:main}
Let $q, r : 2^V \rightarrow \Z$ be two skew-supermodular functions 
and $p:2^V\rightarrow \Z$ be the function defined as $p(X) \coloneqq \max \left\{q(X), r(X)\right\}$ for every $X\subseteq V$.  
Furthermore, let $m:V\rightarrow\Z_{\ge 0}$ be a non-negative function satisfying the following two conditions: 
\begin{enumerate}[label=$(\alph*)$, ref=(\alph*)]
    \item \label{thm:WeakCoverTwoFunctionsViaUniformHypergraph:main:(a)}$m(X) \ge p(X)$ for every $X \subseteq V$ and 
    \item \label{thm:WeakCoverTwoFunctionsViaUniformHypergraph:main:(b)}$m(u) \leq K_p$ for every $u \in V$. 
\end{enumerate}
Then, there exists a hypergraph $\left(H = \left(V, E\right), w:E\rightarrow\Z_+\right)$ satisfying the following five properties:
\begin{enumerate}[label=$(\arabic*)$, ref=(\arabic*)]
    \item \label{thm:WeakCoverTwoFunctionsViaUniformHypergraph:main:(1)} $b_{(H, w)}(X) \geq p(X)$ for every $X\subseteq V$,
    \item \label{thm:WeakCoverTwoFunctionsViaUniformHypergraph:main:(2)}  $b_{(H, w)}(u) = m(u)$ for every $u \in V$, 
    \item \label{thm:WeakCoverTwoFunctionsViaUniformHypergraph:main:(3)} $\sum_{e\in E}w(e) = K_p$,
    \item \label{thm:WeakCoverTwoFunctionsViaUniformHypergraph:main:(4)} if $K_p>0$, then 
    $|e|\in \{\lfloor m(V)/K_p \rfloor, \lceil m(V)/K_p \rceil\}$ 
    for every $e \in E$, and 
    \item \label{thm:WeakCoverTwoFunctionsViaUniformHypergraph:main:(5)} $|E| = O(|V|^2)$.
\end{enumerate}
Furthermore, given a function $m: V\rightarrow \Z_{\ge 0}$ and access to \functionMaximizationOracleStrongCover{q} and \functionMaximizationOracleStrongCover{r} of skew-supermodular functions $q, r: 2^V\rightarrow \Z$ where $m$, $q$, and $r$ satisfy conditions (a) and (b), there exists an algorithm that runs in time $\poly(|V|)$ using $\poly(|V|)$ queries to \functionMaximizationOracleStrongCover{q} and \functionMaximizationOracleStrongCover{r} 
to return a hypergraph satisfying the above five properties. 
The run-time  includes the time to construct the hypergraphs that are
used as inputs to \functionMaximizationOracleStrongCover{q} \functionMaximizationOracleStrongCover{r}. 
Moreover, for each query to \functionMaximizationOracleStrongCover{q} and \functionMaximizationOracleStrongCover{r}, the hypergraph $(G_0, c_0)$ used as input to the query has $O(|V|)$ vertices and $O(|V|^2)$ hyperedges. 
\end{restatable}


We note that conditions (a) and (b) and properties (1)-(4) in Theorems \ref{thm:WeakCoverViaUniformHypergraph:main} and \ref{thm:WeakCoverTwoFunctionsViaUniformHypergraph:main} are identical, but the hypothesis about the function $p$ is different between the two theorems. Moreover, property (5) also differs between the two theorems. 

Theorem \ref{thm:SzigetiWeakCover:main} is a strengthening of Szigeti's result \cite{Szi99} in two ways. Szigeti showed that if $m$ and $p$ are two functions satisfying conditions (a) and (b) of the theorem, then there exists a hypergraph satisfying properties (1)-(3). Our Theorem \ref{thm:SzigetiWeakCover:main} strengthens his result by additionally showing property (4) (i.e., the number of hyperedges is polynomial in the number of vertices) and designing a strongly polynomial time algorithm that returns a hypergraph satisfying the four properties.
Theorems \ref{thm:WeakCoverViaUniformHypergraph:main} and \ref{thm:WeakCoverTwoFunctionsViaUniformHypergraph:main} are strengthenings of Bern\'ath and Kir\'aly's results \cite{Bernath-Kiraly} in two ways. Bern\'ath and Kir\'aly showed that if conditions (a) and (b) of the theorem hold, then there exists a hypergraph satisfying properties (1)-(4). Our Theorems \ref{thm:WeakCoverViaUniformHypergraph:main} and \ref{thm:WeakCoverTwoFunctionsViaUniformHypergraph:main} strengthen their result by additionally showing property (5) (i.e., the number of hyperedges is polynomial in the number of vertices) and designing a strongly polynomial time algorithm that returns a hypergraph satisfying the five properties. 
Although Szigeti's result \cite{Szi99} and Bern\'{a}th and Kir\'{a}ly's results \cite{Bernath-Kiraly} are existential, their proof of their result is algorithmic. However, their proof is based on induction on $K_p$ and consequently, the upper bound on the number of hyperedges in the returned solution as well as their run-time is $O(K_p)$. Thus, the best run-time known for their algorithms is only pseudo-polynomial. Our main contribution is analyzing an adaptation of their algorithm to bound the number of hyperedges and the run-time. 

Our algorithm for Theorem \ref{thm:SzigetiWeakCover:main} is combinatorial while the algorithms for Theorems \ref{thm:WeakCoverViaUniformHypergraph:main} and \ref{thm:WeakCoverTwoFunctionsViaUniformHypergraph:main} are LP-based. We show that the underlying LPs are solvable in strongly polynomial-time in the technical sections. 
Theorems \ref{thm:SzigetiWeakCover:main}, \ref{thm:WeakCoverViaUniformHypergraph:main}, and \ref{thm:WeakCoverTwoFunctionsViaUniformHypergraph:main} imply Theorems \ref{thm:strongly-poly-for-dsswch}, \ref{thm:strongly-poly-for-dsswcuh}, and \ref{thm:strongly-poly-for-dssswcuh} respectively.

\subsubsection{Proof Technique for Theorem \ref{thm:SzigetiWeakCover:main}}
In this section, we describe our proof technique for \Cref{thm:SzigetiWeakCover:main}. The algorithmic result in \Cref{thm:SzigetiWeakCover:main} follows from our techniques for the existential result using known tools for submodular functions. So, we focus on describing our proof technique for the existial result here. 
Let $p:2^V\rightarrow\Z$ be a skew-supermodular function and $m: V\rightarrow \Z_{\ge 0}$ be a non-negative function such that $m(X)\ge p(X)$ for every $X\subseteq V$ and $m(u)\le K_p$ for every $u\in V$. Our goal is to show that there exists a hypergraph $(H=(V, E), w: E\rightarrow \Z_+)$ such that
    (1) $(H, w)$ weakly covers the function $p$, 
    (2) $b_{(H, w)}=m(u)$ for every $u\in V$, 
    (3) $\sum_{e\in E}w(e) = K_p$, and 
    (4) $|E|=O(|V|)$.
Our proof of the existential result builds on the techniques of Szigeti \cite{Szi99} who proved the existence of a hypergraph $(H=(V, E), w: E\rightarrow \Z_+)$ satisfying properties (1)-(3), so we briefly recall his techniques. 
His proof proceeds by picking a minimal counterexample and arriving at a contradiction. Consequently, the algorithm implicit in the proof is naturally recursive. 
We present the algorithmic version of the proof since it will be useful for our purposes. 

For the purposes of the algorithmic proof, we assume that $m$ is a positive-valued function. We show that this assumption is without loss of generality since an arbitrary instance can be reduced to such an instance (for details, see the proof overview paragraph immediately following the restatement of \Cref{thm:SzigetiWeakCover:main} in \Cref{sec:SzigetiAlgorithm}).
The main insight underlying the proof is the following characterization of hyperedges in a feasible hypergraph. 
\begin{proposition}
\label{prop:Szigeti:hyperedge-characterization}
        Let $p:2^V\rightarrow\Z$ be a skew-supermodular function and $m:V\rightarrow\Z_+$ be a positive function 
        such that $m(X) \geq p(X)$ for every $X\subseteq V$ and $m(u)\leq K_p$ for every $u \in V$. Let $A \subseteq V$. Then, there exists a hypergraph $(H=(V, E), w: E\rightarrow \Z_+)$ satisfying properties (1)-(3) such that $A \in E$ if and only if $A$ satisfies the following:
        \begin{enumerate}[label=(\roman*)]
            \item $A$ is a transversal for the family of $p$-maximizers,
            \item $A$ contains the set $\{u \in V : m(u) = K_p\}$, 
            \item $m(v) \geq 1$ for each $v \in A$, and 
            \item $m(X) - |A\cap X| \geq p(X) - 1$ for every $X \subseteq V$.
        \end{enumerate}
    \end{proposition}

\paragraph{Algorithm of \cite{Szi99}.} \Cref{prop:Szigeti:hyperedge-characterization} leads to the following natural recursive strategy to compute a feasible hypergraph. If $K_p = 0$, then the algorithm is in its base case and returns the empty hypergraph (with no vertices) -- here, we note that $0 < m(u) \leq K_p = 0$ for every $u \in V$, and consequently, $V=\emptyset$ and thus, the empty hypergraph satisfies properties (1)-(3) as desired. Alternatively, if $K_p > 0$, then the algorithm recurses on appropriately revised versions of the input functions $p$ and $m$.
In particular, the algorithm picks an arbitrary minimal transversal $T$ for the family of $p$-maximizers and computes the set $A := T \cup \{u \in V : m(u) = K_p\}$. 
It can be shown that this set $A$ satisfies properties (i)-(iv) of \Cref{prop:Szigeti:hyperedge-characterization}; consequently, there exists a feasible hypergraph containing the hyperedge $A$. 
In order to revise the input functions, the algorithm defines $(H_0, w_0)$ to be the hypergraph on vertex set $V$ consisting of the single hyperedge $A$ with weight $w_0(A) = 1$ and constructs the set $\zeros := \{u \in A : m(u) = 1\}$. 
Next, the algorithm defines revised functions $m'':V-\zeros\rightarrow\Z$  and $p'':2^{V - \zeros}\rightarrow\Z$ as 
$m''(u):=m(u)-1$ if $u\in A-\zeros$ and $m''(u):=m(u)$ if $u\in V-A-\zeros$ and 
$p''(X):=\max\{p(X\cup R)-b_{(H_0,w_0)}(X\cup R): R\subseteq \zeros\}$ for every $X\subseteq V-\zeros$. 
The algorithm recurses on the revised input functions $m''$ and $p''$
to obtain a hypergraph $(H'', w'')$. Finally, 
the algorithm obtains the hypergraph $(G, c)$ by adding vertices $\zeros$ to $(H'', w'')$, and returns the hypergraph $(G+H_0, c+w_0)$. 
It can be shown that $p''$ is a skew-supermodular function and $m''$ is a positive function such that $m''(X) \geq p''(X)$ for all $X\subseteq V- \zeros$ and $m''(u) \leq K_{p''}$ for every $u \in V - \zeros$. We note that $K_{p''} = K_p - 1$ by the definition of the function $p''$ and the choice of set $A$ being a transversal for the family of $p$-maximizers. 
Consequently, by induction on $K_p$ and Proposition 1, the algorithm can be shown to terminate in $K_p$ recursive calls and return a hypergraph $(H, w)$ satisfying properties \ref{thm:SzigetiWeakCover:main:(1)}-\ref{thm:SzigetiWeakCover:main:(3)}. 
Furthermore, we observe that the number of distinct hyperedges added by the algorithm is at most the number of recursive calls since each recursive call adds at most one new hyperedge to $(H, w)$. Thus, the number of distinct hyperedges in $(H, w)$ is also at most $K_p$. Consequently, in order to reduce the number of distinct hyperedges, it suffices to reduce the recursion depth of Szigeti's algorithm.
We note that there exist inputs for which Szigeti's algorithm can indeed witness an execution with exponential recursion depth, and consequently may construct only exponential sized hypergraphs on those inputs (see example in Remark \ref{remark:dsswch-exponential-example}). So, we necessarily have to modify his algorithm to reduce the recursion depth. 

\paragraph{Our Algorithm.} 
We now describe our modification of Szigeti's algorithm to reduce the recursion depth. We  observe that the hyperedge $A$ chosen during a recursive call of Szigeti's algorithm could
also be the hyperedge chosen during a subsequent recursive call. In fact, this could repeat for several consecutive recursive calls before the algorithm cannot pick the hyperedge $A$ anymore. We avoid such a sequence of consecutive recursive calls by picking as many copies of the hyperedge $A$ as possible into the hypergraph $(H_0, w_0)$ (i.e., set the weight to be the number of copies picked) and revising the input functions $m$ and $p$ accordingly for recursion.
We describe this formally now. Let $(p, m)$ be the input tuple and $A \subseteq V$ be as defined by Szigeti's algorithm.
Let
$$\alpha = \min\begin{cases}
        \alpha^{(1)} \coloneqq  \min\left\{m(u) : u \in A\right\} \\
        \alpha^{(2)} \coloneqq  \min\left\{K_p - p(X) : X \subseteq V - A\right\} \\
        \alpha^{(3)} \coloneqq  \min\left\{K_p - m(u) : u \in V - A\right\} \end{cases}$$
We construct the hypergraph 
$(H_0, w_0)$ on vertex set $V$ consisting of a single hyperedge $A$ with weight $w_0(A)=\alpha$. Next, we 
proceed similarly to Szigeti's algorithm as follows:  We construct the sets $\zeros:=\{u\in A: m(u)=\alpha\}$ and $V'':=V-\zeros$. Next, we define 
the set function $p'': 2^{V''}\rightarrow \Z$ as $p''(X):=\max\{p(X\cup R)-b_{(H_0, w_0)}(X\cup R): R\subseteq \zeros\}$ for every $X\subseteq V''$, and the function $m'': V''\rightarrow \Z$ as $m''(u):=m(u)-\alpha\indicator_{u\in A}$ for every $u\in V''$, where $\indicator_{u\in A}$ evaluates to one if $u\in A$ and evaluates to zero otherwise. 
Next, we recurse on the input tuple $(p'',m'')$ to obtain a hypergraph $(H'', w'')$; obtain the hypergraph $(G, c)$ by adding vertices $\zeros$ to $(H'', w'')$, and return $(G+H_0, c+w_0)$. 

\paragraph{Recursion Depth Analysis.} 
By induction on $K_p$ (generalizing Szigeti's proof), it can be shown that our algorithm returns a hypergraph satisfying properties \ref{thm:SzigetiWeakCover:main:(1)}-\ref{thm:SzigetiWeakCover:main:(3)} 
 of \Cref{thm:SzigetiWeakCover:main} and also terminates within finite number of recursive calls. 
We now sketch our proof to show a strongly polynomial bound on the recursion depth of our modified algorithm. For this, we consider how the value $\alpha$ is computed. We recall that by our min-max relation, $\alpha$ is the minimum of the three values $\{\alpha^{(1)}, \alpha^{(2)}, \alpha^{(3)}\}$. Using this, we identify a potential function which strictly increases with the recursion depth. For this, we consider three set families that we define now. 
Let $\ell$ be the recursion depth of the algorithm on an input instance and let $i\in [\ell]$. Let $\zeros_{i}$ be the set $\zeros$ and $\calD_i$ be the set $\{u\in V: m(u)=K_p\}$ in the $i^{th}$ recursive call. 
Let $\mathcal{F}_i$ be the family of inclusionwise minimal $p$-maximizers where $p$ is the set function input to the $i^{th}$ recursive call (a set $X$ is a $p$-maximizer if $p(X)=K_p$).  
Let $\zeros_{\le i}:=\cup_{j=1}^{i}\zeros_j$ and $\mathcal{F}_{\le i}:=\cup_{j=1}^i \mathcal{F}_j$. 
We consider the potential function $\phi(i):=|\zeros_{\le i}|+|\mathcal{F}_{\le i}|+|\calD_i|$ and show that each of the three terms in the function is non-decreasing with $i$. Furthermore, if $\alpha$ is determined by $\alpha^{(1)}$, then $|\zeros_{\le i}|$ strictly increases; 
if $\alpha$ is determined by $\alpha^{(2)}$, then $|\mathcal{F}_{\le i}|$ strictly increases; and
if $\alpha$ is determined by $\alpha^{(3)}$, then $|\calD_i|$ strictly increases (see \Cref{lem:SzigetiAlgorithm:progression-of-set-families:main}). 
We note that $\zeros_{\le \ell}\subseteq V$ and $\calD_{\ell}\subseteq V$. Moreover, we examine how the input functions $p$ across recursive calls relate to each other (see \Cref{lem:SzigetiAlgorithm:p_i-from-p_1}) and exploit skew-supermodularity of $p$ to show that the family $\mathcal{F}_{\le \ell}$ is \emph{laminar} over the ground set $V$ (see \Cref{lem:UncrossingProperties:Cumulative-Minimal-p-Maximizer-Family-Laminar}). These facts together imply that the recursion depth is at most $4|V| - 1$.
\begin{remark}
    Our main algorithmic contribution is the modification to Szigeti's algorithm to pick as many copies of a chosen hyperedge as possible during a recursive call, i.e., pick hyperedge $A$ with weight $\alpha$ as opposed to weight $1$. 
    Without our modification, there exist inputs for which Szigeti's original algorithm can indeed witness an execution with exponential recursion depth, and consequently may construct only exponential sized hypergraphs on those inputs (see example in Remark \ref{remark:dsswch-exponential-example}).
    Our main analysis contribution is identifying an appropriate potential function to show that our modified algorithm indeed has linear recursion depth.
\end{remark}

\subsubsection{Proof Technique for Theorems \ref{thm:WeakCoverViaUniformHypergraph:main} and \ref{thm:WeakCoverTwoFunctionsViaUniformHypergraph:main}}
In this section, we describe our proof technique for Theorems \ref{thm:WeakCoverViaUniformHypergraph:main} and \ref{thm:WeakCoverTwoFunctionsViaUniformHypergraph:main}. 
The algorithmic result in both these theorems follow from our techniques for the existential result using known tools for submodular functions. So, we focus on describing our proof technique for the existential results here. 
We will consider the setting of \Cref{thm:WeakCoverTwoFunctionsViaUniformHypergraph:main} since it is more general than the setting of \Cref{thm:WeakCoverViaUniformHypergraph:main}. 
We will distinguish between the two settings later after we set up the background. 

Let $p: 2^V\rightarrow \Z$ be a function such that $p(X):=\max\{q(X), r(X)\}$ for every $X\subseteq V$, where $q,r: 2^V\rightarrow \Z$ are skew-supermodular functions (we note that if $p$ itself is skew-supermodular, then it corresponds to the setting of \Cref{thm:WeakCoverViaUniformHypergraph:main}). 
Let $m: V\rightarrow \Z_{\ge 0}$ be a non-negative function satisfying \hypertarget{WeakCover(a)}{(a) $m(X)\ge p(X)$ for every $X\subseteq V$} and \hypertarget{WeakCover(b)}{(b) $m(u)\le K_p$ for every $u\in V$}. Our goal is to show that there exists a hypergraph $(H=(V, E), w: E\rightarrow \Z_+)$ such that \hypertarget{WeakCover(1)}{(1) $(H, w)$ weakly covers the function $p$}, \hypertarget{WeakCover(2)}{(2) $b_{(H, w)}=m(u)$ for every $u\in V$}, \hypertarget{WeakCover(3)}{(3) $\sum_{e\in E}w(e) = K_p$}, \hypertarget{WeakCover(4)}{(4) $|e|\in \{\lfloor m(V)/K_{p}\rfloor, \lceil m(V)/K_{p}\rceil\}$ for every $e\in E$} if $K_p>0$, and \hypertarget{WeakCover(5)}{(5) $|E|=O(|V|)$}. 
Our proof of the existential result builds on the techniques of Bern\'{a}th and Kir\'{a}ly \cite{Bernath-Kiraly} who proved the existence of a hypergraph $(H=(V, E), w: E\rightarrow \Z_+)$ satisfying properties \hyperlink{WeakCover(1)}{(1)}--\hyperlink{WeakCover(4)}{(4)}, so we briefly recall their techniques. We present the algorithmic version of their proof since it will be useful for our purposes. 

For the purposes of the algorithmic proof, we assume that $m$ is a positive-valued function. We show that this assumption is without loss of generality since an arbitrary instance can be reduced to such an instance
(for details, see the proof overview paragraph immediately following the restatement of \Cref{thm:WeakCoverViaUniformHypergraph:main} and \Cref{thm:WeakCoverTwoFunctionsViaUniformHypergraph:main} in Sections \ref{sec:WeakCoverViaUniformHypergraph} and \ref{sec:WeakCoverTwoFunctionsViaUniformHypergraph} respectively).
The proof in \cite{Bernath-Kiraly} is inductive, and consequently, the algorithm implicit in their proof is recursive. 
Central to their proof is the following polyhedron: 
\begin{equation*}
    Q(p, m) \coloneqq  \left\{ x\in \R^{V}\ \middle\vert 
        \begin{array}{l}
            {(\text{i})}{\ \ \ 0 \leq x_u \leq \min\{1, m(u)\}} \hfill {\qquad \forall \ u \in V} \\
            {(\text{ii})}{\ \ x(Z) \geq 1} \hfill {\qquad\qquad\qquad\qquad\qquad \forall\  Z\subseteq V:\ p(Z) = K_p} \\
            {(\text{iii})}{ \ x(u) = 1} \hfill {\qquad\qquad\qquad\qquad\qquad \forall\ u\in V:\ m(u) = K_p} \\
            {(\text{iv})}{ \ x(Z) \leq m(Z) - p(Z) + 1} \hfill { \ \ \ \forall\ Z \subseteq V} \\
            {(\text{v})}{\ \ \left \lfloor \frac{m(V)}{K_p} \right \rfloor \le x(V) \le \left \lceil \frac{m(V)}{K_p} \right \rceil} \hfill {} \\
        \end{array}
        \right\}. 
\end{equation*}
Bern\'ath and Kir\'aly showed the following three properties of the polyhedron $Q(p,m)$: 
\begin{enumerate}[label=(\Roman*),ref=(\Roman*)]
    \item \label{WeakCover:Q:PropertyI} An integer vector $y$ is in $Q(p,m)$ if and only if it is the indicator vector of a hyperedge\footnote{The indicator vector of a hyperedge $e$ of a hypergraph $(H=(V, E), w: E\rightarrow \Z_+)$ is the vector $\chi_e\in \{0,1\}^V$ such that $\chi_e(u)=1$ if and only if $u\in e$.} in a hypergraph $(H=(V, E), w: E\rightarrow \Z_+)$ satisfying properties \hyperlink{WeakCover(1)}{(1)}--\hyperlink{WeakCover(4)}{(4)}, 
    \item \label{WeakCover:Q:PropertyII} the polyhedron $Q(p,m)$ is non-empty, and
    \item \label{WeakCover:Q:PropertyIII} The polyhedron $Q(p,m)$ is the intersection of two generalized polymatroids (see \Cref{appendix:sec:Function-Maximization-Oracles:Qpm-oracle} for the definition of generalized polymatroid) and consequently, all its extreme points are integral. 
\end{enumerate}
\paragraph{Algorithm of \cite{Bernath-Kiraly}.} 
With these three properties, Bern\'{a}th and Kir\'{a}ly's algorithm proceeds as follows. If $K_p = 0$, then the algorithm is in its base case and it returns the empty hypergraph (with no vertices). Otherwise, $K_p > 0$; the algorithm picks a set $A\subseteq V$ whose indicator vector $\chi_A$ is in $Q(p,m)$ -- we note that such a set $A$ exists owing to Properties \ref{WeakCover:Q:PropertyII} and \ref{WeakCover:Q:PropertyIII}; with this choice of the set $A$, the algorithm defines $(H_0, w_0)$ to be the hypergraph on vertex set $V$ consisting of the single hyperedge $A$ with weight $w_0(A) = 1$. 
The fact that we have one hyperedge of the needed hypergraph allows us to revise the functions $m$ and $p$ and recursively solve the problem by relying on Property \ref{WeakCover:Q:PropertyI}. We elaborate on how to revise and recursively solve it now. 
The algorithm constructs the set $\zeros:=\{u\in A: m(u)=1\}$. Next, the algorithm defines revised functions $m'': V-\zeros\rightarrow \Z$ and $p'':2^{V-\zeros}\rightarrow \Z$ as 
$m''(u):=m(u)-\indicator_{u\in A}$ for every $u\in V-\zeros$, where $\indicator_{u\in A}$ evaluates to one if $u\in A$ and evaluates to zero otherwise, and 
$p''(X):=\max\{p(X\cup R)-b_{(H_0, w_0)}(X\cup R): R\subseteq \zeros\}$ for every $X\subseteq V-\zeros$. 
Next, the algorithm recurses on the inputs $m''$ and $p''$ to obtain a hypergraph $(H_0^*, w_0^*)$. Finally, the algorithm obtains the hypergraph $(G, c)$ by adding vertices $\zeros$ to $(H_0^*, w_0^*)$, and returns the hypergraph $(G + H_0, c + w_0)$.
It can be shown that $m''$ is a positive function and there exist skew-supermodular functions $q'', r'':2^V\rightarrow\Z$ such that $p''(X) = \max\{q''(X), r''(X)\}$ for each $X\subseteq V$ (if $p$ is skew-supermodular, then $p''$ is also skew-supermodular). Moreover, it can also be shown that the functions $p'', m''$ satisfy conditions \hyperlink{WeakCover(a)}{(a)} and \hyperlink{WeakCover(b)}{(b)}, and $K_{p''} = K_p - 1$. Consequently, by induction on $K_p$ and Property \ref{WeakCover:Q:PropertyI}, the algorithm can be shown to terminate in $K_p$ recursive calls and returns a hypergraph $(H, w)$ satisfying properties \hyperlink{WeakCover(1)}{(1)}-\hyperlink{WeakCover(4)}{(4)}. Furthermore, we observe that the number of new hyperedges added by the algorithm is at most the number of recursive calls since each recursive call adds at most one new hyperedge to $(H, w)$. Thus, the number of hyperedges in $(H, w)$ is also at most $K_p$. Consequently, in order to reduce the number of new hyperedges, it suffices to reduce the recursion depth of this algorithm.

\paragraph{Our Algorithm.} 
We now describe our modification of the above-mentioned Bern\'{a}th and Kir\'{a}ly's algorithm to reduce the recursion depth.
We make two major modifications to their algorithm. 
First, we observe that the hyperedge $A$ chosen for the input tuple $(p,m)$ could be such that $\chi_A\in Q(p'',m'')$, where $(p'',m'')$ is the input for the next recursive call and consequently, could be chosen as the hyperedge for the next recursive call. In fact, this could repeat for several consecutive recursive calls before we cannot use the hyperedge $A$ anymore. We avoid such a sequence of consecutive recursive calls by picking as many copies of the hyperedge $A$ as possible into the hypergraph $(H_0, w_0)$ (i.e., set the weight
to be the number of copies picked) and revising $m$ and $p$ accordingly. We describe this formally now: 
 Let $(p, m)$ be the input tuple and $A \subseteq V$ be as defined by Szigeti's algorithm.
For $\beta \in \Z_+$, 
let $(H_0^\beta, w_0^\beta)$ denote the hypergraph on vertex set $V$ consisting of a single hyperedge $A$ with weight $w_0^{\beta}(A) = \beta$. Let
$$\alpha := \max\left\{\beta\in \Z_+ : \chi_A\in Q\left(p-b_{(H_0^{\beta}, w_0^{\beta})}, m-\beta \chi_A\right)\right\}.$$
We construct the hypergraph $(H_0, w_0):=(H_0^{\alpha}, w_0^{\alpha})$, 
sets $\zeros:=\{u\in A: m(u)=\alpha\}$, $V'':=V-\zeros$, and proceed similar to Bern\'{a}th and Kir\'{a}ly's  algorithm. In particular, we define the revised inputs $p'': 2^{V''}\rightarrow \Z$ as $p''(X):=\max\{p(X\cup R)-b_{(H_0, w_0)}(X\cup R): R\subseteq \zeros\}$ and $m'': V''\rightarrow \Z$ as $m''(u):=m(u)-\alpha\indicator_{u\in A}$ for every $u\in V''$, where $\indicator_{u\in A}$ evaluates to one if $u\in A$ and evaluates to zero otherwise. 
Our algorithm then recurses on the revised input tuple $(p'',m'')$ to obtain a hypergraph $(H_0^*, w_0^*)$, obtains the hypergraph $(G, c)$ by adding vertices $\zeros$ to $(H_0^*, w_0^*)$, and returns $(G+H_0, c+w_0)$.   

Our second modification to Bern\'{a}th and Kir\'{a}ly's  algorithm is regarding the hyperedge $A''$ chosen in the subsequent recursive call. 
To initialize for an arbitrary input $(p,m)$, we pick an arbitrary hyperedge $A$ such that $\chi_A\in Q(p,m)$. However, in the subsequent recursive call, we force the chosen hyperedge to be correlated with the hyperedge chosen in the current call. Their algorithm would have picked an arbitrary set $A''\subseteq V''$ whose indicator vector $\chi_{A''}$ is in $Q(p'',m'')$. Instead, we pick a set $A''$ that has the \emph{largest overlap} with the set $A$ chosen in the immediate previous recursive call subject to $\chi_{A''}\in Q(p'',m'')$.    

\paragraph{Recursion Depth Analysis.} 
By induction on $K_p$ (generalizing Bern\'{a}th and Kir\'{a}ly's proof), it can be shown that our algorithm returns a hypergraph satisfying \hyperlink{WeakCover(1)}{(1)}--\hyperlink{WeakCover(4)}{(4)} and also terminates within finite number of recursive calls. 
Here, we sketch our proof to show that the recursion depth of our modified algorithm is polynomial in the number of vertices. For this, we consider how the value $\alpha$ is computed:  
We have five constraints describing the polyhedron $Q(p,m)$. Consequently, the largest positive integer $\alpha$ such that $\chi_A\in Q(p-b_{(H_0^{\alpha}, w_0^{\alpha})}, m-\alpha\chi_A)$ is determined (i.e., constrained tightly) by one of the five constraints. We show that $\alpha$ can be determined by constraint (v) in at most one recursive call (\Cref{lem:alpha-5}). Hence, it suffices to bound the number of recursive calls where $\alpha$ is determined by one of the constraints (i)--(iv). Next, we bound the number of recursive calls where $\alpha$ is determined by one of the constraints (i)--(iii). For this, we consider three set families that we define now: 
Let $\ell$ be the recursion depth of the algorithm on an input instance and let $i\in [\ell]$. Let $\zeros_{i}$ be the set $\zeros$ and $\calD_i$ be the set $\{u\in V: m(u)=K_p\}$ in the $i^{th}$ recursive call. 
Let $\mathcal{F}_i$ be the family of inclusionwise minimal $p$-maximizers where $p$ is the set function input to the $i^{th}$ recursive call.  
Let $\zeros_{\le i}:=\cup_{j=1}^{i}\zeros_j$ and $\mathcal{F}_{\le i}:=\cup_{j=1}^i \mathcal{F}_j$. 
We consider the potential function $\phi(i):=|\zeros_{\le i}|+|\mathcal{F}_{\le i}|+|\calD_i|$ and show that each of the three terms in the function is non-decreasing with $i$. Furthermore, if $\alpha$ is determined by constraint (i), then $|\zeros_{\le i}|$ strictly increases; 
if $\alpha$ is determined by constraint (ii), then $|\mathcal{F}_{\le i}|$ strictly increases; 
if $\alpha$ is determined by constraint (iii), then $|\calD_i|$ strictly increases (see \Cref{lem:Progression-of-set-families:main}). 
We note that $\zeros_{\le \ell}\subseteq V$ and $\calD_{\ell}\subseteq V$. We exploit skew-supermodularity to show that $\mathcal{F}_{\le \ell}$ is the union of two laminar families over the ground set $V$ (see \Cref{lem:UncrossingProperties:Cumulative-Minimal-p-Maximizer-Family-Laminar} and \Cref{claim:WeakCoverTwoFunctionsViaHypergraphs:calFp_leqell-4|V|-2} in  \Cref{lem:WeakCoverTwoFunctionsViaUniformHypergraph:recursion-depth:main}). These facts together imply that the number of recursive calls where $\alpha$ is determined by one of the constraints (i)--(iii) is at most $6|V|$. It remains to bound the number of recursive calls where $\alpha$ is determined by constraint (iv). This is the challenging part of the proof.
Our proof strategy for covering one function differs from that for covering two functions and we will describe them next. 

\paragraph{Covering one function.}
Here, we consider the case where the input function $p$ is skew-supermodular. As mentioned before, the algorithm's input function in each recursive call is also skew-supermodular. 
To bound the number of recursive calls where $\alpha$ is determined by constraint (iv), we consider another set family. For an input tuple $(p,m)$, we say that a set $X$ is \emph{$(p, m)$-tight} if $p(X)=m(X)$. Let $i\in [\ell]$. We define $T_i$ to be the family of $(p, m)$-tight sets where the tuple $(p, m)$ is the input to the $i^{th}$ recursive call and $\mathcal{T}_{i}$ to be the family of inclusionwise maximal sets in $T_i$. In contrast to $\mathcal{F}_{\le i}=\cup_{j=1}^i\mathcal{F}_j$, it is possible that $\cup_{j=1}^i\mathcal{T}_{j}$ is not laminar. To overcome this challenge, we define $\mathcal{T}_{\le i}$ differently by considering the projection of all tight sets to the ground set $V_i$ of the $i^{th}$ recursive call: we define $\mathcal{T}_{\le 1}:=\mathcal{T}_1$ and 
$\mathcal{T}_{\le i}:=\mathcal{T}_i\cup \{X\cap V_i: X\in \mathcal{T}_{\le i-1}\}$ for integers $i$ where $2\le i \le \ell$. We show that $\mathcal{T}_{\le i}$ is laminar for every $i\in [\ell]$ (see \Cref{lem:Cumulative-Projected-Maximal-Tight-Set-Family-Laminar}). 
However, $|\mathcal{T}_{\le i}|$ is not necessarily non-decreasing with $i$ since projection of a set family to a subset could result in the loss of sets from the family. To circumvent this issue, we consider the potential function $\Phi(i):=|\mathcal{F}_{\le i}| + |\mathcal{T}_{\le i}|+3|\zeros_{\le i-1}|$ and show that it is non-decreasing with $i$ (see \Cref{lem:WeakCoverViaUniformHypergraph:runtime:num-createsteps}). Next, it is tempting to show that if $\alpha$ is determined by constraint (iv), then $\Phi(i)$ strictly increases, but this is not necessarily true. Instead, by exploiting the second modification in our algorithm -- i.e., by correlating the hyperedge chosen in the 
current recursive call with the hyperedge chosen in the immediate previous recursive call, 
we show that if $\Phi(i)=\Phi(i+1)$ and $\alpha$ is determined by constraint (iv) (but none of the other constraints) in the $i^{th}$ recursive call, then $\Phi(i+2)>\Phi(i)$ (see \Cref{lem:WeakCoverViaUniformHypergraph:runtime:num-createsteps} which in turn relies on Lemmas \ref{lem:good-set:alpha=alpha4:maximalTightSetFamily-changes} and \ref{lem:good_vector}). We emphasize that this is the most challenging part of the proof -- it   involves showing that new tight sets are formed at the end of the $(i+1)^{th}$ recursive call assuming no progress to the potential function in the $i^{th}$ recursive call; we prove this by exploiting the correlated choice of $A''$ with $A$. The above arguments together imply that the number of recursive calls where $\alpha$ is determined by constraint (iv) is $O(|V|)$.  

\paragraph{Covering two functions.}
Here, we consider the case where the input function $p: 2^V\rightarrow \Z$ is such that $p(X)=\max\{q(X), r(X)\}$ for every $X\subseteq V$, where $q, r: 2^V\rightarrow \Z$ are skew-supermodular functions. We recall that the goal is to bound the number of recursive calls where $\alpha$ is determined by constraint (iv). In the setting of covering two functions, we are unable to adapt the strategy mentioned above for covering one function -- we were unable to show that new tight sets are formed at the end of the $(i+1)^{th}$ recursive call assuming $\Phi(i)=\Phi(i+1)$ and $\alpha$ is determined by constraint (iv) in the $i^{th}$ recursive call. So, our proof strategy here is different. 

For an input tuple $(p,m)$, we define the pairwise $(p, m)$-slack function $\gamma_{p,m}(\{u, v\}):=\min\{m(X)-p(X): u, v\in X\subseteq V\}$ for every pair $\{u, v\}\in \binom{V}{2}$. As mentioned before, the algorithm's input function $(p_i: 2^{V_i}\rightarrow \Z, m_i: V_i \rightarrow \Z_{\ge 0})$ in the $i^{th}$ recursive call satisfies $m_i(X)\ge p_i(X)$ for every $X\subseteq V_i$. Consequently, the pairwise slack function $\gamma_{p_i,m_i}(\{u, v\})$ is non-negative for every pair $u, v\in \binom{V_i}{2}$. We show that $\gamma_{p_{i+1}, m_{i+1}}(\{u, v\})\le \gamma_{p_i, m_i}(\{u, v\})$ for every pair $\{u, v\} \in \binom{V_{i+1}}{2}$ -- i.e., the pairwise slack function is non-increasing (see \Cref{lem:slack_monotone}). Moreover, if $\alpha$ is determined by constraint (iv) in the $i^{th}$ recursive call, then there exists a set $Z\subseteq V$  such that $|A\cap Z|\ge 2$ and $\gamma_{p_{i+1}, m_{i+1}}(\{u, v\})<\min\{|A\cap Z|-1, \gamma_{p_i, m_i}(\{u, v\})\}$ for every pair $\{u, v\} \in \binom{A \cap Z}{2}$, where $A$ is the hyperedge chosen in the $i^{th}$ recursive call (see \Cref{lem:slack:witness-set-properties}). Now, we define the potential function $\Phi(i):=\sum_{\{u, v\}\in \binom{V_i}{2}}\sum_{j=1}^{ \gamma_{p_i, m_i}(u,v)}(1/j^2)$ and show that it is non-increasing. Moreover, if $\alpha$ is determined by constraint (iv) in the $i^{th}$ recursive call, then $\Phi(i+1)\le \Phi(i)-1/4$ (\Cref{claim:WeakCoverTwoFunctionsViaHypergraphs:num-createsteps-alpha=alpha4} in \Cref{lem:WeakCoverTwoFunctionsViaUniformHypergraph:recursion-depth:main}). Finally, we have that $\Phi(1)=O(|V|^2)$. Consequently, the number of recursive calls where $\alpha$ is determined by constraint (iv) is $O(|V|^2)$. 
We emphasize that our proofs for covering two functions \emph{do not} exploit the third modification that we made to Bern\'{a}th and Kir\'{a}ly's algorithm. 

\begin{remark}
Our algorithmic contributions are the two modifications to Bern\'{a}th and Kir\'{a}ly's algorithm: (1) We allow picking as many copies of a chosen hyperedges as possible during a recursive call (i.e., pick hyperedge $A$ with weight $\alpha$ as opposed to weight $1$). 
(2) We correlate the choice of the hyperedge $A$ in the 
current recursive call with the immediate previous recursive call. 
The first modification is necessary to bound the recursion depth to be polynomial in the size of the ground set $V$. 
The second modification is not necessary to bound the recursion depth to be polynomial in the size of the ground set $V$. The first modification already suffices for this purpose. 
We exploit the second modification in the setting of covering one function to improve the recursion depth from quadratic to linear in the size of the ground set. 
\end{remark}

\begin{remark}
    Our main analysis contribution is identifying appropriate potential functions to measure progress of our modified algorithm. For covering two functions, we identified the notion of pairwise slack functions, showed progress for these functions, and devised a potential function that combined pairwise slack functions to show global progress. 
    For the improved bound for covering one function, we exploited the correlated choice of the hyperedges chosen in two consecutive recursive calls to show quantifiable progress within two recursive calls. 
\end{remark}

%% file: related-work.tex
\subsection{Related Work}
\label{sec:related-work}
In network design problems, the goal is to construct a graph that has certain pre-specified properties. This is a broad family of problems encompassing polynomial-time solvable problems like minimum cost spanning tree and shortest $(s, t)$-path as well as NP-hard problems like Steiner Tree and Tree Augmentation.
The polynomial-time solvable network design problems have been studied from the perspective of fast algorithms while the NP-hard network design problems have been studied from the perspective of approximation algorithms leading to a rich toolkit of techniques on both fronts (e.g., see \cite{Frank-book, nagamochi_ibaraki_2008_book, Schrijver-book, Vazirani-book, SW-book, SRL-book}). In this section, we mention related works on \emph{hypergraph} network design problems. 

Hypergraph cut/coverage/spectral sparsifiers can be viewed as special cases of hypergraph network design problems. Here, the input is a hypergraph and the goal is to construct a sub-hypergraph (i.e., the hyperedges in the constructed hypergraph form a subset of the hyperedges of the input hypergraph and their weights are allowed to be different from the input hypergraph) with few hyperedges in order to approximately preserve cut and/or  coverage values of all subsets of vertices and/or spectral values of all vectors on $V$ coordinates where $V$ is the vertex set of the input hypergraph. 
Recent works have shown the existence of cut/coverage/spectral sparsifiers where the number of hyperedges is near-linear in the number of vertices and that they can be constructed in time that is near-linear in the size of the input hypergraph \cite{KK15, CX18, SY19, BST19, CKN20, KKTY21-FOCS, KKTY21-STOC, Lee22, JLS22, Quanrud22}. These results do not have direct implications for the problems that we study in our work since we focus on degree-specified network design problems. Even if we relax the degree-specification requirement, these results only imply the \emph{existence} of hypergraphs with near-linear number of hyperedges that satisfy certain kinds of cut/coverage requirements \emph{approximately}. In contrast, our work focuses on efficient algorithmic construction of small-sized hypergraphs that satisfy cut/coverage requirements exactly. 

The rank of a hypergraph is the size of its largest hyperedge. 
Degree-specified hypergraph local connectivity augmentation using \emph{rank-$t$} hyperedges of minimum total weight problem is a variant of  \dshlcah 
where the returned hypergraph is additionally required to have rank at most $t$ (where $t$ is an input integer). 
Thus, \dshlcae is a special case of degree-specified hypergraph local connectivity augmentation using rank-$2$ hyperedges. Although \dsglcae is polynomial-time solvable \cite{WN87, CS89, Fra94, Fra92, Gab94, BHTP03, NGM97, Ben99, BK00, LY13}, \dshlcae is NP-hard (i.e., hypergraph local connectivity augmentation using edges of minimum total weight is NP-hard) \cite{CJK10}. 
In a surprising twist, independent works by Cosh \cite{Cosh-thesis}, Nutov \cite{Nut10}, and Bern\'{a}th and Kir\'{a}ly \cite{BK12} have shown that feasible instances of \dshlcah admit a solution with at most one hyperedge of size at least $3$ while all other hyperedges are in fact size-$2$ edges. 
Their results are also applicable to \dssssch. We note that their proof technique does not generalize to obtain near-uniform hyperedges or to address simultaneous cover of two functions. 

%% file: preliminaries.tex
\section{Preliminaries}
\paragraph{Projection of Laminar Families.} We recall that two sets $X, Y \subseteq V$ are said to cross if the sets $X - Y, Y - X, X \cap Y$ are non-empty. A
family $\calL \subseteq 2^V$ is \emph{laminar} if no two sets in $\calL$ cross. We will be interested in the \emph{projection} of a laminar family, i.e., the restriction of a laminar family $\calL \subseteq 2^V$ to a ground set $V - \calZ$ for some $\calZ\subseteq V$. The next lemma from \cite{bérczi2023splittingoff} says that the size of the projection is comparable to that of the original family $\calL$.

\begin{lemma}[Lemma 2.1 of \cite{bérczi2023splittingoff}]\label{lem:CoveringAlgorithm:projection-laminar-family}
    Let $\calL \subseteq 2^V$ be a laminar family and $\calZ \subseteq V$ be a subset of elements. Let $\calL' \coloneqq  \{X - \calZ : X \in \calL\} - \{\emptyset\}$. Then, the family $\calL'$ is a laminar family and $|\calL| \leq |\calL'| + 3|\calZ|$.
\end{lemma}

\paragraph{Maximization Oracles.} We recall that in our algorithmic problems, we have access to the input set function $p$ via \functionMaximizationOracleStrongCover{p}. We describe two additional oracles that  will simplify our proofs. We will show that both oracles defined below can be implemented in strongly polynomial time using \functionMaximizationOracleStrongCover{p}. 
\begin{definition}[$p$-maximization oracles]
    Let $p: 2^V\rightarrow \Z$ be a set function. 
    We define two oracles for the set function $p$. Both oracles take as input disjoint sets $S_0, T_0\subseteq V$, and a vector ${y_0} \in \R^V$. 
    \begin{enumerate}
        \item $\functionMaximizationOracle{p}\left(\left({G_0},{w_0}\right), S_0, T_0, {y_0}\right)$ additionally takes a hypergraph $(G_0=(V, E_0), c_0: E_0\rightarrow \mathbb{Z}_{\ge 0})$ as input and returns a tuple $(Z, p(Z))$, where $Z$ is an optimum solution to the following problem:
        \begin{align*}
        \max & \left\{p(Z)-b_{(G_0, c_0)}(Z)+y_0(Z): S_0\subseteq Z\subseteq V-T_0\right\}. \tag{\functionMaximizationOracle{p}}\label{tag:oracle-main}
        \end{align*}
         
        \item $\functionMaximizationEmptyOracle{p}(S_0, T_0, y_0)$ returns a tuple $(Z, p(Z))$, where $Z$ is an optimum solution to the following problem:
        \begin{align*}
        \max & \left\{p(Z)+y_0(Z): S_0\subseteq Z\subseteq V-T_0\right\}. \tag{\functionMaximizationEmptyOracle{p}}\label{tag:oracle-empty}
        \end{align*}
    \end{enumerate}
\end{definition}

The evaluation oracle for a function $p$ can be implemented using one query to  
\functionMaximizationEmptyOracle{p};  \functionMaximizationEmptyOracle{p} can be implemented using one query to \functionMaximizationOracle{p} where the input hypergraph is the empty hypergraph. The next lemma from \cite{bérczi2023splittingoff} shows that \functionMaximizationOracle{p} can be implemented using at most $|V|+1$ queries to \functionMaximizationOracleStrongCover{p} where the hypergraphs used as input to  \functionMaximizationOracleStrongCover{p} queries have size of the order of the size of the hypergraph input to \functionMaximizationOracle{p}. We note that the \functionMaximizationOracle{p} as defined in \cite{bérczi2023splittingoff} does not include the vector $y_0$ as an input parameter. The proof of the next lemma follows by modifying the proof of the lemma from \cite{bérczi2023splittingoff} to incorporate the parameter $y_0$. We omit details here for brevity. 

\begin{lemma}[Lemma 2.2 of \cite{bérczi2023splittingoff}]\label{lem:Preliminaries:wc-oracle-from-sc-oracle}
Let $p: 2^V\rightarrow \Z$ be a set function, $(G = (V, E), w:E\rightarrow\Z_+)$ be a hypergraph, and $S, T \subseteq V$ be disjoint sets. Then, $\functionMaximizationOracle{p}((G, w), S, T)$ can be implemented to run in $O(|V|(|V|+|E|))$ time using at most $|V|+1$ queries to $\functionMaximizationOracleStrongCover{p}$.
The run-time includes the time to construct the hypergraphs used as input to the queries to \functionMaximizationOracleStrongCover{p}. Moreover, each query to \functionMaximizationOracleStrongCover{p} is on an input hypergraph $(G_0, c_0)$ that has at most $|V|$ vertices and $|E|$ hyperedges.    
\end{lemma}

%% file: set-families.tex
\section{Minimal Maximizers of Skew-Supermodular Function Sequences}\label{sec:AnalysisTools:StructuredSetFamilies}
In this section, we consider set families of maximizing sets of \emph{skew-supermodular} functions. We state certain structural properties of these families which we leverage in subsequent sections to show rapid convergence of our algorithms. 

For a function $p:2^V\rightarrow\Z$, a set $X\subseteq V$ is a \emph{$p$-maximizer} if $p(X)=K_p$. We let $\calF_p$ denote the family of \emph{minimal} $p$-maximizers. The following lemma says that the family $\calF_p$ is disjoint if the function $p$ is skew-supermodular. An important consequence of this lemma is that a transversal for the family $\calF_p$ has polynomial (in particular, at most linear) size, a fact that we will use in our analyzing the runtime of one of our algorithms in subsequent sections.

\begin{restatable}{lemma}{calFpDisjoint}\label{lem:UncrossingProperties:calFp-disjoint}
    Let $p:2^V\rightarrow\Z$  be a skew-supermodular function. Then, the family $\calF_p$ is a disjoint set-family.
\end{restatable}

The structural property of \Cref{lem:UncrossingProperties:calFp-disjoint} is in fact a special case of a more general structural property of certain sequences of skew-supermodular functions. Let $p_1, p_2, \ldots, p_\ell$ be a sequence of skew-supermodular functions defined over a ground set $V$. We will be interested in the structure of the family of minimal maximizers across the entire sequence of functions. In particular, we will focus on the family $\calF_{p_{\leq \ell}}$, where the family $\calF_{p_{\leq i}}$ for $i\in [\ell]$ is defined as follows: $$\calF_{p_{\leq i}} \coloneqq  \bigcup_{j \in [i]} \calF_{p_{i}}$$
We note that for a general sequence of skew-supermodular functions, the sets of the family $\calF_{p_{\leq \ell}}$ may not be well-structured. Consequently, \Cref{lem:UncrossingProperties:calFp-disjoint} applied to each function in the sequence implies that $|\calF_{p_{\leq \ell}}| \leq O(\ell|V|)$.
However, we will focus on a specific structured sequence of skew-supermodular functions that will arise in the analyses of all our algorithms in later sections. The next lemma shows that for such a sequence, the family $\calF_{p_{\leq \ell}}$ is \emph{laminar} (which in turn, will be helpful to conclude later that $\left|\calF_{p_{\leq \ell}}\right|=O(|V|)$).

\begin{restatable}{lemma}{cumulativeMinimalMaximizerFamilyLaminar}\label{lem:UncrossingProperties:Cumulative-Minimal-p-Maximizer-Family-Laminar}
    Let $\ell \in \Z_+$ be a positive integer. Let the sequence $V_1, V_2, \ldots, V_\ell$ and $\zeros_1, \zeros_2, \ldots, \zeros_\ell$ be such that $V_{i+1} \subseteq V_{i}$ 
    and $\zeros_{i} \coloneqq  V_{i} - V_{i+1}$ for $i \in [\ell - 1]$.  Let $g_1, g_2, \ldots, g_\ell$ be a sequence of non-negative monotone submodular functions such that $g_i : 2^{V_i} \rightarrow \Z_{\geq 0}$ for each $i \in [\ell]$. Let $p_1, p_2, \ldots, p_\ell$ be a sequence of functions such that $p_i : 2^{V_i}\rightarrow \Z$ for $i\in[\ell]$,  the function $p_1$ is skew-supermodular, and $p_{i+1} \coloneqq  \functionContract{(p_i - g_i)}{\zeros_i}$ for each $i \in [\ell - 1]$. Then, the family $\calF_{p_{\leq \ell}}$ is a laminar family over the ground set $V_1$.
\end{restatable}

We prove \Cref{lem:UncrossingProperties:Cumulative-Minimal-p-Maximizer-Family-Laminar} in \Cref{appendix:sec:uncrossing-properties:cumulative-minimal-maximizers-family-across-function-sequences}. For completeness, we also include a simplified proof of the special case \Cref{lem:UncrossingProperties:calFp-disjoint} in the same section.

%% file: weak-cover-szigeti.tex
\section{Weak Cover with Linear Number of Hyperedges}\label{sec:SzigetiAlgorithm}
In this section we prove Theorem \ref{thm:SzigetiWeakCover:main}. We restate it below. 

\thmSzigetiWeakCover*

We first describe our proof of \Cref{thm:SzigetiWeakCover:main} under the assumption that the input function $m:V\rightarrow\Z_{+}$ is a \emph{positive} function. 
In \Cref{sec:SzigetiAlgorithm:algorithm}, we present our algorithm (see \Cref{alg:SzigetiAlgorithm}). 
In \Cref{sec:SzigetiAlgorithm:termination-and-partial-correctness}, we show that our algorithm terminates within a finite (pseudo-polynomial) number of recursive calls and returns a hypergraph satisfying properties (1), (2) and (3) of \Cref{thm:SzigetiWeakCover:main} (see \Cref{lem:SzigetiAlgorithm:termination-and-correctness}). In \Cref{sec:SzigetiAlgorithm:recursion-depth-and-hypergraph-support-size}, we give a tighter bound on the number of recursive calls witnessed by our algorithm and show that the hypergraph returned by algorithm also satisfies property (4) of \Cref{thm:SzigetiWeakCover:main} (see \Cref{lem:SzigetiAlgorithm:recursion-depth-and-support-size}). Finally, in \Cref{sec:SzigetiAlgorith:runtime-skew-supermodular-functions}, we show that our algorithm runs in strongly polynomial time, given the appropriate function evaluation oracle (see \Cref{lem:SzigetiAlgorithm:runtime}). 
Lemmas 
\ref{lem:SzigetiAlgorithm:termination-and-correctness}, \ref{lem:SzigetiAlgorithm:recursion-depth-and-support-size}, and \ref{lem:SzigetiAlgorithm:runtime} together complete the proof of \Cref{thm:SzigetiWeakCover:main} under the assumption that the input function $m:V\rightarrow\Z_+$ is a \emph{positive} function. 

We now briefly remark on how to circumvent the positivity assumption on the input function $m$ in the above proof. We note that if $\zeros \coloneqq  \{u \in V : m(u) = 0\} \not = \emptyset$, then $\functionContract{p}{\zeros} : 2^{V-\zeros} \rightarrow\Z$ is a skew-supermodular function and $\functionRestrict{m}{\zeros}:V - \zeros \rightarrow\Z_+$ is a \emph{positive} function satisfying the two hypothesis conditions of \Cref{thm:SzigetiWeakCover:main}, i.e. $\functionRestrict{m}{\zeros}(X) \geq \functionContract{p}{\zeros}(X)$ for every $X \subseteq V - \zeros$ and $\functionRestrict{m}{\zeros}(u) \leq K_{\functionContract{p}{\zeros}}$ for every $u \in V - \zeros$. Furthermore, we observe that a hypergraph satisfying properties (1)-(4) for the functions $\functionContract{p}{\zeros}$ and $\functionRestrict{m}{\zeros}$ also satisfies the four properties for the functions $p$ and $m$. Finally, \functionMaximizationOracleStrongCover{\functionContract{p}{\zeros}} can be implemented using \functionMaximizationOracleStrongCover{p} in strongly polynomial time. 

\subsection{The Algorithm}\label{sec:SzigetiAlgorithm:algorithm}
Our algorithm takes as input (1) a skew-supermodular function $p:2^V \rightarrow\Z$, and (2) a \emph{positive} function $m:V \rightarrow \Z_{+}$ and returns a hypergraph $\left(H= \left(V, E\right), w\right)$. We note that in contrast to the non-negative function $m$ appearing in the statement of \Cref{thm:SzigetiWeakCover:main}, the function $m$ that is input to our algorithm should be positive. 

We now give an informal description of the algorithm (refer to \Cref{alg:SzigetiAlgorithm} for a formal description). Our algorithm is recursive.
If $V=\emptyset$, then the algorithm is in its base case and returns the empty hypergraph. Otherwise, the algorithm is in its recursive case. First, the algorithm computes an arbitrary minimal transversal $T \subseteq V$ for the family of $p$-maximizers -- i.e., an inclusionwise minimal set $T\subseteq V$ such that $T\cap X\neq \emptyset$ for every $X\subseteq V$ with $p(X)=K_p$. Next, the algorithm computes the set $\calD \coloneqq  \{u\in V : m(u) = K_p\}$ and defines the set $A \coloneqq  T \cup \calD$. 
Next, the algorithm uses the set $A$ to compute the following three intermediate quantities $\alpha^{(1)}, \alpha^{(2)}, \alpha^{(3)}$:
\begin{itemize}
        \item $\alpha^{(1)} \coloneqq  \min\{m(u) : u\in A\}$, 
        \item $\alpha^{(2)} \coloneqq  \min\left\{K_p - p(X) : X \subseteq V - A\right\}$, and 
        \item $\alpha^{(3)} \coloneqq  \min\left\{K_p - m(u) : u \in V - A\right\}$,
    \end{itemize}
    and defines $\alpha \coloneqq  \min\left\{ \alpha^{(1)}, \alpha^{(2)}, \alpha^{(3)}\right\}$. 
    Next, the algorithm computes the set $\zeros \coloneqq  \{u \in V : m(u) - \alpha = 0\}$ and defines $(H_0, w_0)$ to be the hypergraph on vertex set $V$ consisting of a single hyperedge $A$ with weight $w_0(A) = \alpha$. Next, the algorithm defines the two functions $m':V\rightarrow \Z$ and $m'':V - \zeros\rightarrow\Z$ as $m' := m - \alpha\chi_A$, and $m'' := \functionRestrict{m'}{\zeros}$. Next, the algorithm defines the two functions $p':2^V \rightarrow\Z$ and $p'':2^{V - \zeros}\rightarrow\Z$ as 
    $p'\coloneqq p - b_{(H_0, w_0)}$ and $p''\coloneqq \functionContract{p'}{\zeros}$. The algorithm then recursively calls itself with the input tuple $(p'', m'')$ to obtain a hypergraph $\left(H'' = \left(V'' \coloneqq  V - \zeros, E''\right), w''\right)$. Finally, the algorithm extends the hypergraph $H''$ with the set $\zeros$ of vertices, adds the set $A$ with weight $\alpha$ as a new hyperedge to the hyperedge set $E''$, and returns the resulting hypergraph. If the hyperedge $A$ already has non-zero weight in $(H'', w'')$, then the algorithm increases the weight of the hyperedge $A$ by $\alpha$, and returns the resulting hypergraph. 

\begin{algorithm}[!htb]
\caption{Weak covering with hyperedges}\label{alg:SzigetiAlgorithm}
\textsc{Input:} Skew-supermodular function $p:2^V\rightarrow\Z$ and positive function $m:V\rightarrow\Z_+$\\
\textsc{Output:} Hypergraph $(H = (V,E), w:E\rightarrow\Z_+)$\\

$\mathbf{\textsc{Algorithm}}\left(p, m\right):$
\begin{algorithmic}[1]
    \If{$V = \emptyset$} \Return{Empty Hypergraph  $\left(\left(\emptyset, \emptyset\right), \emptyset\right)$.}\label{algstep:SzigetiAlgorithm:base-case}
    \Else:\label{algstep:SzigetiAlgorithm:else}
    \State{$\calD \coloneqq  \{u \in V : m(u) = K_p\}$}\label{algstep:SzigetiAlgorithm:def:calD}
    \State{$T\coloneqq $ an arbitrary minimal transversal for the family of $p$-maximizers}\label{algstep:SzigetiAlgorithm:def:T}
    \State{$A \coloneqq  T\cup \calD$}\label{algstep:SzigetiAlgorithm:def:A}
    \State{$\alpha \coloneqq  \min\begin{cases}
        \alpha^{(1)} \coloneqq  \min\left\{m(u) : u \in A\right\} \\
        \alpha^{(2)} \coloneqq  \min\left\{K_p - p(X) : X \subseteq V - A\right\} \\
        \alpha^{(3)} \coloneqq  \min\left\{K_p - m(u) : u \in V - A\right\} \end{cases}$}\label{algstep:SzigetiAlgorithm:def:alpha}
    \State{$\zeros \coloneqq  \left\{u \in A : m(u) - \alpha = 0\right\}$}\label{algstep:SzigetiAlgorithm:def:zeros}
    \State{Construct $\left(H_0:=\left(V, E_0:=\left\{A\right\}\right), w_0: E_0\rightarrow \left\{\alpha\right\}\right)$}\label{algstep:SzigetiAlgorithm:def:Tilde_H-Tilde_w}
    \State{$m'\coloneqq m - \alpha\chi_{A}$ and $m''\coloneqq \functionRestrict{m'}{\zeros}$}\label{algstep:SzigetiAlgorithm:def:m'-and-m''}
    \State{$p'\coloneqq p - b_{(H_0, w_0)}$ and $p''\coloneqq \functionContract{p'}{\zeros}$}\label{algstep:SzigetiAlgorithm:def:p'-and-p''}
    \State{$(H'',w'')\coloneqq \textsc{Algorithm}\left(p'', m''\right)$}\label{algstep:SzigetiAlgorithm:recursion}
    \State{Obtain hypergraph $(G, c)$ from $(H'', w'')$ by adding vertices $\zeros$.}\label{algstep:SzigetiAlgorithm:def:G-and-c}
    \State{\Return $(G+H_0, c+w_0)$}\label{algstep:SzigetiAlgorithm:return}
    \EndIf
\end{algorithmic}
\end{algorithm}


\subsection{Termination and Partial Correctness}\label{sec:SzigetiAlgorithm:termination-and-partial-correctness}
In this section, we show that \Cref{alg:SzigetiAlgorithm} terminates within a finite (pseudo-polynomial) number of recursive calls, and moreover, returns a hypergraph satisfying properties (1)-(3) of \Cref{thm:SzigetiWeakCover:main}. The ideas in this section are modifications of the results due to Szigeti \cite{Szi99}; we suggest first time readers to skip this section and return to it as needed. 
\Cref{lem:SzigetiAlgorithm:termination-and-correctness} is the main result of this section and is included at the end of the section. The next lemma states useful properties of \Cref{alg:SzigetiAlgorithm}.

\begin{lemma}\label{lem:SzigetiAlgorithm:properties}
Suppose that $V \not = \emptyset$ and the input to  \Cref{alg:SzigetiAlgorithm} is a tuple $(p,m)$, where $p:2^V\rightarrow\Z$ is a skew-supermodular function and $m:V\rightarrow\Z_{+}$ is a positive function such that $m(X) \ge p(X)$ for all $X \subseteq V$ and $m(u) \leq K_p$ for all $u \in V$. Let $\alpha, \zeros, m''$ be as defined by \Cref{alg:SzigetiAlgorithm}. Then, we have that 
\begin{enumerate}[label=(\alph*)]
    \item $\alpha \geq 1$,
    \item $m''(V - \zeros) < m(V)$,
    \item $m(X) \geq p(X) + \alpha|A\cap X| - \alpha$ for every $X \subseteq V$.
\end{enumerate}
\end{lemma}
\begin{proof} We note that $K_p \geq m(u)\geq 1$ for every $u\in V$, where the final inequality is because $m$ is a positive function, and consequently $K_p > 0$ since $V \not = \emptyset$. We now prove each property separately below. 
\begin{enumerate}[label=(\alph*)]
    \item We note that
    $\alpha^{(1)} \in \Z_+$ since the function $m$ is a positive integer valued function. Thus, it suffices to show that $\alpha^{(2)}, \alpha^{(3)} \in \Z_+$. Our strategy to show this will be to assume otherwise and arrive at a contradiction to the algorithm's choice of the set $A$.

    First, we recall that $\alpha^{(2)} = \min\{K_p - p(Z) : Z \subseteq V - A\}$. We note that $\alpha^{(2)} \in \Z$ since $p:2^V \rightarrow \Z$. Thus, it suffices to show that $\alpha^{(2)} \geq 1$. By way of contradiction, say $\alpha^{(2)} \leq 0$.  Since $K_p \geq p(X)$ for all $X \subseteq V$ by definition, we have that $\alpha^{(2)} \geq 0$. Thus, $\alpha^{(2)} = 0$. Let $X\subseteq V$ be a set that witnesses $\alpha^{(2)} = 0$, i.e. $\alpha^{(2)} = K_p - p(X)$ and $X \subseteq V - A$. We note that $X \not = \emptyset$, since otherwise we have that $0 = \alpha^{(2)} = K_p - p(\emptyset) \geq K_p - m(\emptyset) = K_p > 0$, a contradiction.  Since $\alpha^{(2)} = 0$, we have that $p(X) = K_p$ and thus, the set $X$ is a $p$-maximizer that is disjoint from the set $A$, contradicting that $A$ contains a transversal for the family of $p$-maximizers by Steps \ref{algstep:SzigetiAlgorithm:def:T} and \ref{algstep:SzigetiAlgorithm:def:A}.

    Next, we recall that $\alpha^{(3)} = \min\{K_p - m(u) : u \in V - A\}$. We note that $\alpha^{(3)} \in \Z$ since $K_p \in \Z$ and $m:V\rightarrow \Z$. Thus, it suffices to show that $\alpha^{(3)} \geq 1$. By way of contradiction, say $\alpha^{(3)} \leq 0$. Since $K_p \geq m(u)$ for all $u \in V$, we have that $\alpha^{(3)} \geq 0$. Thus, $\alpha^{(3)} = 0$. Consequently, there exists a vertex $v \in V$ that witnesses $\alpha^{(3)} = 0$, i.e. $\alpha^{(3)} = 0 = K_p - m(v)$ and $v \in V - A$, contradicting Steps \ref{algstep:SzigetiAlgorithm:def:calD} and \ref{algstep:SzigetiAlgorithm:def:A}.
    
    \item Let $m', m'', \zeros$ be as defined by \Cref{alg:SzigetiAlgorithm}. We note that by Steps \ref{algstep:SzigetiAlgorithm:def:T} and \ref{algstep:SzigetiAlgorithm:def:A}, the set $A$ contains a transversal for the family of $p$-maximizers. If $\emptyset$ is the only $p$-maximizer, then we have that $0 < K_p = p(\emptyset) \leq m(\emptyset) = 0$, a contradiction. Thus, there exists a non-empty $p$-maximizer and consequently $A \not = \emptyset$. Moreover, by part (a) of the current lemma, we have that $\alpha \ge 1$. Then, the following shows the claim.
    $$m''(V - \zeros) = \functionRestrict{m'}{\zeros}(V - \zeros) = m'(V) = m(V) - \alpha\chi_A < m(V).$$

    \item Let $X\subseteq V$ be arbitrary. We will show the claim by induction on $|X|$. We first consider the base case $|X| = 0$. Then, we have that $m(X) \geq p(X) > p(X) - \alpha = p(X) + \alpha|A\cap X| - \alpha$, and so the claim holds.

    Next, we consider the inductive case of $|X| \geq 1$. Let $T$ and $\calD$ be as defined by \Cref{alg:SzigetiAlgorithm} for input tuple $(p, m)$. Let $D := \calD - T$ so that $A = T \uplus D$. We now prove the claim by considering two cases. First, suppose that $D\cap X \not = \emptyset$. Then, we have the following:
    \begin{align*}
        m(X)& \geq m(A\cap X)&\\
        & = m(T\cap X) + m(D\cap X)&\\
        & \geq \alpha|T\cap X| + K_{p}\cdot |D\cap X|&\\
        &\geq K_{p} + \alpha|T\cap X| + \alpha(|D\cap X| - 1)&\\
        &\geq p(X) + \alpha|A\cap X| - \alpha.&
    \end{align*}
    Here, the second inequality is because $m(u) \geq \alpha^{(1)} \geq \alpha$ for each $u \in A$ by Step \ref{algstep:SzigetiAlgorithm:def:alpha} and $m(u)=K_p$ for each $u\in D$. 
    The third inequality is because $K_{p} \geq m(u) \geq \alpha^{(1)} \geq \alpha$ for each $u \in A$ by Step \ref{algstep:SzigetiAlgorithm:def:alpha}. The final inequality is because $K_{p}\geq p(X)$ and $A = T\uplus D$.

    Next, suppose that $D\cap X = \emptyset$. Suppose that $T\cap X = \emptyset$. Then, we have that $A\cap X = \emptyset$ since $A = T\uplus D$. Consequently, we have that
    $m(X) \geq p(X) > p(X) - \alpha = p(X) + \alpha|A\cap X| - \alpha,$ and hence, the claim holds. Here, the first inequality is because $m(Z) \geq p(Z)$ for every $Z \subseteq V$.
    We henceforth assume that $T\cap X \not = \emptyset$. Let $y \in T\cap X$ be an arbitrary element. Since the set $T$ is a minimal transversal of $p$-maximizers, there exists a $p$-maximizer $Y\subseteq V$ such that $T\cap Y = \{y\}$. In particular, we have that $|A\cap (X - Y)| = |A\cap X| - 1$. 
    We now consider two cases based on the skew-supermodularity behavior of the function $p$ at the sets $X$ and $Y$.

    First, suppose that $p(X) + p(Y) \leq p(X\cap Y) + p(X\cup Y)$. Then, we have the following:
    \begin{align*}
        p(X) + p(Y)& \leq p(X\cap Y) + p(X\cup Y)&\\
        &\leq p(X\cap Y) + p(Y)&\\
        & \leq m(X\cap Y) + p(Y)&\\
        & = m(X) - m(X - Y) + p(Y)&\\
        &\leq m(X) - m(A \cap (X - Y)) + p(Y)&\\
        & \leq m(X) - \alpha|A\cap (X - Y)| + p(Y)&\\
        & = m(X)  - \alpha (|A \cap X| - 1) + p(Y).&
    \end{align*}
    Here, the second inequality is because the set $Y$ is a $p$-maximizer. The third inequality is because $m(Z) \geq p(Z)$ for every $Z\subseteq V$. The fourth inequality is because $m$ is a positive function. 
    The fifth inequality is because $m(u) \geq \alpha^{(1)} \geq \alpha$ for each $u\in A$ by Step \ref{algstep:SzigetiAlgorithm:def:alpha}. The final equality is because $|A\cap (X - Y)| = |A\cap X| - 1$. We note that rearranging the terms gives us the claim.
    
   Next, suppose that $p(X) + p(Y) \leq p(X - Y) + p(Y - X)$. Then, we have the following:
   \begin{align*}
       p(X) + p(Y)& \leq p(X - Y) + p(Y - X)&\\
       &\leq p(X - Y) + p(Y)&\\
       & \leq m(X - Y) - \alpha|A\cap (X - Y)| + \alpha + p(Y)&\\
       &= m(X) - m(X \cap Y) - \alpha|A\cap (X - Y)| + \alpha + p(Y)&\\
       &= m(X) - m(X\cap Y) - \alpha|A\cap X| + \alpha|A\cap X\cap Y| + \alpha + p(Y)&\\
       &\leq m(X) - \alpha|A\cap X| + \alpha + p(Y).&
   \end{align*}
   Here, the second inequality is because the set $Y$ is a $p$-maximizer. The third inequality is by the inductive hypothesis for the set $X - Y$: we note that $|X - Y| < |X|$ since $y \in X\cap Y \not = \emptyset$. The final equality is because $|A\cap (X - Y)| = |A\cap X| - |A\cap X\cap Y|$. The final inequality is because $\alpha|A\cap X\cap Y| \leq m(A\cap X\cap Y)\leq m(X\cap Y)$ since $m(u)\geq \alpha^{(1)} \geq \alpha$ for each $u\in A$ by Step \ref{algstep:SzigetiAlgorithm:def:alpha}.

\end{enumerate}
\end{proof}

The next two lemmas will allow us to conclude that if the input functions $p$ and $m$ to a recursive call of \Cref{alg:SzigetiAlgorithm} satisfy conditions  \Cref{thm:SzigetiWeakCover:main}(a) and (b), then the intermediate functions $p', m'$ and the input functions to the subsequent recursive call $p'', m''$ as constructed by \Cref{alg:SzigetiAlgorithm} also  satisfy conditions \Cref{thm:SzigetiWeakCover:main}(a) and (b). Such recursive properties will be useful in inductively proving the existential version of \Cref{thm:SzigetiWeakCover:main}.

\begin{lemma}\label{lem:SzigetiAlgorithm:(p'm')-hypothesis}
Suppose that $V\not=\emptyset$ and the input to  \Cref{alg:SzigetiAlgorithm} is a tuple $(p,m)$, where $p:2^V\rightarrow\Z$ is a skew-supermodular function and $m:V\rightarrow\Z_{+}$ is a positive function such that $m(X) \ge p(X)$ for all $X \subseteq V$ and $m(u) \leq K_p$ for all $u \in V$. Let $\alpha, p', m'$ be as defined by \Cref{alg:SzigetiAlgorithm}. Then, we have that 
\begin{enumerate}[label=(\alph*)]
    \item The function $m':V\rightarrow\Z_{\geq 0}$ is a non-negative function,
    \item $K_{p'} \coloneqq  K_p - \alpha$,
    \item $m'(u) \leq K_{p'}$ for all $u \in V'$, and
    \item $m'(X) \ge p'(X)$ for all $X \subseteq V'$.
\end{enumerate}
\end{lemma}
\begin{proof} We prove each property separately below.
\begin{enumerate}[label=(\alph*)]
    \item Let $u\in V$ be a vertex. If $u \not \in A$, then $m'(u) = m(u) \geq 0$. Otherwise $u\in A$ and we have that $m'(u) = m(u)- \alpha \geq m(u) - \alpha^{(1)} \geq 0$, where the last inequality is by definition of $\alpha^{(1)}$.
    
    \item Let $X \subseteq V$ be a $p$-maximizer and $Y \subseteq V$ be a $p'$-maximizer. 
    By Steps \ref{algstep:SzigetiAlgorithm:def:T} and \ref{algstep:SzigetiAlgorithm:def:A}, the set $A$ is a transversal for the family of $p$-maximizers, and consequently, $|A\cap X| \geq 1$. Thus, $K_{p'}=p'(Y) \geq p'(X) = p(X) - \alpha =K_p-\alpha$. Here, the inequality is because $Y$ is a $p'$-maximizer, and the equality $p'(X) = p(X) - \alpha $ is by the definition of the function $p'$ and because  $|A\cap X| \geq 1$. We now show that $K_{p'} \leq K_p - \alpha$. If the set $Y$ is also a $p$-maximizer, then the claim holds since we have that $K_{p'} = p'(Y) = p(Y) - \alpha = K_{p} - \alpha$, where the second equality is because the set $A$ is a transversal for the family of $p$-maximizers. Thus, we may assume that $Y$ is not a $p$-maximizer. If $|A\cap Y| \geq 1$, then we have that $$K_p-\alpha = p(X) - \alpha =p'(X)\leq K_p'=p'(Y) = p(Y) - \alpha \leq p(X) - \alpha,$$ where the third equality is because $Y$ is a $p'$-maximizer, the fourth equality is because of $|A\cap Y| \geq 1$ and the definition of $p'$, and the final inequality is because the set $X$ is a $p$-maximizer. Thus, all inequalities are equations and the claim holds. Otherwise, suppose that $|A \cap Y| = 0$. Here, we have that $$K_p-\alpha = p(X) - \alpha \geq p(X) - \alpha^{(2)} \geq p(X) - \left(p(X) - p(Y)\right) = p(Y) = p'(Y)=K_{p'},$$
    where the first inequality is by definition of $\alpha$, the second inequality is by definition of $\alpha^{(2)}$ and the fact that the set $X$ is a $p$-maximizer, and the third equality is because of $|A\cap Y| = 0$ and definition of the function $p'$.
    
    \item First, consider a vertex $u \in A$. Then, we have that $m'(u) = m(u) - \alpha \leq K_p - \alpha = K_{p'}$, where the inequality is because $m(v) \leq K_p$ for every $v \in V$, and the final equality is by part (b) of the current lemma. Next, consider a vertex $u  \in V - A$. Thus, $m'(u) = m(u)$. In particular, we have that $K_p - m'(u) = K_p - m(u) \geq \alpha^{(3)} \geq \alpha$. Rearranging the terms, we get that $K_{p'} = K_p - \alpha \geq m'(u)$, where the first equality is again by part (b) of the current lemma.
    \item Let $X\subseteq V$ be arbitrary and let the hypergraph $(H_0, w_0)$ be as defined in Step \ref{algstep:SzigetiAlgorithm:def:Tilde_H-Tilde_w}. First, suppose that $|X\cap A| = 0$. Then, we have that $b_{(H_0, w_0))}(X) = 0$, and consequently $m'(X) = m(X) \geq p(X) = p'(X)$. Thus, the claim holds. Next, suppose that $|X\cap A| \geq 1$. Then, we have the following:
    $$m'(X) = m(X) - \alpha|X\cap A| \geq p(X) - \alpha = p(Z) -  b_{(H_0, w_0)}(X) = p'(X).$$
    Here, the inequality holds by \Cref{lem:SzigetiAlgorithm:properties}(c). The final equality is because $|X\cap A| \geq 1$. 
\end{enumerate}
\end{proof}

\begin{lemma}\label{lem:SzigetiAlgorithm:(p''m'')-hypothesis}
Suppose that $V\not=\emptyset$ and the input to  \Cref{alg:SzigetiAlgorithm} is a tuple $(p,m)$, where $p:2^V\rightarrow\Z$ is a skew-supermodular function and $m:V\rightarrow\Z_{+}$ is a positive function such that $m(X) \ge p(X)$ for all $X \subseteq V$ and $m(u) \leq K_p$ for all $u \in V$. Let $\alpha, \zeros, p'', m''$ be as defined by \Cref{alg:SzigetiAlgorithm}. Then, we have that
\begin{enumerate}[label=(\alph*)]
    \item $p'' = \functionContract{(p-b_{(H_0,w_0)})}{\zeros}$; moreover, the function $p''$ is skew-supermodular,
    \item the function $m'':V - \zeros\rightarrow\Z_{+}$ is a positive function,
    \item $K_{p''} = K_p - \alpha$,
    \item $m''(u) \leq K_{p''}$ for all $u \in V''$, and
    \item $m''(X) \ge p''(X)$ for all $X \subseteq V''$.
\end{enumerate} 
\end{lemma}
\begin{proof} We prove each property separately below.
\begin{enumerate}[label=(\alph*)]
    \item We note that $p'' = \functionContract{(p-b_{(H_0,w_0)})}{\zeros}$ by Step \ref{algstep:SzigetiAlgorithm:def:p'-and-p''}. Moreover, the function $p''$ is skew-supermodular because $b_{(H_0, w_0)}$ is submodular and function contraction is known to preserve skew-supermodularity (see Proposition 4 of \cite{Szi99}).
    \item  We have that $m''(u) = m'(u) > 0$ for each $u \in V - \zeros$, where the inequality is by \Cref{lem:SzigetiAlgorithm:(p'm')-hypothesis}(a) and the definition of the set $\zeros$ in Step \ref{algstep:SzigetiAlgorithm:def:zeros}.
    \item By \Cref{lem:SzigetiAlgorithm:(p'm')-hypothesis}(b), it suffices to show that $K_{p''} = K_{p'}$. Let $X\uplus R_X$ be a $p'$-maximizer, where $R_X \subseteq\zeros$ and $X \subseteq V -\zeros$. Furthermore, let $Y\subseteq V -\zeros$ be a $p''$-maximizer such that $p''(Y) = p'(Y\uplus R_Y)$ where $R_Y \subseteq\zeros$. Then, we have the following:
    $$K_{p'} = p'(X\uplus R_X) \leq \max\left\{p'(X\uplus R') : R'\subseteq\zeros\right\} = p''(X) \leq K_{p''} = p''(Y) = p'(Y\uplus R_Y) \leq K_{p'}.$$
    Thus, all inequalities are equations.
    
    \item Let $u \in V - \zeros$. Then we have that $m''(u) = m'(u) \leq K_{p'} = K_{p''}$, where the inequality is by \Cref{lem:SzigetiAlgorithm:(p'm')-hypothesis}(a) and the final equality is by \Cref{lem:SzigetiAlgorithm:(p'm')-hypothesis}(b) and part (c) of the current lemma.
    
    \item Let $X\subseteq V - \zeros$ be arbitrary. Let $R \subseteq\zeros$ such that $p''(X) = p'(X\uplus R)$. Then, we have that $m''(X) = m'(X\uplus R) \geq p'(X\uplus R) = p''(X)$. Here, the first equality is because $R\subseteq\zeros$ and the inequality is by \Cref{lem:SzigetiAlgorithm:(p'm')-hypothesis}(a).
\end{enumerate}
\end{proof}

We conclude the section with the main result of the section.

\begin{lemma}\label{lem:SzigetiAlgorithm:termination-and-correctness}
 Suppose that the input to  \Cref{alg:SzigetiAlgorithm} is a tuple $(p,m)$, where $p:2^V\rightarrow\Z$ is a skew-supermodular function and $m:V\rightarrow\Z_{+}$ is a positive function such that $m(X) \ge p(X)$ for all $X \subseteq V$ and $m(u) \leq K_p$ for all $u \in V$.
Then, \Cref{alg:SzigetiAlgorithm} terminates within a finite (pseudo-polynomial) number of recursive calls. Furthermore, the hypergraph $\left(H = \left(V, E\right), w:E\rightarrow\Z_+\right)$ returned by \Cref{alg:SzigetiAlgorithm} 
satisfies the following three properties:
\begin{enumerate}
    \item $b_{(H,w)}(X) \geq p(X)$ for all $X\subseteq V$,
    \item $b_{(H, w)}(u) = m(u)$ for all $u \in V$, and
    \item $\sum_{e\in E}w(e) = K_p$.
\end{enumerate}
\end{lemma}
\begin{proof} 
We prove the lemma by induction on the potential function $\phi(m) \coloneqq  m(V)$. We note that $\phi(m) \geq 0$ since $m$ is a positive function. For the base case of induction, suppose that $\phi(m) = 0$. Then, we have that $V = \emptyset$ since $m$ is a positive function. Consequently, \Cref{alg:SzigetiAlgorithm} is in its base case (Step \ref{algstep:SzigetiAlgorithm:base-case}) and terminates. Moreover, Step \ref{algstep:SzigetiAlgorithm:base-case} returns an empty hypergraph which satisfies properties (1)-(3) and so the lemma holds. 

For the inductive case, suppose that $\phi(m) > 0$. Since $m$ is a positive function, we have that $V \not = \emptyset$ and so \Cref{alg:CoveringAlgorithm} is in its recursive case. Thus, by \Cref{lem:SzigetiAlgorithm:(p''m'')-hypothesis}(a), (b), (d) and (e), the tuple $(p'', m'')$ constructed by \Cref{alg:SzigetiAlgorithm} satisfy the hypothesis of the current lemma. We note that the tuple $(p'', m'')$ is the input to the subsequent recursive call of \Cref{alg:SzigetiAlgorithm} by Step \ref{algstep:SzigetiAlgorithm:recursion}. Moreover, $\phi(m'') < \phi(m)$ by \Cref{lem:SzigetiAlgorithm:properties}(b). Thus, by induction, the subsequent recursive call to \Cref{alg:SzigetiAlgorithm} terminates and returns a hypergraph $(H'' = (V'', E''), w'')$ satisfying properties (1)-(3) for the tuple $(p'', m'')$. Consequently, the entire execution of \Cref{alg:SzigetiAlgorithm} terminates within a finite number of recursive calls. Let $(H = (V, E), w)$ be the hypergraph returned by Step \ref{algstep:SzigetiAlgorithm:recursion} and let $(H_0, w_0), \zeros, A, \alpha$ be as defined by \Cref{alg:SzigetiAlgorithm} for the input tuple $(p, m)$. We first show that the hypergraph $(H, w)$ satisfies property (1). Let $X\subseteq V$ be arbitrary. Then, we have that 
\begin{align*}
    b_{(H,w)}(X)& = b_{(H'', w'')}(X -\zeros) + b_{(H_0, w_0)}(X) &\\
    &\geq p''(X -\zeros)  + b_{(H_0, w_0)}(X)&\\
    &\geq p'(X) +  b_{(H_0, w_0)}(X)&\\
    &= p(X).&
\end{align*}
Here, the first inequality holds by the inductive hypothesis property (1) and the second inequality is because $p'' = \functionContract{p'}{\zeros}$ by Step \ref{algstep:SzigetiAlgorithm:def:p'-and-p''}.
Next, we show that hypergraph $(H, w)$ satisfies property (2). Let $u \in V$ be arbitrary. Then, we have the following:
\begin{align*}
    b_{(H, w)}(u) & = b_{(H'',w'')}(u) + b_{(H_0, w_0)}(u) &\\
    &= b_{(H'',w'')}(u) + \alpha\chi_A(u)&\\
    & = m''(u) +\alpha\chi_A(u)&\\
    &= m(u).&
\end{align*}
Here, the third equality is by the inductive hypothesis property (2). Next, we show that the hypergraph $(H, w)$ satisfies property (3). We have the following:
$$\sum_{e\in E}w(e) = \sum_{e \in E''}w(e) + w(A) = \sum_{e \in E''}w''(e) + \alpha = K_{p''} + \alpha = K_p.$$
Here, the third equality is by induction hypothesis property (3), and the final equality is by \Cref{lem:SzigetiAlgorithm:(p''m'')-hypothesis}(c). 
\end{proof}

\subsection{Recursion Depth and Hypergraph Support Size}\label{sec:SzigetiAlgorithm:recursion-depth-and-hypergraph-support-size}
In this section, we give an upper bound on  the number of recursive calls witnessed by an execution of \Cref{alg:SzigetiAlgorithm}. The main result of this section is \Cref{lem:SzigetiAlgorithm:recursion-depth-and-support-size} which we show at the end of the section.
We begin with some convenient notation that will be used in this section as well as the next section. 

\paragraph{Notation.} 
By \Cref{lem:SzigetiAlgorithm:termination-and-correctness}, the number of recursive calls made by \Cref{alg:SzigetiAlgorithm} is finite.
We will use $\ell$ to denote the depth of recursion. We will refer to the recursive call at depth $i \in [\ell]$ as \emph{recursive call $i$} or the \emph{$i^{th}$ recursive call}. We let $V_i$ denote the ground set at the start of recursive call $i$, and $p_i:2^{V_i}\rightarrow\Z$ and $m_i:V_i\rightarrow\Z_{\geq 0}$ denote the input functions to recursive call $i$. Furthermore, we let $\calD_i, T_i, A_i, \alpha_i, \alpha_i^{(1)}, \alpha_i^{(2)}, \alpha_i^{(3)}, \zeros_i, (H^i_0, w^i_0), p'_i, m'_i, p''_i, m''_i$ denote the quantities $\calD, T, A, \alpha, \alpha^{(1)}, \alpha^{(2)}, \alpha^{(3)}, \zeros, (H_0, w_0),$ $p', m', p'', m''$ as defined by \Cref{alg:SzigetiAlgorithm} during the $i^{th}$ recursive call for $i \in [\ell - 1]$. For convenience, we also define the relevant sets, values and functions during the base case ($\ell^{th}$ recursive call) as follows: $A_\ell, \calD_\ell, \zeros_\ell := \emptyset$, $\alpha_\ell,\alpha_\ell^{(1)}, \alpha_\ell^{(2)}, \alpha_\ell^{(3)}:=0$, $m_\ell', m_\ell'',m_{\ell+1} := m_\ell$, and $p_\ell', p_{\ell}'',p_{\ell+1} := p_{\ell}$. Finally, let
$\left(H_i = \left(V_i, E_i\right), w_i\right)$ denote the hypergraph returned by the $i^{th}$ recursive call. 

We note that \Cref{lem:SzigetiAlgorithm:(p''m'')-hypothesis} and induction on the recursion depth $i$ immediately imply the following lemma which says that for every $i \in [\ell]$, the input tuple $(p_i, m_i)$ satisfies the hypothesis of \Cref{thm:SzigetiWeakCover:main}.

\begin{lemma}\label{lem:SzigetiAlgorithm:(p_im_i)-hypothesis}
     Suppose that the input to  \Cref{alg:SzigetiAlgorithm} is a tuple $(p_1,m_1)$, where $p_1:2^{V_1}\rightarrow\Z$ is a skew-supermodular function and $m_1:V_1\rightarrow\Z_{+}$ is a positive function such that $m_1(X) \ge p_1(X)$ for all $X \subseteq V_1$ and $m_1(u) \leq K_{p_1}$ for all $u \in V_1$. Let $\ell \in \Z_+$ be the number of recursive calls witnessed by the execution of \Cref{alg:SzigetiAlgorithm} and, for all $i\in[\ell]$, let $(p_i, m_i)$ be the input tuple to the $i^{th}$ recursive call of the execution. Then, for all $i \in [\ell]$ we have that $p_i:2^{V_i}\rightarrow\Z$ is a skew-supermodular function and $m_i:V_i\rightarrow\Z_+$ is a positive function such that $m_i(X) \ge p_i(X)$ for all $X \subseteq V_i$ and $m_i(u) \leq K_{p_i}$ for all $u \in V_i$.
\end{lemma}

\paragraph{Set Families.} For our analysis, we will focus on certain set families associated with an execution of \Cref{alg:SzigetiAlgorithm}. Let $i \in [\ell]$ be a recursive call of \Cref{alg:SzigetiAlgorithm}. We define $\cumulativeZeros{i} \coloneqq  \cup_{j \in [i]}\zeros_{i}$. We use $\minimalMaximizerFamily{i}$ and $ \minimalMaximizerFamilyAfterCreate{i}$ to denote the families of \emph{minimal} $p_i$-maximizers and $p_i'$-maximizers respectively, i.e., $\minimalMaximizerFamily{i}$ is the collection of inclusionwise minimal sets in the family $\{X\subseteq V_i: p_i(X)=K_{p_i}\}$ and $\minimalMaximizerFamilyAfterCreate{i}$ is the collection of inclusionwise minimal sets in the family $\{X\subseteq V_i: p_i'(X)=K_{p_i'}\}$. \Cref{lem:SzigetiAlgorithm:persistance-of-maximizers} below shows the progression of these families across recursive calls of an execution of \Cref{alg:SzigetiAlgorithm}. We will also be interested in families of \emph{all} minimal maximizers of the input functions witnessed by the algorithm up to a given recursive call. Formally, we define the family $\cumulativeMinimalMaximizerFamily{i}\coloneqq  \cup_{j \in [i]} \minimalMaximizerFamily{j}$. 
\Cref{lem:SzigetiAlgorithm:progression-of-set-families:main} below summarizes useful properties of the stated families. 

\begin{lemma}\label{lem:SzigetiAlgorithm:persistance-of-maximizers}
   Suppose that the input to  \Cref{alg:SzigetiAlgorithm} is a tuple $(p_1,m_1)$, where $p_1:2^V\rightarrow\Z$ is a skew-supermodular function and $m_1:V\rightarrow\Z_{+}$ is a positive function such that $m_1(X) \ge p_1(X)$ for all $X \subseteq V_1$ and $m_1(u) \leq K_{p_1}$ for all $u \in V_1$. Let $\ell \in \Z_+$ be the number of recursive calls witnessed by the execution of \Cref{alg:SzigetiAlgorithm} and, for all $i\in[\ell]$, let $(p_i, m_i)$ be the input tuple to the $i^{th}$ recursive call of the execution. Then, for all $i \in [\ell]$, we have the following:
   \begin{enumerate}[label=(\alph*)]
       \item if $Y \subseteq V_i$ is a $p_i$-maximizer, then $Y$ is also a $p_{i}'$-maximizer, and 
     \item if $Y \subseteq V_i$ is a $p_i'$-maximizer such that $Y - \zeros_i \not = \emptyset$, then $Y - \zeros_i$ is a $p_{i+1}$-maximizer.
   \end{enumerate}
\end{lemma}
\begin{proof}
     We note that both properties trivially hold if $i = \ell$, and so we consider the case where when $i \leq \ell - 1$ (i.e. \Cref{alg:SzigetiAlgorithm} is in its recursive case). We recall that by \Cref{lem:SzigetiAlgorithm:(p_im_i)-hypothesis}, $p_i$ is a skew-supermodular function and $m_i$ is a positive function such that $m_i(X) \ge p_i(X)$ for all $X \subseteq V_i$ and $m_i(u) \leq K_{p_i}$ for all $u \in V_i$. We prove both parts separately below.
    \begin{enumerate}[label=(\alph*)]
        \item Let $Y \subseteq V_i$ be a $p_i$-maximizer. We recall that by Steps \ref{algstep:SzigetiAlgorithm:def:T} and \ref{algstep:SzigetiAlgorithm:def:A} during the $i^{th}$ recursive call, the set $A_i$ is a transversal for the family of $p_i$-maximizers. Thus, we have that $p_i'(Y) = p_i(Y) - \alpha_i = K_{p_i} - \alpha_i = K_{p'_i}$, where the first equality is because $A_i\cap Y \not = \emptyset$ and the third equality is by \Cref{lem:SzigetiAlgorithm:(p_im_i)-hypothesis} and \Cref{lem:SzigetiAlgorithm:(p'm')-hypothesis}(b).
        \item Let $Y \subseteq V_i$ be a $p_i'$-maximizer such that $Y - \zeros_i\neq \emptyset$. We have the following:
    $$K_{p_{i+1}} \geq p_{i+1}(Y - \zeros_i) = \max\{p_i'(Y\cup R):R\subseteq \zeros_i\} \geq p_i'(Y) = K_{p'_i} = K_{p_{i+1}},$$
    where the final equality is by \Cref{lem:SzigetiAlgorithm:(p'm')-hypothesis}(b) and \Cref{lem:SzigetiAlgorithm:(p''m'')-hypothesis}(c). Thus, all inequalities are equations and we have that $p_{i+1}(Y - \zeros_i) = K_{p_{i+1}}$.
    \end{enumerate}
\end{proof}

\begin{lemma}\label{lem:SzigetiAlgorithm:progression-of-set-families:main}
Suppose that the input to  \Cref{alg:SzigetiAlgorithm} is a tuple $(p_1,m_1)$, where $p_1:2^V\rightarrow\Z$ is a skew-supermodular function and $m_1:V\rightarrow\Z_{+}$ is a positive function such that $m_1(X) \ge p_1(X)$ for all $X \subseteq V_1$ and $m_1(u) \leq K_{p_1}$ for all $u \in V_1$. Let $\ell \in \Z_+$ be the number of recursive calls witnessed by the execution of \Cref{alg:SzigetiAlgorithm} and, for all $i\in[\ell]$, let $(p_i, m_i)$ be the input tuple to the $i^{th}$ recursive call of the execution. Then, for all $i \in [\ell-1]$, we have the following:
    \begin{enumerate}[label=(\alph*)]
        \item 
   $\zeros_{\leq i} \subseteq \zeros_{\leq i+1}$; furthermore, $\alpha_i = \alpha_{i}^{(1)}$ if and only if $\zeros_{i} \not = \emptyset$ (i.e., $\zeros_{\leq i} \subsetneq \zeros_{\leq i+1}$),
   \item $\cumulativeMinimalMaximizerFamily{i}\subseteq  \cumulativeMinimalMaximizerFamily{i+1}$; furthermore, if $\alpha_{i} = \alpha_{i}^{(2)} < \alpha_{i}^{(1)}$, then $\cumulativeMinimalMaximizerFamily{i}\subsetneq 
   \cumulativeMinimalMaximizerFamily{i+1}$,
   \item 
   $\calD_{i}\subseteq \calD_{i}' \subseteq \calD_{ i+1}$; furthermore, if $\alpha_i = \alpha_{i}^{(3)}$, then $\calD_{i}\subsetneq  \calD_{i}'$.
    \end{enumerate}
\end{lemma}
\begin{proof}  We recall that by \Cref{lem:SzigetiAlgorithm:(p_im_i)-hypothesis}, $p_i$ is a skew-supermodular function and $m_i$ is a positive function such that $m_i(X) \ge p_i(X)$ for all $X \subseteq V_i$ and $m_i(u) \leq K_{p_i}$ for all $u \in V_i$. We prove each property separately below.
    \begin{enumerate}[label=(\alph*)]
        \item We have that $\calZ_{\leq i} \subseteq \calZ_{\leq i+1}$ by definition. We now show the second part of the claim. The reverse direction follows because the function $m$ is a positive function. For the forward direction, suppose that $\alpha_i = \alpha_i^{(1)}$. Let $u \in A_i$ be a vertex  that witnesses $\alpha_i = \alpha_i^{(1)}$, i.e. $m_i(u) = \alpha_i$. Then, we have that $m_{i}'(u) = m_i(u) - \alpha_i\chi_{A_i}(u) = 0$. Thus, $u \in \zeros_{i}$. 

        \item We note that $\cumulativeMinimalMaximizerFamily{i} \subseteq \cumulativeMinimalMaximizerFamily{i+1}$ follows by definition.  We now show the second part of the claim. Suppose that $\alpha_i = \alpha_i^{(2)} < \alpha_{i}^{(1)}$. Then by part (a) of the current lemma, we have that $\zeros_i = \emptyset$. Consequently $p_i' = p_{i+1}$ by the definition of the two functions. We consider the family $ \minimalMaximizerFamilyAfterCreate{i}$ of minimal $p_i'$-maximizers. 
        By \Cref{lem:SzigetiAlgorithm:persistance-of-maximizers}(a), we have that $\minimalMaximizerFamily{i}\subseteq \minimalMaximizerFamilyAfterCreate{i} = \minimalMaximizerFamily{i+1}$, and thus, $\cumulativeMinimalMaximizerFamily{i} \subseteq \cumulativeMinimalMaximizerFamily{i} \cup \minimalMaximizerFamilyAfterCreate{i} = \cumulativeMinimalMaximizerFamily{i+1}$, where the equalities in both the previous expressions are because $p_i' = p_{i+1}$. We now show that the first inclusion is strict. For convenience, we let $\cumulativeMinimalMaximizerFamilyAfterCreate{i} \coloneqq  \cumulativeMinimalMaximizerFamily{i} \cup \minimalMaximizerFamilyAfterCreate{i}$. 
        By way of contradiction, suppose that $\cumulativeMinimalMaximizerFamily{i} = \cumulativeMinimalMaximizerFamilyAfterCreate{i}$. 
        Let $X \subseteq V_i - A_i$ be a set such that  $\alpha_i = \alpha_{i}^{(2)} = K_{p_i} - p_i(X)$ (such a set exists since $\alpha_i=\alpha_i^{(2)}$). Then, we have that $p_{i}'(X) = p_i(X) = K_{p_i}  - \alpha_i^{(2)} = K_{p_i'}.$
        Here, the first equality is because $X \subseteq V_i - A_i$, the second equality is because our choice of $X$ satisfies $\alpha_{i}^{(2)} = K_{p_i} - p_i(X)$, and the final inequality is by \Cref{lem:SzigetiAlgorithm:(p'm')-hypothesis}(b). Thus, the set $X$ is a $p_{i}'$-maximizer. Furthermore, $X$ is not a $p_i$-maximizer since $A_i$ is a transversal for the family of $p_i$-maximizers by Step \ref{algstep:SzigetiAlgorithm:def:A}, but $A_i\cap X = \emptyset$.
        Consequently, there exists a set $Y \subseteq X$ such that $Y \in \minimalMaximizerFamilyAfterCreate{i} - \minimalMaximizerFamily{i}$. Since $\cumulativeMinimalMaximizerFamily{i} = \cumulativeMinimalMaximizerFamilyAfterCreate{i}$, we have that $Y \in \cumulativeMinimalMaximizerFamily{i-1}$. In particular, there exists a recursive call $j \in [i-1]$ such that $Y \in \minimalMaximizerFamily{j}$. Since $Y\subseteq V_i$, by \Cref{lem:SzigetiAlgorithm:persistance-of-maximizers}(a), (b) and induction on $j$, we have that $Y \in \minimalMaximizerFamily{i}$. Thus, the set $Y$ is a $p_i$-maximizer and so $A_i \cap Y \not = \emptyset$ by Step \ref{algstep:SzigetiAlgorithm:def:A}, a contradiction to $A_i \cap X = \emptyset$.

        \item First, we show that $\calD_{i} \subseteq \calD_{i}'$. Let $u \in \calD_{i}$. Then we have the following:
        $$m_{i}(u)' = m_i(u) - \alpha_i\chi_{A_i}(u) = K_{p_i} - \alpha_i\chi_{A_i}(u) = K_{p_i} - \alpha_i  = K_{p_{i}'}.$$
        Here, the second equality is because $u \in \calD_i$, the third equality is because $\calD_i \subseteq A_i$, and the final equality is by \Cref{lem:SzigetiAlgorithm:(p'm')-hypothesis}(b). Thus, $u \in \calD_{i}'$ and we have that $\calD_i \subseteq \calD_{i}'$. Next, we show that $\calD_{i}'\subseteq \calD_{i+1}$. Let $u \in\calD_i'$. We note that if $u \in \zeros_i$, then we have that $0 = m_i'(u) = K_{p_i'}$, and thus $i = \ell$, contradicting $i \in [1, \ell - 1]$. Thus, $u \not \in \zeros_i$, i.e. $u \in V_{i+1}$.  Then, we have that $m_{i+1}(u) = m_{i}'(u) = K_{p_i'} = K_{p_{i+1}}$, where the final equality is by \Cref{lem:SzigetiAlgorithm:(p'm')-hypothesis}(b) and \Cref{lem:SzigetiAlgorithm:(p''m'')-hypothesis}(c). Thus, $u \in \calD_{i+1}$ and we have that $\calD_{i}'\subseteq \calD_{i+1}$. 
        
        We now show the second part of the claim. Suppose that $\alpha_{i} = \alpha_i^{(3)}$ and let $v \in V_i - A_i$ be the vertex which witnesses $\alpha_i = \alpha_i^{(3)}$, i.e. $\alpha_i = K_{p_i} - m_i(v)$. Then, we have that:
        $$m_{i}'(v) = m_i(v) - \alpha_i\chi_{A_i}(v) = m_i(v) = K_{p_i} - \alpha_i = K_{p_{i+1}}.$$
        Here, the second equality is because $v \in V_i - A_i$. The third equality is because the vertex $v$ witnesses $\alpha_i = \alpha_{i}^{(3)}$. The final equality is by \Cref{lem:SzigetiAlgorithm:(p''m'')-hypothesis}(c). Thus, $v\in \calD_{i+1} \backslash \calD_{i}$.
    \end{enumerate}
\end{proof}

The next lemma shows how the function input to an arbitrary recursive call is related to the function input to the first recursive call. This expression will be useful in implementing \Cref{alg:SzigetiAlgorithm} since it  will allow us to construct function maximization oracle for $p_i$ using only polynomially many queries to the function maximization oracle of $p_1$.

\begin{lemma}\label{lem:SzigetiAlgorithm:p_i-from-p_1}
Suppose that the input to  \Cref{alg:SzigetiAlgorithm} is a tuple $(p_1,m_1)$, where $p_1:2^V\rightarrow\Z$ is a skew-supermodular function and $m_1:V\rightarrow\Z_{+}$ is a positive function such that $m_1(X) \ge p_1(X)$ for all $X \subseteq V_1$ and $m_1(u) \leq K_{p_1}$ for all $u \in V_1$. Let $\ell \in \Z_+$ be the number of recursive calls witnessed by the execution of \Cref{alg:SzigetiAlgorithm}. For $i\in[\ell]$, let $(p_i, m_i)$ be the input tuple to the $i^{th}$ recursive call; moreover, for $i \in [\ell - 1]$, let $\zeros_i, (H^i_0, w^i_0)$ be as defined by \Cref{alg:SzigetiAlgorithm} for the input $(p_i, m_i)$.
Then, for each $i \in [\ell]$ we have that
$$p_i = \functionContract{\left(p_1 - \sum_{j \in [i-1]}b_{(\Tilde{H}_j, \Tilde{w}_{j})}\right)}{\bigcup_{j \in [i - 1]}\zeros_j},$$
where $(\Tilde{H}_0^{i}, \Tilde{w}_0^{i})$ denotes the hypergraph obtained by adding the vertices $\cup_{j\in [i-1]}\zeros_j$ to the hypergraph $(H^i_0, w^i_0)$.
\end{lemma}
\begin{proof}
    We show the lemma by induction on $i$. The base case $i = 1$ trivially holds. For induction, suppose that $i \geq 2$. Then, we have the following:
    \begin{align*}
        p_i = \functionContract{\left(p_{i-1} - b_{({H}^{i-1}_0, {w}^{i-1}_0)}\right)}{\zeros_{i-1}}& = \functionContract{\left(\functionContract{\left(p_1 - \sum_{j \in [i-2]}b_{(\Tilde{H}_0^j, \Tilde{w}_0^{j})}\right)}{\bigcup_{j \in [i - 2]}\zeros_j}- b_{({H}^{i-1}_0, {w}^{i-1}_0)}\right)}{\zeros_{i-1}}&\\
        & = \functionContract{\left(\functionContract{\left(p_1 - \sum_{j \in [i-2]}b_{(\Tilde{H}_0^j, \Tilde{w}_0^{j})}\right)}{\bigcup_{j \in [i - 2]}\zeros_j}- b_{(\Tilde{H}^{i-1}_0, \Tilde{w}^{i-1}_0)}\right)}{\zeros_{i-1}}&\\
        & = \functionContract{\left(\functionContract{\left(p_1 - \sum_{j \in [i-1]}b_{(\Tilde{H}_0^j, \Tilde{w}_0^{j})}\right)}{\bigcup_{j \in [i - 2]}\zeros_j}\right)}{\zeros_{i-1}}&
    \end{align*}
    Here, the first equality is by Step \ref{algstep:SzigetiAlgorithm:def:p'-and-p''}, the second equality is by the inductive hypothesis, and the third and fourth equalities are because $b_{(\Tilde{H}_0^{i-1}, \Tilde{w}_0^{i-1})}(X) = b_{({H}_0^{i-1}, {w}_0^{i-1})}(X \cap V_{i-1})$ for all $X \subseteq V$ by definition of the coverage function, where $V_i$ is the ground set at the $i^{th}$ recursive call. Simplifying the final RHS gives the claimed expression.
\end{proof}

We now show the main result of the section which says that an execution of \Cref{alg:SzigetiAlgorithm} witnesses $O(|V|)$ recursive calls. Since every recursive call adds at most one new hyperedge to the hypergraph returned by the execution, this also implies that the number of hyperedges in a solution returned by \Cref{alg:SzigetiAlgorithm} is $O(|V|)$. This shows property (4) of \Cref{thm:SzigetiWeakCover:main}.

\begin{lemma}\label{lem:SzigetiAlgorithm:recursion-depth-and-support-size}
   Suppose that $V_1 \not=\emptyset$ and the input to  \Cref{alg:SzigetiAlgorithm} is a tuple $(p_1,m_1)$, where $p_1:2^{V_1}\rightarrow\Z$ is a skew-supermodular function and $m_1:V_1\rightarrow\Z_{+}$ is a positive function such that $m_1(X) \ge p_1(X)$ for all $X \subseteq V_1$ and $m_1(u) \leq K_{p_1}$ for all $u \in V_1$. Let $\ell \in \Z_+$ be the number of recursive calls witnessed by the execution of \Cref{alg:SzigetiAlgorithm}. Then, 
   \begin{enumerate}
    \item the recursion depth $\ell$ of \Cref{alg:SzigetiAlgorithm} is at most $4|V_1| - 1$, and
       \item the number of hyperedges in the hypergraph $(H=(V_1, E), w)$ returned by \Cref{alg:SzigetiAlgorithm} is at most $4|V_1| - 1$.
   \end{enumerate}
\end{lemma}
\begin{proof}
    We note that part (2) of the current lemma follows from part (1) since every recursive call adds at most one new hyperedge to the hypergraph returned by the execution in Step \ref{algstep:SzigetiAlgorithm:recursion}. Thus, it suffices to show part (1). We define a potential function $\phi:[\ell]\rightarrow\Z_{\geq 0}$ as follows: for each $i \in [\ell]$,
    $$\phi(i) \coloneqq  |\zeros_{\leq i}| + |\cumulativeMinimalMaximizerFamily{i}| + |\calD_{i}|.$$
    By \Cref{lem:SzigetiAlgorithm:progression-of-set-families:main}, we have that $\phi$ is a monotone increasing function, since each of the three terms is non-decreasing and at least one of the three terms strictly increases with increasing $i \in [\ell]$. Consequently, the number of recursive calls witnessed by the execution of \Cref{alg:SzigetiAlgorithm} is at most $\phi(\ell) - \phi(0) \leq |\zeros_{\leq \ell}| + |\calF_{{\leq \ell}}| +|\calD_{\ell}|\leq 2|V| + |\calF_{{\leq \ell}}| \leq 4|V| - 1$. Here, the final inequality is because the family $\calF_{p_{\leq \ell}}$ is laminar by the following: let $(p_i, m_i)$ denote the input to the $i^{th}$ recursive call of \Cref{alg:SzigetiAlgorithm}. By  \Cref{lem:SzigetiAlgorithm:(p''m'')-hypothesis}(a) and induction on $i$, we have that $p_{i+1} = \functionContract{(p_{i} - b_{(H_0^{i}, w_0^{i})})}{\zeros_i}$ for every recursive call $i \in [\ell - 1]$. Then, laminarity of the family $\calF_{p_{\leq \ell}}$ is obtained by applying \Cref{lem:UncrossingProperties:Cumulative-Minimal-p-Maximizer-Family-Laminar} to the sequence of functions $p_1, \ldots, p_{\ell}$.
\end{proof}

\subsection{Runtime}\label{sec:SzigetiAlgorith:runtime-skew-supermodular-functions}
In this section, we show that \Cref{alg:SzigetiAlgorithm} can be implemented to run in strongly polynomial time provided that the input function $p$ is skew-supermodular. 
\Cref{lem:SzigetiAlgorithm:runtime} is the main lemma of this section. Let $i \in [\ell - 1]$ be a recursive call of an execution of \Cref{alg:SzigetiAlgorithm}.
In order to show that the $i^{th}$ recursive call can be implemented in polynomial time, we will require the ability to compute a minimal transversal for the family of $p_i$-maximizers in strongly polynomial time (Step \ref{algstep:SzigetiAlgorithm:def:T}). We recall that the function $p_i$ is skew-supermodular by \Cref{lem:SzigetiAlgorithm:(p_im_i)-hypothesis}. Consequently, $\calF_{p_i}$ is a disjoint family by \Cref{lem:UncrossingProperties:calFp-disjoint}. Thus, in order to compute a transversal for the family $\calF_{p_i}$ in strongly polynomial time, it suffices to explicitly compute the family $\calF_{p_i}$ itself in strongly polynomial time, and pick an arbitrary vertex from each set of the family. We note that disjointness of the family guarantees that it contains at most a linear number of sets. The next lemma shows that $\calF_{p_i}$ can be computed in strongly polynomial time using the function maximization oracle for $p_i$.

\begin{lemma}\label{lem:SzigetiAlgorithm:computingCalFp}
     Let $p:2^V\rightarrow\Z$  be a skew-supermodular function. Then, the family $\minimalMaximizerFamily{p}$ can be computed in $O(|V|^3)$ time and $O(|V|^2)$ queries to \functionMaximizationEmptyOracle{p}. The run-time includes the time to construct the input to the queries to \functionMaximizationEmptyOracle{p}.
\end{lemma}
\begin{proof}  We describe a recursive procedure that takes a $p$-maximizer $X \subseteq V$ as input and returns a minimal $p$-maximizer contained in $X$. For every $u \in X$, let $Z^u$ be the set returned by querying $\functionMaximizationEmptyOracle{p}(S_0 := \emptyset, T_0 := (V - X) \cup \{u\}, y_0 := \{0\}^V)$. If there exists a vertex $u \in X$ such that the set $Z^u$ is a $p$-maximizer, then recursively compute a minimal $p$-maximizer contained in $Z^{u}$. Otherwise, return the set $X$. 
    
    The correctness of the above procedure is because of the following. The above procedure terminates (in at most $|V|$ recursive calls) because the cardinality of the input set $X$ strictly decreases during each recursive call (i.e., $|X'| > |X|$). Let $Y \subseteq X$ be a minimal $p$-maximizer contained in the set $X$. If $Y = X$, then there are no $p$-maximizers that are strictly contained in the set $X$. Here, the procedure correctly returns the set $X$. Alternatively, suppose that $Y \subsetneq X$ and let $u \in X - Y$. Then, there exists a $p$-maximizer $Z$ such that $Y \subseteq Z \subseteq X - \{u\}$. Consequently, the set $Z^u$ computed by the procedure will be a $p$-maximizer that is strictly contained in the set $X$. Using this fact, the procedure can be shown to return a minimal $p$-maximizer contained in the set $X$ by induction on $|X|$.
    
We now describe a slight modification of the above procedure that leads to an improved runtime and a smaller number of queries to \functionMaximizationEmptyOracle{p}. Let $u \in X$ be a vertex. If the set $Z^{u}$ is a $p$-maximizer, then we include $u$ in the input set $T_0$ to all $\functionMaximizationEmptyOracle{p}$ queries in subsequent recursive calls of the procedure. Alternatively, if the set $Z^{u}$ is not a $p$-maximizer, then we include $u$ in the input set $S_0$ to all $\functionMaximizationEmptyOracle{p}$ queries in subsequent recursive calls procedure. Consequently, the number of queries to $\functionMaximizationEmptyOracle{p}$ (and number of recursive calls of the procedure) is at most $|V|$. We note that constructing the input to each query and verifying whether the set returned by the query is a $p$-maximizer can be done in $O(|V|)$ time. Thus, the procedure terminates in $O(|V|^2)$ time. Here, the runtime includes the time to construct the inputs to the \functionMaximizationEmptyOracle{p} queries.

We note that by \Cref{lem:UncrossingProperties:calFp-disjoint}, the family $\calF_p$ is disjoint. Moreover, one can compute a set from the family $\calF_p$ using the procedure from above. Thus, the entire family $\calF_p$ can be computed by iteratively applying the procedure from above -- each time a set from the family $\calF_p$ is found, it will be included in the input set $T_0$ to all $\functionMaximizationEmptyOracle{p}$ queries in subsequent iterations of the procedure. The lemma follows because each iteration of this new procedure runs in $O(|V|^2)$ time and uses at most $|V|$ queries to $\functionMaximizationEmptyOracle{p}$; moreover, there are at most $O(|V|)$ iterations of this new procedure since $|\calF_p| = O(|V|)$ by the disjointness of $\calF_p$.
\end{proof}

We now show the main lemma of the section which says that \Cref{alg:SzigetiAlgorithm} can be implemented in strongly polynomial time and polynomial number of calls to \functionMaximizationOracleStrongCover{p}. We remark that the maximization oracle used in Lemma \ref{lem:SzigetiAlgorithm:runtime} is stronger than the one used in Lemma \ref{lem:SzigetiAlgorithm:computingCalFp}. This stronger version is needed for the input function $p$ in order to be able to construct the required \functionMaximizationEmptyOracle{p_i} for all recursive calls $i \in [\ell]$.

\begin{lemma}\label{lem:SzigetiAlgorithm:runtime}
    Suppose that the input to  \Cref{alg:SzigetiAlgorithm} is a tuple $(p,m)$, where $p:2^{V}\rightarrow\Z$ is a skew-supermodular function and $m:V\rightarrow\Z_{+}$ is a positive function such that $m(X) \ge p(X)$ for all $X \subseteq V$ and $m(u) \leq K_{p}$ for all $u \in V$. Then, \Cref{alg:SzigetiAlgorithm}  can be implemented to run in $O(|V|^5)$ time using $O(|V|^4)$ queries to \functionMaximizationOracleStrongCover{p}. The run-time includes the time to construct the hypergraphs used as input to the queries to \functionMaximizationOracleStrongCover{p}.
Moreover, for each query to \functionMaximizationOracleStrongCover{p}, the hypergraph $(G_0, c_0)$ used as input to the query has $|V|$ vertices and $O(|V|)$ hyperedges.
\end{lemma}
\begin{proof}
We recall that $\ell \in \Z_+$ denotes the number of recursive calls witnessed by the execution of \Cref{alg:SzigetiAlgorithm} and that $(p_i, m_i)$ denotes the input tuple to the $i^{th}$ recursive call of the execution. Let $i\in [\ell]$ be a recursive call. 
 The next claim shows that all steps of \Cref{alg:SzigetiAlgorithm} except for Steps \ref{algstep:SzigetiAlgorithm:def:p'-and-p''} and \ref{algstep:SzigetiAlgorithm:recursion} in the $i^{th}$ recursive call can be implemented to run in strongly polynomial time and strongly polynomial number of queries to \functionMaximizationEmptyOracle{p_i}.
   \begin{claim}\label{claim:SzigetiAlgorithm:runtime:all-steps-except-recursion}
        All steps except for Steps \ref{algstep:SzigetiAlgorithm:def:p'-and-p''} and \ref{algstep:SzigetiAlgorithm:recursion} in the $i^{th}$ recursive call can be implemented to run in $O(|V_i|^3)$ time using $O(|V_i|^2)$ queries to \functionMaximizationEmptyOracle{p_i}. The run-time includes the time to construct the inputs for the queries to \functionMaximizationEmptyOracle{p_i}.
    \end{claim}
    \begin{proof}
    We note that the value $K_{p_i}$ can be computed by querying $\functionMaximizationEmptyOracle{p_i}(S_0 := \emptyset, T_0:=\emptyset, y_0:=\{0\}^{V_i})$.  Step \ref{algstep:SzigetiAlgorithm:def:calD} can be computed in $O(|V_i|)$ time by explicitly checking the condition $m_i(u) = K_{p_i}$ for every $u \in V_i$. By \Cref{lem:SzigetiAlgorithm:(p_im_i)-hypothesis} we have that the function $p_i$ is skew-supermodular. Consequently, the family of minimal $p_i$-maximizers is a disjoint family by \Cref{lem:UncrossingProperties:calFp-disjoint}. Moreover, by \Cref{lem:SzigetiAlgorithm:computingCalFp}, this family can be computed in $O(|V_i|^3)$ time using $O(|V_i|^2)$ queries to \functionMaximizationEmptyOracle{p_i}, where the runtime includes the time to construct the inputs for queries to \functionMaximizationEmptyOracle{p_i}. Consequently, by picking an arbitrary vertex from each set of the family, Steps \ref{algstep:SzigetiAlgorithm:def:T} and \ref{algstep:SzigetiAlgorithm:def:A} can be implemented in $O(|V_i|^3)$ time using $O(|V_i|^2)$ queries to \functionMaximizationEmptyOracle{p_i}, where the runtime includes the time to construct the inputs for queries to \functionMaximizationEmptyOracle{p_i}.  The values $\alpha_i^{(1)}$ and $\alpha_i^{(3)}$ can be computed in $O(|V_i|)$ time by iterating over all vertices in $V_i$. The value $\alpha_i^{(2)}$ can be computed by the single query $\functionMaximizationEmptyOracle{p_i}(S_0 := \emptyset, T_0:= A, y_0:=\{0\}^{V_i})$. Thus, Step \ref{algstep:SzigetiAlgorithm:def:alpha} can be implemented in $O(|V_i|)$ time using a single query to \functionMaximizationEmptyOracle{p_i}, where the runtime includes the time to construct the inputs for the query to \functionMaximizationEmptyOracle{p_i}. Steps \ref{algstep:SzigetiAlgorithm:def:zeros} and \ref{algstep:SzigetiAlgorithm:def:m'-and-m''} can be computed in $O(|V_i|)$ time by iterating over all vertices in $V_i$. We note that the number of distinct hyperedges in the hypergraph obtained from recursion (Step \ref{algstep:SzigetiAlgorithm:recursion}) is at most $\ell - i$ since each recursive call adds at most one distinct hyperedge to the returned. Thus, Steps \ref{algstep:SzigetiAlgorithm:def:G-and-c} and \ref{algstep:SzigetiAlgorithm:return} can be implemented in $O(|V_i| + \ell) = O(|V_1|)$ time, where the equality is because $\ell = O(|V_1|)$ by \Cref{lem:SzigetiAlgorithm:recursion-depth-and-support-size}.
    \end{proof}
 
Next, we focus on the time to implement a query to \functionMaximizationEmptyOracle{p_i} using \functionMaximizationOracleStrongCover{p_1}, where we recall that $p_1 = p$ is the input function to the initial call to \Cref{alg:SzigetiAlgorithm}.
    
    \begin{claim}\label{claim:SzigetiAlgorithm:runtime:p_i-oracle-using-p-weak-oracle}
        For disjoint subsets $S, T\subseteq V_i$ and a vector $y \in \Z^{V_i}$, the answer to the query $\functionMaximizationEmptyOracle{p_i}(S_0 := S, T_0 := T, y_0 := y)$ can be computed in $O(|V_1|^2)$ time using at most $|V_1|+1$ queries to \functionMaximizationOracleStrongCover{p_1}. The run-time includes the time to construct the inputs for the queries to \functionMaximizationOracleStrongCover{p_1}. Moreover, for each query to \functionMaximizationOracleStrongCover{p_1}, the hypergraph $(G_0, c_0)$ used as input to the query has $|V_1|$ vertices and $O(|V_1|)$ hyperedges.
    \end{claim}
    \begin{proof}
    We prove the claim in two steps. In the first step, we will use the expression for the function $p_i$ given by \Cref{lem:SzigetiAlgorithm:p_i-from-p_1} to construct an answer for the query $\functionMaximizationEmptyOracle{p_i}(S, T, y)$ using a single query to $\functionMaximizationOracle{p_1}$. In the second step, we will obtain the desired runtime using \functionMaximizationOracleStrongCover{p_1} by invoking \Cref{lem:Preliminaries:wc-oracle-from-sc-oracle}.
    
    For every $j \in [i]$, let $(\Tilde{H}_0^{j}, \Tilde{w}_0^{j})$ denote the hypergraph obtained by adding the vertices $\cup_{k\in [j-1]}\zeros_k$ to the hypergraph $(H^k_0, w^k_0)$. We now show the first step of the proof via the following procedure.  First, construct the hypergraph $(G, c)$ and the vector $y_1 \in \Z^{V_1}$, where $G:= \sum_{j \in [i-1]}\Tilde{H}_0^{j}$,  $c:= \sum_{j \in [i-1]}, \Tilde{w}_0^{j}$ and $y_1(u) := y_0(u)$ for every $u \in V_i$, $y_1(u) := 0$ otherwise Next, obtain $(Z, p_1(Z))$ by querying $\functionMaximizationOracle{p_1}((G, c), S, T, y_1)$. Finally, return $(Z-\cup_{j=1}^{i-1}\zeros_{j}, p(Z))$ as the answer to the query $\functionMaximizationEmptyOracle{p_i}(S, T, y)$. We note that correctness of the procedure is because  $p_i = \functionContract{\left(p_1 - \sum_{j \in [i-1]}b_{(\Tilde{H}_j, \Tilde{w}_{j})}\right)}{\bigcup_{j \in [i - 1]}\zeros_j}$ by \Cref{lem:SzigetiAlgorithm:p_i-from-p_1}.
    
    We now show the second step of the proof. We observe that the hypergraph $(G, c)$ constructed above has vertex set $V_1$ since for every $j \in [i - 1]$, the vertex set of $(\Tilde{H}_0^j, \Tilde{w}_0^j)$ is contained in $V_1$. Moreover, $(G, c)$
    has $i - 1 \leq \ell = O(|V_1|)$ hyperedges because for every $j \in [i - 1]$, the hypergraph $(\Tilde{H}_0^j, \Tilde{w}_0^j)$ has exactly $1$ hyperedge by Step \ref{algstep:SzigetiAlgorithm:def:Tilde_H-Tilde_w}, and $\ell = O(|V_1|)$ by \Cref{lem:SzigetiAlgorithm:recursion-depth-and-support-size}. Consequently, by \Cref{lem:Preliminaries:wc-oracle-from-sc-oracle}, the answer to the query $\functionMaximizationOracle{p_1}((G, c), S, T, y_1)$ in the above procedure can be computed in $O(|V_1|^2)$ time using $|V_1| + 1$ queries to \functionMaximizationOracleStrongCover{p_1}, where the runtime includes the time to construct the inputs to \functionMaximizationOracleStrongCover{p_1}; moreover, the hypergraph used as input to each \functionMaximizationOracleStrongCover{p_1} query has at most $|V_1|$ vertices and $O(|V_1|)$ hyperedges. Finally, we note that the construction of hypergraph $(G, c)$ and returning the (set, value) pair returned by the above procedure can be implemented in $O(|V_1|^2)$ time and so the claim holds.
    \end{proof}
    We now complete the proof using the two claims above. By \Cref{claim:SzigetiAlgorithm:runtime:all-steps-except-recursion}, all steps except for Steps \ref{algstep:SzigetiAlgorithm:def:p'-and-p''} and \ref{algstep:SzigetiAlgorithm:recursion} in the $i^{th}$ recursive call can be implemented to run in $O(|V_i|^3)$ time using $O(|V_i|^2)$ queries to \functionMaximizationEmptyOracle{p_i}.  By \Cref{claim:SzigetiAlgorithm:runtime:p_i-oracle-using-p-weak-oracle}, an answer to a single \functionMaximizationEmptyOracle{p_i} query can be computed in $O(|V_1|^2)$ time using at most $|V_1|+1$ queries to \functionMaximizationOracleStrongCover{p_1}. Thus, all steps except for Steps \ref{algstep:SzigetiAlgorithm:def:p'-and-p''} and \ref{algstep:SzigetiAlgorithm:recursion} in the $i^{th}$ recursive call can be implemented to run in $O(|V_1|^4)$ time and $O(|V_1|^3)$ queries to \functionMaximizationOracleStrongCover{p_1}. By Claims \ref{claim:SzigetiAlgorithm:runtime:all-steps-except-recursion} and \ref{claim:SzigetiAlgorithm:runtime:p_i-oracle-using-p-weak-oracle}, the run-time also includes the time to construct the inputs for the queries to \functionMaximizationOracleStrongCover{p_1}. Moreover, by \Cref{claim:SzigetiAlgorithm:runtime:p_i-oracle-using-p-weak-oracle}, the hypergraph $(G_0, c_0)$ used as input to each \functionMaximizationOracleStrongCover{p_1} query has $|V_1|$ vertices and $O(|V_1|)$ hyperedges.

    We note that Steps \ref{algstep:SzigetiAlgorithm:def:p'-and-p''} and \ref{algstep:SzigetiAlgorithm:recursion} need not be implemented explicitly for the purposes of the algorithm. Instead, \functionMaximizationEmptyOracle{p_{i+1}} can be used to execute all steps except for Steps \ref{algstep:SzigetiAlgorithm:def:p'-and-p''} and \ref{algstep:SzigetiAlgorithm:recursion} in the $(i+1)^{th}$ recursive call. We note that \Cref{claim:SzigetiAlgorithm:runtime:p_i-oracle-using-p-weak-oracle} also enables us to answer a \functionMaximizationEmptyOracle{p_{i+1}} query in $O(|V_1|^2)$ time using at most $|V_1|+1$ queries to \functionMaximizationOracleStrongCover{p_1}. Thus, we have shown that each recursive call can be implemented to run in $O(|V_1|^4)$ time and $O(|V_1|^3)$ queries to \functionMaximizationOracleStrongCover{p_1}. Moreover, each query to \functionMaximizationOracleStrongCover{p_1} is on an input hypergraph $(G_0, c_0)$ that has $|V_1|$ vertices and $O(|V_1|)$ hyperedges. By \Cref{lem:SzigetiAlgorithm:recursion-depth-and-support-size}, the number of recursive calls witnessed by the execution of \Cref{alg:SzigetiAlgorithm} is $O(|V_1|)$. Thus, \Cref{alg:SzigetiAlgorithm} can be implemented to run in $O(|V_1|^5)$ time and $O(|V_1|^4)$ queries to \functionMaximizationOracleStrongCover{p_1}, where each query to \functionMaximizationOracleStrongCover{p_1} is on an input hypergraph $(G_0, c_0)$ that has $|V_1|$ vertices and $O(|V_1|)$ hyperedges. 
\end{proof}

%% file: covering-algorithm.tex
\newcommand{\FeasibleHyperedge}{\textsc{FeasibleHyperedge}}
\newcommand{\HypedgeFeasibility}[1]{\text{HypEdge-Feasibility\ensuremath{(#1)}\xspace}}

\section{Weak Cover using Near-Uniform Hyperedges}\label{sec:CoveringAlgorithm}

In this section, we present an algorithm to obtain a degree-specified near-uniform weak cover of a function.  We emphasize that our algorithm does not require any assumption on the structure of the input set-function $p$ other than it being integer-valued and satisfying certain conditions with the input degree bounds. In subsequent sections where we prove 
\Cref{thm:WeakCoverViaUniformHypergraph:main} and \Cref{thm:WeakCoverTwoFunctionsViaUniformHypergraph:main}, we will show that this algorithm, when executed on inputs satisfying the theorem-specific hypotheses (regarding skew-supermodularity of $p$), returns hypergraphs with the claimed properties of the respective theorems. 

In \Cref{sec:CoveringAlgorithm:algorithm}, we present our algorithm (See \Cref{alg:CoveringAlgorithm}). Our algorithm is recursive in nature. We define the notion of a \emph{feasible execution} of our algorithm in the same section. In \Cref{sec:CoveringAlgorithm:termination-partial-correctness}, we prove that a feasible execution of the algorithm terminates within finitely many recursive calls. 
Here, we also show that it returns a hypergraph satisfying properties (1)--(4) of  \Cref{thm:WeakCoverViaUniformHypergraph:main} and \Cref{thm:WeakCoverTwoFunctionsViaUniformHypergraph:main}. In \Cref{sec:CoveringAlgorithm:number-recursive-calls-with-alpha=alpha1-alpha2-alpha3}, we bound the number of recursive calls of a specific type during feasible executions of our algorithm---we use this bound in subsequent sections to show properties of our algorithm when the input function is skew-supermodular or the maximum of two skew-supermodular functions.
In \Cref{sec:skew-supermodularity-based-properties-of-algorithm}, we turn our attention to the case where the input is a skew-supermodular function or the maximum of two skew-supermodular functions. Here, we show that the execution of our algorithm is a feasible execution on such inputs. 
In \Cref{sec:run-time-assuming-bounded-recursion-depth}, we bound the run-time of the algorithm as a function of the recursion depth for skew-supermodular functions and maximum of two skew-supermodular functions as inputs. 

Finally, in Sections \ref{sec:WeakCoverViaUniformHypergraph}  and \ref{sec:WeakCoverTwoFunctionsViaUniformHypergraph}, we use the results from the preceding sections to prove our main results. In \Cref{sec:WeakCoverViaUniformHypergraph}, we bound the number of recursive calls of type not considered in \Cref{sec:CoveringAlgorithm:number-recursive-calls-with-alpha=alpha1-alpha2-alpha3} for when the input function $p$ is a skew-supermodular function and complete the proof of \Cref{thm:WeakCoverViaUniformHypergraph:main}. In \Cref{sec:WeakCoverTwoFunctionsViaUniformHypergraph}, we bound the number of recursive calls of type not considered in \Cref{sec:CoveringAlgorithm:number-recursive-calls-with-alpha=alpha1-alpha2-alpha3} for when the input function $p$ is the maximum of two skew-supermodular functions and complete the proof of \Cref{thm:WeakCoverTwoFunctionsViaUniformHypergraph:main}.

\subsection{The Algorithm}\label{sec:CoveringAlgorithm:algorithm}
Our algorithm takes as input (1) a function $p:2^V \rightarrow\Z$, (2) a positive function $m:V \rightarrow \Z_{+}$, and (3) a set $J\subseteq V$ 
and returns a hypergraph $\left(H= \left(V, E\right), w\right)$. We note that in contrast to the non-negative function $m$ appearing in the statements of \Cref{thm:WeakCoverViaUniformHypergraph:main} and \Cref{thm:WeakCoverTwoFunctionsViaUniformHypergraph:main}, the function $m$ that is input to our algorithm should be positive. However, the algorithm in this section will not require the function $p$ to be skew-supermodular. The algorithm is recursive. The input set $J\subseteq V$ is used to share information between subsequent recursive calls and is thus a consequence of the recursive nature of the algorithm. This input parameter can be suppressed if the algorithm is implemented in an iterative manner. 

Our algorithm is LP-based. We define the polyhedron associated with the LP now. 
For a set function $p:2^V \rightarrow\Z$ and a function $m: V\rightarrow \R_{\ge 0}$, we define the polyhedron $Q(p,m)$ as:
\begin{equation}\tag{$Q$-polyhedron}\label{eqn:Q(p,m)}
    Q(p, m) \coloneqq  \left\{ x\in \R^{V}\ \middle\vert 
        \begin{array}{l}
            {(\text{i})}{\ \ \ 0 \leq x(u) \leq \min\{1, m(u)\}} \hfill {\qquad \forall \ u \in V} \\
            {(\text{ii})}{\ \ x(Z) \geq 1} \hfill {\qquad\qquad\qquad\qquad\qquad \forall\  Z\subseteq V:\ p(Z) = K_p} \\
            {(\text{iii})}{ \ x(u) = 1} \hfill {\qquad\qquad\qquad\qquad\qquad \forall\ u\in V:\ m(u) = K_p} \\
            {(\text{iv})}{ \ x(Z) \leq m(Z) - p(Z) + 1} \hfill { \ \ \ \forall\ Z \subseteq V} \\
            {(\text{v})}{\ \ \left \lfloor \frac{m(V)}{K_p} \right \rfloor \le x(V) \le \left \lceil \frac{m(V)}{K_p} \right \rceil} \hfill {} \\
        \end{array}
        \right\}. 
\end{equation} 
We now give an informal description of the algorithm (refer to \Cref{alg:CoveringAlgorithm} for a formal description).
If $V=\emptyset$, then the algorithm returns the empty hypergraph. Otherwise, the algorithm computes an extreme point optimum solution $y$ to the LP $\max\{\sum_{u\in J}y_u: y\in Q(p,m)\}$ and defines $A$ to be the support of $y$. 
The algorithm uses the set $A$ to compute the following five intermediate quantities $\alpha^{(1)}, \alpha^{(2)}, \alpha^{(3)}, \alpha^{(4)}, \alpha^{(5)}$: 
\begin{itemize}
        \item $\alpha^{(1)} \coloneqq  \min\{m(u) : u\in A\}$, 
        \item $\alpha^{(2)} \coloneqq  \min\left\{K_p - p(X) : X \subseteq V - A\right\}$, 
        \item $\alpha^{(3)} \coloneqq  \min\left\{K_p - m(u) : u \in V - A\right\}$, 
        \item $\alpha^{(4)} \coloneqq  \min\left\{\floor{\frac{m(X) - p(X)}{|A\cap X| - 1}} : |A\cap X| \geq 2\right\}$, and 
        \item $\alpha^{(5)} \coloneqq \begin{cases}
            m(V)\mod K_p & \text{if $|A| = \ceil{m(V)/K_p} > m(V)/K_p$,}\\
            K_p - (m(V)\mod K_p) & \text{if $|A| = \floor{m(V)/K_p} < m(V)/K_p$,}\\
            +\infty & \text{otherwise}\end{cases},$
    \end{itemize}
and  defines $\alpha \coloneqq  \min\left\{ \alpha^{(1)}, \alpha^{(2)}, \alpha^{(3)}, \alpha^{(4)}, \alpha^{(5)}\right\}$.
Next, the algorithm computes the set $\zeros \coloneqq  \{u \in V : m(u) - \alpha = 0\}$ and defines $(H_0, w_0)$ to be the hypergraph on vertex set $V$ consisting of a single hyperedge $A$ with weight $w_0(A) = \alpha$. Next, the algorithm defines the two functions $m':V\rightarrow \Z$ and $m'':V - \zeros\rightarrow\Z$ as $m' := m - \alpha\chi_A$, and $m'' := \functionRestrict{m'}{\zeros}$. Next, the algorithm defines the two functions $p':2^V \rightarrow\Z$ and $p'':2^{V - \zeros}\rightarrow\Z$ as 
    $p'\coloneqq p - b_{(H_0, w_0)}$ and $p''\coloneqq \functionContract{p'}{\zeros}$. The algorithm then defines the set $J''\coloneqq A-\zeros$ and recursively calls itself with the input tuple $(p'', m'', J'')$ to obtain a hypergraph $\left(H'' = \left(V'' \coloneqq  V - \zeros, E''\right), w''\right)$. Finally, the algorithm extends the hypergraph $H''$ with the set $\zeros$ of vertices, adds the set $A$ with weight $\alpha$ as a new hyperedge to the hyperedge set $E''$, and returns the resulting hypergraph. If the hyperedge $A$ already has non-zero weight in $(H'', w'')$, then the algorithm increases the weight of the hyperedge $A$ by $\alpha$, and returns the resulting hypergraph. 

The definition of the polyhedron $Q(p,m)$ immediately implies the following lemma. 
\begin{lemma}\label{lem:hyperedge-feasibility-for-integral-polyhedron}
    Let $p:2^V\rightarrow \Z$ be a set function with $K_p>0$ and $m:V\rightarrow \Z_+$ be a function. 
    Let $y$ be an integral vector in $Q(p, m)$ 
    and $A\coloneqq \{u\in V: y_u>0\}$. 
    Then, the following properties hold: 
    \begin{enumerate}[label={(\arabic*)}]
        \item The set $A$ is a transversal for the family of $p$-maximizers, \label{HypedgeFeasibility:property:1}
        \item $\{u \in V : m(u) = K_p\} \subseteq A$, \label{HypedgeFeasibility:property:2}
        \item $|A\cap Z| \leq m(Z) - p(Z) + 1$ for all $Z \subseteq V$, and \label{HypedgeFeasibility:property:3}
        \item $|A|\in [\lfloor m(V)/K_p\rfloor, \lceil m(V)/K_p\rceil]$. \label{HypedgeFeasibility:property:4}
    \end{enumerate}
\end{lemma}

We recall that a polyhedron is integral if all its extreme points are integral. Non-emptiness and integrality of the polyhedron $Q(p,m)$ will be an important requirement to show properties about \Cref{alg:CoveringAlgorithm}. We will denote recursive calls of \Cref{alg:CoveringAlgorithm} during which $Q(p,m)$ is a non-empty integral polyhedron as \emph{feasible recursive calls} and say that an execution of \Cref{alg:CoveringAlgorithm} is a \emph{feasible execution} if every recursive call during the execution is feasible or $V=\emptyset$ in that recursive call. We note that if $V=\emptyset$ in a recursive call, then the algorithm terminates with that recursive call. 

\begin{algorithm}[htb]
\caption{Weak covering with hyperedges}\label{alg:CoveringAlgorithm}
\textsc{Input:} Function $p:2^V\rightarrow\Z$ and positive function $m:V\rightarrow\Z_+$\\
\textsc{Output:} Hypergraph $(H = (V,E), w:E\rightarrow\Z_+)$\\

$\mathbf{\textsc{Algorithm}}\left(p, m, J\right):$
\begin{algorithmic}[1]
    \If{$V = \emptyset$} 
    \Return{Empty Hypergraph  $\left(\left(\emptyset, \emptyset\right), \emptyset\right)$.}\label{algstep:CoveringAlgorithm:base-case}
    \Else:\label{algstep:CoveringAlgorithm:else}
    \State{Compute an extreme point optimum solution $y$ to $\max\left\{\sum_{u\in J}y_u: y\in Q(p, m)\right\}$}\label{algstep:CoveringAlgorithm:def:y}
    \State{$A \coloneqq  \{u\in V: y_u> 0\}$}\label{algstep:CoveringAlgorithm:def:A}
    \State{$\alpha \coloneqq  \min\begin{cases}
        \alpha^{(1)} \coloneqq  \min\left\{m(u) : u \in A\right\} \\
        \alpha^{(2)} \coloneqq  \min\left\{K_p - p(X) : X \subseteq V - A\right\} \\
        \alpha^{(3)} \coloneqq  \min\left\{K_p - m(u) : u \in V - A\right\} \\
        \alpha^{(4)} \coloneqq  \min\left\{\floor{\frac{m(X) - p(X)}{|A\cap X| - 1}} : |A\cap X| \geq 2\right\}\\
        \alpha^{(5)} \coloneqq \begin{cases}
            m(V)\mod K_p & \text{if $|A| = \ceil{m(V)/K_p} > m(V)/K_p$,}\\
            K_p - (m(V)\mod K_p) & \text{if $|A| = \floor{m(V)/K_p} < m(V)/K_p$,}\\
            +\infty & \text{otherwise.}\end{cases}
\end{cases}$}\label{algstep:CoveringAlgorithm:def:alpha}
    \State{$\zeros \coloneqq  \left\{u \in A : m(u) = \alpha\right\}$}\label{algstep:CoveringAlgorithm:def:zeros}
    \State{Let $(H_0\coloneqq (V, E_0\coloneqq \{A\}, w_0: E_0\rightarrow \{\alpha\})$}\label{algstep:CoveringAlgorithm:def:H_0-and-w_0}
    \State{$m'\coloneqq m - \alpha\chi_{A}$ and $m''\coloneqq \functionRestrict{m'}{\zeros}$}\label{algstep:CoveringAlgorithm:def:m'-and-m''}
    \State{$p'\coloneqq p - b_{(H_0, w_0)}$ and $p''\coloneqq \functionContract{p'}{\zeros}$}\label{algstep:CoveringAlgorithm:def:p'-and-p''}
     \State{$(H'', w'')\coloneqq \textsc{Algorithm}\left(p'', m'', J''\coloneqq A-\zeros\right)$}\label{algstep:CoveringAlgorithm:recursion}
    \State{Obtain $(G, c)$ from $(H'', w'')$ by adding vertices $\zeros$}\label{algstep:CoveringAlgorithm:def:G-and-c}
    \State{\Return $(G+H_0, c+w_0)$}
    \label{algstep:CoveringAlgorithm:return}
    \EndIf
\end{algorithmic}
\end{algorithm}

\subsection{Termination and Partial Correctness of \Cref{alg:CoveringAlgorithm}}\label{sec:CoveringAlgorithm:termination-partial-correctness}
This section is devoted to proving that  
a feasible execution of \Cref{alg:CoveringAlgorithm} terminates within finite number of recursive calls and returns a hypergraph satisfying properties (1)--(4) of \Cref{thm:WeakCoverViaUniformHypergraph:main} and \Cref{thm:WeakCoverTwoFunctionsViaUniformHypergraph:main}--see \Cref{lem:CoveringAlgorithm:main}. 
We also prove certain additional properties in this section about the execution of \Cref{alg:CoveringAlgorithm} that will help in subsequent analysis. 
We emphasize that all lemmas that appear in this section do not require the function $p$ to be skew-supermodular.
The ideas to prove \Cref{lem:CoveringAlgorithm:main} are modifications of the results due to Bern\'{a}th and Kir\'{a}ly  \cite{Bernath-Kiraly}; we suggest first time readers to skip reading this section until this lemma and return to them as needed. 

We begin with the following lemma which shows that the value $\alpha$ computed in Step \ref{algstep:CoveringAlgorithm:def:alpha} is positive and the function value of the entire ground set for function $m$ strictly decreases in each feasible recursive call. 

\begin{lemma}\label{lem:CoveringAlgorithm:properties}
Suppose that the input to  \Cref{alg:CoveringAlgorithm} is a tuple $(p,m,J)$, where $J\subseteq V$ is an arbitrary set, $p:2^V\rightarrow\Z$ is a set function, and $m:V\rightarrow\Z_{+}$ is a function such that $m(X)\ge p(X)$ for all $X\subseteq V$, $m(u)\le K_p$ for all $u\in V$, and $Q(p,m)$ is a non-empty integral polyhedron. Let $\alpha, \zeros$ and $m''$ be as defined by \Cref{alg:CoveringAlgorithm}. Then, we have that 
\begin{enumerate}[label=(\alph*)]
    \item $\alpha \geq 1$, and
    \item $m''(V - \zeros) < m(V)$.
\end{enumerate}
\end{lemma}
\begin{proof} We note that since $Q(p,m)$ is non-empty, we have that $V \not=\emptyset$. Consequently, \Cref{alg:CoveringAlgorithm} is in its recursive case. We prove both properties separately below.
\begin{enumerate}[label=(\alph*)]
    \item Let $y, A, \alpha^{(1)}, \alpha^{(2)}, \alpha^{(3)}, \alpha^{(4)}, \alpha^{(5)}$ be as defined by \Cref{alg:CoveringAlgorithm}.
    Since $Q(p,m)$ is a non-empty integral polyhedron, the vector $y$ is the indicator vector of the set $A$. We note that $\alpha^{(1)} \in \Z_+$ since the function $m$ is a positive integer valued function. Thus, it suffices to show that $\alpha^{(2)}, \alpha^{(3)}, \alpha^{(4)}, \alpha^{(5)} \in \Z_+$. 

First, we recall that $\alpha^{(2)} = \min\{K_p - p(Z) : Z \subseteq V - A\}$.
We note that $\alpha^{(2)} \in \Z$ since $p:2^V \rightarrow \Z$. Thus, it suffices to show that $\alpha^{(2)} \geq 1$. By way of contradiction, say $\alpha^{(2)} \leq 0$.  Since $K_p \geq p(X)$ for all $X \subseteq V$ by definition, we have that $\alpha^{(2)} \geq 0$. Thus, $\alpha^{(2)} = 0$. Let $X\subseteq V$ be a set that witnesses $\alpha^{(2)} = 0$, i.e. $\alpha^{(2)} = K_p - p(X)$ and $X \subseteq V - A$. Since $\alpha^{(2)} = 0$, we have that $p(X) = K_p$ and thus, the set $X$ is a $p$-maximizer that is disjoint from the set $A$. This contradicts \Cref{lem:hyperedge-feasibility-for-integral-polyhedron}(1).

Next, we recall that $\alpha^{(3)} = \min\{K_p - m(u) : u \in V - A\}$.
We note that $\alpha^{(3)} \in \Z$ since $K_p \in \Z$ and $m:V\rightarrow \Z$. Thus, it suffices to show that $\alpha^{(3)} \geq 1$.
By way of contradiction, say $\alpha^{(3)} \leq 0$. Since $K_p \geq m(u)$ for all $u \in V$, we have that $\alpha^{(3)} \geq 0$. Thus, $\alpha^{(3)} = 0$. Let $v \in V$ be a  vertex that witnesses $\alpha^{(3)} = 0$, i.e. $\alpha^{(3)} = K_p - m(v)$ and $v \in V - A$. Since $\alpha^{(3)} = 0$, we have that $m(v) = K_p$. However, this contradicts \Cref{lem:hyperedge-feasibility-for-integral-polyhedron}(2).

Next, we recall that $\alpha^{(4)} = \min\left\{\floor{\frac{m(X) - p(X)}{|A\cap X| - 1}} : |A\cap X| \geq 2\right\}$.
    We note that $\alpha^{(4)} \in \Z$ by definition and so it suffices to show that $\alpha^{(4)} \geq 1$.
    By way of contradiction, say $\alpha^{(4)} \leq 0$. Since $m(X) \geq p(X)$ for every $X\subseteq V$, we have that $\alpha^{(4)} \geq 0$. Thus, $\alpha^{(4)} = 0$. Let $Z\subseteq V$ be a set that witnesses $\alpha^{(4)} = 0$, i.e. $Z \coloneqq  \arg\min\left\{\floor{\frac{m(X) - p(X)}{|A\cap X| - 1}} : |A\cap X| \geq 2\right\}$. Consequently, we have that $\frac{m(Z) - p(Z)}{|A\cap Z| - 1} < 1$. However, this gives us that $m(Z) - |A\cap Z| < p(Z) - 1$, contradicting \Cref{lem:hyperedge-feasibility-for-integral-polyhedron}(3).

Finally, we recall that 
\[
\alpha^{(5)} \coloneqq \begin{cases}
            m(V)\mod K_p & \text{if $|A| = \ceil{m(V)/K_p} > m(V)/K_p$,}\\
            K_p - (m(V)\mod K_p) & \text{if $|A| = \floor{m(V)/K_p} < m(V)/K_p$,}\\
            +\infty & \text{otherwise.}
            \end{cases}
\]
If $\alpha^{(5)}=\infty$, then we are done. Suppose that $\alpha^{(5)}\neq \infty$. Then, $m(V)/K_p$ is not an integer. By definition, this immediately implies that $\alpha^{(5)}\ge 1$. 

\item Let $A, m'$ be as defined by \Cref{alg:CoveringAlgorithm}. Since $Q(p,m)$ is non-empty, we have that $A \neq \emptyset$.  
    By part (1) of the current lemma, we have that $\alpha\ge 1$. Hence, $m'(V) = m(V) - \alpha\chi_A < m(V)$ by Step \ref{algstep:CoveringAlgorithm:def:m'-and-m''}. Since the function $m'':{V - \zeros}\rightarrow\Z$ is defined as $m'' = \functionRestrict{(m - \alpha\chi_A )}{\zeros}$ by Step \ref{algstep:SzigetiAlgorithm:def:m'-and-m''}, we have that $m''(V-\zeros)=m'(V)<m(V)$ and the claim holds. 
\end{enumerate}
\end{proof}

The next two lemmas will allow us to conclude that if the input functions $p$ and $m$ to a feasible recursive call of \Cref{alg:CoveringAlgorithm} satisfy 
certain conditions, 
then the intermediate functions $p', m'$ and the input functions to the subsequent recursive call $p'', m''$ also  satisfy the same conditions. 
Such recursive properties will be useful in inductively proving properties about \Cref{alg:CoveringAlgorithm}. 

\begin{lemma}\label{lem:Covering-Algorithm:hypothesis-after-createstep}
Suppose that the input to  \Cref{alg:CoveringAlgorithm} is a tuple $(p,m,J)$, where $J\subseteq V$ is an arbitrary set, $p:2^V\rightarrow\Z$ is a set function, and $m:V\rightarrow\Z_{+}$ is a function such that $m(X)\ge p(X)$ for all $X\subseteq V$, $m(u)\le K_p$ for all $u\in V$, and $Q(p,m)$ is a non-empty integral polyhedron. 
Let $p'$, $m'$ and $\alpha$ be as defined by \Cref{alg:CoveringAlgorithm}. Then, the following hold: 
    \begin{enumerate}[label=(\alph*)]
        \item The function $m':V'\rightarrow\Z_{\geq 0}$ is a non-negative function,
        \item $K_{p'} \coloneqq  K_p - \alpha$,
         \item $m'(u) \leq K_{p'}$ for all $u \in V'$, and
        \item $m'(X) \ge p'(X)$ for all $X \subseteq V'$. 
    \end{enumerate}
\end{lemma}
\begin{proof}
    We note that since $Q(p,m)$ is non-empty, we have that $V \not=\emptyset$. Consequently, \Cref{alg:CoveringAlgorithm} is in its recursive case. Moreover, since since $Q(p,m)$ is an integral polyhedron, the vector $y$ is the indicator vector of the set $A$, where $y$ and $A$ are as defined by the algorithm. 
    Let $\alpha^{(1)}, \alpha^{(2)}, \alpha^{(3)}, \alpha^{(4)}, \alpha^{(5)}, \alpha$, and $(H_0=(V, E_0\coloneqq \{A\}, w_0: E_0\rightarrow \{\alpha\})$ be as defined by the algorithm. 
    We show that the functions $m'$ and $p'$ satisfy each of the properties (a)-(d) below.
    \begin{enumerate}[label=(\alph*)]
    \item 
    The function $m'$ is non-negative due to the following: if the vertex $u \not \in A$, then $m'(u) = m(u) \geq 0$. Otherwise $u\in A$ and we have that $m'(u) = m(u)- \alpha \geq m(u) - \alpha^{(1)} \geq 0$, where the last inequality is by definition of $\alpha^{(1)}$ in Step \ref{algstep:CoveringAlgorithm:def:alpha}.
    
    \item Let $X \subseteq V$ be a $p$-maximizer and $Y \subseteq V$ be a $p'$-maximizer. 
    We recall that the set $A$ is a transversal for the family of $p$-maximizers by  \Cref{lem:hyperedge-feasibility-for-integral-polyhedron}(1), and consequently, $|A\cap X| \geq 1$. Thus, $K_{p'}=p'(Y) \geq p'(X) = p(X) - \alpha =K_p-\alpha$, where the inequality is because the set $Y$ is a $p'$-maximizer, and the second equality is because $|A\cap X| \geq 1$ and the definition of the function $p'$ in Step \ref{algstep:CoveringAlgorithm:def:p'-and-p''}. We now show that $K_{p'} \leq K_p - \alpha$. If the set $Y$ is also a $p$-maximizer, then the claim holds since we have that $K_{p'} = p'(Y) = p(Y) - \alpha = K_{p} - \alpha$, where the second equality is because the set $A$ is a transversal for the family of $p$-maximizers. Thus, we may assume that $Y$ is not a $p$-maximizer. If $|A\cap Y| \geq 1$, then we have that $$p(X) - \alpha \leq p'(Y) = p(Y) - \alpha \leq p(X) - \alpha,$$ where the first inequality is because we have already shown that $K_{p'}\ge K_p-\alpha$, the equality is because of $|A\cap Y| \geq 1$ and the definition of $p'$ in Step \ref{algstep:CoveringAlgorithm:def:p'-and-p''}, and the final inequality is because the set $X$ is a $p$-maximizer. Thus, all inequalities are equations and the claim holds. Otherwise, suppose that $|A \cap Y| = 0$. Here, we have that $$p(X) - \alpha \geq p(X) - \alpha^{(2)} \geq p(X) - \left(p(X) - p(Y)\right) = p(Y) = p'(Y),$$
    where the first inequality is because $\alpha \leq \alpha^{(2)}$ in Step \ref{algstep:CoveringAlgorithm:def:alpha}, the second inequality is by definition of $\alpha^{(2)}$ in Step \ref{algstep:CoveringAlgorithm:def:alpha} and the fact that the set $X$ is a $p$-maximizer, and the final equality is because of $|A\cap Y| = 0$ and definition of the function $p'$ in Step \ref{algstep:CoveringAlgorithm:def:p'-and-p''}.
    
    \item Suppose $u \in A$. Then, we have that $m'(u) = m(u) - \alpha \leq K_p - \alpha = K_{p'}$, where the first equality is by definition of $m'$ in Step \ref{algstep:CoveringAlgorithm:def:m'-and-m''} and the final equality is by part (b) of the current lemma. Next, suppose $u  \in V - A$, implying $m'(u) = m(u)$. In particular, we have that $K_p - m'(u) = K_p - m(u) \geq \alpha^{(3)} \geq \alpha$, where the first inequality is by definition of $\alpha^{(3)}$ and the second inequality is by definition of $\alpha$ in Step \ref{algstep:CoveringAlgorithm:def:alpha}. Rearranging the terms, we get that $K_{p'} = K_p - \alpha \geq m'(u)$, where the first equality is again by part (b) of the current lemma.
    
    \item Let $X\subseteq V$ be arbitrary. 
    Suppose that $|X\cap A| = 0$. Then, $b_{(H_0, w_0)}(X) = 0$ and hence $m'(X) = m(X) \geq p(X) = p'(X)$,  so the claim holds. Suppose that $|X\cap A| \geq 1$. Then, we have the following:
    $$m'(X) = m(X) - \alpha|X\cap A| \geq p(X) - \alpha = p(X) - b_{(H_0, w_0)}(X) = p'(X).$$
    Here, the inequality holds trivially if $|X\cap A|=1$, while for $|X\cap A|\ge 2$ it holds because $\alpha \geq \alpha^{(4)}$ in Step \ref{algstep:CoveringAlgorithm:def:alpha}. The last equality is because $|X\cap A| \geq 1$. 
    \end{enumerate}
\end{proof}

\begin{lemma}\label{lem:CoveringAlgorithm:(p''m'')-hypothesis}
Suppose that the input to  \Cref{alg:CoveringAlgorithm} is a tuple $(p,m,J)$, where $J\subseteq V$ is an arbitrary set, $p:2^V\rightarrow\Z$ is a set function, and $m:V\rightarrow\Z_{+}$ is a function such that $m(X)\ge p(X)$ for all $X\subseteq V$, $m(u)\le K_p$ for all $u\in V$, and $Q(p,m)$ is a non-empty integral polyhedron. 
Let $\alpha, \alpha^{(5)}, \zeros, p''$ and $m''$ be as defined by \Cref{alg:CoveringAlgorithm}. Then, the following hold: 
\begin{enumerate}[label=(\alph*)]
        \item The function $m'':V - \zeros\rightarrow\Z_{+}$ is a positive function,
        \item $K_{p''} = K_p - \alpha$,
        \item $m''(u) \leq K_{p''}$ for all $u \in V''$, 
        \item $m''(X) \ge p''(X)$ for all $X \subseteq V''$, and 
        \item either $K_{p''}=0$ or $m''(V-\zeros)/K_{p''}\in [\lfloor m(V)/K_p\rfloor, \lceil m(V)/K_p \rceil]$; furthermore, if $K_{p''} > 0$ and $\alpha < \alpha^{(5)}$, then we have that 
        $$\floor{(m''(V-\zeros)/K_{p''}} = \floor{(m(V)/K_{p}} \text{ and } \ceil{(m''(V-\zeros)/K_{p''}} = \ceil{m(V)/K_p}.$$
    \end{enumerate}
\end{lemma}
\begin{proof} 
We note that since $Q(p,m)$ is non-empty, we have that $V \not=\emptyset$. Consequently, \Cref{alg:CoveringAlgorithm} is in its recursive case. Let $A$, $p': 2^V\rightarrow \Z$, and $m': V\rightarrow \Z$ be as defined by the \Cref{alg:CoveringAlgorithm}.
We show that the functions $m''$ and $p''$ satisfy each of the properties (a)-(d) below.
    \begin{enumerate}[label=(\alph*)]
    \item We have that $m''(u) = m'(u) > 0$ for every $u \in V - \zeros$, where the inequality is by \Cref{lem:Covering-Algorithm:hypothesis-after-createstep}(a) and the definition of the set $\zeros$ in Step \ref{algstep:CoveringAlgorithm:def:zeros}.
        \item By \Cref{lem:Covering-Algorithm:hypothesis-after-createstep}(b), it suffices to show that $K_{p''} = K_{p'}$. Let $X\uplus R_X$ be a $p'$-maximizer, where $R_X \subseteq\zeros$ and $X \subseteq V -\zeros$. Furthermore, let $Y\subseteq V -\zeros$ be a $p''$-maximizer such that $p''(Y) = p'(Y\uplus R_Y)$ where $R_Y \subseteq\zeros$. Then, we have the following:
    $$K_{p'} = p'(X\uplus R_X) \leq \max\left\{p'(X\uplus R') : R'\subseteq\zeros\right\} = p''(X) \leq K_{p''} = p''(Y) = p'(Y\uplus R_Y) \leq K_{p'}.$$
    Thus, all inequalities are equations.

    \item  Let $u \in V - \zeros$. Then we have that $m''(u) = m'(u) \leq K_{p'} = K_{p''}$, where the inequality is by \Cref{lem:Covering-Algorithm:hypothesis-after-createstep}(c) and the final equality is by \Cref{lem:Covering-Algorithm:hypothesis-after-createstep}(b) and part (b) of the current lemma.

    \item Let $X\subseteq V - \zeros$ be arbitrary. Let $R \subseteq\zeros$ such that $p''(X) = p'(X\uplus R)$. Then, we have that $m''(X) = m'(X\uplus R) \geq p'(X\uplus R) = p''(X)$. Here, the first equality is because $R\subseteq\zeros$ and the inequality is by \Cref{lem:Covering-Algorithm:hypothesis-after-createstep}(d).

    \item  We note that $K_{p''}\geq 0$ because $K_{p''} = K_{p} - \alpha \geq K_{p} - \alpha^{(1)} \geq K_{p} - m(u) \geq 0$ for every $u \in A$, where the final inequality is because $m(u)\le K_p$ for every $u\in V$. Suppose that $K_{p''}>0$. By \Cref{lem:hyperedge-feasibility-for-integral-polyhedron}(4), we know that $|A| \in [\lfloor m(V)/K_p \rfloor, \lceil m(V)/K_p\rceil]$. First, we consider the case where $m(V)/K_p\in \Z$. Then, we have that $|A|=m(V)/K_p$ and hence, 
    $$\frac{m''(V-\zeros)}{K_{p''}} = \frac{m(V) - \alpha|A|}{K_{p''}} = \frac{m(V) - \alpha\frac{m(V)}{K_p}}{K_{p''}} = \frac{m(V) - \alpha\frac{m(V)}{K_p}}{K_p - \alpha} = \frac{m(V)}{K_p},$$
    where the third equality is by \Cref{lem:Covering-Algorithm:hypothesis-after-createstep}(b). Here, we also observe that the second part of the claim holds (regardless of the value of $\alpha$).
    Next, suppose that $m(V)/K_p\not\in \Z$. 
    Let $d, r \in \Z_+$  be positive integers such that $m(V) = d\cdot K_p + r$ with $r<K_p$ (such a $d$ and $r$ exist since $m$ is a positive function, $m(Z)\ge p(Z)$ for every $Z\subseteq V$, and $m(V)/K_p\not\in \Z$). We note that $d=\lfloor m(V)/K_p\rfloor$ and $d+1 = \lceil m(V)/K_p \rceil$ and hence, it suffices to show that $m''(V-\zeros)/K_p\in [d, d+1]$. 
    We consider two cases. For he first case, suppose that $|A|=\lfloor m(V)/K_p \rfloor$. This implies that $|A|=d$ and $\alpha\le \alpha^{(5)}=K_p-r$. Then, we have the following: 
    \[
    \frac{m''(V-\zeros)}{K_{p''}} = \frac{m(V) - \alpha|A|}{K_{p}-\alpha} = \frac{dK_p+r-\alpha d}{K_p-\alpha} \in [d, d+1]. 
    \]
    For the second case, suppose that $|A|=\lceil m(V)/K_p \rceil$. This implies that $|A|=d+1$ and $\alpha\le \alpha^{(5)}=r$. Then, we have the following: 
    \[
    \frac{m''(V-\zeros)}{K_{p''}} = \frac{m(V) - \alpha|A|}{K_{p}-\alpha} = \frac{dK_p+r-\alpha (d+1)}{K_p-\alpha} \in [d, d+1]. 
    \]
    Then, the first part of the claim holds by the final inclusions in both cases above. Furthermore, we observe that in both cases, the final inclusions are in the range $(d, d+1)$ if $\alpha < \alpha^{(5)}$. Consequently, $\floor{m''(V'')/K_{p''}} = d$ and $\ceil{m''(V'')/K_{p''}} = d+1$ and so the second part of the claim also holds.
    \end{enumerate}
\end{proof}

We recall that an execution of \Cref{alg:CoveringAlgorithm} is feasible if the $Q(p,m)$ polyhedron is integral for every recursive call of the algorithm. We now show the main lemma of the section which says that a feasible execution of \Cref{alg:CoveringAlgorithm} terminates within a finite number of recursive calls and returns a hypergraph satisfying properties (1)–(4) of \Cref{thm:WeakCoverViaUniformHypergraph:main} and \Cref{thm:WeakCoverTwoFunctionsViaUniformHypergraph:main}.

\begin{restatable}{lemma}{lemCoveringAlgorithmMain}\label{lem:CoveringAlgorithm:main}
Suppose that the input to  \Cref{alg:CoveringAlgorithm} is a tuple $(p,m,J)$, where $J\subseteq V$ is an arbitrary set, $p:2^V\rightarrow\Z$ is a set function, and $m:V\rightarrow\Z_{+}$ is a function such that $m(X)\ge p(X)$ for all $X\subseteq V$ and $m(u)\le K_p$ for all $u\in V$, 
and the execution of \Cref{alg:CoveringAlgorithm} for the input tuple $(p, m, J)$ is feasible. 
Then, \Cref{alg:CoveringAlgorithm} terminates in a finite (weakly-polynomial) number of recursive calls. Furthermore, the hypergraph $\left(H = \left(V, E\right), w:E\rightarrow\Z_+\right)$ returned by \Cref{alg:CoveringAlgorithm} 
satisfies the following properties:
\begin{enumerate}[label={(\arabic*)}]
    \item $b_{(H,w)}(X) \geq p(X)$ for all $X\subseteq V$,
    \item $b_{(H, w)}(u) = m(u)$ for all $u \in V$, 
    \item $\sum_{e\in E}w(e) = K_p$, and 
    \item if $K_p > 0$, then $|e|\in [\lfloor m(V)/K_p\rfloor, \lceil m(V)/K_p \rceil]$ for all $e\in E$.
\end{enumerate}
\end{restatable}
\begin{proof}
We prove the lemma by induction on the potential function $\phi(m) \coloneqq  m(V)$. We note that $\phi(m) \geq 0$ since $m$ is a positive function. For the base case of induction, suppose that $\phi(m) = 0$. Then, we have that $V = \emptyset$ since $m$ is a positive function. Consequently, \Cref{alg:CoveringAlgorithm} is in its base case (Step \ref{algstep:CoveringAlgorithm:base-case}) and terminates. 
Moreover, Step \ref{algstep:CoveringAlgorithm:base-case} returns an empty hypergraph which satisfies properties (1)-(4) and so the lemma holds. Here, we note that $K_p \leq 0$ since otherwise, we have that $0 = m(\emptyset) \geq p(\emptyset) > 0$, a contradiction.

For the inductive case, suppose that $\phi(m) > 0$. We note that since $m$ is a positive function, we have that $V \not= \emptyset$ and $K_p \geq m(u) > 0$ for every $u \in V$. Consequently, \Cref{alg:CoveringAlgorithm} is in its recursive case. Since the execution of \Cref{alg:CoveringAlgorithm} is feasible for the input tuple $(p, m, J)$, we have that $Q(p, m)$ is a non-empty integral polyhedron. By \Cref{lem:CoveringAlgorithm:(p''m'')-hypothesis}(a), (c) and (d), the tuple $(p'', m'', J'')$ constructed by \Cref{alg:CoveringAlgorithm} satisfy the hypothesis of the current lemma. We note that the tuple $(p'', m'', J'')$ is the input to the subsequent recursive call of \Cref{alg:CoveringAlgorithm} by Step \ref{algstep:CoveringAlgorithm:recursion}. Moreover, $\phi(m'') < \phi(m)$ by \Cref{lem:CoveringAlgorithm:properties}(b). Thus, by induction, the subsequent recursive call to \Cref{alg:CoveringAlgorithm} terminates and returns a hypergraph $(H'' = (V'', E''), w'')$ satisfying properties (1)-(4) for the tuple $(p'', m'', J'')$. Consequently, the entire execution of \Cref{alg:CoveringAlgorithm} terminates within a finite number of recursive calls. Let $(H = (V, E), w)$ be the hypergraph returned by Step \ref{algstep:CoveringAlgorithm:recursion} and let $(H_0, w_0), \zeros, A, \alpha$ be as defined by \Cref{alg:CoveringAlgorithm} for the input tuple $(p, m, J)$. We first show that the hypergraph $(H, w)$ satisfies property (1). Let $X\subseteq V$ be arbitrary. Then, we have that 
\begin{align*}
    b_{(H,w)}(X)& = b_{(H'', w'')}(X -\zeros) + b_{(H_0, w_0)}(X) &\\
    &\geq p''(X -\zeros)  + b_{(H_0, w_0)}(X)&\\
    &\geq p'(X) +  b_{(H_0, w_0)}(X)&\\
    &= p(X).&
\end{align*}
Here, the first inequality holds by the inductive hypothesis property (1) and the second inequality is because $p'' = \functionContract{p'}{\zeros}$ by Step \ref{algstep:CoveringAlgorithm:def:p'-and-p''}.
Next, we show that hypergraph $(H, w)$ satisfies property (2). Let $u \in V$ be arbitrary. Then, we have the following:
\begin{align*}
    b_{(H, w)}(u) & = b_{(H'',w'')}(u) + b_{(H_0, w_0)}(u) &\\
    &= b_{(H'',w'')}(u) + \alpha\chi_A(u)&\\
    & = m''(u) +\alpha\chi_A(u)&\\
    &= m(u).&
\end{align*}
Here, the third equality is by the inductive hypothesis property (2). Next, we show that the hypergraph $(H, w)$ satisfies property (3). We have the following:
$$\sum_{e\in E}w(e) = \sum_{e \in E''}w(e) + w(A) = \sum_{e \in E''}w''(e) + \alpha = K_{p''} + \alpha = K_p.$$
Here, the third equality is by induction hypothesis property (3), and the final equality is by \Cref{lem:CoveringAlgorithm:(p''m'')-hypothesis}(b). 
Next, we show that property (4) holds. Let $e \in E$ be an arbitrary hyperedge. We consider two cases.
First, suppose that $e \in E - E''$. Then, we have that $e = A$ by Steps \ref{algstep:CoveringAlgorithm:def:H_0-and-w_0}, \ref{algstep:CoveringAlgorithm:def:G-and-c} and \ref{algstep:CoveringAlgorithm:return}. Consequently, we have that $|e| = |A|\in [\lfloor m(V)/K_p\rfloor, \lceil m(V)/K_p \rceil]$ by \Cref{lem:hyperedge-feasibility-for-integral-polyhedron}(4). Alternatively, suppose that $e \in E''$. We note that $|E''| > 0$ implies that $K_{p''} > 0$ by property (3) of the inductive hypothesis. Moreover, $|e|\in [\lfloor m''(V'')/K_{p''}\rfloor, \lceil m''(V'')/K_{p''} \rceil]$ by the induction hypothesis property (4). By \Cref{lem:CoveringAlgorithm:(p''m'')-hypothesis}(e), we have that $m''(V'')/K_{p''}\in [\lfloor m(V)/K_p\rfloor, \lceil m(V)/K_p \rceil]$. Thus, the interval $[\lfloor m''(V'')/K_{p''}\rfloor, \lceil m''(V'')/K_{p''} \rceil]$ is contained in the interval $[\lfloor m(V)/K_p\rfloor, \lceil m(V)/K_p \rceil]$ and so property (4) holds.
\end{proof}

We conclude the section by showing two properties of \Cref{alg:CoveringAlgorithm} which we will leverage in subsequent sections while analyzing the recursion depth and runtime of the algorithm. The first lemma below shows that any feasible execution of \Cref{alg:CoveringAlgorithm} can witness $\alpha = \alpha^{(5)}$ in at most one recursive call. The second lemma below shows that the set $A$ chosen during a feasible recursive call of \Cref{alg:CoveringAlgorithm} is not feasible for the subsequent recursive call. 
As a consequence, the value $\alpha$ computed in Step \ref{algstep:CoveringAlgorithm:def:alpha} during a feasible recursive call of \Cref{alg:CoveringAlgorithm} is the maximum value such that the input tuple $(p'', m'', J'')$ to the subsequent recursive call in Step \ref{algstep:CoveringAlgorithm:recursion} satisfy the hypothesis of \Cref{lem:CoveringAlgorithm:main}.

\begin{lemma}\label{lem:alpha-5}
Suppose that the input to  \Cref{alg:CoveringAlgorithm} is a tuple $(p,m,J)$, where $J\subseteq V$ is an arbitrary set, $p:2^V\rightarrow\Z$ is a set function, and $m:V\rightarrow\Z_{+}$ is a function such that $m(X)\ge p(X)$ for all $X\subseteq V$, $m(u)\le K_p$ for all $u\in V$, and $Q(p,m)$ is a non-empty integral polyhedron. 
Let $A, \alpha^{(5)}, \alpha, \zeros, p'', m''$ be as defined by \Cref{alg:CoveringAlgorithm}. If $\alpha= \alpha^{(5)}$ and $K_{p''}>0$, then $m(V)/K_p\not\in \Z$, $m''(V-\zeros)/K_{p''}\in \Z$, and $m''(V-\zeros)/K_{p''}\neq |A|$. 
\end{lemma}
\begin{proof}
We note that since $Q(p,m)$ is non-empty, we have that $V \not=\emptyset$ and consequently \Cref{alg:CoveringAlgorithm} is in its recursive case. Since $\alpha=\alpha^{(5)}$, it follows that $\alpha^{(5)}<\infty$ (because $\alpha^{(5)} = \alpha \leq \alpha^{(1)} \leq \max_{u\in V} m(u) < \infty$, where final inequality is because the range of $m$ is $\Z_+$). and hence, $m(V)/K_p\not\in \Z$ by definition of $\alpha^{(5)}$ in Step \ref{algstep:CoveringAlgorithm:def:alpha}. Next, we prove the conclusions about $m''(V-\zeros)/K_{p''}$. Let $d, r\in \Z_+$ be positive integers such that $m(V) = d\cdot K_p + r$ with $r<K_p$ (such a $d$ and $r$ exist since $m$ is a positive function, $m(Z)\ge p(Z)$ for every $Z\subseteq V$, and $m(V)/K_p\not\in \Z$). We note that $d=\lfloor m(V)/K_p \rfloor$ and $d+1=\lceil m(V)/K_p \rceil$.  
We consider two cases: Firstly, suppose that $|A|=\lfloor m(V)/K_p \rfloor$. It implies that $|A|=d$ and $\alpha=\alpha^{(5)}=K_p-r$. Then, we have the following:
\[
\frac{m''(V-\zeros)}{K_{p''}} 
= \frac{m(V)-\alpha |A|}{K_{p}-\alpha} 
= \frac{dK_p+r-(K_p-r)d}{K_{p}-K_p-r} 
= d+1,
\]
and hence, $m''(V-\zeros)/K_{p''}\in \Z$ and $m''(V-\zeros)/K_{p''}\neq |A|$. In the above, the first equality is by definition of $m''$ in Step \ref{algstep:CoveringAlgorithm:def:m'-and-m''} and by \Cref{lem:CoveringAlgorithm:(p''m'')-hypothesis}(b). 
Next, suppose that $|A|=\lceil m(V)/K_p \rceil$. It implies that $|A|=d+1$ and $\alpha = \alpha^{(5)}=r$. Then, we have the following: 
\[
\frac{m''(V-\zeros)}{K_{p''}} 
= \frac{m(V)-\alpha |A|}{K_{p}-\alpha} 
= \frac{dK_p+r-r(d+1)}{K_{p}-r} 
= d,
\]
and hence, $m''(V-\zeros)/K_{p''}\in \Z$ and $m''(V-\zeros)/K_{p''}\neq |A|$. In the above, the first equality is by definition of $m''$ in Step \ref{algstep:CoveringAlgorithm:def:m'-and-m''} and by \Cref{lem:CoveringAlgorithm:(p''m'')-hypothesis}(b).

\end{proof}

\begin{lemma}\label{lem:CoveringAlgorithm:A-not-feasible-for-p''-and-m''}
Suppose that the input to  \Cref{alg:CoveringAlgorithm} is a tuple $(p,m,J)$, where $J\subseteq V$ is an arbitrary set, $p:2^V\rightarrow\Z$ is a set function, and $m:V\rightarrow\Z_{+}$ is a function such that $m(X)\ge p(X)$ for all $X\subseteq V$, $m(u)\le K_p$ for all $u\in V$, and $Q(p,m)$ is a non-empty integral polyhedron. 
Let $A, \zeros, p'', m''$ be as defined by \Cref{alg:CoveringAlgorithm}. If $A\cap \zeros = \emptyset$ and $K_{p''}>0$, then the indicator vector $\chi_A\in \{0, 1\}^{V-\zeros}$ is not in $Q(p'', m'')$.  
\end{lemma}
\begin{proof}
We note that since $Q(p,m)$ is non-empty, we have that $V \not=\emptyset$ and consequently \Cref{alg:CoveringAlgorithm} is in its recursive case. By way of contradiction, suppose that $A\cap \zeros = \emptyset$ and the indicator vector $\chi_A\in \{0, 1\}^{V-\zeros}$ is in $Q(p'', m'')$. Consequently, the set $A$ satisfies the properties \Cref{lem:hyperedge-feasibility-for-integral-polyhedron}(1)-(4) with respect to the functions $p''$ and $m''$.  
We recall the definition of $\alpha$ below:

$$\alpha \coloneqq  \min\begin{cases}
        \alpha^{(1)} \coloneqq  \min\left\{m(u) : u \in A\right\} \\
        \alpha^{(2)} \coloneqq  \min\left\{K_p - p(X) : X \subseteq V - A\right\} \\
        \alpha^{(3)} \coloneqq  \min\left\{K_p - m(u) : u \in V - A\right\} \\
        \alpha^{(4)} \coloneqq  \min\left\{\floor{\frac{m(X) - p(X)}{|A\cap X| - 1}} : |A\cap X| \geq 2\right\}\\
        \alpha^{(5)} \coloneqq \begin{cases}
            m(V)\mod K_p & \text{if $|A| = \ceil{m(V)/K_p} > m(V)/K_p$,}\\
            K_p - (m(V)\mod K_p) & \text{if $|A| = \floor{m(V)/K_p} < m(V)/K_p$,}\\
            +\infty & \text{otherwise.}\end{cases}
    \end{cases}$$

First, suppose that $\alpha = \alpha^{(1)}$. Let $u \in A$ be a vertex such that $\alpha = m(u)$. Then, we have that $u \in \zeros \not = \emptyset$. Thus, we have that $u \in A \subseteq V - \zeros \subseteq V - \{u\}$, a contradiction.

Next, suppose that $\alpha = \alpha^{(2)}$. Let $X\subseteq V - A$ be a set such that $\alpha = K_p - p(X)$. Then, we have that 
$p''(X) = p(X) = K_{p} - \alpha = K_{p''}.$
Here, the first equality is because $X \subseteq V - A$ and the final equality is by  
\Cref{lem:CoveringAlgorithm:(p''m'')-hypothesis}(b). Thus, the set $X$ is a $p''$-maximizer and hence $X\cap A \not = \emptyset$ since the set $A$ satisfies \Cref{lem:hyperedge-feasibility-for-integral-polyhedron}(1) with respect to the functions $p''$ and $m''$, 
a contradiction to $X\subseteq V-A$.

Next, suppose $\alpha = \alpha^{(3)}$. Let $u \in V - A$ be a vertex such that $\alpha = K_p - m(u)$. Then, we  have that $m''(u) = m(u) = K_p - \alpha = K_{p''}$. Here, the first equality is because $u \in V - A$, and the final equality is by \Cref{lem:CoveringAlgorithm:(p''m'')-hypothesis}(b). Thus, $m''(u) = K_{p''}$ and consequently $u \in A$ since the set $A$ satisfies  \Cref{lem:hyperedge-feasibility-for-integral-polyhedron}(2) with respect to the functions $p''$ and $m''$, a contradiction to $u\in V-A$.

Next, suppose that $\alpha = \alpha^{(4)}$. Let $X \subseteq V$ be such that $|A\cap X| \geq 2$ and $\alpha = \floor{\frac{m(X) - p(X)}{|A\cap X| - 1}}$. First, we consider the case where $\alpha =\frac{m(X)-p(X)}{|A\cap X|-1}$.
Then, we have that 
$m''(X) = m(X) - \alpha|A\cap X| = p(X) - \alpha = p''(X)$, where the first and final equalities are because $|A\cap X| \geq 1$. This gives us the required contradiction as follows: $p''(X) - 1 \leq m''(X) - |A \cap X| = p''(X) - |A \cap X| < p''(X) - 1$, a contradiction. Here, the first inequality is because the set $A$ satisfies  \Cref{lem:hyperedge-feasibility-for-integral-polyhedron}(3) with respect to the functions $p''$ and $m''$ and the final inequality is because $|A\cap X| \geq 2$.
Next, we consider the case where $\alpha < \frac{m(X)-p(X)}{|A\cap X|-1}$. Then, it follows that $\alpha+1 > \frac{m(X)-p(X)}{|A\cap X|-1}$. Rewriting this  shows that $p(X)-\alpha+|A\cap X|-1> m(X)-\alpha |A\cap X|$. Consequently, $m''(X)=m(X)-\alpha |A\cap X|< p(X)-\alpha + |A\cap X|-1=p''(X)+|A\cap X|-1$. This gives us the required contradiction as follows: $p''(X)-1\le m''(X)-|A\cap X|< p''(X)-1$, a contradiction. Here, the first inequality is because the set $A$ satisfies \Cref{lem:hyperedge-feasibility-for-integral-polyhedron}(3) with respect to the functions $p''$ and $m''$.

Next, suppose that $\alpha = \alpha^{(5)}$. By \Cref{lem:alpha-5}, we have that $m''(V-\zeros)/K_{p''}$ is an integer and $m''(V-\zeros)/K_{p''} \neq |A|$. This contradicts \Cref{lem:hyperedge-feasibility-for-integral-polyhedron}(4) with respect to the functions $p''$ and $m''$. 
\end{proof}

\subsection{Number of Recursive Calls of \Cref{alg:CoveringAlgorithm} with $\alpha \in \{\alpha^{(1)}, \alpha^{(2)}, \alpha^{(3)}, \alpha^{(5)}\}$}\label{sec:CoveringAlgorithm:number-recursive-calls-with-alpha=alpha1-alpha2-alpha3}
In this section, we bound the number of recursive calls of \Cref{alg:CoveringAlgorithm} that witness $\alpha \in \{\alpha^{(1)}, \alpha^{(2)}, \alpha^{(3)}, \alpha^{(5)}\}$. The main result of this section is \Cref{cor:num-createsteps-alpha=alpha1_alpha2_alpha3} which we show at the end of the section.

\paragraph{Notation. }We first set up some convenient notation that will be used in the remainder of this section as well as in subsequent sections. 
By \Cref{lem:CoveringAlgorithm:main}, the number of recursive calls made by the algorithm is finite.
We will use $\ell$ to denote the depth of recursion of the algorithm. We will refer to the recursive call at depth $i$ as \emph{recursive call $i$} or the \emph{$i^{th}$ recursive call}. Let $V_i$ denote the ground set at the start of recursive call $i$, and $p_i:2^{V_i}\rightarrow\Z$,  $m_i:V_i\rightarrow\Z_{+}$, and $J_i \subseteq V_i$ denote the inputs to recursive call $i$.  
Furthermore, let $A_i, \alpha_i, \alpha^{(1)}_i, \alpha^{(2)}_i, \alpha^{(3)}_i, \alpha^{(4)}_i, \alpha^{(5)}_i, \zeros_i, p'_i, m'_i$ denote the quantities $A, \alpha, \alpha^{(1)}, \alpha^{(2)}, \alpha^{(3)}, \alpha^{(4)}, \alpha^{(5)}, \zeros, p', m'$ defined by \Cref{alg:CoveringAlgorithm} at recursive call $i$. For notational convenience, we define $\zeros_\ell, \calD_{\ell} = \emptyset$.
Finally, let 
$\left(H_i = \left(V_i, E_i\right), w_i\right)$ denote the hypergraph returned by recursive call $i$.

We note that \Cref{lem:CoveringAlgorithm:(p''m'')-hypothesis} and induction on the recursion depth $i$ immediately imply the following lemma which says that for every $i \in [\ell]$, the input tuple $(p_i, m_i, J_i)$ satisfies the hypothesis of \Cref{lem:CoveringAlgorithm:main}.

\begin{lemma}\label{lem:CoveringAlgorithm:(p_im_iJ_i)-hypothesis}
Consider a feasible execution of \Cref{alg:CoveringAlgorithm} with $\ell \in \Z_+$ recursive calls. For every $i \in [\ell]$, suppose that the input to  the $i^{th}$ recursive call 
 of \Cref{alg:CoveringAlgorithm} is a tuple $(p_i,m_i,J_i)$, where $J_1 \subseteq V_1$ is an arbitrary set, and the functions  $p_1:2^{V}\rightarrow\Z$ and
 $m_1:V\rightarrow\Z_{+}$ 
 are such that $m_1(X)\ge p_1(X)$ for every $X\subseteq V$ and $m_1(u)\le K_{p_1}$ for every $u\in V_1$. Then, for every $i \in [\ell]$, we have that the function $p_i:2^{V_i}\rightarrow\Z$  and positive function $m_i:V_i\rightarrow\Z_+$ satisfy $m_i(X) \ge p_i(X)$ for all $X \subseteq V_i$ and $m_i(u) \leq K_{p_i}$ for all $u \in V_i$.
\end{lemma}
\paragraph{Set Families. } We now define certain set families that will be useful in the analysis. We recall that $\zeros_i = \{u \in V_i : m_i'(u) = 0\}$. We let the family $\cumulativeZeros{i} \coloneqq  \cup_{j \in [i]}\zeros_{i}$.
We also let $\calD_i \coloneqq  \{u \in V_i : m_i(u) = K_{p_i}\}$ and $\calD_{i}' \coloneqq  \{u \in V_i : m_i'(u) = K_{p'_i}\}$.
We use $\minimalMaximizerFamily{i}$ to denote the families of \emph{minimal} $p_i$-maximizers. \Cref{lem:persistance-of-maximizers} below shows the progression of minimal $p_i$-maximizer families across recursive calls of an execution of \Cref{alg:CoveringAlgorithm}. We will also be interested in families of all minimal
maximizers of the input functions witnessed by the algorithm up to a given recursive call. Formally, we
define the family
$\cumulativeMinimalMaximizerFamily{i}\coloneqq  \cup_{j \in [i]} \minimalMaximizerFamily{j}$.  
\Cref{lem:Progression-of-set-families:main} below summarizes the properties of the stated families that will help in subsequent analysis. 

\begin{lemma}\label{lem:persistance-of-maximizers}
    Consider a feasible execution of \Cref{alg:CoveringAlgorithm} with $\ell \in \Z_+$ recursive calls. For every $i \in [\ell]$, suppose that the input to  the $i^{th}$ recursive call 
 of \Cref{alg:CoveringAlgorithm} is a tuple $(p_i,m_i,J_i)$, where $J_1 \subseteq V_1$ is an arbitrary set, and the functions  $p_1:2^{V}\rightarrow\Z$ and
 $m_1:V\rightarrow\Z_{+}$ 
 are such that $m_1(X)\ge p_1(X)$ for every $X\subseteq V$ and $m_1(u)\le K_{p_1}$ for every $u\in V_1$. Then, for every $i \in [\ell-1]$ we have the following:
 \begin{enumerate}[label=(\alph*)]
    \item if $Y \subseteq V_i$ is a $p_i$-maximizer, then $Y$ is also a $p_{i}'$-maximizer, and 
     \item if $Y \subseteq V_i$ is a $p_i'$-maximizer such that $Y - \zeros_i \not = \emptyset$, then $Y - \zeros_i$ is a $p_{i+1}$-maximizer.
 \end{enumerate}
\end{lemma}
\begin{proof}
    By \Cref{lem:CoveringAlgorithm:(p''m'')-hypothesis} (a), (c), and (d), we have that for every $i \in [\ell]$, the function $m_i: V_i\rightarrow \Z$ is a positive function, and the functions $p_i:2^{V_i}\rightarrow\Z$ and 
    $m_i$ 
    are such that $m_i(X)\ge p_i(X)$ for every $X\subseteq V_i$ and $m_i(u)\le K_{p_i}$ for every $u\in V_i$. We prove both claims below for a fixed $i\in [\ell]$.
    \begin{enumerate}[label=(\alph*)]
        \item Let $Y \subseteq V_i$ be a $p_i$-maximizer. By \Cref{lem:hyperedge-feasibility-for-integral-polyhedron}(1), the set $A_i$ is a transversal for the family of $p_i$-maximizers. Thus, we have that $p_i'(Y) = p_i(Y) - \alpha_i = K_{p_i} - \alpha_i = K_{p_i'}$, where the first equality is because $A_i\cap Y \not = \emptyset$ and the third equality is by \Cref{lem:Covering-Algorithm:hypothesis-after-createstep}(b).
        \item Let $Y \subseteq V_i$ be a $p_i'$-maximizer such that $Y -\zeros_i \neq \emptyset$.
     We have the following:
    $$K_{p_{i+1}} \geq p_{i+1}(Y - \zeros_i) = \max\{p_i'(Y\cup R):R\subseteq \zeros_i\} \geq p_i'(Y) = K_{p_i'} = K_{p_{i+1}},$$
    where the final equality is by \Cref{lem:CoveringAlgorithm:(p''m'')-hypothesis}(b). Thus, all inequalities are equations and we have that $p_{i+1}(Y - \zeros_i) = K_{p_{i+1}}$.
    \end{enumerate}
\end{proof}

\begin{lemma}\label{lem:Progression-of-set-families:main}
Consider a feasible execution of \Cref{alg:CoveringAlgorithm} with $\ell \in \Z_+$ recursive calls. For every $i \in [\ell]$, suppose that the input to  the $i^{th}$ recursive call 
 of \Cref{alg:CoveringAlgorithm} is a tuple $(p_i,m_i,J_i)$, where $J_1\subseteq V$ is an arbitrary set, and the functions $p_1:2^{V}\rightarrow\Z$ and 
 $m_1:V\rightarrow\Z_{+}$ 
 are such that $m_1(X)\ge p_1(X)$ for every $X\subseteq V$ and $m_1(u)\le K_{p_1}$ for every $u\in V_1$. Then, for every $i \in [\ell-1]$ we have the following:
    \begin{enumerate}[label=(\alph*)]
        \item 
   $\zeros_{\leq i} \subseteq \zeros_{\leq i+1}$; furthermore, $\alpha_i = \alpha_{i}^{(1)}$ if and only if $\zeros_{i} \not = \emptyset$ (i.e., $\zeros_{\leq i} \subsetneq \zeros_{\leq i+1}$),
   \item $\cumulativeMinimalMaximizerFamily{i}\subseteq  \cumulativeMinimalMaximizerFamily{i+1}$; furthermore, if $\alpha_{i} = \alpha_{i}^{(2)} < \alpha_{i}^{(1)}$, then $\cumulativeMinimalMaximizerFamily{i}\subsetneq 
   \cumulativeMinimalMaximizerFamily{i+1}$,
   \item 
   $\calD_{i}\subseteq \calD_{i}' \subseteq \calD_{ i+1}$; furthermore, if $\alpha_i = \alpha_{i}^{(3)}$, then $\calD_{i}\subsetneq  \calD_{i}'$, and
   \item  if $\cumulativeMinimalMaximizerFamily{i} = \cumulativeMinimalMaximizerFamily{i+1}$ and $\zeros_i = \emptyset$, then $\minimalMaximizerFamily{i} = \minimalMaximizerFamily{i+1}$.
    \end{enumerate}
\end{lemma}
\begin{proof}
By \Cref{lem:CoveringAlgorithm:(p''m'')-hypothesis} (a), (c), and (d), we have that for every $i \in [\ell]$, the function $m_i: V_i\rightarrow \Z$ is a positive function, and the functions $p_i:2^{V_i}\rightarrow\Z$ and $m_i$ are such that $m_i(X)\ge p_i(X)$ for every $X\subseteq V_i$ and $m_i(u)\le K_{p_i}$ for every $u\in V_i$. We prove each claim separately below.
\begin{enumerate}[label=(\alph*)]
        \item We have that $\calZ_{\leq i} \subseteq \calZ_{\leq i+1}$ by definition. We now show the second part of the claim. For the forward direction, suppose that $\alpha_i = \alpha_i^{(1)}$. Let $u \in A_i$ be a vertex  such that $m_i(u) = \alpha_i$ (such a vertex exists since $\alpha_i = \alpha_i^{(1)}$). Then, we have that $m_{i}'(u) = m_i(u) - \alpha_i\chi_{A_i}(u) = 0$. Thus, $u \in \zeros_{i}$. The reverse direction of the claim follows because the function $m$ is a positive function.

\item  We note that $\cumulativeMinimalMaximizerFamily{i} \subseteq \cumulativeMinimalMaximizerFamily{i+1}$ follows by definition.  We now show the second part of the claim. Suppose that $\alpha_i = \alpha_i^{(2)} < \alpha_{i}^{(1)}$. Then by part (a) of the current lemma, we have that $\zeros_i = \emptyset$, and consequently $p_i' = p_{i+1}$ by the definition of the two functions.
We consider the family $ \minimalMaximizerFamilyAfterCreate{i}$ of minimal $p_i'$-maximizers. By \Cref{lem:persistance-of-maximizers}(a), we have that $\minimalMaximizerFamily{i}\subseteq \minimalMaximizerFamilyAfterCreate{i} = \minimalMaximizerFamily{i+1}$, and thus, $\cumulativeMinimalMaximizerFamily{i} \subseteq \cumulativeMinimalMaximizerFamily{i} \cup \minimalMaximizerFamilyAfterCreate{i} = \cumulativeMinimalMaximizerFamily{i+1}$, where the equalities in both the previous expressions are because $p_i' = p_{i+1}$. We now show that the first inclusion is strict. For convenience, we let $\cumulativeMinimalMaximizerFamilyAfterCreate{i} \coloneqq  \cumulativeMinimalMaximizerFamily{i} \cup \minimalMaximizerFamilyAfterCreate{i}$. 
By way of contradiction, suppose that $\cumulativeMinimalMaximizerFamily{i} = \cumulativeMinimalMaximizerFamilyAfterCreate{i}$. 
Let $X \subseteq V_i - A_i$ be a set such that  $\alpha_i = \alpha_{i}^{(2)} = K_{p_i} - p_i(X)$ (such a set exists since $\alpha_i=\alpha_i^{(2)}$). Then, we have that $p_{i}'(X) = p_i(X) = K_{p_i}  - \alpha_i^{(2)} = K_{p_i'}.$
Here, the first equality is because $X \subseteq V_i - A_i$, the second equality is because our choice of $X$ satisfies $\alpha_{i}^{(2)} = K_{p_i} - p_i(X)$, and the final inequality is by \Cref{lem:Covering-Algorithm:hypothesis-after-createstep}(b). Thus, the set $X$ is a $p_{i}'$-maximizer. Furthermore, the set $X$ is not a $p_i$-maximizer since the set $A_i$ is a transversal for the family of $p_i$-maximizers by \Cref{lem:hyperedge-feasibility-for-integral-polyhedron}(1), but $A_i\cap X = \emptyset$.
Consequently, there exists a set $Y \subseteq X$ such that $Y \in \minimalMaximizerFamilyAfterCreate{i} - \minimalMaximizerFamily{i}$. Since $\cumulativeMinimalMaximizerFamily{i} = \cumulativeMinimalMaximizerFamilyAfterCreate{i}$, we have that $Y \in \cumulativeMinimalMaximizerFamily{i-1}$. In particular, there exists a recursive call $j \in [i-1]$ such that $Y \in \minimalMaximizerFamily{j}$. Since $Y\subseteq V_i$, by \Cref{lem:persistance-of-maximizers}(a), (b) and induction on $j$, we have that $Y \in \minimalMaximizerFamily{i}$. Thus, the set $Y$ is a $p_i$-maximizer and consequently $A_i \cap Y \not = \emptyset$ by \Cref{lem:hyperedge-feasibility-for-integral-polyhedron}(1), a contradiction to $A_i \cap X = \emptyset$.

        \item First, we show that $\calD_{i} \subseteq \calD_{i}'$ for every $i \in [\ell]$. Let $u \in \calD_{i}$. Then we have the following:
    $$m_{i}'(u) = m_i(u) - \alpha_i\chi_{A_i}(u) = K_{p_i} - \alpha_i\chi_{A_i}(u) = K_{p_i} - \alpha_i  = K_{p_{i}'}.$$
    Here, the second equality is because $u \in \calD_i$, the third equality is because $\calD_i \subseteq A_i$ by \Cref{lem:hyperedge-feasibility-for-integral-polyhedron}(2), and the final equality is by \Cref{lem:Covering-Algorithm:hypothesis-after-createstep}(b). Thus, $u \in \calD_{i}'$ and we have that $\calD_i \subseteq \calD_{i}'$. Next, we show that $\calD_{i}'\subseteq \calD_{i+1}$. Let $u \in\calD_i'$. We note that if $u \in \zeros_i$, then we have that $0 = m_i'(u) = K_{p_i'}$, and thus $i = \ell$, contradicting $i \in [1, \ell - 1]$. Thus, $u \not \in \zeros_i$, i.e. $u \in V_{i+1}$.  Then, we have that $m_{i+1}(u) = m_{i}'(u) = K_{p_i'} = K_{p_{i+1}}$, where the final equality is by \Cref{lem:CoveringAlgorithm:(p''m'')-hypothesis}(b). Thus, $u \in \calD_{i+1}$ and we have that $\calD_{i}'\subseteq \calD_{i+1}$. 
    
    We now show the second part of the claim. Suppose that $\alpha_{i} = \alpha_i^{(3)}$ and let $v \in V_i - A_i$ be a vertex such that $\alpha_i = K_{p_i} - m_i(v)$ (such a vertex exists since $\alpha_i=\alpha_i^{(3)}$). Then, we have that:
    $$m_{i}'(v) = m_i(v) - \alpha_i\chi_{A_i}(v) = m_i(v) = K_{p_i} - \alpha_i = K_{p_{i+1}}.$$
    Here, the second and third equalities are by the choice of $v$ and the final equality is by \Cref{lem:Covering-Algorithm:hypothesis-after-createstep}(b). Thus, $v\in \calD_{i+1} \backslash \calD_{i}$.

    \item We note that since $\zeros_i = \emptyset$, the function $p_i' = p_{i+1}$. Consequently, $\minimalMaximizerFamily{i+1} = \minimalMaximizerFamilyAfterCreate{i}$ and $\cumulativeMinimalMaximizerFamily{i+1} = \cumulativeMinimalMaximizerFamilyAfterCreate{i}$. Since $\cumulativeMinimalMaximizerFamily{i+1} = \cumulativeMinimalMaximizerFamily{i}$, we have that $\cumulativeMinimalMaximizerFamily{i} = \cumulativeMinimalMaximizerFamilyAfterCreate{i}$. By \Cref{lem:persistance-of-maximizers}(a), we have that $\minimalMaximizerFamily{i} \subseteq \minimalMaximizerFamilyAfterCreate{i}$. Thus, it suffices to show the reverse inclusion $\minimalMaximizerFamilyAfterCreate{i} \subseteq \minimalMaximizerFamily{i}$
     By way of contradiction, let $X \in \minimalMaximizerFamilyAfterCreate{i} - \minimalMaximizerFamily{i} \not = \emptyset$. Since $\minimalMaximizerFamilyAfterCreate{i} \subseteq \cumulativeMinimalMaximizerFamilyAfterCreate{i} = \cumulativeMinimalMaximizerFamily{i}$, we have that $X \in \cumulativeMinimalMaximizerFamily{i}$. By definition of the family $\cumulativeMinimalMaximizerFamily{i}$ and our choice of the set $X$, there exists $j \in [i-1]$ and $Y \in \minimalMaximizerFamily{j} \cup \minimalMaximizerFamilyAfterCreate{j}$ such that $Y\cap V_{i} = X$. By \Cref{lem:persistance-of-maximizers}(a), (b) and induction on $j$, we have that $X \in \minimalMaximizerFamily{i}$, a contradiction to the choice of set $X$.
    \end{enumerate}
\end{proof}


We now show the main result of the section which says that the number of recursive calls of a feasible execution of \Cref{alg:CoveringAlgorithm} for which $\alpha \in \left\{\alpha^{(1)}, \alpha^{(2)}, \alpha^{(3)}, \alpha^{(5)}\right\}$ is at most $|\cumulativeMinimalMaximizerFamily{\ell}| + 2|V|+1$. We derive this as a consequence of \Cref{lem:Progression-of-set-families:main}.

\begin{corollary}\label{cor:num-createsteps-alpha=alpha1_alpha2_alpha3}
Consider a feasible execution of \Cref{alg:CoveringAlgorithm} with $\ell \in \Z_+$ recursive calls. Suppose that the input to \Cref{alg:CoveringAlgorithm} is a tuple $(p_1,m_1,J_1)$, where $J_1 \subseteq V$ is an arbitrary set, and the functions $p_1:2^{V}\rightarrow\Z$ and $m_1:V\rightarrow\Z_{+}$  are such that $m_1(X)\ge p_1(X)$ for every $X\subseteq V$ and $m_1(u)\le K_{p_1}$ for every $u\in V_1$. 
Then, we have that $$\left|\left\{i \in [\ell -1]: \alpha_i \in \left\{\alpha_i^{(1)}, \alpha_i^{(2)}, \alpha_i^{(3)}, \alpha_i^{(5)}\right\}\right\}\right| \leq \left|\cumulativeMinimalMaximizerFamily{\ell}\right| + 2|V|+1.$$

\end{corollary}
\begin{proof}
    We first show that the number of iterations $i$ for which $\alpha_i=\alpha_i^{(5)}$ is at most one. Let $i\in [\ell-1]$ be the least index such that $\alpha_i = \alpha_i^{(5)}$. Then, by \Cref{lem:alpha-5}, we have that $m_i(V_i)/K_{p_i}\not\in Z$ and one of the following holds: either $K_{p_{i+1}}=0$ or $m_{i+1}(V_{i+1})/K_{p_{i+1}}\in \Z$. If $K_{p_{i+1}}=0$, then $i=\ell-1$ and this is the only iteration where $\alpha_i = \alpha_i^{(5)}$. If $m_{i+1}(V_{i+1})/K_{p_{i+1}}\in \Z$, then  $\alpha_j^{(5)}=\infty$ for every $j\ge i+1$ by definition and consequently, $\alpha\neq \alpha_j^{(5)}$ for every $j\ge i+1$. 

    By the above arguments, it suffices to show that the number of iterations $i$ for which $\alpha_i\in \{\alpha_i^{(1)}, \alpha_i^{(2)}, \alpha_i^{(3)}\}$ is at most $\left|\cumulativeMinimalMaximizerFamily{\ell}\right| + 2|V|$. For this, we define a potential function $\phi:[\ell]\rightarrow\Z_{\geq 0}$ as follows: for every $i \in [\ell]$,
    $$\phi(i) \coloneqq  |\zeros_{\leq i}| + |\cumulativeMinimalMaximizerFamily{i}| + |\calD_{i}|.$$
    By \Cref{lem:Progression-of-set-families:main}, we have that the function $\phi$ is non-decreasing as each of the three terms of $\phi$ are non-decreasing. By \Cref{lem:Progression-of-set-families:main}, we also have that if $\alpha_i \in \left\{\alpha_i^{(1)}, \alpha_i^{(2)}, \alpha_i^{(3)}\right\}$, then $\phi(i) < \phi(i+1)$. In particular, we have that the number of recursive calls of the algorithm with $\alpha \in \left\{\alpha^{(1)}, \alpha^{(2)}, \alpha^{(3)}\right\}$ is at most $\phi(\ell) - \phi(0) \leq |\zeros_{\leq \ell}| + |\calF_{p_{\leq \ell}}| +|\calD_{\ell}|\leq 2|V| + |\calF_{{\leq \ell}}|$.
\end{proof}

\subsection{Feasibility of Execution for Skew-supermodular Functions}\label{sec:skew-supermodularity-based-properties-of-algorithm}
In this section, we show that the execution of \Cref{alg:CoveringAlgorithm} is feasible under a skew-supermodularity assumption on the input function. We recall that a call to \Cref{alg:CoveringAlgorithm} is feasible if the input functions $p$ and $m$ are such that the polyhedron $Q(p,m)$ is non-empty and integral. Moreover, the execution of \Cref{alg:CoveringAlgorithm} is feasible if for every recursive call to the algorithm during its execution, either \Cref{alg:CoveringAlgorithm} is in its base case or recursive call is feasible 
Bern\'{a}th and Kir\'{a}ly showed the following result. 

\begin{lemma}[\hspace{-1sp}\cite{Bernath-Kiraly}]\label{lem:WeakCoverTwoFunctionsViaUniformHypergraph:BK-integral-g-polymatroid}
    Let $q, r:2^V \rightarrow\Z$ be skew-supermodular functions and
    $p:2^V \rightarrow\Z$ be the function defined as $p(X) \coloneqq  \max\{q(X), r(X)\}$ for every $X\subseteq V$. Furthermore, let $m:V\rightarrow\Z_{\geq 0}$ be a non-negative function satisfying $m(X)\ge p(X)$ for every $X\subseteq V$ and $m(u)\le K_p$ for every $u\in V$. Then, the polyhedron $Q(p,m)$ is a non-empty integral polyhedron. 
\end{lemma}

Using \Cref{lem:WeakCoverTwoFunctionsViaUniformHypergraph:BK-integral-g-polymatroid} in conjunction with \Cref{lem:CoveringAlgorithm:(p_im_iJ_i)-hypothesis} leads to the following corollary showing feasibility of the execution of \Cref{alg:CoveringAlgorithm}.  
\begin{corollary}\label{coro:feasibility-of-execution-for-simultaneous}
Let $q, r:2^V \rightarrow\Z$ be skew-supermodular functions. 
Suppose that the input to \Cref{alg:CoveringAlgorithm}  is a tuple $(p,m,J)$, where $J\subseteq V$ is an arbitrary set,  $p:2^V\rightarrow\Z$ is a function defined by $p(X)\coloneqq \max\{q(X), r(X)\}$ for every $X\subseteq V$ with $K_p>0$ and $m:V\rightarrow\Z_{+}$ is a positive integer-valued function such that $m(X)\ge p(X)$ for every $X\subseteq V$ and $m(u)\le K_p$ for every $u\in V$. Let $\zeros, (H_0, w_0), p'', m'', J''$ be as defined by \Cref{alg:CoveringAlgorithm}. Then, 
the input tuple $(p'', m'', J'')$ for the subsequent recursive call in the execution of the algorithm either has $K_{p''} = 0$ or satisfies the following properties:
\begin{enumerate}[label = {(\arabic*)}]
    \item the functions $q''$ and $r''$ defined by 
    $q''\coloneqq \functionContract{(q - b_{(H_0, w_0)})}{\zeros}$ and $r''\coloneqq \functionContract{(r - b_{(H_0, w_0)})}{\zeros}$, where $(H_0, w_0), \zeros$ are as defined by \Cref{alg:CoveringAlgorithm}, 
    are skew-supermodular; furthermore $p''(X)=\max\{q''(X), r''(X)\}$ for every $X\subseteq V$, 
    \item $m''$ is a positive integer-valued function,
    \item $m''(X)\ge p''(X)$ for every $X\subseteq V$,
    \item $m''(u)\le K_{p''}$ for every $u\in V''$, and 
    \item $Q(p'', m'')$ is a non-empty integral polyhedron. 
\end{enumerate}
In particular, the execution of \Cref{alg:CoveringAlgorithm} on the input tuple $(p, m, J)$ is a feasible execution with finite recursion depth. 
\end{corollary}
\begin{proof}
We note that the function $b_{(H_0, w_0)}$ is submodular, and consequently, the functions $q''$ and $r''$ are skew-supermodular by definition. Let $V'' := V - \zeros$. Consider a set $X\subseteq V''$. Then, we have the following:
        \begin{align*}
           p''(X) &= \functionContract{(p - b_{(H_0, w_0)})}{\zeros} (X)&\\
            &= \functionContract{\left(\max\{q, r\} -  b_{(H_0, w_0)}\right)}{\zeros}(X)&\\
            &= \functionContract{\left(\max\{q- b_{(H_0, w_0)}, r- b_{(H_0, w_0)}\} \right)}{\zeros}(X)&\\
            &=\max\{q''(X), r''(X)\},&
        \end{align*}
        where the first equality is by Step \ref{algstep:CoveringAlgorithm:def:p'-and-p''}.
        Thus, property (1) of the claim holds. Furthermore, properties (2), (3) and (4) of the claim hold by \Cref{lem:CoveringAlgorithm:(p''m'')-hypothesis}. Finally, property (5) of the claim holds by \Cref{lem:WeakCoverTwoFunctionsViaUniformHypergraph:BK-integral-g-polymatroid} and the previously shown properties (2), (3) and (4) of the claim. 
        
        We now show that the execution of \Cref{alg:CoveringAlgorithm} on the input tuple $(p, m, J)$ is a feasible execution with finite recursion depth using (reverse) induction on the value $K_p$. For the base case, suppose that $K_p = 0$. Then, we have that $V = \emptyset$ since $m:V\rightarrow\Z_+$ is a positive function with $m(u) \leq K_p$ for every $u \in V$. Thus, Step \ref{algstep:CoveringAlgorithm:base-case} returns the empty hypergraph and the claim holds. For the induction step, suppose that $K_p > 0$.
We recall that $K_{p''} = K_p - \alpha \leq K_p - 1$, where the equality is by \Cref{lem:CoveringAlgorithm:(p''m'')-hypothesis} and the inequality is by \Cref{lem:CoveringAlgorithm:properties}. By the inductive hypothesis, we have that the execution of \Cref{alg:CoveringAlgorithm} on the input tuple $(p'', m'', J'')$, where $J'' \subseteq V''$ is the set defined by \Cref{alg:CoveringAlgorithm}, is a feasible execution with finite recursion depth. Furthermore, by \Cref{lem:WeakCoverTwoFunctionsViaUniformHypergraph:BK-integral-g-polymatroid}, $Q(p,m)$ is a non-empty integral polyhedron. Consequently the execution of \Cref{alg:CoveringAlgorithm} on the input tuple $(p, m, J)$ is also feasible execution with finite recursion depth.
\end{proof}

\Cref{lem:WeakCoverTwoFunctionsViaUniformHypergraph:BK-integral-g-polymatroid} implies the following for skew-supermodular functions. 

\begin{lemma}[\hspace{-1sp}\cite{Bernath-Kiraly}]\label{thm:WeakCoverViaUniformHypergraph:BK-integral-g-polymatroid}
    Let $p:2^V \rightarrow\Z$ be a skew-supermodular function and $m:V\rightarrow\Z_{\geq 0}$ be a non-negative function satisfying $m(X)\ge p(X)$ for every $X\subseteq V$ and $m(u)\le K_p$ for every $u\in V$. Then, the polyhedron $Q(p,m)$ is a non-empty integral polyhedron. 
\end{lemma}

Using \Cref{thm:WeakCoverViaUniformHypergraph:BK-integral-g-polymatroid} in conjunction with \Cref{lem:CoveringAlgorithm:(p_im_iJ_i)-hypothesis}  leads to the following corollary (similar to the proof of \Cref{coro:feasibility-of-execution-for-simultaneous}). 

\begin{corollary}\label{coro:feasibility-of-execution}
Suppose that the input to \Cref{alg:CoveringAlgorithm}  is a tuple $(p,m,J)$, where $J\subseteq V$ is an arbitrary set,  $p:2^V\rightarrow\Z$ is a skew-supermodular function with $K_p > 0$, and $m:V\rightarrow\Z_{+}$ is a positive integer-valued function such that $m(X)\ge p(X)$ for every $X\subseteq V$ and $m(u)\le K_p$ for every $u\in V$. Then, 
the input $(p'': 2^{V''}\rightarrow \Z, m'': V'' \rightarrow \Z, J''\subseteq V'')$ for every recursive call in the execution of the algorithm either has $K_{p''} = 0$ or satisfies the following properties:
\begin{enumerate}[label = {(\arabic*)}]
    \item $p''$ is skew-supermodular,
    \item $m''$ is a positive integer-valued function,
    \item $m''(X)\ge p''(X)$ for every $X\subseteq V$,
    \item $m''(u)\le K_{p''}$ for every $u\in V''$, and 
    \item $Q(p'', m'')$ is a non-empty integral polyhedron. 
\end{enumerate}
In particular, the execution of the algorithm is a feasible execution with finite recursion depth. 
\end{corollary}

%% file: runtime-for-covering-with-uniform-hypergraph.tex
\subsection{Run-time for Skew-supermodular Functions}\label{sec:run-time-assuming-bounded-recursion-depth}
 In this section, we bound the run-time of \Cref{alg:CoveringAlgorithm} for skew-supermodular functions assuming a bound on the recursion depth of the algorithm. \Cref{lem:WeakCoverViaUniformHypergraph:strongly-polytime:main} and \Cref{lem:simultaneuous-WeakCoverViaUniformHypergraph:strongly-polytime:main} are the main lemmas of this section. 
 
In order to show that each recursive call of \Cref{alg:CoveringAlgorithm} can be implemented in polynomial time, we will require the ability to solve certain optimization problems in polynomial time.  
In the next two lemmas, we summarize these optimization problems and show that they can be solved in polynomial time. The proofs of these lemmas use combinations of several known tools in submodular minimization.
We note that Lemma \ref{lem:alpha4-oracle} below does not require the input function $p$ to be skew-supermodular while \Cref{lem:optimizing-over-Q-polyhedron-intersection} below requires the input functions $q$ and $r$ to be skew-supermodular. 

\Cref{lem:optimizing-over-Q-polyhedron-intersection} concerns optimizing over the polyhedron $Q(p,m)$.  
We note that \functionMaximizationEmptyOracle{p} is the separation oracle for constraint (iv) of the polyhedron $Q(p,m)$. Consequently, given access to \functionMaximizationEmptyOracle{p}, we can optimize over the polyhedron $Q(p,m)$ in \emph{weakly} polynomial time using the ellipsoid algorithm. 
The following lemma shows that given access to \functionMaximizationEmptyOracle{q} and \functionMaximizationEmptyOracle{r}, where functions $q,r:2^V\rightarrow\Z$ are \emph{skew-supermodular}, we can optimize over the Q-polyhedron associated with the function $p$ defined as $p(X) := \max\{q(X), r(X)\}$ for every $X \subseteq V$ in \emph{strongly} polynomial time. We defer the proof of the lemma to  \Cref{appendix:sec:Function-Maximization-Oracles:Qpm-oracle}.

\begin{restatable}{lemma}{lemQpmOracle}\label{lem:optimizing-over-Q-polyhedron-intersection}
Let $q,r :2^V \rightarrow\Z$ be skew-supermodular functions
and let $p:2^V\rightarrow\Z$ be the function defined as $p(X) := \max\{q(X), r(X)\}$ for every $X \subseteq V$.
Then, the following optimization problem can be solved in $\poly(|V|)$ time using $\poly(|V|)$ queries to \functionMaximizationEmptyOracle{q} and \functionMaximizationEmptyOracle{r}: for a given non-negative function $m:V\rightarrow\Z_{\ge 0}$ such that $Q(p, m)\neq \emptyset$ and a given cost vector $c\in \R^V$, find an extreme point optimum solution to the following linear program: 
$$\max\left\{\sum_{u \in V}c_ux_u : x \in Q(p,m) \right\}.$$
\end{restatable}

\Cref{lem:alpha4-oracle} defines a ratio-maximization problem and shows that it can be solved in strongly polynomial time using polynomial number of maximization-oracle queries. This will be helpful in computing the value of $\alpha^{(4)}$ in \Cref{alg:CoveringAlgorithm} Step \ref{algstep:CoveringAlgorithm:def:alpha}. We defer the proof of the lemma to \Cref{appendix:sec:Function-Maximization-Oracles:AlphaFour-oracle}.
\begin{restatable}{lemma}{lemAlphaFourOracle}\label{lem:alpha4-oracle}
For a function $p :2^V \rightarrow\Z$, the following optimization problem can be solved in $\poly(|V|)$ time using $\poly(|V|)$ queries to \functionMaximizationEmptyOracle{p}: for a given vector $y \in \R^V$ and a given set $A\subseteq V$, compute a set $Z$ satisfying $|A\cap Z|\ge 2$ that maximizes $(p(Z)-y(Z))/(|A\cap Z|-1)$ i.e., 
    $$\arg \max\left\{\frac{p(Z) - y(Z)}{|A\cap Z| - 1} : |A\cap Z| \geq 2\right\}.$$
\end{restatable}

The next lemma shows how the function input to an arbitrary recursive call of \Cref{alg:CoveringAlgorithm} is related to the function input to the first recursive call. We will use this expression in proving the main lemma of the section. The proof of the lemma follows by induction on the recursion depth and is exactly the same as that of \Cref{lem:SzigetiAlgorithm:p_i-from-p_1}. We omit the proof for brevity. 

\begin{lemma}\label{lem:CoveringAlgorithm:p_i-from-p_1}
Suppose that the input to  \Cref{alg:CoveringAlgorithm} is a tuple $(p_1,m_1)$, where $p_1:2^{V_1}\rightarrow\Z$ and $m_1:V_1\rightarrow\Z_{+}$ are functions such that $m_1(X) \ge p_1(X)$ for all $X \subseteq V_1$ and $m_1(u) \leq K_{p_1}$ for all $u \in V_1$. Let $\ell \in \Z_+$ be the number of recursive calls witnessed by the execution of \Cref{alg:CoveringAlgorithm}. For $i\in[\ell]$, let $(p_i, m_i)$ be the input tuple to the $i^{th}$ recursive call; moreover, for $i \in [\ell - 1]$, let $\zeros_i, (H^i_0, w^i_0)$ be as defined by \Cref{alg:CoveringAlgorithm} for the input $(p_i, m_i)$.
Then, for each $i \in [\ell]$ we have that
$$p_i = \functionContract{\left(p_1 - \sum_{j \in [i-1]}b_{(\Tilde{H}_j, \Tilde{w}_{j})}\right)}{\bigcup_{j \in [i - 1]}\zeros_j},$$
    where $(\Tilde{H}_0^{i}, \Tilde{w}_0^{i})$ denotes the hypergraph obtained by adding the vertices $\cup_{j\in [i-1]}\zeros_j$ to the hypergraph $(H^i_0, w^i_0)$.
\end{lemma}

We now bound the run-time of \Cref{alg:CoveringAlgorithm} in terms of its recursion depth. In subsequent sections, we will get strongly polynomial bounds on the recursion depth of the \Cref{alg:CoveringAlgorithm} when the input function $p$ is skew-supermodular (or is defined to be the maximum of two skew-supermodular functions). Then, \Cref{lem:WeakCoverViaUniformHypergraph:strongly-polytime:main} (or respectively, \Cref{lem:simultaneuous-WeakCoverViaUniformHypergraph:strongly-polytime:main}) below will imply an overall strongly polynomial runtime for \Cref{alg:CoveringAlgorithm}, given access to the appropriate function maximization oracle for $p$.
\begin{lemma}\label{lem:WeakCoverViaUniformHypergraph:strongly-polytime:main}
Suppose that the input to \Cref{alg:CoveringAlgorithm} 
 is a tuple $(p,m,J)$, where $J\subseteq V$ is a set,  $p:2^V\rightarrow\Z$ is a skew-supermodular function, and $m:V\rightarrow\Z_{+}$ is a positive integer-valued function such that $m(X)\ge p(X)$ for every $X\subseteq V$ and $m(u)\le K_p$ for every $u\in V$. Let $\ell \in \Z_+$ denote the recursion depth of \Cref{alg:CoveringAlgorithm} on the input tuple $(p, m, J)$ is $\ell$. 
 Then, \Cref{alg:CoveringAlgorithm}  can be implemented to run in time poly$(|V|, \ell)$ using poly$(|V|, \ell)$ queries to \functionMaximizationOracleStrongCover{p}. The run-time includes the time to construct the hypergraphs used as input to the queries to \functionMaximizationOracleStrongCover{p}. Moreover, for each query to \functionMaximizationOracleStrongCover{p}, the hypergraph $(G_0, c_0)$ used as input to the query has $|V|$ vertices and $O(\ell)$ hyperedges.
\end{lemma}
\begin{proof}
By \Cref{coro:feasibility-of-execution}, the execution of \Cref{alg:CoveringAlgorithm} is a feasible execution and $\ell$ is finite. 
For $i\in [\ell]$, let the tuple $(p_i, m_i, J_i)$ denote the input to the $i^{th}$ recursive call -- we note that $(p_1,m_1,J_1) = (p, m, J)$. 
The next claim shows that all steps of \Cref{alg:CoveringAlgorithm} except for Steps \ref{algstep:CoveringAlgorithm:def:p'-and-p''} and \ref{algstep:CoveringAlgorithm:recursion} in the $i^{th}$ recursive call can be implemented to run in $\poly(|V|, \ell)$ time and $\poly(|V|)$  number of queries to \functionMaximizationEmptyOracle{p_i}.
   \begin{claim}\label{claim:CoveringAlgorithm:runtime:all-steps-except-recursion}
        All steps except for Steps \ref{algstep:CoveringAlgorithm:def:p'-and-p''} and \ref{algstep:CoveringAlgorithm:recursion} in the $i^{th}$ recursive call can be implemented to run in $\poly(|V_i|) + O(\ell)$ time using $\poly(|V_i|)$ queries to \functionMaximizationEmptyOracle{p_i}. The run-time includes the time to construct the inputs for the queries to \functionMaximizationEmptyOracle{p_i}.
    \end{claim}
    \begin{proof}
    Step \ref{algstep:CoveringAlgorithm:def:y} can be implemented in $\poly(|V_i|)$ time using $\poly(|V_i|)$ queries to \functionMaximizationEmptyOracle{p_i} since we can find an extreme point optimum solution $y$ to the LP $\max\{\sum_{u\in J_i}y_u : y \in Q(p_i, m_i)\}$ in $\poly(|V_i|)$ time using $\poly(|V_i|)$ queries to \functionMaximizationEmptyOracle{p_i} by \Cref{lem:optimizing-over-Q-polyhedron-intersection}. 
Step \ref{algstep:CoveringAlgorithm:def:A} can be implemented in $O(|V_i|)$ time by explicitly checking $y_u > 0$ for every $u \in V_i$.
Next, we show that Step \ref{algstep:CoveringAlgorithm:def:alpha} can be implemented in poly$(|V_i|)$ time using poly$(|V_i|)$ queries to \functionMaximizationEmptyOracle{p_i}.
We note that the value $K_{p_i}$ can be computed by querying $\functionMaximizationEmptyOracle{p_i}(S_0 := \emptyset, T_0:=\emptyset, y_0:=\{0\}^{V_i})$. Then, the values $\alpha_i^{(1)}$, $\alpha^{(3)}$, $\alpha^{(3)}$ and $\alpha^{(5)}$ can be computed in $O(|V_i|)$ time as follows: $\alpha_i^{(1)}$ and $\alpha_i^{(3)}$ can be computed by iterating over all vertices in $V_i$, $\alpha_i^{(2)}$ can be computed by the single query $\functionMaximizationEmptyOracle{p_i}(S_0 := \emptyset, T_0:= A, y_0:=\{0\}^{V_i})$, and $\alpha^{(5)}$ can be computed by checking the relevant conditions. Moreover, $\alpha^{(4)}$ can be computed in $\poly(|V_i|)$ time using $\poly(|V_i|)$ queries to \functionMaximizationEmptyOracle{p_i} by \Cref{lem:alpha4-oracle}. Consequently, Step \ref{algstep:CoveringAlgorithm:def:alpha} can be implemented in in $\poly(|V_i|)$ time using $\poly(|V_i|)$ queries to \functionMaximizationEmptyOracle{p_i}. Steps \ref{algstep:CoveringAlgorithm:def:zeros}, \ref{algstep:CoveringAlgorithm:def:H_0-and-w_0} and \ref{algstep:CoveringAlgorithm:def:m'-and-m''} can be implemented in $O(|V_i|)$ time by iterating over $V_i$. We note that the number of distinct hyperedges in the hypergraph obtained from recursion (Step \ref{algstep:CoveringAlgorithm:recursion}) is at most $\ell - i$ since each recursive call adds at most one distinct hyperedge to the returned. Thus, Steps \ref{algstep:CoveringAlgorithm:def:G-and-c} and \ref{algstep:CoveringAlgorithm:return} can be implemented in $O(|V| + \ell)$ time.
    \end{proof}

Next, we focus on the time to implement a query to \functionMaximizationEmptyOracle{p_i}. 

\begin{claim}\label{claim:CoveringAlgorithm:runtime:p_i-oracle-using-p-weak-oracle}
        For disjoint subsets $S, T\subseteq V_i$ and a vector $y \in \Z^{V_i}$, the answer to the query $\functionMaximizationEmptyOracle{p_i}(S_0 := S, T_0 := T, y_0 := y)$ can be computed in $O(|V_1|(|V_1| + \ell))$ time using at most $|V_1|+1$ queries to \functionMaximizationOracleStrongCover{p_1}. The run-time includes the time to construct the inputs for the queries to \functionMaximizationOracleStrongCover{p_1}. Moreover, for each query to \functionMaximizationOracleStrongCover{p_1}, the hypergraph $(G_0, c_0)$ used as input to the query has $|V_1|$ vertices and $O(i)$ hyperedges.
    \end{claim}
    \begin{proof}
    We prove the claim in two steps. In the first step, we will use the expression for the function $p_i$ given by \Cref{lem:CoveringAlgorithm:p_i-from-p_1}(b) to construct an answer for the query $\functionMaximizationEmptyOracle{p_i}(S, T, y)$ using a single query to $\functionMaximizationOracle{p_1}$. In the second step, we will obtain the desired runtime using \functionMaximizationOracleStrongCover{p_1} by invoking \Cref{lem:Preliminaries:wc-oracle-from-sc-oracle}.
    
    For every $j \in [i]$, let $(\Tilde{H}_0^{j}, \Tilde{w}_0^{j})$ denote the hypergraph obtained by adding the vertices $\cup_{k\in [j-1]}\zeros_k$ to the hypergraph $(H^k_0, w^k_0)$. We now show the first step of the proof via the following procedure.  First, construct the hypergraph $(G, c)$ and the vector $y_1 \in \Z^{V_1}$, where $G:= \sum_{j \in [i-1]}\Tilde{H}_0^{j}$,  $c:= \sum_{j \in [i-1]}, \Tilde{w}_0^{j})$ and $y_1(u) := y_0(u)$ for every $u \in V_i$, $y_1(u) := 0$ otherwise.  Next, obtain $(Z, p_1(Z))$ by querying $\functionMaximizationOracle{p_1}((G, c), S, T, y_1)$. Finally, return $(Z-\cup_{j=1}^{i-1}\zeros_{j}, p(Z))$ as the answer to the query $\functionMaximizationEmptyOracle{p_i}(S, T, y)$. We note that correctness of the procedure is because  $p_i = \functionContract{\left(p_1 - \sum_{j \in [i-1]}b_{(\Tilde{H}_j, \Tilde{w}_{j})}\right)}{\bigcup_{j \in [i - 1]}\zeros_j}$ by \Cref{lem:CoveringAlgorithm:p_i-from-p_1}(b).
    
    We now show the second step of the proof. We observe that the hypergraph $(G, c)$ constructed above has vertex set $V_1$ since for every $j \in [i - 1]$, the vertex set of $(\Tilde{H}_0^j, \Tilde{w}_0^j)$ is contained in $V_1$. Moreover, $(G, c)$ has at most $i - 1$ distinct hyperedges because for every $j \in [i - 1]$, the hypergraph $(\Tilde{H}_0^j, \Tilde{w}_0^j)$ has exactly $1$ hyperedge by Step \ref{algstep:CoveringAlgorithm:def:H_0-and-w_0}. Consequently, by \Cref{lem:Preliminaries:wc-oracle-from-sc-oracle}, the answer to the query $\functionMaximizationOracle{p_1}((G, c), S, T, y_1)$ in the above procedure can be computed in $O(|V_1|(|V_1| + i))$ time using $|V_1| + 1$ queries to \functionMaximizationOracleStrongCover{p_1}, where the runtime includes the time to construct the inputs to \functionMaximizationOracleStrongCover{p_1}; moreover, the hypergraph used as input to each \functionMaximizationOracleStrongCover{p_1} query has at most $|V_1|$ vertices and $O(i)$ hyperedges. Finally, we note that the construction of hypergraph $(G, c)$ and returning the (set, value) pair returned by the above procedure can be implemented in $O(|V_1|(|V_1| + i))$ time and so the claim holds.
    \end{proof}

We now complete the proof using the two claims above. By \Cref{claim:CoveringAlgorithm:runtime:all-steps-except-recursion}, all steps except for Steps \ref{algstep:CoveringAlgorithm:def:p'-and-p''} and \ref{algstep:CoveringAlgorithm:recursion} in the $i^{th}$ recursive call can be implemented to run in $\poly(|V_i|) + O(\ell)$ time using $\poly(|V_i|)$ queries to \functionMaximizationEmptyOracle{p_i}.  By \Cref{claim:CoveringAlgorithm:runtime:p_i-oracle-using-p-weak-oracle}, an answer to a single \functionMaximizationEmptyOracle{p_i} query can be computed in $O(|V_1|(|V_1| + \ell))$ time using at most $O(|V_1|)$ queries to \functionMaximizationOracleStrongCover{p_1}. Thus, all steps except for Steps \ref{algstep:CoveringAlgorithm:def:p'-and-p''} and \ref{algstep:CoveringAlgorithm:recursion} in the $i^{th}$ recursive call can be implemented to run in $\poly(|V_1|, \ell)$ time and $\poly(|V_1|)$ queries to \functionMaximizationOracleStrongCover{p_1}. By Claims \ref{claim:CoveringAlgorithm:runtime:all-steps-except-recursion} and \ref{claim:CoveringAlgorithm:runtime:p_i-oracle-using-p-weak-oracle}, the run-time also includes the time to construct the inputs for the queries to \functionMaximizationOracleStrongCover{p_1}. Moreover, by \Cref{claim:CoveringAlgorithm:runtime:p_i-oracle-using-p-weak-oracle}, the hypergraph $(G_0, c_0)$ used as input to each \functionMaximizationOracleStrongCover{p_1} query has $|V_1|$ vertices and $O(\ell)$ hyperedges.

    We note that Steps \ref{algstep:CoveringAlgorithm:def:p'-and-p''} and \ref{algstep:CoveringAlgorithm:recursion} need not be implemented explicitly for the purposes of the algorithm. Instead, \functionMaximizationEmptyOracle{p_{i+1}} can be used to execute all steps except for Steps \ref{algstep:CoveringAlgorithm:def:p'-and-p''} and \ref{algstep:CoveringAlgorithm:recursion} in the $(i+1)^{th}$ recursive call. We note that \Cref{claim:CoveringAlgorithm:runtime:p_i-oracle-using-p-weak-oracle} also enables us to answer a \functionMaximizationEmptyOracle{p_{i+1}} query in $O(|V_1|(|V_1| + \ell))$ time using at most $|V_1|+1$ queries to \functionMaximizationOracleStrongCover{p_1}. Thus, we have shown that each recursive call can be implemented to run in $\poly(|V_1|, \ell)$ time and $\poly(|V_1|)$ queries to \functionMaximizationOracleStrongCover{p_1}. Moreover, each query to \functionMaximizationOracleStrongCover{p_1} is on an input hypergraph $(G_0, c_0)$ that has $|V_1|$ vertices and $O(\ell)$ hyperedges. 
    Since there are $\ell$ recursive calls, \Cref{alg:CoveringAlgorithm} can be implemented to run in $\poly(|V_1|, \ell)$ time and $\ell\cdot\poly(|V_1|) = \poly(|V_1|, \ell)$ queries to \functionMaximizationOracleStrongCover{p_1} 
    where each query to \functionMaximizationOracleStrongCover{p_1} is on an input hypergraph $(G_0, c_0)$ that has $|V_1|$ vertices and $O(\ell)$ hyperedges.
\end{proof}

For simultaneous covering of two skew-supermodular functions, we have the following lemma whose proof is similar to the proof of \Cref{lem:WeakCoverViaUniformHypergraph:strongly-polytime:main}. We omit the proof in the interests of brevity. 
\begin{lemma}\label{lem:simultaneuous-WeakCoverViaUniformHypergraph:strongly-polytime:main}
Let $q, r: 2^V\rightarrow \Z$ be skew-supermodular functions. 
Suppose that the input to \Cref{alg:CoveringAlgorithm} 
 is a tuple $(p,m,J)$, where $J\subseteq V$ is a set, $p:2^V\rightarrow\Z$ is the function defined as $p(X):=\max\{q(X), r(X)\}$ for every $X\subseteq V$,
  and $m:V\rightarrow\Z_{+}$ is a  positive integer-valued function such that $m(X)\ge p(X)$ for every $X\subseteq V$ and $m(u)\le K_p$ for every $u\in V$. 
  Let $\ell\in\Z_+$ denote the recursion depth of \Cref{alg:CoveringAlgorithm} on the input tuple $(p, m, J)$.
 Then, \Cref{alg:CoveringAlgorithm}  can be implemented to run in time poly$(|V|, \ell)$ using poly$(|V|, \ell)$ queries to \functionMaximizationOracle{q} and \functionMaximizationOracle{r}.
 The run-time includes the time to construct the hypergraphs used as input to the queries to \functionMaximizationOracleStrongCover{q} and \functionMaximizationOracleStrongCover{r}. Moreover, for each query to \functionMaximizationOracleStrongCover{q} and \functionMaximizationOracleStrongCover{r}, the hypergraph $(G_0, c_0)$ used as input to the query has $|V|$ vertices and $O(\ell)$ hyperedges.
\end{lemma}

%% file: BK-covering-1-function-with-uniform-hypergraph.tex
\subsection{Weak Cover of One Function with Linear Number of Near-Uniform Hyperedges}\label{sec:WeakCoverViaUniformHypergraph}
In this section, we prove the following result.

\thmWeakCoverViaUniformHypergraph*

We first describe how to prove \Cref{thm:WeakCoverViaUniformHypergraph:main} under the assumption that the input function $m:V\rightarrow\Z_{+}$ is a \emph{positive} function. Under this extra assumption, we prove \Cref{thm:WeakCoverViaUniformHypergraph:main} by running \Cref{alg:CoveringAlgorithm} on the input instance $(p, m, \emptyset)$. By \Cref{coro:feasibility-of-execution}, the execution of \Cref{alg:CoveringAlgorithm} is a feasible execution with a finite recursion depth. Furthermore, by \Cref{lem:CoveringAlgorithm:main}, the hypergraph returned by \Cref{alg:CoveringAlgorithm} satisfies properties (1), (2), (3) and (4) of \Cref{thm:WeakCoverViaUniformHypergraph:main}. In this section, we show that the hypergraph returned by \Cref{alg:CoveringAlgorithm} also satisfies property (5) of \Cref{thm:WeakCoverViaUniformHypergraph:main} and that 
the recursion depth of the algorithm is polynomial. 
By \Cref{lem:WeakCoverViaUniformHypergraph:strongly-polytime:main} and the observation that each recursive call of the algorithm adds at most one hyperedge, it suffices to bound the recursion depth of the \Cref{alg:CoveringAlgorithm} by $O(|V|)$. 

The rest of this section is devoted to bounding the  recursion depth. In \Cref{sec:WeakCoverViaUniformHypergraph:properties-of-maximal-projected-tight-sets}, we define a family of \emph{cumulative projected maximal $(p, m)$-tight sets} (w.r.t. an execution of \Cref{alg:CoveringAlgorithm}) and prove certain properties about this family that will be useful in bounding the recursion depth. In  \Cref{sec:WeakCoverViaUniformHypergraph:good_vector}, we define a notion of \emph{good vectors} and prove certain properties which we subsequently use to bound the number of recursive calls not considered by the bound in \Cref{cor:num-createsteps-alpha=alpha1_alpha2_alpha3}. In \Cref{sec:WeakCoverViaUniformHypergraph:NumRecursiveCalls}, we use the tools which we developed in \Cref{sec:WeakCoverViaUniformHypergraph:properties-of-maximal-projected-tight-sets} and \Cref{sec:WeakCoverViaUniformHypergraph:good_vector} to show that the recursion depth of the algorithm is linear in the size of the ground set (see \Cref{lem:WeakCoverViaUniformHypergraph:runtime:num-createsteps}). Thus, the returned hypergraph also satisfies property (5) of \Cref{thm:WeakCoverViaUniformHypergraph:main}. \Cref{lem:WeakCoverViaUniformHypergraph:runtime:num-createsteps} and \Cref{lem:WeakCoverViaUniformHypergraph:strongly-polytime:main} together show that \Cref{alg:CoveringAlgorithm} can be implemented to run in \emph{strongly polynomial} time  
 given the appropriate function evaluation oracle (see \Cref{lem:WeakCoverViaUniformHypergraph:strongly-polytime:main}). \Cref{coro:feasibility-of-execution}, \Cref{lem:CoveringAlgorithm:main},  \Cref{lem:WeakCoverViaUniformHypergraph:runtime:num-createsteps}, and \Cref{lem:WeakCoverViaUniformHypergraph:strongly-polytime:main}
 complete the proof of \Cref{thm:WeakCoverViaUniformHypergraph:main} under the assumption that the input function $m$ is positive. 

We now briefly remark on how to circumvent the positivity assumption on the input function $m$ in the above proof. We note that if $\zeros \coloneqq  \{u \in V : m(u) = 0\} \not = \emptyset$, then $\functionContract{p}{\zeros} : 2^{V-\zeros} \rightarrow\Z$ is a skew-supermodular function with $K_{\functionContract{p}{\zeros}} \coloneqq  \max\{\functionContract{p}{\zeros}(X) : X \subseteq V - \zeros\}$ and $\functionRestrict{m}{\zeros}:V-\zeros\rightarrow\Z_{+}$ is a \emph{positive} function such that $\functionRestrict{m}{\zeros}(X)\ge \functionContract{p}{\zeros}(X)$ for every  $X\subseteq V - \zeros$ and $\functionRestrict{m}{\zeros}(u)\le K_{\functionContract{p}{\zeros}}$ for every  $u\in V-\zeros$ (i.e., the functions $\functionContract{p}{\zeros}$ and $\functionRestrict{m}{\zeros}$ satisfy conditions (a) and (b) of \Cref{thm:WeakCoverViaUniformHypergraph:main}). Furthermore, we observe that a hypergraph satisfying properties (1)--(5) of \Cref{thm:WeakCoverViaUniformHypergraph:main} for the functions $\functionContract{p}{\zeros}$ and $\functionRestrict{m}{\zeros}$ also satisfies the five properties for the functions $p$ and $m$. Thus, our final algorithm is to return the hypergraph returned by 
\Cref{alg:CoveringAlgorithm} on the instance $(\functionContract{p}{\zeros}, \functionRestrict{m}{\zeros}, \emptyset)$. Finally, \functionMaximizationOracleStrongCover{\functionContract{p}{\zeros}} can be implemented using \functionMaximizationOracleStrongCover{p} in strongly polynomial time. This completes the proof of \Cref{thm:WeakCoverViaUniformHypergraph:main}.

\subsubsection{Laminarity of Projected Maximal Tight Sets Across Recursive Calls}\label{sec:WeakCoverViaUniformHypergraph:properties-of-maximal-projected-tight-sets}
In this section, we define $(p, m)$-tight set families and investigate how they interact with the family of minimal $p$-maximizers. We will also investigate the progression of these families during an execution of \Cref{alg:CoveringAlgorithm}. The properties we prove in this section will be useful in the next section where we bound the runtime of \Cref{alg:CoveringAlgorithm} as well as the number of hyperedges in the hypergraph returned by the algorithm.

\begin{definition}\label{def:pm-tight-sets}
    Let $p:2^V\rightarrow\Z$ and $m:V\rightarrow\Z$ be set functions. A set $X\subseteq V$ is said to be \emph{$(p, m)$-tight} if $p(X) = m(X)$. We let $\tightSetFamily{p,m}$ denote the family of $(p, m)$-tight sets and $\maximalTightSetFamily{p,m}$ denote the family of \emph{inclusion-wise maximal} $(p, m)$-tight sets.
\end{definition}

 The next lemma says that if $p$ is a skew-supermodular function and $m$ is positive function satisfying $m(Z)\geq p(Z)$ for each $Z\subseteq V$, then the family $\calT_{p,m}$ is disjoint. Furthermore, if the set $X\in \calT_{p, m}$ is a maximal tight set w.r.t. functions $p$ and $m$, and the set $Y\in \calF_{p}$ is a minimal $p$-maximizer, then either the sets $X$ and $Y$ are disjoint, or the maximal tight set $X$ is contained in the minimal $p$-maximizer $Y$.

\begin{restatable}{lemma}{calTpmDisjoint}\label{lem:PropertiesOfTightSets:calTpmDisjoint}
    Let $p:2^V\rightarrow\Z$ be a skew-supermodular function and $m:V\rightarrow\Z_{+}$ be a positive function such that $m(Z) \geq p(Z)$ for all $Z\subseteq V$. Then, we have the following:
    \begin{enumerate}[label=(\alph*)]
        \item if $X, Y \in T_{p,m}$ such that $X\cap Y \not = 0$, then $X\cap Y, X\cup Y \in T_{p,m}$,
        \item the family $\calT_{p, m}$ is a disjoint set-family, and
        \item if $X\in \calT_{p, m}$ and $Y\in \calF_{p}$, then either $X \cap Y = \emptyset$ or $ X \subseteq Y$.
    \end{enumerate} 
\end{restatable}
\begin{proof} We prove the claims below.
\begin{enumerate}[label=(\alph*)]
    \item If $X \subseteq Y$ or $Y \subseteq X$, the claim holds. Suppose that $X - Y, Y - X \not = \emptyset$. We consider two cases based on the behavior of the function $p$ at the sets $X$ and $Y$. First, suppose that $p(X) + p(Y) \leq p(X - Y) + p(Y - X)$. We have the following:
    $$p(X) + p(Y) \leq p(X - Y) + p(Y - X) \leq m(X - Y) + m(Y - X)< m(X) + m(Y) = p(X) + p(Y),$$ which gives a contradiction. Here, the second inequality is because $m(Z) \geq p(Z)$ for every $Z\subseteq V$. The third inequality is because the function $m:V\rightarrow Z_{+}$ is a positive function and $X\cap Y \not = \emptyset$. The final equality is because $X,Y \in T_{p, m}$. Next, suppose that $p(X) + p(Y) \leq p(X\cup Y) + p(X\cap Y)$. We have the following:
    $$p(X) + p(Y) \leq p(X\cup Y) + p(X\cap Y) \leq m(X\cup Y) + m(X\cap Y) = m(X) + m(Y) = p(X) + p(Y).$$
    Here, the second inequality is because $m(Z)\ge p(Z)$ for every $Z\subseteq V$. The final equality is because the sets $X, Y \in T_{p,m}$. Thus, we have that all inequalities are equations. In particular, the second inequality is an equation and consequently the set $p(X\cup Y) = m(X\cup Y)$ and $p(X\cap Y) = m(X\cap Y)$.

    \item By way of contradiction, let $X, Y \in \calT_{p,m}$ be distinct maximal tight sets such that $X\cap Y \not = \emptyset$. We note that since the sets $X, Y$ are distinct and maximal, $X - Y, Y - X \not = \emptyset$. Since $X,Y \in T_{p,m}$, by part (a) of the current lemma (shown above), we have that $X \cup Y \in T_{p,m}$, contradicting maximality of the sets $X, Y \in\calT_{p,m}$.

    \item Let $X \cap Y \not = \emptyset$ and suppose by way of contradiction that $X - Y \not= \emptyset$. We consider two cases based on the behavior of the function $p$ at the sets $X, Y$ and arrive at a contradiction in both cases. First, suppose that $p(X) + p(Y) \leq p(X\cup Y) + p(X\cap Y)$. Then, we have that $$p(X) + p(Y) \leq p(X\cup Y) + p(X\cap Y) \leq p(Y) + m(X\cap Y) < p(Y) + m(X) = p(Y) + p(X),$$ which is a contradiction. Here, the second inequality is because the set $Y \in\calF_p$ is a $p$-maximizer and $m(Z)\ge p(Z)$ for every $Z\subseteq V$. The third inequality is because the function $m:V\rightarrow\Z_+$ is a positive function and $X - Y \not = \emptyset$. The final equality is because $X \in\calT_{p,m}$. Next, suppose that  $p(X) + p(Y) \leq p(X - Y) + p(Y - X)$. Then, we have that:
$$p(X) + p(Y) \leq p(X - Y) + p(Y - X) < m(X - Y) + p(Y) < m(X) + p(Y) = p(X) + p(Y),$$
which is a contradiction. Here, the second inequality is because the set $Y \in\calF_p$ is a $p$-maximizer and $m(Z)\ge p(Z)$ for every $Z\subseteq V$. The third inequality is the function $m:V\rightarrow\Z_+$ is a positive function and $X - Y \not = \emptyset$. The final equality is because $X \in\calT_{p,m}$. 
    \end{enumerate}
\end{proof}

In the remainder of the section, we will examine the progression of the maximal $(p,m)$-tight set family across an execution of \Cref{alg:CoveringAlgorithm}.
Consider an execution of \Cref{alg:CoveringAlgorithm} with $\ell \in \Z_+$ recursive calls. Suppose that the input to the $i^{th}$ recursive call is the tuple $(p_i, m_i, J_i)$.
 For every  $i \in [\ell - 1]$, we let $T_i, T_i', \maximalTightSetFamily{i}$, and $\maximalTightSetFamilyAfterCreate{i}$ denote the families $T_{p_i, m_i}, T_{p_i', m_i'}, \maximalTightSetFamily{p_i, m_i}$, and $\maximalTightSetFamily{p_i', m_i'}$ respectively. For convenience, we also define the  families $\tightSetFamily{\ell} , \tightSetFamilyAfterCreate{\ell} \coloneqq  \tightSetFamily{p_{\ell}, m_{\ell}}$ and $ \maximalTightSetFamily{\ell}, \maximalTightSetFamilyAfterCreate{\ell}, \coloneqq  \maximalTightSetFamily{p_{\ell}, m_{\ell}}$. Furthermore, we define the family of \emph{cumulative projected maximal $(p_i, m_i)$-tight sets} as
\begin{align*}
    \cumulativeMaximalTightSetFamily{i}&\coloneqq  \begin{cases}
    \maximalTightSetFamily{1}& \text{if $i = 1$,}\\
   \cumulativeMaximalTightSetFamily{i-1}|_{V_i} \cup 
    \maximalTightSetFamily{i}& \text{otherwise.}
\end{cases}& \text{$\forall i \in [\ell]$.}
\end{align*}
In the next lemma, we prove several useful properties of these set families during feasible executions of \Cref{alg:CoveringAlgorithm}. 

\begin{lemma}\label{lem:Cumulative-Projected-Maximal-Tight-Set-Family-Laminar}
    Consider a feasible execution of \Cref{alg:CoveringAlgorithm} on an input instance that terminates in $\ell \in \Z_+$ recursive calls. For every  $i \in [\ell]$, suppose that the input to  the $i^{th}$ recursive call 
 of \Cref{alg:CoveringAlgorithm} is a tuple $(p_i,m_i,J_i)$, where $p_1:2^{V}\rightarrow\Z$ is a skew-supermodular function with $K_{p_1} > 0$, $m_1:V\rightarrow\Z_{+}$ is a positive function, and $J_1\subseteq V$ is an arbitrary set such that $m_1(X)\ge p_1(X)$ for every $X\subseteq V$ and $m_1(u)\le K_{p_1}$ for every $u\in V$. Then, for every  $i \in [\ell - 1]$ we have that
 \begin{enumerate}[label = (\alph*)]
    \item $\tightSetFamily{i} \subseteq \tightSetFamilyAfterCreate{i}$,
    \item  $\projTightSetFamilyAfterCreate{i}{i+1} \subseteq \tightSetFamily{i+1}$,  
     \item $\cumulativeMaximalTightSetFamily{i} \subseteq \tightSetFamily{i}$,
     \item the family \cumulativeMaximalTightSetFamily{i} is laminar, and
     \item if $\zeros_i = \emptyset$ and  $\cumulativeMaximalTightSetFamily{i} = \cumulativeMaximalTightSetFamily{i+1}$, then $\maximalTightSetFamily{i} = \maximalTightSetFamily{i+1}$.
 \end{enumerate}
\end{lemma}
\begin{proof} 
Let $i \in [\ell -1]$ be a recursive call. By \Cref{coro:feasibility-of-execution} and induction on $i$, 
we have that $p_{i}:2^{V_{i}}\rightarrow\Z$ is  skew-supermodular function with $K_{p_i} > 0$, $m_i:V_i\rightarrow\Z_+$ is a positive function such that $m_i(X) \geq p_i(X)$ for every  $X \subseteq V_i$, $m_i(u) \leq K_{p_i}$ for every  $u \in V_i$ and $Q(p_i, m_i)$ is a non-empty integral polyhedron. We now prove each claim separately below.
    \begin{enumerate}[label=(\alph*)]
        \item Let $X\in \tightSetFamily{i}$. We first show that $|A_i \cap X| \leq 1$. By way of contradiction, let $|A_i\cap X| \geq 2$. Then, we have the following:
        \begin{align*}
            m_{i}'(X) &= m_i(X) - \alpha_i\chi_{A_i}(X)&\\
            &= p_i(X) - \alpha_i\chi_{A_i}(X)&\\
            &\leq p_i(X) - 2\alpha_i&\\
            &< p_i(X) - \alpha_i b_{(V_i, \{A_i\})}(X)&\\
            & = p_{i}'(X),&
        \end{align*}
    which is a contradiction to the functions $m_i'(X) \geq p_i'(X)$ 
    by \Cref{lem:Covering-Algorithm:hypothesis-after-createstep}(d). Here, the second equality is because $X \in \tightSetFamily{i}$, the first inequality is by our assumption that $|A_i\cap X|\ge 2$ and the second inequality is because $b_{(V_i, \{A_i\})}(X) = 1$. We now show that $X \in \tightSetFamilyAfterCreate{i}$. We have the following: 
    $$m_i'(X) = m_i(X) - \alpha_i|A_i\cap X| = p_i(X) - \alpha_i b_{(V_i, \{A_i\})}(X) = p_i'(X).$$
    Here, the second equality holds because $X \in \tightSetFamily{i}$, and $b_{(V_i, \{A_i\})}(X) = |A_i\cap X|$ since $|A_i\cap X|\leq 1$. 

    \item  We note that each set in the family $\projTightSetFamilyAfterCreate{i}{i+1}$ is of the form $Y - \zeros_i$ for some set $Y \in \tightSetFamilyAfterCreate{i}$. Thus, it suffices to consider an arbitrary set $X \in \tightSetFamilyAfterCreate{i}$ and show that $X - \zeros_i \in \tightSetFamily{i+1}$. We have the following:
    $$m_{i+1}(X - \zeros_i) = m_i'(X) = p_i'(X) \leq p_{i+1}(X - \zeros_i) \leq m_{i+1}(X - \zeros_i).$$
        Thus, all inequalities are equations and we have that $X - \zeros \in \tightSetFamily{i+1}$. Here, the first equality is because $m'(\zeros_i) = 0$. The second equality is because $X \in \tightSetFamilyAfterCreate{i}$. The first inequality is because $p_{i+1} = \functionContract{p_i'}{\zeros_i}$. The final inequality is because the functions $m_{i+1}$ and $p_{i+1}$ satisfy condition $m_{i+1}(Z) \geq p_{i+1}(Z)$ for every  $Z \subseteq V_{i+1}$.
        We remark that the statement (and proof) of this part can be strengthened to say that $\projTightSetFamilyAfterCreate{i}{i+1} = \tightSetFamily{i+1}$, but the weaker statement suffices for our purposes and so we omit the details. 
        
        \item By way of contradiction, let $i \in [\ell - 1]$ be the smallest recursive call such that the claim is false. We note that $i \geq 2$ because $\cumulativeMaximalTightSetFamily{1} = \maximalTightSetFamily{1} \subseteq \tightSetFamily{1}$. Thus, we have that $\cumulativeMaximalTightSetFamily{i} = \projCumulativeMaximalTightSetFamily{i-1}{i} \cup \maximalTightSetFamily{i}$. We observe that  $\cumulativeMaximalTightSetFamily{i-1}\subseteq \tightSetFamily{i-1} \subseteq \tightSetFamilyAfterCreate{i-1}$, where the first inclusion is by our choice of $i \in [\ell - 1]$, and the second inclusion is by part (1) of the current lemma (shown above). Furthermore, we have that  $\projCumulativeMaximalTightSetFamily{i-1}{i}  \subseteq \projTightSetFamilyAfterCreate{i-1}{i} \subseteq \tightSetFamily{i}$ 
        , where the second inclusion is by part  (2) of the current lemma (shown above). Thus, $\cumulativeMaximalTightSetFamily{i} \subseteq \tightSetFamily{i}$.

        
        \item By way of contradiction, let $i \in [\ell - 1]$ be the smallest recursive call such that the family $\cumulativeMaximalTightSetFamily{i}$ 
        is not laminar. 
        We note that $i \geq 2$ since the family $\cumulativeMaximalTightSetFamily{1} = \maximalTightSetFamily{1}$ is disjoint by \Cref{lem:PropertiesOfTightSets:calTpmDisjoint}(b). Thus, we have that $\cumulativeMaximalTightSetFamily{i} = \projCumulativeMaximalTightSetFamily{i-1}{i} \cup \maximalTightSetFamily{i}$.
        We note that the family $\maximalTightSetFamily{i}$ is also disjoint by \Cref{lem:PropertiesOfTightSets:calTpmDisjoint}(b). Furthermore, the family $\cumulativeMaximalTightSetFamily{i-1}$ is laminar by our choice of $i\in[\ell - 1]$. Since the family $\cumulativeMaximalTightSetFamily{i}$ is not laminar, there exist distinct sets $X \in \projCumulativeMaximalTightSetFamily{i-1}{i}$ and $Y \in \maximalTightSetFamily{i}$ such that $X - Y, Y - X, X \cap Y \not = \emptyset$. By part (c) of the current lemma (shown above), we have that $X \in \tightSetFamily{i}$. Consequently, there exists $Z\in \maximalTightSetFamily{i}$ such that $X \subseteq Z$. We note that $Y \not = Z$ because $X -Z=\emptyset$ and $X - Y \not = \emptyset$. Thus, by \Cref{lem:PropertiesOfTightSets:calTpmDisjoint}(b) we have that $Z\cap Y = \emptyset$. Therefore, $X\cap Y = \emptyset$, contradicting our choice of sets $X$ and $Y$.

        \item We first show that $\maximalTightSetFamily{i+1} \subseteq \maximalTightSetFamily{i}$. By way of contradiction, let $X \in \maximalTightSetFamily{i+1} - \maximalTightSetFamily{i}\not = \emptyset$. We note that since $\maximalTightSetFamily{i+1} \subseteq \cumulativeMaximalTightSetFamily{i+1} = \cumulativeMaximalTightSetFamily{i}$, the set $X \in \cumulativeMaximalTightSetFamily{i}$. Furthermore, since $X \not \in \maximalTightSetFamily{i}$, we have that $X \in \cumulativeMaximalTightSetFamily{i-1}|_{V_i}$ by definition of the family $\cumulativeMaximalTightSetFamily{i}$. By part (c) of the current lemma (shown above) and $\zeros_i = \emptyset$, we have that $X \in \tightSetFamily{i}$. Thus, there exists a $Y \in \maximalTightSetFamily{i}$ such that $X \subsetneq Y$, where the inclusion is strict because $X \not \in \maximalTightSetFamily{i}$. By parts (a) and (b) of the current lemma (shown above), we have that $Y \in \tightSetFamily{i+1}$, contradicting maximality of $X \in \maximalTightSetFamily{i+1}$. 

        Next, we show  that $\maximalTightSetFamily{i} \subseteq \maximalTightSetFamily{i+1}$. By way of contradiction, let $X \in \maximalTightSetFamily{i} - \maximalTightSetFamily{i+1} \not = \emptyset$. By parts (a) and (b) of the current lemma (shown above) and $\zeros_i = \emptyset$, we have that $X \in \tightSetFamily{i+1}$. Thus, there exists $Y \in \maximalTightSetFamily{i+1}$ such that $X \subsetneq Y$, where the strict inclusion is because $X \not \in \maximalTightSetFamily{i+1}$. By the previous paragraph, we have that $\maximalTightSetFamily{i+1} \subseteq \maximalTightSetFamily{i}$. Consequently $Y \in \maximalTightSetFamily{i}$, contradicting maximality of the set $X \in\maximalTightSetFamily{i}$.
    \end{enumerate}
\end{proof}

\subsubsection{Adding Good Hyperedges}\label{sec:WeakCoverViaUniformHypergraph:good_vector}
In this section, we show an important property of the set $A$ chosen during a recursive call of \Cref{alg:CoveringAlgorithm}.
This property will be useful to bound the number of recursive calls of the algorithm and the number of hyperedges returned by the algorithm. We now formally define the property.

\begin{definition}[\alphagood Sets]\label{def:uniform-hyperedges:alpha-good}
Let $p:2^V\rightarrow \Z$ and $m:V\rightarrow \Z$ be functions defined over a finite set $V$ and $\alpha$ be a non-negative real value. 
A set $A \subseteq V$ is \emph{$\alpha$-good for a set $Z \subseteq V$ w.r.t. the functions $p$ and $m$} if the following three properties hold:
 \begin{align*}
   (1)\ & \text{$\chi_A \in Q(p, m)$,} & \\
   (2)\ & \text{$\alpha|A\cap Z| = m(Z) - p(Z) + \alpha$, and} \tag{$\alpha$-Good-Set}\label{tag:conditions:good-set} & \\
   (3)\ & \text{there is no $Y\in \maximalTightSetFamily{p, m}$ such that $Z \subseteq Y$.}&
\end{align*}
Furthermore, the set $A\subseteq V$ is \emph{$\alpha$-good w.r.t the functions $p$ and $m$} if there exists a set $Z\subseteq V$ such that $A$ is $\alpha$-good for the set $Z$ w.r.t functions $p$ and $m$.
\end{definition}

Before showing the main lemma of the section, we show the following lemma which highlights the use of \alphagood sets.
It shows that during a feasible recursive call where the set $A$ chosen by the algorithm is \alphagood for a set $Z\subseteq V$, some superset of $Z$ is added to the maximal tight set family (prior to function contractions).

\begin{lemma}\label{lem:good-set:alpha=alpha4:maximalTightSetFamily-changes}
Suppose that the input to \Cref{alg:CoveringAlgorithm} is a tuple $(p,m,J)$, where $p:2^V\rightarrow\Z$ is a skew-supermodular function with $K_p > 0$, $m:V\rightarrow\Z_{+}$ is a positive function, and $J\subseteq V$ is an arbitrary set such that $m(X)\ge p(X)$ for every $X\subseteq V$, $m(u)\le K_p$ for every $u\in V$ and $Q(p, m)$ is an integral polyhedron.  Let $A, \alpha,p', m'$ be as defined by \Cref{alg:CoveringAlgorithm}. Suppose that the set $A$ is \alphagood w.r.t. the functions $p$ and $m$. Then, we have that 
$\maximalTightSetFamily{p', m'} - \maximalTightSetFamily{p, m} \not = \emptyset$.
\end{lemma}
\begin{proof}
    Let $Z\subseteq V$ be such that the set $A$ is \alphagood for the set $Z$ w.r.t. the functions $p$ and $m$.
    We first show that $|A\cap Z|\geq 2$.
    We note that by (\ref{tag:conditions:good-set}) property (3), the set $Z$ is not a tight set for the functions $p$ and $m$, i.e. $m(Z) > p(Z)$. Thus, we have that
    $|A\cap Z| = \frac{1}{\alpha}\left(m(Z) - p(Z) + \alpha\right) > 1,$
    where the equality is by (\ref{tag:conditions:good-set}) property (2) and \Cref{lem:CoveringAlgorithm:properties}(1) -- we note that $\chi_A \in Q_{p,m} \not = \emptyset$ by (\ref{tag:conditions:good-set}) property (1).
    
Then, the following sequence of equations shows that the set $Z$ is a tight set for the functions $p'$ and $m'$. This implies the existence of a maximal tight set $Y \in \maximalTightSetFamily{p', m'}$ such that $Z \subseteq Y$ 
and the claim follows as a consequence of (\ref{tag:conditions:good-set}) property (3).
    $$p'(Z) = p(Z) - \alpha b_{(V, \{A\})}(Z) = p(Z) - \alpha = m(Z) - \alpha|A\cap Z| = m'(Z).$$
    Here, the second equality is because $|A\cap Z| \geq 1$. The third equality is due to (\ref{tag:conditions:good-set}) property (2).
\end{proof}

The following is the main lemma of this section. It shows that if a recursive call of \Cref{alg:CoveringAlgorithm} witnesses $\alpha = \alpha^{(4)} < \min\{\alpha^{(1)}, \alpha^{(2)}, \alpha^{(3)}, \alpha^{(5)}\}$, then at least one of the following happens: (1) $\minimalMaximizerFamily{p''} - \minimalMaximizerFamily{p}\neq \emptyset$, (2) $\maximalTightSetFamily{p'', m''} -\maximalTightSetFamily{p, m} \neq \emptyset$, and (3) the hyperedge selected by the algorithm in the subsequent recursive call is $1$-good w.r.t. the functions $p''$ and $m''$. 

\begin{lemma}\label{lem:good_vector}
    Suppose that the input to \Cref{alg:CoveringAlgorithm} is a tuple $(p,m,J)$, where $p:2^V\rightarrow\Z$ is a skew-supermodular function with $K_p > 0$, $m:V\rightarrow\Z_{+}$ is a positive function, and $J\subseteq V$ is an arbitrary set such that $m(X)\ge p(X)$ for every $X\subseteq V$ and $m(u)\le K_p$ for every $u\in V$.  Let $A, \alpha, \alpha^{(1)}, \alpha^{(2)}, \alpha^{(3)}, \alpha^{(4)}, \alpha^{(5)}, p', m'$, $p'', m''$ be as defined by \Cref{alg:CoveringAlgorithm}. 
    Suppose that 
    $\alpha = \alpha^{(4)} < \min\{\alpha^{(1)}, \alpha^{(2)}, \alpha^{(3)}, \alpha^{(5)}\}$ and $K_{p''} > 0$. Then, at least one of the following properties hold:
    \begin{enumerate}[label=(\arabic*)]
        \item $\minimalMaximizerFamily{p''} -  \minimalMaximizerFamily{p} \not = \emptyset$, 
        \item $\maximalTightSetFamily{p'',m''} -\maximalTightSetFamily{p,m} \not = \emptyset$, and  
        \item let $y$ be an extreme point optimum solution to $\max\left\{\sum_{u\in A}y_u: y \in Q(p'', m'') \right\}$; then the set $T \coloneqq  \{u\in V: y_u > 0\}$ is $1$-good w.r.t. functions $p''$ and $m''$.
    \end{enumerate}
\end{lemma}
\begin{proof}
Suppose by way of contradiction that the claimed properties (1)-(3) of the claim do not hold. Let $\zeros$ be the set defined by \Cref{alg:CoveringAlgorithm}.
Since $\alpha < \alpha^{(1)}$, we have that $\zeros=\emptyset$. Consequently, $V'' = V$, $p' = p''$ and $m' = m''$.
By \Cref{coro:feasibility-of-execution}, we have that the execution of \Cref{alg:CoveringAlgorithm} is a feasible execution.
Then, by \Cref{lem:persistance-of-maximizers}(a) and (b), our contradiction assumption, and minimality of sets in the families $\minimalMaximizerFamily{p}$ and $\minimalMaximizerFamily{p''}$, we have that $\minimalMaximizerFamily{p} = \minimalMaximizerFamily{p''}$. 
Similarly, by \Cref{lem:Cumulative-Projected-Maximal-Tight-Set-Family-Laminar}(a) and (b), our contradiction assumption, and maximality of sets in the families \maximalTightSetFamily{p,m} and \maximalTightSetFamily{p'',m''}, we have that $\maximalTightSetFamily{p,m} = \maximalTightSetFamily{p'',m''}$. We note that by \Cref{coro:feasibility-of-execution}, we also have the following: $p''$ is a skew-supermodular function, $m''$ is a positive function, $m''(X)\ge p''(X)$ for every $X\subseteq V$, $m''(u)\le K_{p''}$ for every $u\in V$ and $Q(p'', m''), Q(p, m)$ are non-empty integral polyhedra. 

Let $y$ be an extreme point optimum solution to the LP $\max\{\sum_{u\in A}y_u: y \in Q(p'', m'')$ and let $T \coloneqq  \{u\in V: y_u > 0\}$. Since the polyhedron $Q(p'', m'')$ is integral, we have that $y=\chi_T$, where $\chi_T\in \{0, 1\}^V$ is the indicator vector of the set $T$. 
Let $\chi_A\in \{0, 1\}^V$ be the indicator vector of the set $A$. 
We recall that $\chi_A\in Q(p, m)$ by definition of the set $A$ and since the polyhedron $Q(p, m)$ is integral. However, $\chi_A\not\in Q(p'', m'')$ by \Cref{lem:CoveringAlgorithm:A-not-feasible-for-p''-and-m''}. Thus, we have that $A\neq T$. Since $\chi_A\in Q(p, m)$, it follows by \Cref{lem:hyperedge-feasibility-for-integral-polyhedron}(1) that the set $A$ contains a transversal for the family $\minimalMaximizerFamily{p} = \minimalMaximizerFamily{p''}$. Similarly, since $\chi_T\in Q(p'', m'')$, it follows by \Cref{lem:hyperedge-feasibility-for-integral-polyhedron}(1) that the set $T$ contains a transversal for the family $\minimalMaximizerFamily{p} = \minimalMaximizerFamily{p''}$. Let $T_0\subseteq T$ and $A_0\subseteq A$ be minimal transversals for the family $\minimalMaximizerFamily{p} = \minimalMaximizerFamily{p''}$ such that:
\begin{enumerate}[label=(\arabic*), ref=(\arabic*)]
    \item \label{lem:good_vector:T_0-A_0-choice:(1)} $|T_0 \cap A_0|$ is maximum, and
    \item \label{lem:good_vector:T_0-A_0-choice:(2)} the set $\left\{\{u,v\}: u \in T_0\cap X, v \in A_0\cap X \text{ for some } X \in \maximalTightSetFamily{p'', m''} = \maximalTightSetFamily{p, m}\right\}$ is inclusion-wise maximal subject to condition (1).
\end{enumerate}

Let $T_1\coloneqq T-T_0$ and $A_1\coloneqq A-A_0$. 
\begin{claim}\label{claim:A-T-difference-in-size}
    We have that $|A_0|=|T_0|$ and $|A_1|-|T_1|\in \{0, 1,-1\}$. 
\end{claim}
\begin{proof}
    The sets $A_0$ and $T_0$ are minimal transversals for the family $\minimalMaximizerFamily{p} = \minimalMaximizerFamily{p''}$ which is a disjoint family by \Cref{lem:UncrossingProperties:calFp-disjoint}. This implies that $|\minimalMaximizerFamily{p}| = |A_0| = |T_0|$.

We note that $|A|\in [\lfloor m(V)/K_p \rfloor, \lceil m(V)/K_p \rceil]$. Moreover, we have that $$|T|\in \left[\left\lfloor m''(V)/K_{p''} \right\rfloor, \left\lceil m''(V)/K_{p''} \right\rceil\right] = \left[\left\lfloor m(V)/K_p \right\rfloor, \left\lceil m(V)/K_p \right\rceil\right],$$ where the equality is by \Cref{lem:CoveringAlgorithm:(p''m'')-hypothesis}(e). Hence, $|A|-|T|\in \{0, 1, -1\}$. Since $A=A_0\uplus A_1$, $T= T_0\uplus T_1$, and $|A_0|=|T_0|$, it follows $|A_1|-|T_1|\in \{0, 1, -1\}$ . 
\end{proof}

Based on the sets $T_0, A_0, T_1$, and $A_1$, we consider different cases to pick vertices $t\in T$ and $a\in A$ and define a set $B$ as follows:
\begin{enumerate}[label = (\roman*)]
    \item If $T_0 \not= A_0$, then we pick $a \in A_0 - T_0$ and $t \in T_0 - A_0$ such that both vertices $a$ and $t$ are contained in the same minimal $p$-maximizer. Such a pair exists due to the following: the sets $A_0$ and $T_0$ are minimal transversals for the family $\minimalMaximizerFamily{p} = \minimalMaximizerFamily{p''}$ which is a disjoint family by \Cref{lem:UncrossingProperties:calFp-disjoint}. This implies that $|\minimalMaximizerFamily{p}| = |A_0| = |T_0|$ and $|A_0 \cap X| = |T_0\cap X| = 1$ for every  minimal $p$-maximizer $X \in \minimalMaximizerFamily{p}$. Since $A_0 \not = T_0$, there exists a minimal $p$-maximizer $X \in \minimalMaximizerFamily{p}$ such that $A_0\cap X \not = T_0\cap X$. Then, we may choose the pair $a, t$ as $\{a\} \coloneqq  A_0\cap X$ and $\{t\} \coloneqq  T_0\cap X$. 
    
    Here, we note that $a \not \in T_1$ and $t \not \in A_1$ by condition \ref{lem:good_vector:T_0-A_0-choice:(1)} of our choice of sets $T_0$ and $A_0$ which maximize $|T_0\cap A_0|$. We define $B\coloneqq T-t+a$. 

    \item Suppose $T_0 = A_0$. Since $A\neq T$ but $T_0=A_0$, it follows that either $A_1-T_1\neq \emptyset$ or $T_1-A_1\neq \emptyset$. We consider three cases:
    \begin{enumerate}
        \item If $A_1-T_1$ and $T_1-A_1$ are non-empty, then we pick $a \in A_1 - T_1$ and $t \in T_1 - A_1$ such that both vertices $a$ and $t$ are contained in the same maximal tight set if such a pair of vertices exists. Otherwise, we choose a pair of vertices $a \in A_1 - T_1$ and $t \in T_1 - A_1$ arbitrarily. We define $B\coloneqq T-t+a$. 
        
        \item If $A_1-T_1=\emptyset$, then 
        $A_1\subseteq T_1$. Since $A\neq T$, it follows that $A_1\neq T_1$. \Cref{claim:A-T-difference-in-size} implies that $T_1$ contains exactly one element apart from the elements of $A_1$. Let $t$ be the unique element of $T_1-A_1$. We define $B\coloneqq T-t$. We note that $B=A$. 
        
        \item If $T_1-A_1=\emptyset$, then 
        $T_1\subseteq A_1$. Since $A\neq T$, it follows that $A_1\neq T_1$. \Cref{claim:A-T-difference-in-size} implies that $A_1$ contains exactly one element apart from the elements of $T_1$. Let $a$ be the unique element of $A_1-T_1$. We define $B\coloneqq T+a$. We note that $B=A$. 
    \end{enumerate}
\end{enumerate}

In all cases above, we will show in \Cref{claim:chiA'-feasible-Q(p'm')} below  that the indicator vector $\chi_B\in \{0, 1\}^V$ of the set $B$ is a feasible point in the polyhedron $Q(p'', m'')$. In cases (i) and (ii)(a), this leads to a contradiction as follows: 
$$|T\cap A| < |B\cap A| \leq \max\left\{\sum_{u \in A}x_u :x \in Q(p'', m'')\right\} = |T\cap A|.$$
In cases (ii)(b) and (ii)(c), the fact that $\chi_B=\chi_A\in Q(p'',m'')$ contradicts \Cref{lem:CoveringAlgorithm:A-not-feasible-for-p''-and-m''}. 

\end{proof}

\begin{claim}\label{claim:chiA'-feasible-Q(p'm')}
    $\chi_{B} \in Q(p'', m'')$.
\end{claim}
\begin{proof}
We recall the definition of $Q(p'', m'')$:
$$Q(p'', m'') \coloneqq  \left\{ x\in \R^{|V|}\ \middle\vert 
        \begin{array}{l}
            {(\text{i})}{\ \ \ 0 \leq x_u \leq \min\{1, m''(u)\}} \hfill {\qquad \forall\ u\in V} \\
            {(\text{ii})}{\ \ x(Z) \geq 1} \hfill {\qquad\qquad\qquad\qquad\qquad \forall\ Z\subseteq V:\ p''(Z) = K_{p''}} \\
            {(\text{iii})}{ \ x(u) = 1} \hfill {\qquad\qquad\qquad\qquad\qquad \ \ \forall\ u\in V:\ m''(u) = K_{p''}} \\
            {(\text{iv})}{ \ x(Z) \leq m''(Z) - p''(Z) + 1} \hfill { \ \ \ \forall\ Z \subseteq V} \\
            {(\text{v})}{\ \ \left\lfloor \frac{m''(V)}{K_{p''}} \right\rfloor\le x(V) \le \left\lceil \frac{m''(V)}{K_{p''}} \right\rceil} \hfill {} \\
        \end{array}
        \right\}$$
    We will show that the vector $\chi_{B}$ satisfies constraint (iv) in \Cref{claim:good_vector:swap:constraint(iv)}. Thus, it suffices to show that constraints (i), (ii), (iii) and (v) are satisfied.
    The vector $\chi_{B}$ satisfies constraint (i) since $m''(a) = m(a) - \alpha > m(a) - \alpha^{(1)} \geq 0$, where the first equality is because $a \in A$, and the first inequality is because $\alpha < \alpha^{(1)}$. 
    
    Next, we show that $\chi_{B}$ satisfies constraint (ii), i.e. the set $B$ is a transversal for the set $\minimalMaximizerFamily{p''}$. If $A_0 \not = T_0$, then constraint (ii) holds by choice of $a, t$ being contained in the same minimal $p$-maximizer. Otherwise, $A_0 = T_0$ and so we have that $A_0 = T_0 \subseteq B$ by choice of $a \in A_1$ and $t \in T_1$. Thus the set $B$ is a transversal for $\minimalMaximizerFamily{p''}$ and constraint (ii) holds.

    Next we show that constraint (iii) holds. Since $T - t \subseteq B$, it suffices to show that $m''(t) \not = K_{p''}$. By way of contradiction, suppose $m''(t) = K_{p''}$. Thus, we have that $$K_{p} - \alpha^{(4)} > K_{p} - \alpha^{(3)} \geq m(t) = m''(t) = K_{p''} = K_{p} - \alpha = K_p - \alpha^{(4)},$$ 
    which is a contradiction. Here, the first inequality is due to $\alpha^{(4)} < \alpha^{(3)}$. The second inequality is by the definition of $\alpha^{(3)}$. The first equality is because $t \not\in A$ by our choice of $t \in T - A$. The third equality is by \Cref{lem:CoveringAlgorithm:(p''m'')-hypothesis}(b)
 and the final equality is because $\alpha = \alpha^{(4)}$. 

    Finally, we show that constraint (v) holds. We consider the different cases that define our set $B$ and show that the vector $\chi_B$ satisfies constraint (v) of the polyhedron $Q(p'', m'')$ in each case. First, suppose that either (i) $A_0\neq T_0$ or (ii)(a) $A_0=T_0$, $A_1-T_1\neq \emptyset$, and $T_1-A_1\neq \emptyset$ holds. Then, $|B|=|T|$, and since $\chi_T$ satisfies constraint (v) of the polyhedron $Q(p'', m'')$, it follows that $\chi_B$ also satisfies constraint (v) of the polyhedron $Q(p'', m'')$. Next, suppose that (ii)(b) $A_0=T_0$ and $A_1-T_1=\emptyset$ holds. Then, $|B|=|A|$ and $|T|=|A|+1$. Since $|A| \in \{\floor{m(V)/K_p}, \ceil{m(V)/K_p}\}$ and $|T| \in \{\floor{m''(V'')/K_{p''}}, \ceil{m''(V'')/K_{p''}}\} = \{\floor{m(V)/K_p}, \ceil{m(V)/K_p}\}$, where the final equality is by \Cref{lem:CoveringAlgorithm:(p''m'')-hypothesis}(e) and $\alpha < \alpha^{(5)}$, it follows that $|A|=\lfloor m(V)/K_p\rfloor$ and $|T|=|A|+1=\lceil m(V)/K_{p}\rceil$. Consequently, $|B|=|A|= \floor{m(V)/K_p} = \floor{m''(V'')/K_{p''}}$ and so constraint (v) holds.  Next, suppose that (ii)(c) $A_0=T_0$ and $T_1-A_1=\emptyset$ holds. Then, $|B|=|A|$ and $|T|=|A|-1$. By the same argument as in the previous case, we have that $|A|=\lceil m(V)/K_p\rceil$ and $|T|=|A|-1=\lfloor m(V)/K_{p}\rfloor$. Consequently, $|B|=|A|= \ceil{m(V)/K_p} = \ceil{m''(V'')/K_{p''}}$ and so constraint (v) holds.
\end{proof}

\begin{claim}\label{claim:good_vector:swap:constraint(iv)}
    $|B\cap Z| \leq m''(Z) - p''(Z) + 1$ for every  $Z \subseteq V$.
\end{claim}
\begin{proof}
    First, suppose that (ii)(b) $A_0=T_0$ and $A_1-T_1=\emptyset$ holds. For an arbitrary set $Z\subseteq V$, we have that 
    \[
    |B\cap Z| = \left|\left(T - t \right)\cap Z\right| = |T\cap Z| - \indicator_{t\in Z} \leq m''(Z) - p''(Z) + 1 - \indicator_{t\in Z}
    \le m''(Z) - p''(Z) + 1. 
    \]
    where the first inequality is because $\chi_{T} \in Q(p'', m'')$ satisfies constraint (iv) of the polyhedron $Q(p'', m'')$. For the rest of the proof, we will assume that either (i) $A_0\neq T_0$ or (ii)(a) $A_0=T_0$, $A_1-T_1\neq \emptyset$, and $T_1-A_1\neq \emptyset$ or (ii)(c) $A_0=T_0$ and $T_1-A_1=\emptyset$ holds. If (ii)(c) holds, then we define the element $t$ to be a dummy element that is not in $V$. 
    
    Let $Z \subseteq V$ be a counter-example to the claim such that $|Z|$ is maximum. We will need the following observation.
    \begin{observation}\label{obs:good_vector:assumptions}
    (1) $|T\cap Z| = m''(Z)-p''(Z)+1$, (2) $a \in Z$, and (3) $t \not\in Z$.
    \end{observation}
    \begin{proof}
    We have the following:
    $$|B\cap Z| = \left|\left(T - t + a\right)\cap Z\right| = |T\cap Z| - \indicator_{t\in Z} + \indicator_{a \in Z} \leq m''(Z) - p''(Z) + 1 - \indicator_{t\in Z} + \indicator_{a \in Z},$$ where the final inequality is because $\chi_{T} \in Q(p'', m'')$ satisfies constraint (iv) of the polyhedron $Q(p'', m'')$. We note that if any of (1)-(3) are false, then the set $Z$ is not a counter-example to the claim, contradicting the choice of the set $Z$.
    \end{proof}
    
    We first show that there exists a set $Y$ in $\calT_{p'',m''}$ such that $Y$ contains $Z$. 
    By way of contradiction, suppose that there is no set in $\calT_{p'',m''}$ containing the set $Z$. We recall that $\chi_T \in Q(p'', m'')$. By \Cref{obs:good_vector:assumptions}, we have $|T\cap Z| = m''(Z) - p''(Z) + 1$. Thus, the set $T$ satisfies all three (\ref{tag:conditions:good-set}) properties for $\alpha = 1$ and so we have that  the set $T$ is $1$-good for the set $Z$ w.r.t. to the functions $p''$ and $m''$, a contradiction.
    
    Next, we consider two cases based on whether the set $Z$ is a set in $\maximalTightSetFamily{p'', m''}$, and arrive at a contradiction in both cases. First, suppose that there exists a set $Y \in \maximalTightSetFamily{p'',m''}$ such that $Z \subsetneq Y$. Then, we have that
    $$|B\cap Z| \leq |B\cap Y| \leq m''(Y) - p''(Y) + 1 = 1 \leq m''(Z) - p''(Z) + 1,$$
    a contradiction to the set $Z$ being a counterexample to the claim. 
    Here, the first inequality is because $Z \subsetneq Y$. The second inequality is by our choice of the set $Z$ being a counter-example of maximum cardinality. The equality is because the set $Y \in\maximalTightSetFamily{p'',m''}$ and consequently $p''(Y) = m''(Y)$. The final inequality is because the functions $m''$ and $p''$ satisfy $m''(X) \geq p''(X)$ for every  $X\subseteq V''$ by \Cref{coro:feasibility-of-execution}.

    Next, suppose that $Z \in \maximalTightSetFamily{p'', m''}$. Here, we have the following two crucial observations.
    \begin{observation}\label{obs:good_vector:TcapZ=t}
        $A\cap Z = \{a\}$.
    \end{observation}
    \begin{proof}
        We recall that $a \in A$ by choice of the vertex $a$ and $a \in Z$ by \Cref{obs:good_vector:assumptions}. Thus, we have that 
        $a \in A\cap Z$. Furthermore, we have that $\chi_A\in Q(p, m)$ by definition of $A$ and the fact that $Q(p, m)$ is an integral polyhedron. Thus, the vector $\chi_A$ satisfies constraint (iv) of the polyhedron $Q(p, m)$. Then, we have the following:
        $$1 \leq |A\cap Z| \leq m(Z) - p(Z) + 1 = 1.$$
        Here, the first inequality is because $a \in A\cap Z$. The second inequality is because $\chi_A$ satisfies constraint (iv) of the polyhedron $Q(p,m)$. The final equality is because $Z \in \maximalTightSetFamily{p'',m''} = \maximalTightSetFamily{p,m}$ and hence $p(Z) = m(Z)$. Thus, all inequalities are equations and the claim follows.
    \end{proof}
    
    \begin{observation}\label{obs:good_vector:AcapZ=y}
    There exists an element $y \in T - (A\cap Z)$ such that $T\cap Z = \{y\}$ and $y \not \in \{a, t\}$.
    \end{observation}
    \begin{proof}
        We have that $|T\cap Z| =  m''(Z) - p''(Z) + 1 = 1$, where the first equality is by \Cref{obs:good_vector:assumptions} and the second equality is because $p''(Z) = m''(Z)$ since $Z \in \maximalTightSetFamily{p'',m''}$. Let $y$ be the unique element in $T\cap Z$. We note that the vertices $t$ and $y$ are distinct since $y \in Z$ but $t \not \in Z$ by \Cref{obs:good_vector:assumptions}. Furthermore, the vertices $y$ and $a$ are distinct since $y \in T$ but $a \not \in T$ by our choice of the vertex $a \in A - T$. Thus, $y \not \in \{a, t\}$. Furthermore, by \Cref{obs:good_vector:TcapZ=t} we have that $y \not\in A\cap Z$ and so $y \in T - (A\cap Z)$ as claimed. 
    \end{proof}

    With the previous observations established, we now consider the three subcases (i), (ii)(a) and (ii)(c) and arrive at a contradiction in all three subcases. 
    Firstly, we consider subcase (ii)(c): suppose that $A_0=T_0$ and $T_1-A_1=\emptyset$. Then, it follows that $T\subseteq A$. Thus, there cannot be an element $y\in T-(A\cap Z)$ such that $y\in T\cap Z$. This contradicts \Cref{obs:good_vector:AcapZ=y}.

    We refer to \Cref{fig:good-vector:A0=T0} and \Cref{fig:good-vector:A0not=T0} for a visualization of the interaction of the specified sets and elements under subcases (i) and (ii)(a).

\begin{figure}[H]
\centering
\begin{subfigure}[b]{0.47\textwidth}
    \centering
\includegraphics[width=0.5\textwidth]{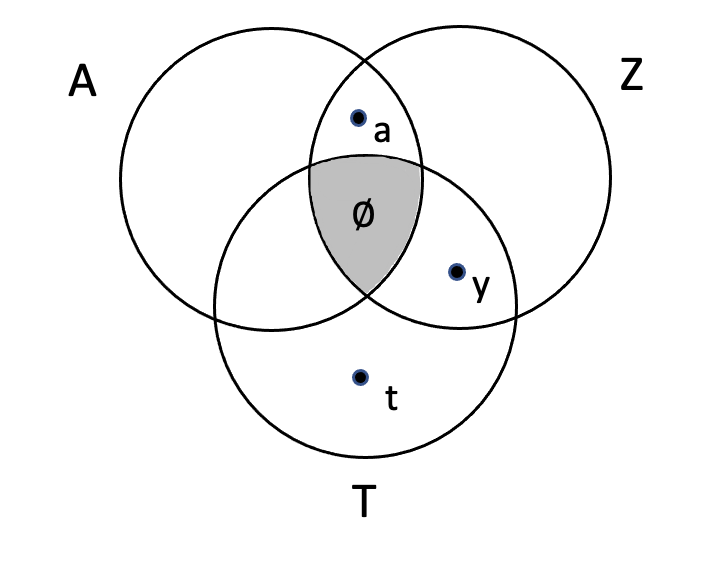}
    \caption{Subcase (ii)(a).}\label{fig:good-vector:A0=T0}
\end{subfigure}
     \hfill
     \begin{subfigure}[b]{0.47\textwidth}
    \centering
\includegraphics[width=0.5\textwidth]{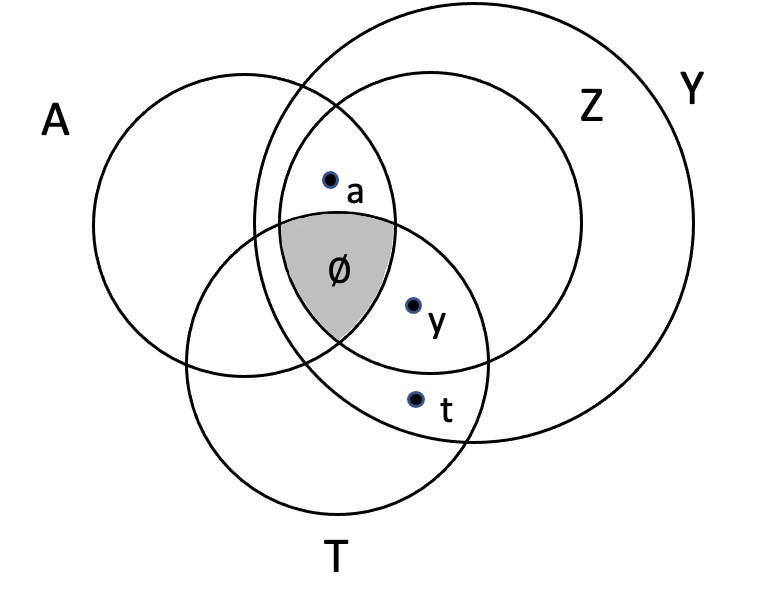}
    \caption{Subcase (i).}
    \label{fig:good-vector:A0not=T0}
\end{subfigure}
\caption{Subcases in the proof of \Cref{claim:good_vector:swap:constraint(iv)} for $Z \in \maximalTightSetFamily{p'', m''}$.}
\end{figure}
    
    Secondly, we consider subcase (ii)(a): suppose that $T_0 = A_0$, $A_1-T_1\neq \emptyset$, and $T_1-A_1\neq \emptyset$ (see \Cref{fig:good-vector:A0=T0}). By \Cref{obs:good_vector:TcapZ=t} and \Cref{obs:good_vector:AcapZ=y},
we have the following three properties:
\begin{enumerate}[label=(\roman{*})]
    \item $y \not \in A$ but $y\in T$ and so $y \in T_1 - A_1$,
    \item $a, y \in Z \in \maximalTightSetFamily{p'', m''}$, and
    \item $t \not \in Z$.
\end{enumerate}
Furthermore, there does not exist a set $Z' \in \maximalTightSetFamily{p'', m''}$ distinct from $Z$ that contains the vertices $a$ and $t$ as otherwise $a \in Z'\cap Z$ and hence $Z'\cap Z \not = \emptyset$, contradicting the disjointness of the family $\maximalTightSetFamily{p'', m''}$ by \Cref{lem:PropertiesOfTightSets:calTpmDisjoint}(b). Then, by property (iii), the pair $a, t$ are not contained in any set of $\maximalTightSetFamily{p'', m''}$. Thus, by properties (i) and (ii), the pair of vertices $a, y$ contradicts the choice of the pair of vertices $a, t$.

Thirdly, we consider subcase (i): suppose that $T_0 \not= A_0$ (see \Cref{fig:good-vector:A0not=T0}). Here, we first 
    recall that the vertices $a$ and $t$ belong to the same minimal $p''$-maximizer by choice of the pair $a,t$. Let $Y \in \minimalMaximizerFamily{p''}$ be the (unique) minimal $p''$-maximizer such that $a,t \in Y$. Next, we recall that $y, a \in Z$ by \Cref{obs:good_vector:TcapZ=t} and \Cref{obs:good_vector:AcapZ=y}. Furthermore, $Z \in \maximalTightSetFamily{p'', m''}$ by our case assumption. Thus, we have that $a \in Y\cap Z$ and hence, $Y\cap Z \not = \emptyset$. In particular, this implies that $Z \subseteq Y$ by \Cref{lem:PropertiesOfTightSets:calTpmDisjoint}(c) and so $y \in Y$. Thus, the set $T_0' \coloneqq  T_0 - t + y$ is also a minimal transversal for $\minimalMaximizerFamily{p''}$ contained in the set $T$. We note that $\left|T'_0 \cap A_0\right| = \left|T_0 \cap A_0\right|$ (i.e., the sets $T_0, A_0$ do not contradict condition (1) of our choice of the sets $T_0, A_0$) since $y\not \in A$ by \Cref{obs:good_vector:AcapZ=y}. However, we have that 
    \begin{align*}
        & \left\{\{u,v\}: u \in T_0, v \in A_0, u \in X, v\in X \text{ for some } X \in \maximalTightSetFamily{p'', m''}\right\} \\
        & \quad \quad \quad \quad \subsetneq \left\{\{u,v\}: u \in T_0', v \in A_0, u \in X, v \in X \text{ for some } X \in \maximalTightSetFamily{p'', m''}\right\},
    \end{align*}
    contradicting the maximality condition \ref{lem:good_vector:T_0-A_0-choice:(2)} in the choice of the sets $A_0$ and $T_0$. Here, we note that the LHS set is contained in the RHS set because $A_0\cap Y = \{a\}$, $T_0 \cap Y = \{t\}$, and the unique set $Z \in \maximalTightSetFamily{p'', m''}$ containing the vertex $a$ does not contain the vertex $t$. Furthermore, the containment is strict
    because the RHS set contains the pair $\{y,a\}$, but the LHS set does not.
\end{proof}

\subsubsection{Number of Recursive Calls and Hypergraph Support Size}\label{sec:WeakCoverViaUniformHypergraph:NumRecursiveCalls}
In this section, we show that the number of recursive calls of \Cref{alg:CoveringAlgorithm} is linear in the size of the ground set $V$. Since \Cref{alg:CoveringAlgorithm} adds at most one new hyperedge in each recursive call, we obtain that the number of hyperedges in the hypergraph returned by the algorithm is linear in the size of the ground set. Consequently, this shows that the hypergraph returned by \Cref{alg:CoveringAlgorithm} satisfies property (5) of \Cref{thm:WeakCoverViaUniformHypergraph:main}. 

\begin{lemma}\label{lem:WeakCoverViaUniformHypergraph:runtime:num-createsteps}
Suppose that the input to \Cref{alg:CoveringAlgorithm} is a tuple $(p,m,J)$, where $p:2^V\rightarrow\Z$ is a skew-supermodular function with $K_p > 0$, $m:V\rightarrow\Z_{+}$ is a positive function, and $J\subseteq V$ is an arbitrary set such that $m(X)\ge p(X)$ for every $X\subseteq V$ and $m(u)\le K_p$ for every $u\in V$. 
   Then, we have that:
   \begin{enumerate}[label=(\arabic*)]
       \item the recursion depth $\ell$ of \Cref{alg:CoveringAlgorithm} is at most $11|V|+1$, and
       \item the number of hyperedges in the hypergraph returned by \Cref{alg:CoveringAlgorithm} is at most $11|V|$.
   \end{enumerate}
\end{lemma}
\begin{proof}
We note that part (2) of the lemma follows from part (1) of the lemma since at most one new hyperedge is added during each recursive call and the base case adds no new hyperedge. We now show part (1) of the lemma. Let $i\in [\ell]$ be a recursive call of the algorithm's execution.
     We recall that $p_{j}$ and $m_j$ are the input functions to the $j^{th}$ recursive call for $j\geq 2$. Since $p_1$ is skew-supermodular, it follows that the function $p_i:2^{V_i} \rightarrow\Z$ is skew-supermodular, $m_i:V_i\rightarrow\Z_+$ is a positive function, $m_i(X)\ge p_i(X)$ for every $X\subseteq V_i$, $m_i(u)\le K_{p_i}$ for every $u\in V_i$, and the execution of \Cref{alg:CoveringAlgorithm} on the instance $(p,m, J)$ is a feasible execution.
We recall that for the sequence of functions $p_1, p_2, \ldots, p_{\ell}$, the family $\minimalMaximizerFamily{i}$ denotes the family of minimal $p_i$-maximizers, and $\cumulativeMinimalMaximizerFamily{i} = \cup_{j \in [i]}\minimalMaximizerFamily{j}$. Then, we have the following:
\begin{align*}
    \ell & = \left|\left\{i \in [\ell - 1] : \alpha_i \in \left\{\alpha_i^{(1)}, \alpha_i^{(2)}, \alpha_i^{(3)}, \alpha_i^{(5)}\right\}\right\}\right| + \left|\left\{i \in [\ell - 1] : \alpha_i = \alpha_i^{(4)} < \min\left\{\alpha_i^{(1)}, \alpha_i^{(2)}, \alpha_i^{(3)}, \alpha_i^{(5)}\right\}\right\}\right| + 1&\\
    & \leq \left(\left|\cumulativeMinimalMaximizerFamily{\ell}\right| + 2\left|V\right|+1\right) + \left(7\left|V\right|\right) + 1&\\
    & \leq 11|V| + 1.&
\end{align*}
Here, the first inequality is
by \Cref{cor:num-createsteps-alpha=alpha1_alpha2_alpha3} and  
\Cref{claim:WeakCoverViaHypergraph:num-createsteps-alpha=alpha4} below. The final inequality is because the family $\cumulativeMinimalMaximizerFamily{\ell}$ is laminar over the ground set $V$ by \Cref{lem:UncrossingProperties:Cumulative-Minimal-p-Maximizer-Family-Laminar} and hence, has size at most $2|V| - 1$.
\end{proof}

\begin{claim}\label{claim:WeakCoverViaHypergraph:num-createsteps-alpha=alpha4}
    $\left|\left\{i \in [\ell - 1] : \alpha_i = \alpha_i^{(4)} < \min\left\{\alpha_i^{(1)}, \alpha_i^{(2)}, \alpha_i^{(3)}, \alpha_i^{(5)}\right\}\right\}\right| \leq 7|V|.$
\end{claim}
\begin{proof} 
We define the potential function $\phi:[\ell] \rightarrow \Z_{\geq 0}$ as follows: for every  $i \in [\ell - 1]$

$$\phi(i) \coloneqq  \left|\cumulativeMinimalMaximizerFamily{i}\right| + \left|\cumulativeMaximalTightSetFamily{i}\right|+ 3\left|\zeros_{\leq i - 1}\right|.$$

First, we show that the potential function $\phi$ is non-decreasing over $[\ell - 1]$. We have the following:

\begin{align*}
    \phi(i) &= \left|\cumulativeMinimalMaximizerFamily{i}\right| + \left|\cumulativeMaximalTightSetFamily{i}\right|+ 3\left|\zeros_{\leq i-1}\right|&\\
    &\leq \left|\cumulativeMinimalMaximizerFamily{i}\right| + \left|\cumulativeMaximalTightSetFamily{i}|_{V_{i+1}}\right| + 3\left|\zeros_{i}\right| + 3\left|\zeros_{\leq i-1}\right|&\\
    &\leq \left|\cumulativeMinimalMaximizerFamily{i} \cup \minimalMaximizerFamily{i+1}\right| + \left|\cumulativeMaximalTightSetFamily{i}|_{V_{i+1}} \cup \maximalTightSetFamily{i+1}\right| + 3\left|\zeros_{i}\right| + 3\left|\zeros_{\leq i-1}\right|&\\
    & = \left|\cumulativeMinimalMaximizerFamily{i+1}\right| + \left|\cumulativeMaximalTightSetFamily{i+1}\right|+ 3\left|\zeros_{\leq i}\right|&\\
    &= \phi(i+1).&
\end{align*}
Here, the first inequality is because of the following. By \Cref{lem:Cumulative-Projected-Maximal-Tight-Set-Family-Laminar}(d), the family $\cumulativeMaximalTightSetFamily{i}$ is laminar. Furthermore, $\cumulativeMaximalTightSetFamily{i}|_{V_{i+1}}$ is the projection of the family $\cumulativeMaximalTightSetFamily{i}$ onto the ground set $V_{i+1} = V_i - \zeros_i$. The inequality then follows by \Cref{lem:CoveringAlgorithm:projection-laminar-family}  --  we note that if $\zeros_i \not = \emptyset$, then \Cref{lem:CoveringAlgorithm:projection-laminar-family} says that the inequality is strict. We will use this observation in the next part of the proof.

Let $i \in [\ell - 3]$ be a recursive call such that $\alpha_i = \alpha_{i}^{(4)} < \min\{\alpha_{i}^{(1)}, \alpha_{i}^{(2)}, \alpha_{i}^{(3)}, \alpha_i^{(5)}\}$. We now show that $\phi(i) < \phi(i+2)$.
For the sake of contradiction, suppose that $\phi(i) \geq \phi(i+2)$. Since the function $\phi$ is non-decreasing, we have that $\phi(i)=\phi(i+1)=\phi(i+2)$. By the observation in the last sentence of the previous paragraph, we have that $\zeros_i, \zeros_{i+1} = \emptyset$, i.e. $V_i = V_{i+1} = V_{i+2}$ and $\zeros_{\leq i-1} = \zeros_{\leq i} = \zeros_{\leq i+1}$. 
By \Cref{lem:Progression-of-set-families:main}(b), we have that $\cumulativeMinimalMaximizerFamily{i} \subseteq \cumulativeMinimalMaximizerFamily{i+1} \subseteq \cumulativeMinimalMaximizerFamily{i+2}$. Furthermore, by the definition of the families $\cumulativeMaximalTightSetFamily{i+1}$ and $\cumulativeMaximalTightSetFamily{i+2}$, we have that $\cumulativeMaximalTightSetFamily{i} \subseteq \cumulativeMaximalTightSetFamily{i+1} \subseteq \cumulativeMaximalTightSetFamily{i+2}$. Thus, $\cumulativeMinimalMaximizerFamily{i} = \cumulativeMinimalMaximizerFamily{i+1} = \cumulativeMinimalMaximizerFamily{i+2}$ and $\cumulativeMaximalTightSetFamily{i} = \cumulativeMaximalTightSetFamily{i+1} = \cumulativeMaximalTightSetFamily{i+2}$ since if any of the inclusions are strict, then $\phi(i) < \phi(i+2)$ by definition of the function $\phi$, a contradiction. Moreover, by \Cref{lem:Progression-of-set-families:main}(d) we have that $\minimalMaximizerFamily{i} = \minimalMaximizerFamily{i+1} = \minimalMaximizerFamily{i+2}$, and by \Cref{lem:Cumulative-Projected-Maximal-Tight-Set-Family-Laminar}(e) we have that $\maximalTightSetFamily{i} = \maximalTightSetFamily{i+1} = \maximalTightSetFamily{i+2}$. Since $i < \ell-2$, we also observe that $K_{p_{i+1}}, K_{p_{i+2}} > 0$. 
Then, by \Cref{lem:good_vector}, the set $A_{i+1}$ added by the \Cref{alg:CoveringAlgorithm} is a $1$-good set w.r.t. the functions $p_{i+1}$ and $m_{i+1}$. Furthermore, by \Cref{lem:good-set:alpha=alpha4:maximalTightSetFamily-changes}, we have that $\maximalTightSetFamilyAfterCreate{i+1} - \maximalTightSetFamily{i+1} \not = \emptyset$. However, we observe that the functions $p_{i+2} = p_{i+1}'$ and $m_{i+2} = m_{i+1}'$ since $\zeros_{i+1} = \emptyset$. Consequently, $\maximalTightSetFamily{i+2} = \maximalTightSetFamilyAfterCreate{i+1}$ and so  $\maximalTightSetFamily{i+2} - \maximalTightSetFamily{i+1} \not = \emptyset$, contradicting $\maximalTightSetFamily{i+1} = \maximalTightSetFamily{i+2}$.

Then, as a consequence of the previous two properties of our potential function, we obtain that:
$$\left|\left\{i \in [\ell - 3] : \alpha_i = \alpha_i^{(4)} < \min\left\{\alpha_i^{(1)}, \alpha_i^{(2)}, \alpha_i^{(3)}, \alpha_i^{(5)}\right\}\right\}\right| \leq \phi(\ell - 3) - \phi(0) \leq 7|V| - 2.$$
Here, the final inequality is because $|\cumulativeMinimalMaximizerFamily{\ell - 3}| \leq 2|V| - 1$ since the family $\cumulativeMinimalMaximizerFamily{\ell - 3}$ is laminar by \Cref{lem:UncrossingProperties:Cumulative-Minimal-p-Maximizer-Family-Laminar}, $\left|\cumulativeMaximalTightSetFamily{\ell - 3}\right| \leq 2|V| - 1$ since the family $\cumulativeMaximalTightSetFamily{\ell - 3}$ is laminar by \Cref{lem:Cumulative-Projected-Maximal-Tight-Set-Family-Laminar}(d), and $|\zeros_{\leq \ell - 3}| \leq |V|$.
\end{proof}

%% file: BK-covering-2-functions-with-uniform-hypergraph.tex
\subsection{Weak Cover of Two Functions with Quadratic Number of Near-Uniform Hyperedges}\label{sec:WeakCoverTwoFunctionsViaUniformHypergraph}

In this section we show the following main result.

\thmWeakCoverTwoFunctionsViaUniformHypergraph*

We first describe how to prove \Cref{thm:WeakCoverTwoFunctionsViaUniformHypergraph:main} under the assumption that the input function $m:V\rightarrow\Z_{+}$ is a \emph{positive} function. Under this extra assumption, we prove \Cref{thm:WeakCoverViaUniformHypergraph:main} by running \Cref{alg:CoveringAlgorithm} on the input instance $(p, m, \emptyset)$. By \Cref{coro:feasibility-of-execution-for-simultaneous}, the execution of \Cref{alg:CoveringAlgorithm} is a feasible execution with finite recursion depth. Furthermore, by \Cref{lem:CoveringAlgorithm:main}, the hypergraph returned by \Cref{alg:CoveringAlgorithm} satisfies properties (1), (2), (3) and (4) of \Cref{thm:WeakCoverTwoFunctionsViaUniformHypergraph:main}. In this section, we show that the hypergraph returned by \Cref{alg:CoveringAlgorithm} also satisfies property (5) of \Cref{thm:WeakCoverTwoFunctionsViaUniformHypergraph:main} and that the recursion depth of the algorithm is polynomial. By \Cref{lem:WeakCoverViaUniformHypergraph:strongly-polytime:main} and the observation that each recursive call of the algorithm adds at most one hyperedge, it suffices to bound the recursion depth of the \Cref{alg:CoveringAlgorithm} by  $O(|V|^2)$. 

The rest of this section is devoted to bounding the recursion depth. 
In \Cref{sec:WeakCoverTwoFunctionsViaUniformHypergraph:slack-functions}, we introduce the notion of \emph{pairwise-slack functions} and prove certain properties about them which we subsequently use in 
\Cref{sec:WeakCoverTwoFunctionsViaUniformHypergraph:recursion-depth} to show that the recursion depth of the algorithm is at most quadratic in the size of the ground set (see \Cref{lem:WeakCoverTwoFunctionsViaUniformHypergraph:recursion-depth:main}). Thus, the returned hypergraph also satisfies property (5) of \Cref{thm:WeakCoverTwoFunctionsViaUniformHypergraph:main}. 
\Cref{lem:WeakCoverViaUniformHypergraph:runtime:num-createsteps} and \Cref{lem:WeakCoverViaUniformHypergraph:strongly-polytime:main} together show that \Cref{alg:CoveringAlgorithm} can be implemented to run in \emph{strongly polynomial} time 
 given the appropriate function evaluation oracle (see \Cref{lem:WeakCoverViaUniformHypergraph:strongly-polytime:main}). 
 \Cref{coro:feasibility-of-execution-for-simultaneous}, 
 \Cref{lem:CoveringAlgorithm:main}, 
 \Cref{lem:WeakCoverTwoFunctionsViaUniformHypergraph:recursion-depth:main}, and \Cref{lem:simultaneuous-WeakCoverViaUniformHypergraph:strongly-polytime:main} complete the proof of \Cref{thm:WeakCoverViaUniformHypergraph:main} under the assumption that the input function $m$ is positive. 
The positivity assumption on the function $m$ can be circumvented in the same manner as in the proof of  \Cref{thm:WeakCoverViaUniformHypergraph:main} (see beginning of  \Cref{sec:WeakCoverViaUniformHypergraph}).

\subsubsection{Slack and Pairwise Slack Functions}\label{sec:WeakCoverTwoFunctionsViaUniformHypergraph:slack-functions}
In this section, we define \emph{slack} and \emph{pairwise-slack} functions. We then show a few key properties about how these  functions behave throughout the execution of the algorithm. In subsequent sections, we will use these properties to define a notion of \emph{progress} of the algorithm in order to bound the number of recursive calls (see \Cref{sec:WeakCoverTwoFunctionsViaUniformHypergraph:recursion-depth}). We remark that the slack function is an intermediate function only defined for convenience and that the proof of \Cref{thm:WeakCoverTwoFunctionsViaUniformHypergraph:main} will only require the properties of the pairwise-slack function.

\begin{remark}
We emphasize that the properties to be shown in this section were not needed in the proof of  \Cref{thm:WeakCoverViaUniformHypergraph:main} where the input function $p$ was skew-supermodular. Unlike the proof of \Cref{thm:WeakCoverViaUniformHypergraph:main}, it is unclear whether \emph{good hyperedges} (see \Cref{sec:WeakCoverViaUniformHypergraph:good_vector}) even exist under the hypothesis of \Cref{thm:WeakCoverTwoFunctionsViaUniformHypergraph:main}  --  we recall that the existence of good hyperedges was crucial in showing that the recursion depth of \Cref{alg:CoveringAlgorithm} is linear when the input function $p$ is skew-supermodular (see \Cref{sec:WeakCoverViaUniformHypergraph:NumRecursiveCalls}). We will subsequently use properties of the pairwise-slack function to show a quadratic bound on the recursion depth of \Cref{alg:CoveringAlgorithm} when the input function $p$ is the maximum of two skew-supermodular functions.
\end{remark}

We now define slack and pairwise-slack functions. 
For functions $m: V\rightarrow \Z $ and $p: 2^V\rightarrow \Z$, the \emph{$(p, m)$-slack} function $\slack_{p, m}:2^V \rightarrow \Z$ is defined as $\slack_{p, m}(X) = m(X) - p(X)$ for every $X\subseteq V$, 
and the \emph{pairwise $(p, m)$-slack function} $\pairwiseSlack_{p, m}:{V\choose 2} \rightarrow \Z$ is defined as 
$\pairwiseSlack_{p, m}(\{u,v\}) \coloneqq  \min\left\{\slack_{p, m}(Z) : u, v \in Z\right\}$ 
for every  pair $\{u,v\} \in {V\choose 2}$. 
The following two lemmas summarize useful properties of the pairwise slack function during the execution of 
\Cref{alg:CoveringAlgorithm}.

\begin{lemma}\label{lem:slack_monotone}
Suppose that the input to \Cref{alg:CoveringAlgorithm} is a tuple $(p,m,J)$, where $p: 2^V\rightarrow \Z$ is a function and $m:V\rightarrow\Z_{+}$ is a positive function, and $J\subseteq V$ is an arbitrary set such that $m(X)\ge p(X)$ for every $X\subseteq V$, $m(u)\le K_p$ for every $u\in V$, and $Q(p,m)$ is a non-empty polyhedron. 
    Let $A, \alpha,p', m', p'', m''$ be as defined by \Cref{alg:CoveringAlgorithm}.
    Then, we have the following:
\begin{enumerate}[label=(\arabic*)]
\item $\slack_{p'', m''}(Z) \leq \slack_{p', m'}(Z) \leq \slack_{p, m}(Z)$ for every  $Z \subseteq V - \zeros$; furthermore, for a set $Z\subseteq V$, we have that $\slack_{p', m'}(Z) < \slack_{p, m}(Z)$ if and only if $|A\cap Z| \geq 2$, and
    \item $\pairwiseSlack_{p'', m''}(\{u,v\})\leq \pairwiseSlack_{p', m'}(\{u,v\}) \leq \pairwiseSlack_{p, m}(\{u,v\})$ for every  $\{u,v\} \in {V - \zeros \choose 2}$; 
    furthermore, $\gamma_{p',m'}(\{u,v\}) < \gamma_{p, m}(\{u,v\})$ for each pair $\{u, v\} \in {A\choose 2}$.
\end{enumerate}
\end{lemma}
\begin{proof}
We prove both properties separately below.
\begin{enumerate}[label=(\arabic*)]
    \item Let $Z\subseteq V - \zeros$. Then, we have the following:
$$\slack_{p'', m''}(Z) = m''(Z) - p''(Z) = \functionRestrict{m'}{\zeros}(Z) - \functionContract{p'}{\zeros}(Z) \leq m'(Z) - p'(Z) = \slack_{p', m'}(Z),$$
where the inequality is because $\functionRestrict{m'}{\zeros}(Z) = m(Z)$ by definition of $\zeros = \{u \in V : m'(u) = 0\}$ and $Z \subseteq V - \zeros$. 
Let $(H_0, w_0)$ be the hypergraph computed by \Cref{alg:CoveringAlgorithm}. Then, we have that 
\[
\slack_{p', m'}(Z) 
= (m - \alpha\chi_A)(Z) - (p - b_{(H_0, w_0)})(Z) 
=m(Z)-p(Z)-\alpha|A\cap Z| + b_{(H_0, w_0)}(Z)\leq m(Z) - p(Z) = \slack_{p, m}(Z). 
\]
The claim then follows by observing that the inequality in the previous sequence is strict if and only if $|A\cap Z| \geq 2$.

\item We note that the first part of the statement, i.e. $\pairwiseSlackAfterContract(\{u,v\})\leq \pairwiseSlackAfterCreate(\{u,v\}) \leq \pairwiseSlack(\{u,v\})$ for every  $\{u,v\} \in {V - \zeros \choose 2}$, follows from part (1) of the current lemma. We now show the second part of the statement. Let $Z$ be a set witnessing $\gamma_{p,m}(\{u,v\})$, i.e. $Z := \arg\min\{m(Z) - p(Z) : u,v \in Z\}$. Then, we have that
\begin{align*}
    \gamma_{p, m}({u,v}) & = \Gamma_{p, m}(Z)&\\
    &= m(Z) - p(Z)&\\
    &= (m'(Z) + \alpha|A\cap Z| ) - (p'(Z) + \alpha) &\\
    &> m'(Z) - p'(Z)&\\
    &= \Gamma_{p',m'}(Z)&\\
    &\geq \gamma_{p',m'}(\{u,v\}).&
\end{align*}
Here, the strict inequality in the previous sequence is because $|A\cap Z| \geq 2$ since $u, v \in A\cap Z$.
\end{enumerate}
\end{proof}

\begin{lemma}\label{lem:slack:witness-set-properties}
Suppose that the input to \Cref{alg:CoveringAlgorithm} is a tuple $(p,m,J)$, where $p: 2^V\rightarrow \Z$ is a function, $m:V\rightarrow\Z_{+}$ is a positive function, and $J\subseteq V$ is an arbitrary set such that $m(X)\ge p(X)$ for every $X\subseteq V$, $m(u)\le K_p$ for every $u\in V$, and $Q(p,m)$ is a non-empty integral polyhedron. Let $A, \alpha, \alpha^{(1)}, \alpha^{(2)}, \alpha^{(3)}, \alpha^{(4)}, \alpha^{(5)}, p''$ and $m''$ be as defined by \Cref{alg:CoveringAlgorithm}.
    Suppose that
    $\alpha = \alpha^{(4)} < \min\{\alpha^{(1)}, \alpha^{(2)}, \alpha^{(3)}, \alpha^{(5)}\}$
    and $K_{p''} > 0$. Then, there exists 
    an optimum solution $Z$ to the problem $\min\left\{\floor{\frac{m(Z) - p(Z)}{|A\cap Z| - 1}} : |A\cap Z| \geq 2, Z\subseteq V\right\}$ such that 
    \begin{enumerate}[label=(\arabic*)]
        \item $\slack_{p'', m''}(Z) < |A\cap Z| - 1$, and
        \item $\pairwiseSlack_{p'', m''}(\{u,v\}) < \min\left\{|A\cap Z| - 1, \pairwiseSlack_{p, m}(\{u,v\})\right\}$ for every  pair $\{u,v\} \in {A\cap Z \choose 2}$.
    \end{enumerate}
\end{lemma}
\begin{proof}
We note that part (2) of the lemma is a consequence of part (1) of the lemma and \Cref{lem:slack_monotone}(2). Hence, we focus on showing part (1) of the lemma. 
We say that a set $Z \subseteq V$ is a \emph{witness set} if it is an optimum solution to the problem 
$$\min\left\{\floor{\frac{m(X) - p(X)}{|A\cap X| - 1}} : |A\cap X| \geq 2, X\subseteq V\right\}.$$
By way of contradiction, suppose that $\slack_{p'', m''}(Z) \geq |A\cap Z| - 1$ for every witness set $Z\subseteq V$. 
Let $(H_0, w_0)$ and $\zeros$ be as defined by \Cref{alg:CoveringAlgorithm}. Since $\alpha < \alpha^{(1)}$, we have that $\zeros=\emptyset$. In particular, we have that $A\cap \zeros = \emptyset$, $p''=p'$, $m''=m'$, and $\slack_{p'', m''}=\slack_{p', m'}$. Our strategy will be to show that the indicator vector $\chi_A\in \{0, 1\}^V$ of the set $A$ is in the polyhedron $Q(p'', m'')$, thus contradicting \Cref{lem:CoveringAlgorithm:A-not-feasible-for-p''-and-m''}. 
We recall the definition of $Q(p'', m'')$ -- we note that constraint (v) is well-defined since $K_{p''} > 0$:
$$Q(p'', m'') \coloneqq  \left\{ x\in \R^{|V|}\ \middle\vert 
        \begin{array}{l}
            {(\text{i})}{\ \ \ 0 \leq x_u \leq \min\{1, m''(u)\}} \hfill {\qquad \forall\ u\in V} \\
            {(\text{ii})}{\ \ x(Z) \geq 1} \hfill {\qquad\qquad\qquad\qquad\qquad \forall\ Z\subseteq V:\ p''(Z) = K_{p''}} \\
            {(\text{iii})}{ \ x(u) = 1} \hfill {\qquad\qquad\qquad\qquad\qquad \ \ \forall\ u\in V:\ m''(u) = K_{p''}} \\
            {(\text{iv})}{ \ x(Z) \leq m''(Z) - p''(Z) + 1} \hfill { \ \ \ \forall\ Z \subseteq V} \\
            {(\text{v})}{\ \ \left\lfloor \frac{m''(V)}{K_{p''}} \right\rfloor\le x(V) \le \left\lceil \frac{m''(V)}{K_{p''}} \right\rceil} \hfill {} \\
        \end{array}
        \right\}$$

Since $\alpha<\alpha^{(1)}$, it follows that $\chi_A$ satisfies constraint (i) of the polyhedron $Q(p'', m'')$. 
Since $\alpha<\alpha^{(2)}$, the set $A$ is a transversal for the family of minimal $p''$-maximizers. Hence, $\chi_A$ satisfies constraint (ii) of the polyhedron $Q(p'', m'')$. 
Since $\alpha<\alpha^{(3)}$, we have that $\{u \in V : m(u) = K_p\} = \{u \in V : m''(u) = K_{p''}\}$. Hence, $\chi_A$ satisfies constraint (iii) of the polyhedron $Q(p'', m'')$. 

Next, we show that $\chi_A$ satisfies constraint (iv) of the polyhedron $Q(p'', m'')$. Consider a set $X\subseteq V$. Suppose that the set $X$ is a witness set. Then, we have that $m''(X) - p''(X) = \slack_{p'', m''}(X) \geq |A\cap X| - 1$, where the final inequality is by our contradiction assumption. Thus, the vector $\chi_A$ satisfies constraint (iv) of the polyhedron $Q(p'', m'')$ for subsets $X$ that are witness sets. Next, suppose that the set $X$ is not a witness set. 
Here, we consider two cases. First, suppose that $|X\cap A| \leq 1$. Then, we have the following:
$$m''(X) - p''(X) = \slack_{p'', m''}(X)  = \slack_{p, m}(X) = m(X) - p(X) \geq |A\cap X| - 1,$$ 
where the second equality is by \Cref{lem:slack_monotone}(1) and the final inequality is because $\chi_A\in Q(p, m)$. Next, suppose that $|X\cap A| \geq 2$. 
By way of contradiction, suppose that constraint (iv) of the polyhedron $Q(p'', m'')$ does not hold for the set $X$, i.e. $m''(X) - p''(X) < |A\cap X| - 1$. Then, we have that $m(X) - \alpha|A\cap X| - p(X) + \alpha  < |A\cap X| - 1$ and hence, $m(X)-p(X)< (\alpha+1)(|A\cap X|-1)$. Thus, 
$$\alpha < \floor{\frac{m(X) - p(X)}{|A\cap X| - 1}} < \alpha + 1,$$
where the first inequality is because the set $X$ is not a witness set. The above is a contradiction since all three quantities $\alpha$, $\floor{\frac{m(X) - p(X)}{|A\cap X| - 1}}$, and $\alpha+1$ are integers, but we cannot have an integer that is strictly larger than the integer $\alpha$ and strictly smaller than the integer $\alpha+1$.

Finally, we show that the vector $\chi_A$ satisfies constraint (v) of the polyhedron $Q(p'', m'')$. Since $\alpha < \alpha^{(5)}$, we have that $\ceil{m''(V'')/K_{p''}} = \ceil{m(V)/K_{p}}$ and $\floor{m''(V'')/K_{p''}} = \floor{m(V)/K_{p}}$ by \Cref{lem:CoveringAlgorithm:(p''m'')-hypothesis}(e), i.e. constraint (v) of $Q(p,m)$ is equivalent to constraint (v) of $Q(p'', m'')$. Since $\chi_A\in Q(p,m)$, it follows that $\chi_A$ satisfies constraint (v) of the polyhedron $Q(p,m)$, and hence, $\chi_A$ also satisfies constraint (v) of the polyhedron $Q(p'',m'')$.
\end{proof}

\subsubsection{Number of Recursive Calls and Hypergraph Support Size}\label{sec:WeakCoverTwoFunctionsViaUniformHypergraph:recursion-depth}

In this section, we show that \Cref{alg:CoveringAlgorithm} witnesses $O(|V|^2)$ recursive calls. Since the algorithm adds at most one new hyperedge during every recursive call, we will obtain that the hypergraph returned by the algorithm has $O(|V|^2)$ distinct hyperedges. 

\begin{remark}
Consider a feasible execution of \Cref{alg:CoveringAlgorithm} with $\ell \in \Z_+$ recursive calls. Let $p_i:2^{V_i} \rightarrow\Z$ denote the input function to the $i^{th}$ recursive call.
We recall that  $\minimalMaximizerFamily{p_i}$ denotes the family of minimal $p_i$-maximizers, and  $\cumulativeFunctionMinimalMaximizerFamily{p}{\ell} \coloneqq  \cup_{j \in [i]}\minimalMaximizerFamily{p_j}$ denotes the family of cumulative minimal maximizers w.r.t. the sequence of functions $p_1, \ldots p_{\ell}$. 
In the previous section, a necessary step towards bounding the number of recursive calls and hypergraph support size was to show that the size of the family $\calF_{p_{\leq \ell}}$ is linear (see \Cref{lem:WeakCoverViaUniformHypergraph:runtime:num-createsteps} in \Cref{sec:WeakCoverViaUniformHypergraph}). 
In particular,  we had that the function $p_1$, and consequently the function $p_i$ for every  $i \in [\ell]$, was skew-supermodular. However, here we have that the function $p_i$ is the maximum of two individually skew-supermodular functions, and so it is not necessarily skew-supermodular. Nevertheless, we will see (in the proof of the next lemma) that the size of the family $\cumulativeFunctionMinimalMaximizerFamily{p}{\ell}$ is still linear in the size of the ground set. Consequently, the bottleneck in bounding the recursion depth of the algorithm is analyzing the number of recursive calls for which $\alpha = \alpha^{(4)} < \{\alpha^{(1)}, \alpha^{(2)}, \alpha^{(3)}, \alpha^{(5)}\}$.
\end{remark}

\begin{lemma}\label{lem:WeakCoverTwoFunctionsViaUniformHypergraph:recursion-depth:main}

Let $q, r:2^V \rightarrow\Z$ be skew-supermodular functions. 
Suppose that the input to \Cref{alg:CoveringAlgorithm}  is a tuple $(p,m,J)$, where $J\subseteq V$ is an arbitrary set,  $p:2^V\rightarrow\Z$ is a function defined by $p(X)\coloneqq \max\{q(X), r(X)\}$ for every $X\subseteq V$ with $K_p>0$ and $m:V\rightarrow\Z_{+}$ is a positive integer-valued function such that $m(X)\ge p(X)$ for every $X\subseteq V$ and $m(u)\le K_p$ for every $u\in V$.
Then, we have that: 
\begin{enumerate}[label = (\arabic*)]
    \item the recursion depth $\ell$ of \Cref{alg:CoveringAlgorithm} is at most $14|V|^2$, and 
    \item the number of hyperedges in the hypergraph returned by \Cref{alg:CoveringAlgorithm} is at most $14|V|^2 - 1$.
\end{enumerate}
\end{lemma}
\begin{proof}
    We note that property (2) of the claim follows from property (1) of the claim since at most one new hyperedge is added during each recursive call of the algorithm and no hyperedges are added during the base case. We now show property (1) of the claim. By \Cref{coro:feasibility-of-execution-for-simultaneous}, the execution of \Cref{alg:CoveringAlgorithm} on the instance $(p,m, J)$ is a feasible execution with finite recursion depth $\ell$. Let $i\in [\ell - 1]$ be a recursive call of the algorithm's execution. 
    We recall that $p_{i}$ and $m_i$ denote the input functions to the recursive call for $i\geq 2$.  
We recall that for the sequence of functions $p_1, p_2, \ldots, p_{\ell}$, the family $\minimalMaximizerFamily{i}$ denotes the family of minimal $p_i$-maximizers, and $\cumulativeMinimalMaximizerFamily{i} = \cup_{j \in [i]}\minimalMaximizerFamily{j}$. Then, we have the following.
\begin{align*}
    \ell & = \left|\left\{i \in [\ell - 1] : \alpha_i \in \left\{\alpha_i^{(1)}, \alpha_i^{(2)}, \alpha_i^{(3)}, \alpha_i^{(5)}\right\}\right\}\right| + \left|\left\{i \in [\ell - 1] : \alpha_i = \alpha_i^{(4)} < \min\left\{\alpha_i^{(1)}, \alpha_i^{(2)}, \alpha_i^{(3)}, \alpha_i^{(5)}\right\}\right\}\right| + 1&\\
    & \leq \left(\left|\cumulativeMinimalMaximizerFamily{\ell}\right| + 2\left|V\right|+1\right) + (8|V|^2) + 1&\\
    & \leq 14|V|^2.&
\end{align*}
Here, the first inequality is
by \Cref{cor:num-createsteps-alpha=alpha1_alpha2_alpha3} and \Cref{claim:WeakCoverTwoFunctionsViaHypergraphs:num-createsteps-alpha=alpha4}
 below. The final inequality is by  \Cref{claim:WeakCoverTwoFunctionsViaHypergraphs:calFp_leqell-4|V|-2} below.
\end{proof}

\begin{claim}\label{claim:WeakCoverTwoFunctionsViaHypergraphs:num-createsteps-alpha=alpha4}
$\left|\left\{i \in [\ell - 1] : \alpha_i = \alpha_i^{(4)} < \min\left\{\alpha_i^{(1)}, \alpha_i^{(2)}, \alpha_i^{(3)}, \alpha_i^{(5)}\right\}\right\}\right| \leq 8|V|^2$.
\end{claim}
\begin{proof}
Let $i \in [\ell - 2]$ be a recursive call of the algorithm. By \Cref{coro:feasibility-of-execution-for-simultaneous} and induction on $i$, we have that $p_i : 2^{V}\rightarrow\Z$ with and $m_i:V_i\rightarrow\Z_+$ are functions such that $m_i(X)\ge p_i(X)$ for every $X\subseteq V_i$, $m_i(u)\le K_{p_i}$ for every $u\in V_i$ and $Q(p_i, m_i)$ is a non-empty integral polyhedron. We note that since $i < \ell - 1$, we have that $K_{p_i}, K_{p_{i+1}} > 0$. For notational convenience, let $\pairwiseSlack_{i} := \pairwiseSlack_{p_i, m_i}$ denote the pairwise $(p_i, m_i)$-slack function.

    We define the potential function $\phi:[\ell-2] \rightarrow\Z_{\ge 0}$ as $$\phi(i) = \sum_{\{u,v\} \in {V_i\choose 2}}\sum_{j \in [\pairwiseSlack_i(\{u,v\})]}\frac{1}{j^2}.$$

    We note that the function $\phi$ is non-increasing over the domain $[\ell - 2]$ because $V_{i+1} \subseteq V_i$ and $\pairwiseSlack_{i+1}(\{u,v\}) \leq \pairwiseSlack_i(\{u,v\})$ for every  $\{u,v\} \in {V_{i+1} \choose 2}$ by \Cref{lem:slack_monotone}(2). We now show that the potential function strictly decreases in value after a recursive call  witnessing $\alpha = \alpha^{(4)} < \min\{\alpha^{(1)}, \alpha^{(2)}, \alpha^{(3)}, \alpha^{(5)}\}$. Let $i \in [\ell - 2]$ be such a recursive call. 
    We recall that a set $Z \subseteq V_i$ is a \emph{witness set} if $Z$ is an optimum solution to the problem $\min\left\{\floor{\frac{m_i(X) - p_i(X)}{|A_i\cap X| - 1}} : |A_i\cap X| \geq 2, X\subseteq V_i\right\}$. By \Cref{lem:slack:witness-set-properties}(2), 
 there exists a witness set $Z \subseteq V_i$ such that $\pairwiseSlack_{i+1}(\{u,v\}) < \min\left\{|A\cap Z| - 1, \pairwiseSlack_i(\{u,v\})\right\}$ for every  pair $\{u,v\} \in {A_i\cap Z \choose 2}$. Then, we have the following:
    \begin{align*}
        \phi(i) - \phi(i+1) & = \sum_{\{u,v\} \in {V_i\choose 2}}\sum_{j \in [\pairwiseSlack_i(\{u,v\})]}\frac{1}{j^2} - \sum_{\{u,v\} \in {V_{i+1}\choose 2}}\sum_{j \in [\pairwiseSlack_{i+1}(\{u,v\})]}\frac{1}{j^2}&\\
        &= \sum_{\{u,v\} \in {V_i\choose 2}}\sum_{j \in \left[\pairwiseSlack_{i+1}(\{u,v\}) + 1, \pairwiseSlack_i(\{u,v\})\right]}\frac{1}{j^2}&\\
        &\geq \sum_{\{u,v\} \in {A_i\cap Z\choose 2}} \frac{1}{\left(1+\pairwiseSlack_{i+1}(\{u,v\})\right)^2} &\\
        &\geq {A_i \cap Z \choose 2}\frac{1}{|A_i\cap Z|^2} &\\
        &\geq 1/4.&
    \end{align*}
    Here, the second equality is because $V_{i+1} = V_i$ and $\pairwiseSlack_{i+1}(\{u,v\}) \leq \pairwiseSlack_{i}(\{u,v\})$ for every  $\{u,v\} \in {V_i \choose 2}$ by \Cref{lem:slack_monotone}(2). The first inequality is because $A_i, Z \subseteq V_i$, the pairwise-slack function is non-negative, and $\pairwiseSlack_{i+1}(\{u, v\})<\pairwiseSlack_i(\{u,v\})$ for every pair $\{u, v\}\in \binom{A_i\cap Z}{2}$.
    The second inequality is because $\pairwiseSlack_{i+1}(\{u,v\}) \leq |A_i \cap Z| - 1$ by \Cref{lem:slack:witness-set-properties}(2). The final inequality is because $|A_i\cap Z|\ge 2$.

    By the previous two properties, we obtain the claimed bound as follows: 
    \begin{align*}
        \left|\left\{i \in [\ell - 2] : \alpha_i = \alpha_i^{(4)} < \min\left\{\alpha_i^{(1)}, \alpha_i^{(2)}, \alpha_i^{(3)}, \alpha_i^{(5)}\right\}\right\}\right|& \leq 4\left(\phi(1) - \phi(\ell - 2)\right)&\\
        &\leq 4\sum_{\{u,v\} \in {V\choose 2}}\sum_{j \in [\pairwiseSlack_1(\{u,v\})]}\frac{1}{j^2}&\\
        &\leq 4\sum_{\{u,v\} \in {V\choose 2}}\sum_{j \ge 1}\frac{1}{j^2} &\\
        &< 8|V|^2,&
    \end{align*}
    where the final inequality holds because $\sum_{j\ge 1}1/j^2 = \pi^2/6 < 2$.
\end{proof}

\begin{claim}\label{claim:WeakCoverTwoFunctionsViaHypergraphs:calFp_leqell-4|V|-2}
 $|\cumulativeFunctionMinimalMaximizerFamily{p}{\ell}| \leq 4|V| - 2$
\end{claim}
\begin{proof}
We recall that the functions $q_1, r_1$ are skew-supermodular, and
$q_i \coloneqq  \functionContract{(q_{i -1} - b_{(H_0^{i-1}, w_0^{i-1})})}{\zeros_{i-1}}$ and $r_i \coloneqq  \functionContract{(r_{i -1} - b_{(H_0^{i-1}, w_0^{i-1})})}{\zeros_{i-1}}$ for $\ell \geq i \geq 2$, where $(H_0^i, w_0^i)$ is the hypergraph $(H_0, w_0)$ and $\zeros_{i}$ is the set $\zeros$ defined by the algorithm during the $i^{th}$ recursive call (we define $(H_0^{\ell}, w_0^{\ell})$ as the empty hypergraph and $\zeros_\ell \coloneqq  \emptyset$). Let $i \in [\ell]$ be a recursive call.
By \Cref{coro:feasibility-of-execution-for-simultaneous} and induction on $i$, 
we have that $p_i(X)=\max\{q_i(X), r_i(X)\}$ for every $X\subseteq V_i$. Consequently, every $p_i$-maximizer is either a $q_i$-maximizer or a $r_i$-maximizer.
We consider the two sequences of functions $q_1, q_2, \ldots, q_{\ell}$ and $r_1, r_2, \ldots, r_{\ell}$, and the corresponding cumulative minimal maximizer families $\cumulativeFunctionMinimalMaximizerFamily{q}{\ell}$ and $\cumulativeFunctionMinimalMaximizerFamily{r}{\ell}$ respectively. 
By the previous observation and the definition of the families, we have that $\calF_{p_{\leq \ell}} \subseteq \calF_{q_{\leq \ell}} \cup \calF_{r_{\leq \ell}}$. In particular, we have the following:
 $$\left|\calF_{p_{\leq \ell}}\right| \leq \left|\calF_{q_{\leq \ell}}\right| + \left|\calF_{r_{\leq \ell}}\right| \leq (2|V|-1) + (2|V| - 1),$$
 where the final inequality is because the families \cumulativeFunctionMinimalMaximizerFamily{q}{\ell} and $\cumulativeFunctionMinimalMaximizerFamily{r}{\ell}$ are laminar over the ground set $V$ by \Cref{lem:UncrossingProperties:Cumulative-Minimal-p-Maximizer-Family-Laminar}, and hence, both of them have size at most $2|V|-1$.
\end{proof}

\begin{remark}
We recall that the proof of \Cref{thm:WeakCoverViaUniformHypergraph:main} (in particular, \Cref{lem:WeakCoverViaUniformHypergraph:runtime:num-createsteps}) crucially leveraged the following property of \Cref{alg:CoveringAlgorithm}: let $p, m$ and $p'', m''$ be input functions during two consecutive recursive calls of \Cref{alg:CoveringAlgorithm} and let $A$ and $A''$ be the hyperedges chosen during these recursive calls respectively; then, $A''$ is a set that has maximum intersection with the set $A$ subject to the indicator vector $\chi_{A''}$ of the set $A''$ satisfying $\chi_{A''} \in Q(p'', m'')$. We note that this property is ensured by  Steps \ref{algstep:CoveringAlgorithm:def:y} and \ref{algstep:CoveringAlgorithm:def:A} of \Cref{alg:CoveringAlgorithm}, where the vector $y$ is chosen to be an optimum (integral) extreme point solution of the polyhedron $Q(p'', m'')$ that maximizes along the objective direction $\chi_A$, the indicator vector of the set $A$. Furthermore, the set $A''$ is defined as the support of this (integral) vector $y$.
 We do not know how to take advantage of this maximum intersection property in the analysis for two functions (i.e., under the hypothesis of \Cref{thm:WeakCoverTwoFunctionsViaUniformHypergraph:main}). We conjecture that \Cref{alg:CoveringAlgorithm} can be used to prove a stronger version of \Cref{thm:WeakCoverTwoFunctionsViaUniformHypergraph:main} -- namely that the number of hyperedges in the support is $O(|V|)$ as opposed to $O(|V|^2)$ and leave this as an open question. 
\end{remark}

%% file: applications.tex
\section{Applications}\label{sec:applications}
In this section, we discuss applications of Theorems 
\ref{thm:SzigetiWeakCover:main}, \ref{thm:WeakCoverViaUniformHypergraph:main}, and \ref{thm:WeakCoverTwoFunctionsViaUniformHypergraph:main}. 
In \Cref{sec:reduce-strong-cover-to-weak-cover}, we discuss a reduction from symmetric skew-supermodular strong cover to symmetric skew-supermodular weak cover problems. This reduction  will be useful in the three connectivity augmentation applications that we discuss subsequently: 
We consider \dshlcah and \dsshlcah in \Cref{sec:local-connectivity-augmentation-of-hypergraphs-using-hyperedges},  hypergraph node-to-area connectivity augmentation using hyperedges in \Cref{sec:node-to-area-connectivity-augmentation}, and degree-constrained mixed-hypergraph global connectivity augmentation using hyperedges in \Cref{sec:global-connectivity-augmentation-of-mixed-hypergraphs}. 




\input{strong-to-weak-cover-reduction-implications}

\input{hypergraph-connectivity-augmentation-2}

\input{hypergraph-node-to-area-2}

\input{mixed-hypergraph-gca}

%% file: strong-to-weak-cover-reduction-implications.tex
\subsection{Strong cover to weak cover reduction for symmetric skew-supermodular functions}\label{sec:reduce-strong-cover-to-weak-cover}
In this section, we recall a reduction from symmetric skew-supermodular \emph{strong} cover using hyperedges problem to symmetric skew-supermodular \emph{weak} cover using hyperedges problem due to Bern\'{a}th and Kir\'{a}ly \cite{Bernath-Kiraly} and state implications of their reduction in conjunction with our main results. These implications will be useful in subsequent applications. 
Bern\'{a}th and Kir\'{a}ly showed \Cref{lemma:BK-reduce-strong-to-weak} below. 

\begin{lemma}[\hspace{-1sp}\cite{Bernath-Kiraly}]\label{lemma:BK-reduce-strong-to-weak}
    Let $p:2^V\rightarrow \Z$ be a symmetric skew-supermodular function and $(H=(V, E), w:E\rightarrow \Z_+)$ be a hypergraph with $\sum_{e\in E}w(e)=K_p$. If 
    $(H, w)$ weakly covers the function $p$, then $(H, w)$ strongly covers the function $p$. 
\end{lemma}
Our results for \dsswch and \dssswch (i.e., Theorems  
\ref{thm:SzigetiWeakCover:main}, 
\ref{thm:WeakCoverViaUniformHypergraph:main}, and \ref{thm:WeakCoverTwoFunctionsViaUniformHypergraph:main}) showed the existence of a hypergraph $(H=(V, E), w)$ that weakly covers the function $p$ such that $\sum_{e\in E}w(e)=K_p$. Hence, using \Cref{lemma:BK-reduce-strong-to-weak} in conjunction with Theorems \ref{thm:SzigetiWeakCover:main}, 
\ref{thm:WeakCoverViaUniformHypergraph:main}, and \ref{thm:WeakCoverTwoFunctionsViaUniformHypergraph:main} lead to the corollaries below regarding \emph{strong} cover for \emph{symmetric} skew-supermodular functions. 

\begin{corollary}\label{coro:SzigetiStrongCoverViaHypergraph:main}
Let $p:2^V\rightarrow\Z$ be a symmetric skew-supermodular function and $m:V\rightarrow\Z_{\ge 0}$ be a non-negative function satisfying the following two conditions: 
\begin{enumerate}[label=(\alph*), ref=\thetheorem(\alph*)]
    \item $m(X) \ge p(X)$ for every $X \subseteq V$,
    \item $m(u) \leq K_p$ for every $u \in V$.
\end{enumerate}
Then, there exists a hypergraph $\left(H = \left(V, E\right), w:E\rightarrow\Z_+\right)$ satisfying the following four properties:
\begin{enumerate}[label=(\arabic*)]
    \item $d_{(H, w)}(X) \geq p(X)$ for every $X\subseteq V$,
    \item $b_{(H, w)}(u) = m(u)$ for every $u \in V$,
    \item $\sum_{e\in E}w(e) = K_p$, and 
    \item $|E| = O(|V|)$.
\end{enumerate}
Furthermore, given a function $m:V\rightarrow \Z_{\ge 0}$ and access to \functionMaximizationOracleStrongCover{p} of a symmetric skew-supermodular function $p:2^V\rightarrow Z$ where $m$ and $p$ satisfy conditions (a) and (b), there exists an algorithm that runs in time $O(|V|^5)$ using $O(|V|^4)$ queries to \functionMaximizationOracleStrongCover{p} and returns a hypergraph satisfying the above four properties. The run-time includes the time to construct the hypergraphs that are used as inputs to \functionMaximizationOracleStrongCover{p}. 
\end{corollary}

\begin{corollary}\label{coro:StrongCoverViaUniformHypergraph:main}
Let $p:2^V\rightarrow\Z$ be a symmetric skew-supermodular function and $m:V\rightarrow\Z_{\ge 0}$ be a non-negative function satisfying the following three conditions: 
\begin{enumerate}[label=(\alph*), ref=\thelemma(\alph*)]
    \item $m(X) \ge p(X)$ for every $X \subseteq V$,
    \item $m(u) \leq K_p$ for every $u \in V$,
\end{enumerate}
Then, there exists a hypergraph $\left(H = \left(V, E\right), w:E\rightarrow\Z_+\right)$ satisfying the following five properties:
\begin{enumerate}[label=(\arabic*)]
    \item $d_{(H, w)}(X) \geq p(X)$ for every $X\subseteq V$, 
    \item $b_{(H, w)}(u) = m(u)$ for every $u \in V$, 
    \item $\sum_{e\in E}w(e) = K_p$,
    \item 
    $|e|\in \{\lfloor m(V)/K_p \rfloor, \lceil m(V)/K_p \rceil\}$ for every $e \in E$, and 
    \item $|E| = O(|V|)$. 
\end{enumerate}
Furthermore, given a function $m:V\rightarrow \Z_{\ge 0}$ and access to \functionMaximizationOracleStrongCover{p} of a symmetric skew-supermodular function $p:2^V\rightarrow Z$ where $m$ and $p$ satisfy conditions (a) and (b), there exists an algorithm that runs in time $\poly(|V|)$ using $\poly(|V|)$ queries to \functionMaximizationOracleStrongCover{p} to return a hypergraph satisfying the above five properties. The run-time includes the time to construct the hypergraphs that are used as inputs to \functionMaximizationOracleStrongCover{p}. 
Moreover, for each query to \functionMaximizationOracleStrongCover{p}, the hypergraph $(G_0, c_0)$ used as input to the query has $O(|V|)$ vertices and $O(|V|)$ hyperedges. 
\end{corollary}

\begin{corollary}\label{coro:StrongCoverTwoFunctionsViaUniformHypergraph:main}
Let $q, r : 2^V \rightarrow \Z$ be two symmetric skew-supermodular functions 
with $K_q=K_r$, 
and $p:2^V\rightarrow \Z$ be the function defined as $p(X) \coloneqq \max \left\{q(X), r(X)\right\}$ for every $X\subseteq V$.  
Furthermore, let $m:V\rightarrow\Z_{\ge 0}$ be a non-negative function satisfying the following three conditions: 
\begin{enumerate}[label=(\alph*), ref=\thelemma(\alph*)]
    \item $m(X) \ge p(X)$ for every $X \subseteq V$,
    \item $m(u) \leq K_p$ for every $u \in V$, and
\end{enumerate}
Then, there exists a hypergraph $\left(H = \left(V, E\right), w:E\rightarrow\Z_+\right)$ satisfying the following five properties:
\begin{enumerate}[label=(\arabic*)]
    \item $d_{(H, w)}(X) \geq p(X)$ for every $X\subseteq V$,
    \item $b_{(H, w)}(u) = m(u)$ for every $u \in V$,
    \item $\sum_{e\in E}w(e) = K_p$,
    \item 
    $|e|\in \{\lfloor m(V)/K_p \rfloor, \lceil m(V)/K_p \rceil\}$ for every $e \in E$, and 
    \item $|E| = O(|V|^2)$.
\end{enumerate}
Furthermore, given a function $m:V\rightarrow \Z_{\ge 0}$ and access to \functionMaximizationOracleStrongCover{q} and \functionMaximizationOracleStrongCover{r} of symmetric skew-supermodular function $q, r:2^V\rightarrow Z$ where $m, q, $ and $r$ satisfy conditions (a) and (b), there exists an algorithm that runs in time $\poly(|V|)$ using $\poly(|V|)$ queries to \functionMaximizationOracleStrongCover{q} and \functionMaximizationOracleStrongCover{r} to return a hypergraph satisfying the above five properties. The run-time includes the time to construct the hypergraphs that are used as inputs to \functionMaximizationOracleStrongCover{q} and \functionMaximizationOracleStrongCover{r}. 
Moreover, for each query to \functionMaximizationOracleStrongCover{p}, the hypergraph $(G_0, c_0)$ used as input to the query has $O(|V|)$ vertices and $O(|V|^2)$ hyperedges. 
\end{corollary}

For the setting of Corollary \ref{coro:SzigetiStrongCoverViaHypergraph:main}, Szigeti \cite{Szi99} showed that conditions (a) and (b) are necessary to achieve properties (1) and (2), and the same conditions are sufficient to achieve properties (1)-(3). Thus, using Lemma \ref{lemma:BK-reduce-strong-to-weak} and Szigeti's results, we have a reduction from \dssssch to \dsswch. Consequently, Corollary \ref{coro:SzigetiStrongCoverViaHypergraph:main}  implies the following result. 

\begin{corollary}\label{coro:strong-poly-SzigetiStrongCoverSymSkewSupmod}
    There exists an algorithm to solve \dssssch that runs in time $O(|V|^5)$ using $O(|V|^4)$ queries to \functionMaximizationOracleStrongCover{p}, where $V$ is the ground set of the input instance. Moreover, if the instance is feasible, then the algorithm returns a solution hypergraph that contains $O(|V|)$ hyperedges. 
    For each query to \functionMaximizationOracleStrongCover{p}, the hypergraph $(G_0, c_0)$ used as input to the query has $O(|V|)$ vertices and $O(|V|)$ hyperedges. 
\end{corollary}

For the setting of Corollary \ref{coro:StrongCoverViaUniformHypergraph:main}, Bern\'{a}th and Kir\'{a}ly \cite{Bernath-Kiraly} showed that conditions (a) and (b) are necessary to achieve properties (1) and (2), and the same conditions are sufficient to achieve properties (1)-(4). Thus, using Lemma \ref{lemma:BK-reduce-strong-to-weak} and their results, we have a reduction from \dsssscuh to \dsswcuh. Consequently, Corollary \ref{coro:strong-poly-StrongCoverSymSkewSupmod-near-uniform} implies the following result. 

\begin{corollary}\label{coro:strong-poly-StrongCoverSymSkewSupmod-near-uniform}
    There exists an algorithm to solve \dssssch that runs in $\poly(|V|)$ time using $\poly(|V|)$ queries to \functionMaximizationOracleStrongCover{p}, where $V$ is the ground set of the input instance. Moreover, if the instance is feasible, then the algorithm returns a solution hypergraph that is near-uniform and contains $O(|V|)$ hyperedges.  
    For each query to \functionMaximizationOracleStrongCover{p}, the hypergraph $(G_0, c_0)$ used as input to the query has $O(|V|)$ vertices and $O(|V|)$ hyperedges. 
\end{corollary}

For the setting of Corollary \ref{coro:StrongCoverTwoFunctionsViaUniformHypergraph:main}, Bern\'{a}th and Kir\'{a}ly \cite{Bernath-Kiraly} showed that conditions (a) and (b) are necessary to achieve properties (1) and (2), and the same conditions are sufficient to achieve properties (1)-(4). Thus, using Lemma \ref{lemma:BK-reduce-strong-to-weak} and their results, we have a reduction 
from \dsssssch to \dssswch (and also 
from \dssssscuh to \dssswcuh). Consequently, Corollary \ref{coro:strong-poly-simul-StrongCoverSymSkewSupmod-near-uniform} implies the following result. 
\begin{corollary}\label{coro:strong-poly-simul-StrongCoverSymSkewSupmod-near-uniform}
    There exists an algorithm to solve \dssssch that runs in $\poly(|V|)$ time using $\poly(|V|)$ queries to \functionMaximizationOracleStrongCover{p}, where $V$ is the ground set of the input instance. Moreover, if the instance is feasible, then the algorithm returns a solution hypergraph that is near-uniform and contains $O(|V|^2)$ hyperedges, where $V$ is the ground set of the input instance. 
    For each query to \functionMaximizationOracleStrongCover{p}, the hypergraph $(G_0, c_0)$ used as input to the query has $O(|V|)$ vertices and $O(|V|^2)$ hyperedges. 
\end{corollary}

The run-times in Corollaries \ref{coro:SzigetiStrongCoverViaHypergraph:main}, \ref{coro:StrongCoverViaUniformHypergraph:main}, and \ref{coro:StrongCoverTwoFunctionsViaUniformHypergraph:main} include the time to construct the hypergraphs that are used as inputs to \functionMaximizationOracleStrongCover{p}, \functionMaximizationOracleStrongCover{q}, and \functionMaximizationOracleStrongCover{r}.

%% file: hypergraph-connectivity-augmentation-2.tex
\subsection{Hypergraph local connectivity augmentation using hyperedges}\label{sec:local-connectivity-augmentation-of-hypergraphs-using-hyperedges}
In this section, we prove Theorems \ref{thm:strongly-poly-for-dshlcah}, \ref{thm:strongly-poly-for-dshlcauh}, and \ref{thm:strongly-poly-for-dsshlcauh}. We will need the following lemma from \cite{bérczi2023splittingoff} showing that a certain set function $p:2^V\rightarrow\Z$ is symmetric skew-supermodular and admits efficient algorithms to implement $\functionMaximizationOracleStrongCover{p}$. We restate and prove \Cref{thm:strongly-poly-for-dshlcah} below the lemma, and remark how to modify the proof to obtain the proofs of Theorems \ref{thm:strongly-poly-for-dshlcauh}, and \ref{thm:strongly-poly-for-dsshlcauh}.


\begin{restatable}[Lemma 3.1 of \cite{bérczi2023splittingoff}]{lemma}{lempmaxoracle}\label{lemma:helper-for-p-max-oracle}
    Let $(G=(V, E), c: E\rightarrow \Z_+)$ be a hypergraph. Let $r:\binom{V}{2}\rightarrow \Z_{\ge 0}$ be a function on pairs of elements of $V$ and $R, p_{(G, c, r)}:2^V\rightarrow \Z$ be functions defined by $R(X)\coloneqq \max\{r(\{u, v\}): u\in X, v\in V\setminus X\}$ for every $X\subseteq V$, $R(\emptyset)\coloneqq 0$, $R(V)\coloneqq 0$, and $p_{(G, c, r)}(X)\coloneqq R(X)-d_{(G, c)}(X)$ for every $X\subseteq V$. 
    \begin{enumerate}
        \item The function $p_{(G, c, r)}$ is symmetric skew-supermodular. 
        \item 
        Given hypergraph $(G, c)$, function $r$, hypergraph $(G_0=(V, E_0), c_0: E_0\rightarrow \Z_{+})$, disjoint sets $S_0, T_0\subseteq V$, and $y_0\in \R^V$, the oracle \functionMaximizationOracleStrongCover{p_{(G, c, r)}}$((G_0, c_0), S_0, T_0, y_0)$ can be computed in $O(|V|^3(|V|+|E_0|+|E|)(|E_0|+|E|))$ time.
    \end{enumerate}
\end{restatable}

\thmStronglyPolyDSHLCAH*
\begin{proof}
Let the input instance be specified by $((G,c), r, m)$, where $(G=(V, E), c:E\rightarrow \Z_+)$ is the input hypergraph, $r:\binom{V}{2}\rightarrow \Z_{\ge 0}$ is the input target connectivity function, and $m:V\rightarrow \Z_{\ge 0}$ is the input degree-specification. 
We consider the set function $R: 2^V\rightarrow \Z_{\ge 0}$ defined by $R(X)\coloneqq \max\{r(\{u, v\}): u\in X, v\in V'\setminus X\}$ for every non-empty proper subset $X\subseteq V$, $R(\emptyset)\coloneqq 0$, and $R(V)\coloneqq 0$. 
Next, we consider the function $p: 2^V\rightarrow \Z$ defined by $p(X)\coloneqq R(X)-d_{(G, c)}(X)$ for every $X\subseteq V$. By \Cref{lemma:helper-for-p-max-oracle}, the function $p$ is symmetric skew-supermodular. Moreover, for a given hypergraph $(G_0=(V, E_0), c_0: E_0\rightarrow \Z_{+})$, and disjoint sets $S_0, T_0\subseteq V$, the oracle   \functionMaximizationOracleStrongCover{p}$((G_0, c_0), S_0, T_0)$ can be implemented to run in $O(n^3(n+m +|E_0|)(m+|E_0|))$ time.

We observe that a hypergraph $(H=(V, E_H), w:E_H\rightarrow \Z_+)$ is a feasible solution for the input instance $((G,c), r, m)$ of \dshlcah if and only if the hypergraph $(H, w)$ is a feasible solution for the instance $(p,m)$ of \dssssch. By \Cref{coro:strong-poly-SzigetiStrongCoverSymSkewSupmod} there exists an algorithm to solve \dssssch that runs in time $O(n^5)$ using $O(n^4)$ queries to \functionMaximizationOracleStrongCover{p}, where for each query to \functionMaximizationOracleStrongCover{p}, the hypergraph $(G_0, c_0)$ used as input to the query has $O(n)$ vertices and $O(n)$ hyperedges. Thus, using our implementation of \functionMaximizationOracleStrongCover{p} from the first paragraph, this algorithm can be implemented to run in $O(n^7(n+m)^2)$ time.
\end{proof}

Theorem \ref{thm:strongly-poly-for-dshlcauh} follows using the same same reduction as in the proof of Theorem \ref{thm:strongly-poly-for-dshlcah} and using Corollary \ref{coro:strong-poly-StrongCoverSymSkewSupmod-near-uniform}. Similarly, Theorem \ref{thm:strongly-poly-for-dsshlcauh} follows using the same reduction as in the proof of Theorem \ref{thm:strongly-poly-for-dshlcah} and using Corollary \ref{coro:strong-poly-simul-StrongCoverSymSkewSupmod-near-uniform}: construct a pair of symmetric skew-supermodular functions $p_1$ and $p_2$ using functions $r_1$ and $r_2$ respectively similar to the proof of \Cref{thm:strongly-poly-for-dshlcah} and apply \Cref{coro:strong-poly-simul-StrongCoverSymSkewSupmod-near-uniform}. 

We note that an optimization variant of hypergraph connectivity augmentation problem is closely related: Szigeti \cite{Szi99} proposed and studied the following optimization variant of \dshlcah: 
the input is a hypergraph $(G=(V, E), c: E\rightarrow \Z_+)$ and a target connectivity function $r: \binom{V}{2}\rightarrow \Z_{\ge 0}$. 
The goal is to find a hypergraph $(H=(V, E_H), w: E_H\rightarrow \Z_+)$ 
with minimum $\sum_{e\in E_H}|e|w(e)$ (i.e., minimum total weighted size), 
such that $\lambda_{(G+H, c+w)}(u, v)\ge r(\{u, v\})$ for every $u, v\in V$. Szigeti gave a pseudo-polynomial time algorithm for this optimization problem. Our results imply a strongly polynomial time algorithm for this optimization problem via a standard reduction (which is well-known in the graph connectivity augmentation literature); moreover our algorithm also returns a near-uniform hypergraph as an optimum solution. We refrain from stating the results for the optimization variant in the interests of brevity. 

%% file: hypergraph-node-to-area-2.tex
\subsection{Hypergraph node-to-area connectivity augmentation using hyperedges}\label{sec:node-to-area-connectivity-augmentation}
The graph node-to-area connectivity augmentation using edges problem was solved by Ishii and Hagiwara \cite{IH06}. Bern\'{a}th and Kir\'{a}ly \cite{Bernath-Kiraly} proposed a hypergraph variant of this problem that we describe now. 
We recall that for a hypergraph $(H=(V, E), w:E\rightarrow \Z_+)$, a subset $W\subseteq V$ and a vertex $u\in V-W$, the min $(u, W)$-cut value is denoted by $\lambda_{(H, w)}(u, W):=\min\{d_{(H, w)}(X): u\in X\subseteq V-W\}$. 
In the degree-specified hypergraph node-to-area connectivity augmentation using hyperedges problem (denoted \dshnacah), the input is a hypergraph $(G=(V, E), c: E\rightarrow \Z_+)$, a collection $\mathcal{W}\subseteq 2^V$ of subsets of $V$, a target connectivity function $r: \mathcal{W}\rightarrow \Z_{+}$, and a degree-specification $m:V\rightarrow \Z_{\ge 0}$ . 
The goal is to verify if there exists a hypergraph $(H=(V, E_H), w: E_H\rightarrow \Z_+)$ such that 
$\lambda_{(G+H, c+w)}(u, W)\ge r(W)$ for every $W\in \mathcal{W}$ and $u\in V-W$ and $b_{(H, w)}(u)=m(u)$ for every $u\in V$, and if so, then find such a hypergraph.  

Bern\'{a}th and Kir\'{a}ly \cite{Bernath-Kiraly} showed that feasible instances of \dshnacah admit a near-uniform hypergraph as a feasible solution and such a solution can be found in pseudo-polynomial time. Our results imply a strongly polynomial-time algorithm for this problem. The next lemma shows that the appropriate function maximization oracle can be implemented in strongly polynomial time. We defer the proof of the lemma to \Cref{sec:appendix-applications-p-max-oracles}. We state and prove the main result of the section below the lemma.

\begin{restatable}{lemma}{lemNodetoAreaConnectivityAugmentationHelper}\label{lemma:helper-for-p-max-oracle-for-node-to-area-connectivity-augmentation}
        Let $(G=(V, E), c: E\rightarrow \Z_+)$ be a hypergraph and $r:\mathcal{W}\rightarrow \Z_{\ge 0}$ be a function defined over a family $\mathcal{W}\subseteq 2^V$. 
        Let $p_{(G, c, r)}:2^V\rightarrow \Z$ be the function defined as follows: for every $X \subseteq V$,
         $$p_{(G, c, r)}(X) := \begin{cases}
         \max\left\{r(W): W\in \mathcal{W}, W\cap X=\emptyset \text{ or } W\subseteq X\right\} - d_{(G, c)}(X)& \text{if } \emptyset \subsetneq X \subsetneq V,\\
        0 & \text{otherwise, i.e. } X \in \{\emptyset, V\}.
    \end{cases}$$
    Then, we have that
    \begin{enumerate}
        \item the function $p_{(G, c, r)}$ is symmetric skew-supermodular, and
        \item given hypergraph $(G, c)$, the family $\mathcal{W}$, the function $r: \mathcal{W}\rightarrow \Z_{\ge 0}$, hypergraph $(G_0=(V, E_0), c_0: E_0\rightarrow \Z_{+})$, disjoint sets $S_0, T_0\subseteq V$, and vector $y_0\in \R^V$, the oracle   \functionMaximizationOracleStrongCover{p_{(G, c, r)}}$((G_0, c_0), S_0, T_0, y_0)$ can be computed in $\poly(|V|, |E|, |E_0|)$ time.
    \end{enumerate}
 \end{restatable}

\begin{theorem}
    There exists a strongly polynomial time algorithm to solve \dshnacah. Moreover, if the instance is feasible, then the algorithm returns a solution hypergraph that is near-uniform and contains $O(n)$ hyperedges, where $n$ is the number of vertices in the input hypergraph. 
\end{theorem}
\begin{proof}
    Let the input instance of \dshnacah be $((G,c), \mathcal{W}, r, m)$, where $(G=(V, E), c: E\rightarrow \Z_+)$ is a hypergraph, $\mathcal{W}\subseteq 2^V$ is a collection of subsets, $r: \mathcal{W}\rightarrow \Z_{+}$ is the target connectivity function, and $m:V\rightarrow \Z_{\ge 0}$ is the degree-specification. 
    We consider the set function $R: 2^V\rightarrow \Z_{\ge 0}$ defined by $R(X)\coloneqq \max\{r(W): W\in \mathcal{W}, W\cap X=\emptyset \text{ or } W\subseteq X\}$ for every non-empty proper subset $X\subseteq V$, $R(\emptyset)\coloneqq 0$, and $R(V)\coloneqq 0$. 
    Next, we consider the function $p: 2^V\rightarrow \Z$ defined by $p(X)\coloneqq R(X)-d_{(G, c)}(X)$ for every $X\subseteq V$. The function $p$ is symmetric skew-supermodular by \Cref{lemma:helper-for-p-max-oracle-for-node-to-area-connectivity-augmentation}. Moreover, for a given hypergraph $(G_0=(V, E_0), c_0: E_0\rightarrow \Z_{+})$, and disjoint sets $S_0, T_0\subseteq V$, and vector $y_0\in \R^V$, the oracle   \functionMaximizationOracleStrongCover{p}$((G_0, c_0), S_0, T_0)$ can be implemented to run in $\poly(|V|, |E|, |E_0|)$ time. 

    We observe that a hypergraph $(H=(V, E_H), w: E_H\rightarrow \Z_+)$ is a feasible solution for the input instance $((G,c), \mathcal{W}, r, m)$ of \dshnacah if and only if the hypergraph $(H, w)$ is a feasible solution for the instance $(p,m)$ of \dssssch. 
     By \Cref{coro:strong-poly-StrongCoverSymSkewSupmod-near-uniform} there exists an algorithm that returns a near-uniform hypergraph solution for \dssssch that runs in time $\poly(|V|)$ using $\poly(|V|)$ queries to \functionMaximizationOracleStrongCover{p}, where for each query to \functionMaximizationOracleStrongCover{p}, the hypergraph $(G_0, c_0)$ used as input to the query has $O(|V|)$ vertices and $O(|V|)$ hyperedges. Thus, using our implementation of \functionMaximizationOracleStrongCover{p} from the first paragraph, this algorithm can be implemented to run in $\poly(|V|, |E|)$ time.
\end{proof}

We note that a \emph{simultaneous} version of \dshnacah can be defined (analogous to \dshlcah and \dsshlcah) and \Cref{coro:StrongCoverTwoFunctionsViaUniformHypergraph:main} can be applied to derive a strongly polynomial time algorithm for it. The optimization variants instead of the degree-specified variants can also be solved in strongly polynomial time. We refrain from stating the details in the interests of brevity.

%% file: mixed-hypergraph-gca.tex
\subsection{Degree-constrained mixed-hypergraph global connectivity augmentation using hyperedges}\label{sec:global-connectivity-augmentation-of-mixed-hypergraphs}
The notion of mixed-hypergraph generalizes undirected graphs/hypergraphs and directed graphs/hypergraphs (e.g., see \cite{Frank-book}). 
A mixed-hypergraph $(G=(V, A), c: A\rightarrow \Z_+)$ consists of a finite set $V$ of vertices and a set $A$ of hyperarcs with weights $c(e)$ for every $e\in A$, where each hyperarc $e\in A$ is specified by an ordered tuple $(\text{tails}(e), \text{ heads}(e), \text{ heads-tails}(e))$ satisfying $\text{tails}(e), \text{ heads}(e), \text{ heads-tails(e)}\subseteq V$ with $\text{tails}(e)\cap \text{heads}(e)=\emptyset$, $\text{tails}(e)\cap \text{heads-tails}(e)=\emptyset$, $\text{heads}(e)\cap \text{heads-tails}(e)=\emptyset$, $\text{tails}(e)\cup \text{heads-tails}(e)\neq \emptyset$ and $\text{heads}(e)\cup \text{heads-tails}(e)\neq \emptyset$. 
For an undirected hypergraph $(H=(V, E_H), w_H: E_H\rightarrow \Z_+)$, we define an associated mixed-hypergraph $(\overrightarrow{H}=(V, \overrightarrow{E_H}), w_{\overrightarrow{H}}: \overrightarrow{E_H}\rightarrow \Z_+)$ where for every
hyperedge $e_H\in E_H$, we introduce a hyperarc $e_{\overrightarrow{H}} \coloneqq (\text{tails}(e_{\overrightarrow{H}})=\emptyset, \text{ heads}(e_{\overrightarrow{H}})=\emptyset, \text{ heads-tails}(e_{\overrightarrow{H}})=e_H)$ into $\overrightarrow{E_H}$ with weight $w_{\overrightarrow{H}}(e_{\overrightarrow{H}})=w_H(e_H)$. 
For a set $X\subseteq V$, we define $\delta_G^{in}(X)\coloneqq \{e\in A: (\text{heads}(e)\cup \text{heads-tails}(e))\cap A \neq \emptyset, (\text{tails}(e)\cup\text{heads-tails}(e))\cap (V\setminus A)\neq \emptyset\}$ and the function $d^{in}_{(G, w)}:2^V\rightarrow \Z_{\ge 0}$ defined by $d^{in}_{(G, w)}(X)\coloneqq \sum_{e\in \delta_G^{in}(X)}w(e)$ for every $X\subseteq V$. 
For distinct vertices $u, v\in V$, we define $\lambda_{(G, w)}(u, v)\coloneqq \min\{d^{in}_{(G, w)}(X): v\in X\subseteq V-\{u\}\}$. We note that $\lambda_{(G, w)}(u, v)$ is not necessarily equal to $\lambda_{(G, w)}(v, u)$ for distinct vertices $u, v\in V$. 
Let $r\in V$ be a designated root vertex. 
If $\lambda_{(G, w)}(v, r)\ge k$ and $\lambda_{(G, w)}(r, v)\ge \ell$ for every $v\in V\setminus \{r\}$, then the mixed-hypergraph is said to be $(k,\ell)$-rooted-arc-connected. 

We consider the degree-constrained mixed-hypergraph global connectivity augmentation using hyperedges problem (denoted \dcmhgcah): the input is a degree-bound function $m: V\rightarrow \Z_{\ge 0}$, a mixed-hypergraph $(G=(V, A), c: A\rightarrow \Z_+)$, a designated root vertex $r\in V$, and values $k, \ell\in \Z_+$. 
The goal is to verify if there exists an undirected hypergraph $(H=(V, E_H), w: E_H\rightarrow \Z_+)$ 
such that 
$b_{(H, w)}(u)\le m(u)$ for every $u\in V$ 
and the mixed-hypergraph $(G+\overrightarrow{H}, c+w_{\overrightarrow{H}})$ is $(k,\ell)$-rooted-arc-connected, where 
$(G+\overrightarrow{H}=(V, E_{G+\overrightarrow{H}}), c+w_{\overrightarrow{H}})$ is the mixed-hypergraph on vertex set $V$ and hyperarc set $E_{G+\overrightarrow{H}}\coloneqq A\cup \overrightarrow{E_H}$ with the weight of every 
hyperarc $e\in A\cap \overrightarrow{E_H}$ being $c(e)+w_{\overrightarrow{H}}(e)$, the weight of every hyperarc $e\in A\setminus \overrightarrow{E_H}$ being $c(e)$ and the weight of every hyperarc $e\in \overrightarrow{E_H}\setminus A$ being $w_{\overrightarrow{H}}(e)$; moreover, the goal involves finding such a hypergraph if it exists. 
Bern\'{a}th and Kir\'{a}ly \cite{Bernath-Kiraly} introduced this problem and showed a complete characterization for the existence of a feasible solution. They also showed that for feasible instances, there exists a feasible solution that is also near-uniform. Our results imply a strongly polynomial-time algorithm for this problem. The next lemma shows that the appropriate function maximization oracle can be implemented in strongly polynomial time. We defer the proof of the lemma to \Cref{sec:appendix-applications-p-max-oracles}. We state and prove the main result of the section below the lemma.

\begin{restatable}{lemma}{lemMixedConnectivityAugmentationHelper}\label{lemma:helper-mixed-connectivity-augmentation}
        Let $g: 2^V\rightarrow \Z$ be a submodular function, $r\in V$, and $k,\ell\in \Z_{\ge 0}$. Let $p: 2^V\rightarrow \Z$ be the function defined by 
        \begin{align*}
    p(X)&\coloneqq 
    \begin{cases}
        0 & \text{ if } X=\emptyset \text{ or } X=V,\\
        k - g(X) & \text{ if } \emptyset \neq X\subseteq V\setminus \{r\},\text{ and}\\ 
        \ell - g(X) & \text{ if } r\in X\subsetneq V,  
    \end{cases}
    \end{align*}
     and let $p_{sym}: 2^V\rightarrow \Z$ be the function defined by $p_{sym}(X)\coloneqq \max\{p(X), p(V\setminus X)\}$ for every $X\subseteq V$. Then we have that
     \begin{enumerate}
        \item the function $p$ is symmetric skew-supermodular, and
        \item given access to the evaluation oracle of the function $g$, a hypergraph $(G_0=(V, E_0), c_0: E_0\rightarrow \Z_{+})$, disjoint sets $S_0, T_0\subseteq V$, and vector $y_0\in \R^V$, the oracle   \functionMaximizationOracleStrongCover{p}$((G_0, c_0), S_0, T_0, y_0)$ can be computed in $\poly(|V|, |E_0|)$ time.
    \end{enumerate}
\end{restatable}

\begin{theorem}
    There exists a strongly polynomial time algorithm to solve \dcmhgcah. Moreover if the instance is feasible, then the algorithm returns a solution hypergraph that is near-uniform and contains $O(n)$ hyperedges, where $n$ is the number of vertices in the input hypergraph. 
\end{theorem}
\begin{proof}
    Let the input instance of \dcmhgcah be $((G, c), r, m, k, \ell)$, where $(G=(V, A), w: A\rightarrow \Z_+)$ is a mixed hypergraph, $r\in V$ is the designated root vertex, $m: V\rightarrow \Z_{\ge 0}$ is the degree-bound function, and $k, \ell\in \Z_+$. 
    We note that the function $d^{in}_{(G, c)}:2^V\rightarrow \Z$ is submodular. 
    Consider the function $p: 2^V\rightarrow \Z$ defined by 
    \begin{align*}
    p(X)&\coloneqq 
    \begin{cases}
        0 & \text{ if } X=\emptyset \text{ or } X=V,\\
        k - d^{in}_{(G, c)}(X) & \text{ if } \emptyset \neq X\subseteq V\setminus \{r\},\text{ and}\\ 
        \ell - d^{in}_{(G, c)}(X) & \text{ if } r\in X\subsetneq V. 
    \end{cases}
    \end{align*}
    We define $p_{sym}: 2^V\rightarrow \Z$ as $p_{sym}(X)\coloneqq \max\{p(X), p(V\setminus X)\}$ for every $X\subseteq V$. The function $p$ is symmetric skew-supermodular by \Cref{lemma:helper-mixed-connectivity-augmentation}. Moreover, for a given hypergraph $(G_0=(V, E_0), c_0: E_0\rightarrow \Z_{+})$, and disjoint sets $S_0, T_0\subseteq V$, and vector $y_0\in \R^V$, the oracle   \functionMaximizationOracleStrongCover{p}$((G_0, c_0), S_0, T_0)$ can be implemented to run in $\poly(|V|, |A|, |E_0|)$ time---we note that the evaluation oracle for the function $d^{in}_{(G_0, c_0)}$ can be implemented in $\poly(|V|, |A|)$ time. 

    We observe that a hypergraph $(H=(V, E_H), w: E_H\rightarrow \Z_+)$ is a feasible solution for the input instance $((G,c), r, m, k, \ell)$ of \dcmhgcah if and only if the hypergraph $(H, w)$ is a feasible solution for the instance $(p,m)$ of \dssssch. 
     By \Cref{coro:StrongCoverViaUniformHypergraph:main} there exists an algorithm that returns a near-uniform hypergraph solution for \dssssch that runs in time $\poly(|V|)$ using $\poly(|V|)$ queries to \functionMaximizationOracleStrongCover{p}, where for each query to \functionMaximizationOracleStrongCover{p}, the hypergraph $(G_0, c_0)$ used as input to the query has $O(|V|)$ vertices and $O(|V|)$ hyperedges. Thus, using our implementation of \functionMaximizationOracleStrongCover{p} from the first paragraph, this algorithm can be implemented to run in $\poly(|V|, |A|)$ time.
\end{proof}

We note that a \emph{simultaneous} version of \dcmhgcah can be defined (analogous to \dshlcah and \dsshlcah) and \Cref{coro:StrongCoverTwoFunctionsViaUniformHypergraph:main} can be applied to derive a polynomial-time algorithm for it. We refrain from stating the details in the interests of brevity. 


%% file: conclusion.tex
\section{Conclusion}\label{sec:conclusion}
The central theme of this work is showing that certain hypergraph network design problems admit solution hypergraphs with polynomial number of hyperedges and moreover, can be solved in strongly polynomial time. 
Our algorithms also return solution hypergraphs that are near-uniform. 
Our results are for certain abstract function cover problems but they have numerous applications; in particular, they enable the first strongly polynomial time algorithms for (i) degree-specified hypergraph connectivity augmentation using hyperedges, (ii) degree-specified hypergraph node-to-area connectivity augmentation using hyperedges, and (iii) degree-constrained mixed-hypergraph global connectivity augmentation using hyperedges.
Previous best-known run-time for these problems were pseudo-polynomial. We believe that the abstract function cover problems might find more applications in the future. 

The abstract function cover problems were introduced by Szigeti \cite{Szi99} and Bern\'{a}th and Kir\'{a}ly \cite{Bernath-Kiraly}. Previously known algorithms for these problems were pseudo-polynomial and there exist instances on which these algorithms necessarily take exponential time. Our main contributions are modifying these algorithms and introducing suitable potential functions to bound the recursion depth of the modified algorithms. 



Our work raises two natural open questions. \emph{Algorithmic question:} Is it possible to solve \dshlcah in near-linear time? \emph{Structural question:} For feasible instances of \dshlcah, does there exist a solution hypergraph whose \emph{size} is linear in the number of vertices? We define the size of a hypergraph to be the sum of the sizes of the hyperedges (and not simply the number of hyperedges). Our results show that 
there exists a solution hypergraph in which 
the number of hyperedges is linear and hence, the size of such a solution hypergraph is quadratic in the number of vertices. We believe that an affirmative answer to the algorithmic question would also imply an affirmative answer to the structural question. On the other hand, answering the structural question would be a helpful stepping stone towards the algorithmic question. We note that recent results have shown that \dsglcae for the case of uniform requirement function can be solved in near-linear time \cite{CLP22-soda, CLP22-stoc, CHLP23}.

%% file: acknowledgement.tex
\medskip
\paragraph{Acknowledgements.} 
Karthekeyan and Shubhang were supported in part by NSF grants CCF-1814613 and CCF-1907937. 
Karthekeyan was supported in part by the Distinguished Guest Scientist Fellowship of the Hungarian Academy of Sciences -- grant number VK-6/1/2022. 
Krist\'{o}f and Tam\'{a}s were supported in part 
by the Lend\"ulet Programme of the Hungarian Academy of Sciences -- grant number LP2021-1/2021, by the Ministry of Innovation and Technology of Hungary from the National Research, Development and Innovation Fund -- grant number ELTE TKP 2021-NKTA-62 funding scheme, and by the Dynasnet European Research Council Synergy project -- grant number ERC-2018-SYG 810115.

%% file: appendix.tex
\input{appendix-applications}
\input{appendix-uncrossing-properties}
\input{appendix-function-maximization-oracles}

%% file: appendix-applications.tex
\section{Maximization oracles for functions arising in applications}\label{sec:appendix-applications-p-max-oracles}
In this section, we show that certain set functions $p: 2^V\rightarrow \Z$ arising in our applications are skew-supermodular and admit efficient algorithms to implement  
\functionMaximizationOracleStrongCover{p}. 
We prove Lemmas \ref{lemma:helper-for-p-max-oracle-for-node-to-area-connectivity-augmentation}, and \ref{lemma:helper-mixed-connectivity-augmentation}. 
These lemmas have been observed in the literature before. We include their proofs for the sake of completeness. 

\lemNodetoAreaConnectivityAugmentationHelper*
    \begin{proof}
    For ease of notation, let $p:=p_{(G, c, r)}$
    We note that the function $p$ was shown to be symmetric skew-supermodular by Bernath and Kiraly \cite{Bernath-Kiraly}. We show the second part of the lemma.
    Let the input to the oracle be the hypergraph $(G_0=(V, E_0), c_0: E_0\rightarrow \Z_+)$, disjoint subsets $S_0, T_0\subseteq V$, and vector $y_0\in \R^V$. For a subset $Z\subseteq V$, we use $y_0(Z)\coloneqq \sum_{z\in Z}y_0(z)$. 
    Let $\mathcal{W}_{S_0}\coloneqq \{W\in \mathcal{W}: W\cap S_0=\emptyset\}$ and $\mathcal{W}_{T_0}\coloneqq \{W\in \mathcal{W}: W\cap T_0=\emptyset\}$. 
    We define the following: 
    \begin{enumerate}
        \item For every  $W\in \mathcal{W}_{T_0}$, let 
        \[
        \lambda_1(W)\coloneqq 
            \min\left\{d_{(G, c)}(Z) + d_{(G_0, c_0)}(Z) -y_0(Z): S_0+W\subseteq Z\subseteq V-T_0\right\}
        \]
        and $Z_1^W$ be a set that achieves the minimum. 
        \item For every  $W\in \mathcal{W}_{S_0}$, let 
        \[
        \lambda_2(W)\coloneqq \min\left\{d_{(G, c)}(Z) + d_{(G_0, c_0)}(Z) -y_0(Z): S_0\subseteq Z\subseteq V-T_0-W\right\}
        \]
        and $Z_2^W$ be a set that achieves the minimum. 
    \end{enumerate}
    For $W\in \mathcal{W}_{S_0}\cup \mathcal{W}_{T_0}$, let 
    \begin{align*}
    \lambda(W) &\coloneqq  
    \begin{cases}
        \min\{\lambda_1(W), \lambda_2(W)\} \text{ if } W\in \mathcal{W}_{S_0}\cap \mathcal{W}_{T_0}\\
        \lambda_1(W) \text{ if } W\in \mathcal{W}_{T_0}\\
        \lambda_2(W) \text{ if } W\in \mathcal{W}_{S_0}. 
    \end{cases}\\
    Z^W &\coloneqq  
    \begin{cases}
        \arg\min\{d_{(G, c)}(Z_1^W) + b_{(G_0, c_0)}(Z_1^W) -y_0(Z_1^W), d_{(G, c)}(Z_2^W) + b_{(G_0, c_0)}(Z_2^W) -y_0(Z_2^W)\} \text{ if } W\in \mathcal{W}_{S_0}\cap \mathcal{W}_{T_0}\\
        Z_1^W \text{ if } W\in \mathcal{W}_{T_0}\\
        Z_2^W \text{ if } W\in \mathcal{W}_{S_0}. 
    \end{cases}
    \end{align*}
    For every  $W\in \mathcal{W}_{S_0}\cup \mathcal{W}_{T_0}$,  the value $\lambda(W)$ and the set $Z^W$ can be computed in $\poly(|V|, |E|, |E_0|)$ time via submodular minimization. For ease of notation, let $R := p + d_{(G, c)}$. Now, we observe that 
    \begin{align*}
    \max&\left\{p(Z)-d_{(G_0, c_0)}(Z)+y_0(Z): S_0\subseteq Z\subseteq V-T_0\right\}\\
    &=\max\left\{R(Z)-d_{(G, c)}(Z)-d_{(G_0, c_0)}(Z)+y_0(Z): S_0\subseteq Z\subseteq V-T_0\right\}\\
    &=\max\left\{r(W)-\lambda(W): W\in \mathcal{W}_{S_0}\cup \mathcal{W}_{T_0}\right\}. 
    \end{align*}
    The RHS problem can be solved in $\poly(|V|, |E|, |E_0|)$ time by iterating over all $W\in \mathcal{W}_{S_0}\cup \mathcal{W}_{T_0}$ and returning a set $W$ for which the objective function $r(W)-\lambda(W)$ is maximum. Let $W\in \mathcal{W}_{S_0}\cup \mathcal{W}_{T_0}$ be a set which achieves the maximum in the RHS problem. Then, we observe that the set $Z^W$ is the maximizer of the LHS problem. Thus, it suffices to return $Z^W$ and the value $p(Z^W)=R(Z^W)-d_{(G_0,c_0)}(Z^W)$ which can also be computed in $\poly(|V|, |E|, |E_0|)$ time. This completes the proof that \functionMaximizationOracle{p} can be implemented in $\poly(|V|, |E|, |E_0|)$ time. 
    \end{proof}

\lemMixedConnectivityAugmentationHelper*
    \begin{proof}
    We note that the function $p$ was shown to be symmetric skew-supermodular by Bernath and Kiraly \cite{Bernath-Kiraly}. We show the second part of the lemma.
        Let the input to the oracle be the hypergraph $(G_0=(V, E_0), c_0: E_0\rightarrow \Z_+)$, disjoint subsets $S_0, T_0\subseteq V$, and vector $y_0\in \R^V$. 
    For a subset $Z\subseteq V$, we use $y_0(Z)\coloneqq \sum_{z\in Z}y_0(z)$. 
    Let $p_1, p_2:2^{V\setminus \{r\}}\rightarrow \Z$ be functions defined by 
    \begin{align*}
        p_1(X)&\coloneqq p(X)-d_{(G_0, c_0)}(X)-y_0(X)\ \forall\ X\subseteq V\setminus \{r\}\ \text{ and}\\
        p_2(X)&\coloneqq p(X\cup\{r\})-d_{(G_0, c_0)}(X\cup\{r\})-y_0(X\cup\{r\})\ \forall\ X\subseteq V\setminus \{r\}. 
    \end{align*}
    Then, the functions $p_1$ and $p_2$ are intersecting supermodular  since the function $g$ and the coverage function $b_{(G_0, c_0)}$ are submodular (see \cite{Frank-book} for the definition of intersecting submodular functions). For an intersecting submodular function $f:2^V\rightarrow \Z$ and given disjoint subsets $S, T\subseteq V$ using the evaluation oracle for $f$, we can find an optimum solution to $\min\{f(X): S\subseteq X\subseteq V-T\}$ in strongly polynomial time \cite{Frank-book}. 
    Consequently, using access to the evaluation oracle of the function $g$, we can find sets $Z_1, Z_2\subseteq (V\setminus \{r\})-T_0$ such that $S_0\subseteq Z_1, Z_2$ and $p_1(Z_1)\ge p_1(X)$ and $p_2(Z_2)\ge p_2(X)$ for every $X\subseteq V\setminus \{r\}$ via submodular minimization in $\poly(|V|, |E_0|)$. Let $Z\coloneqq \arg\min\{p(Z_1), p(Z_2+r)\}$. Then, we observe that the set $Z$ is an optimum solution to 
    \[
    \max\left\{p(Z)-d_{(G_0, c_0)}(Z)+y_0(Z): S_0\subseteq Z\subseteq V-T_0\right\}.
    \]
    Hence, it suffices to return $Z$ and its function value $p(Z)$ which can be computed in $\poly(|V|, |E_0|)$ time using access to the evaluation oracle of the function $g$. 
    \end{proof}

%% file: appendix-uncrossing-properties.tex
\section{Properties of Laminar Families and Skew-Supermodular Functions}\label{appendix:sec:uncrossing-properties}
In this section, we prove \Cref{lem:UncrossingProperties:calFp-disjoint}, \Cref{lem:UncrossingProperties:Cumulative-Minimal-p-Maximizer-Family-Laminar}, \Cref{lem:optimizing-over-Q-polyhedron-intersection} and \Cref{lem:alpha4-oracle}.
\subsection{Cumulative Minimal Maximizer Family Across Function Sequences}\label{appendix:sec:uncrossing-properties:cumulative-minimal-maximizers-family-across-function-sequences}

We recall that for a function $p:2^V\rightarrow\Z$, the family of minimal $p$-maximizers is denoted by $\calF_{p}$. We now restate and prove \Cref{lem:UncrossingProperties:calFp-disjoint} which says that this family is disjoint if $p$ is a skew-supermodular function.

\calFpDisjoint*
\begin{proof}
    By way of contradiction, let $X, Y \in \calF_{p}$ be distinct sets such that $X\cap Y = \emptyset$. We note that $X - Y, Y - X \not = \emptyset$ as otherwise either $X\subset Y$ or $Y\subset X$ which contradicts the minimality of the sets $Y$ and $X$ respectively. We consider two cases based on the behavior of the function $p$ on the sets $X, Y$. First, suppose that $p(X) + p(Y) \leq p(X\cup Y) + p(X\cap Y)$. Then we have that $p(X) + p(Y) \leq p(X\cup Y) + p(X\cap Y) < p(X) + p(Y),$ a contradiction. Here, we have $p(X\cup Y)\leq p(X)$ since $X$ is a $p$-maximizer and $p(X\cap Y) < p(Y)$ since $Y$ is a minimal $p$-maximizer. Next suppose that $p(X) + p(Y) \leq p(X - Y) + p(X - Y)$. Then we have that $p(X) + p(Y) \leq p(X - Y) + p(X - Y) < p(X) + p(Y),$ a contradiction. Here, the inequality is because $p(X - Y)< p(X)$ and $p(Y - X) < p(Y)$ since the sets $X$ and $Y$ are minimal $p$-maximizers.
\end{proof}

Let $\ell \in \Z_+$ be a positive integer. Let $p_1, p_2, \ldots, p_\ell$ be a sequence of skew-supermodular functions defined over a ground set $V$. We recall that the family $\calF_{p_{\leq i}}$ for every $i \in [\ell]$ is defined as follows:
$$\calF_{p_{\leq i}} \coloneqq  \bigcup_{j \in [i]} \calF_{p_{i}}.$$
We now restate and prove \Cref{lem:UncrossingProperties:Cumulative-Minimal-p-Maximizer-Family-Laminar} which says that the family $\calF_{p_{\leq i}}$ is laminar for a specific sequence of functions. 
\cumulativeMinimalMaximizerFamilyLaminar*
\begin{proof}
    We first show a useful claim that establishes a relationship between two arbitrary functions in the sequence  $p_1, p_2, \ldots p_\ell$.
    \begin{claim}\label{lem:calFp_leqell_laminar:p_b-to-p_a}
Let $a, b\in [\ell]$ with $a<b$. Then, we  have that
    $$p_b(X)= \max\left\{ p_a\left( X \uplus\biguplus_{j \in [a, b-1]}R_j\right) - \sum_{j \in [a, b-1]}g_j\left(X\uplus\biguplus_{k \in [j+1, b - 1]} R_k\right) : R_j \subseteq \zeros_j \text{ for } j\in[a, b-1]  \right\}.$$
\end{claim}
\begin{proof}
    We prove the claim by induction on $b - a$. The claim for the base case $b = a + 1$ is because $p_{a+1} =  \functionContract{p_a}{\zeros_a} - g_a$. We consider the inductive case $b > a + 1$. We have the following:
    \begin{align*}
        p_b(X) &= \max\left\{ p_{b-1}\left( X \uplus R_{b-1}\right) - g_{b-1}\left(X\right) : R_{b-1} \subseteq \zeros_{b-1}\right\} &\\
        &= \max\left\{ p_a\left( X \uplus\biguplus_{j \in [a, b-1]}R_j\right) - \sum_{j \in [a, b-1]}g_j\left(X\uplus\biguplus_{k \in [j+1, b - 1]} R_k\right) : R_j \subseteq \zeros_j \text{ for } j\in[a, b-1]  \right\}.&
    \end{align*}
    Here, the first equality is because $p_{b} =  \functionContract{p_{b-1}}{\zeros_{b-1}} - g_{b-1}$ and the second equality is by the inductive hypothesis.
\end{proof}

By way of contradiction, suppose that the family $\calF_{p_{\leq \ell}}$ is not laminar. Then, there exist sets $A, B \in \calF_{p_{\leq \ell}}$ such that $A-B, B-A, A\cap B \not = \emptyset$. Let $a \coloneqq  \min\{i \in [\ell] : A \in \calF_{p_a}\}$ and $b \coloneqq  \min\{i \in [\ell] : B \in \calF_{p_b}\}$. We have that $a \not = b$ as otherwise the sets $A, B$ are disjoint by \Cref{lem:UncrossingProperties:calFp-disjoint}, contradicting our choice of the sets $A, B$. We assume without loss of generality that $a < b$. We note that $B \subseteq V_a$ since $B \subseteq V_b\subseteq V_a$. By \Cref{lem:calFp_leqell_laminar:p_b-to-p_a}, we have that 
    $$p_b(B) = \max\left\{ p_a\left( B \uplus\biguplus_{j \in [a, b-1]}R_j\right) - \sum_{j \in [a, b-1]}g_j\left(B\uplus\biguplus_{k \in [j+1, b - 1]} R_k\right) : R_j \subseteq \zeros_j \text{ for } j\in[a, b-1]  \right\}.$$
    Consequently, 
    \begin{align*}
        &K_{p_a} + K_{p_b}&\\
        &= p_a(A) + p_b(B)&\\
        & = p_a(A) + \max\left\{ p_a\left( B \uplus\biguplus_{j \in [a, b-1]}R_j\right) - \sum_{j \in [a, b-1]}g_j\left(B\uplus\biguplus_{k \in [j+1, b - 1]} R_k\right) : R_j \subseteq \zeros_j \text{ for } j\in[a, b-1]  \right\}&\\
        & = p_a(A) + p_a\left( B \uplus\biguplus_{j \in [a, b-1]}\Tilde{R}_j\right) - \sum_{j \in [a, b-1]}g_j\left(B\uplus\biguplus_{k \in [j+1, b - 1]} \Tilde{R}_k\right),&
    \end{align*}
    for some sets $\Tilde{R}_k \subseteq \zeros_k$ for every $k \in [a, b-1]$. For notational convenience, we denote $S \coloneqq B \uplus\biguplus_{j \in [a, b-1]}\Tilde{R}_j$. Now, we consider two cases based on the behavior of the function $p_a$ at the sets $A, S$. 
    
    First, we consider the case where the function $p_a$ is locally supermodular at the sets $A$ and $S$, i.e. $p_a(A) + p_a(S) \leq p_{a}(A\cup S) + p_a(A\cap S)$. Then we have the following:
    {\small
    \begin{align*}
        &K_{p_a} + K_{p_b}\\
        & = p_a(A) + p_a(S) - \sum_{j \in [a, b-1]}g_j\left(B\uplus\biguplus_{k \in [j+1, b - 1]} \Tilde{R}_k\right)\\
        &\leq p_a(A\cup S) + p_a(A\cap S) - \sum_{j \in [a, b-1]}g_j\left(B\uplus\biguplus_{k \in [j+1, b - 1]} \Tilde{R}_k\right)\\
        &\leq K_{p_a} + p_a(A\cap S) - \sum_{j \in [a, b-1]}g_j\left(\left(A\cap B\right)\uplus\biguplus_{k \in [j+1, b - 1]} \Tilde{R}_k\right)\\
        &= K_{p_a} + p_a\left(A\cap \left(B \uplus\biguplus_{j \in [a, b-1]}\Tilde{R}_j\right)\right) - \sum_{j \in [a, b-1]}g_j\left(\left(A\cap B\right)\uplus\biguplus_{k \in [j+1, b - 1]} \Tilde{R}_k\right)\\
        &= K_{p_a} + p_a\left(\left( A\cap B\right) \uplus \left( A\cap\biguplus_{j \in [a, b-1]}\Tilde{R}_j\right)\right) - \sum_{j \in [a, b-1]}g_j\left(\left(A\cap B\right)\uplus\biguplus_{k \in [j+1, b - 1]} \Tilde{R}_k\right)\\
        &\leq K_{p_a} + \max\left\{ p_a\left( (A\cap B) \uplus\biguplus_{j \in [a, b-1]}R_j\right) - \sum_{j \in [a, b-1]}g_j\left((A\cap B)\uplus\biguplus_{k \in [j+1, b - 1]} R_k\right) : R_j \subseteq \zeros_j \text{ for } j\in[a, b-1]  \right\}\\
        & = K_{p_a} + p_{b}(A\cap B)\\
        &\leq K_{p_a} + K_{p_b}.
    \end{align*}
    }
    Here, the second inequality is by the monotonicity of the functions $g_j$ for every $j \in [\ell]$, while the last equality is by \Cref{lem:calFp_leqell_laminar:p_b-to-p_a}.
    Thus, all inequalities are equations and we have that $p_b(A\cap B) = K_{p_b}$, contradicting minimality of the set $B \in \calF_{p_b}$. 

    Next, we consider the case where the function $p_a$ is locally negamodular at the sets $A$ and $S$, i.e. $p_a(A) + p_a(S) \leq p_a(A - S) + p_a(S - A)$. Then we have the following:
    {\small
    \begin{align*}
        &K_{p_a} + K_{p_b}&\\
        & = p_a(A) + p_a(S) - \sum_{j \in [a, b-1]}g_j\left(B\uplus\biguplus_{k \in [j+1, b - 1]} \overline{R}_k\right)&\\
        &\leq p_a(A - S) + p_a(S - A) - \sum_{j \in [a, b-1]}g_j\left(B\uplus\biguplus_{k \in [j+1, b - 1]} \overline{R}_k\right)&\\
        &\leq K_{p_a} + p_a(S - A) - \sum_{j \in [a, b-1]}g_j\left(\left(B - A\right)\uplus\biguplus_{k \in [j+1, b - 1]} \overline{R}_k\right)&\\
        &= K_{p_a} + p_a\left(\left(B \uplus\biguplus_{j \in [a, b-1]}\overline{R}_j\right) - A\right) - \sum_{j \in [a, b-1]}g_j\left(\left(B - A\right)\uplus\biguplus_{k \in [j+1, b - 1]} \overline{R}_k\right)&\\
        &= K_{p_a} + p_a\left(\left(B - A\right) \uplus \left( \biguplus_{j \in [a, b-1]}\overline{R}_j\right) - A\right) - \sum_{j \in [a, b-1]}g_j\left(\left( B - A \right)\uplus\biguplus_{k \in [j+1, b - 1]} \overline{R}_k\right)&\\
        &\leq K_{p_a} + \max\left\{ p_a\left( (B - A) \uplus\biguplus_{j \in [a, b-1]}R_j\right) - \sum_{j \in [a, b-1]}g_j\left((B - A)\uplus\biguplus_{k \in [j+1, b - 1]} R_k\right) : R_j \subseteq \zeros_j \text{ for } j\in[a, b-1]  \right\}&\\
        & = K_{p_a} + p_{b}(B - A)&\\
        &\leq K_{p_a} + K_{p_b}.&
    \end{align*}
    }
Here, the second inequality is by the monotonicity of the functions $g_j$ for every $j \in [\ell]$ and the last equality is by \Cref{lem:calFp_leqell_laminar:p_b-to-p_a}.
    Thus, all inequalities are equations and we have that $p_b(B - A) = K_{p_b}$, contradicting minimality of the set $B \in \calF_{p_b}$. 
\end{proof}

%% file: appendix-function-maximization-oracles.tex
\subsection{Optimizing over \ref{eqn:Q(p,m)} in Strongly Polynomial Time}\label{appendix:sec:Function-Maximization-Oracles:Qpm-oracle}
In this section, we prove \Cref{lem:optimizing-over-Q-polyhedron-intersection}. Let $p:2^V \rightarrow\Z$ and $m: V\rightarrow \R_{\ge 0}$ be two functions and let $k \geq K_p$. We will require the following generalization of the polyhedron $Q(p,m)$ from \cite{Bernath-Kiraly}.
    \begin{equation*}
    \QpmGeneralization(p, k, m) \coloneqq  \left\{ x\in \R^{V}\ \middle\vert 
        \begin{array}{l}
            {(\text{i})}{\ \ \ 0 \leq x_u \leq \min\{1, m_u\}} \hfill {\qquad \forall \ u \in V} \\
            {(\text{ii})}{\ \ x(Z) \geq 1} \hfill {\qquad\qquad\qquad\qquad\ \ \ \text{if } p(Z) = k} \\
            {(\text{iii})}{ \ x(u) = 1} \hfill {\qquad\qquad\qquad\qquad\ \ \ \ \text{if } m(u) = k} \\
            {(\text{iv})}{ \ x(Z) \leq m(Z) - p(Z) + 1} \hfill { \ \ \ \forall\ Z \subseteq V} \\
            {(\text{v})}{\ \ \floor{\frac{m(V)}{k}} \leq  x(V) \leq \ceil{\frac{m(V)}{k}}} \hfill {} \\
        \end{array}
        \right\}.
\end{equation*}
We note that $Q(p,m) = \QpmGeneralization(p, K_p, m)$.
We will also require the following definitions and results from \cite{Frank-book, frank-tardos-1988}. 
Let $f:2^V\rightarrow\R$ be a set function. 
We define the following two polyhedra w.r.t the function $f$:
\begin{align*}
    \contrapolymatroid(f)& \coloneqq  \left\{x \in \R^V : x \geq 0, x(Z) \geq f(Z) \text{ for every } Z \subseteq V\right\},&\\
\basepolyhedron(f)& \coloneqq  \left\{x \in \R^V : x(V) = f(V), x(Z) \geq f(Z) \text{ for every } Z \subseteq V\right\}.&
\end{align*}
If the set function $f$ is supermodular, then the polyhedron $\contrapolymatroid(f)$ is
called the \emph{contrapolymatroid} with border function $f$. Furthermore, if the function $f$ is also finite, then the polyhedron $\basepolyhedron(f)$ is called the \emph{base-contrapolymatroid} with border function $f$. 
The following  result by \cite{Frank-book} shows properties of the polyhedron $\contrapolymatroid(p)$ whose border function $p:2^V\rightarrow\Z\cup\{-\infty\}$ is skew-supermodular.

\begin{theorem}[\hspace{-1sp}\cite{Frank-book}]\label{thm:Frank:optimizing-over-contrapolymatroid-skew-supermodular-function}
    Let $p:2^V \rightarrow\Z\cup\{-\infty\}$ be a skew-supermodular function. We have the following:
    \begin{enumerate}
        \item for a given vector $c \in \R^{V}$,
    an optimum solution to the problem $\min\{\sum_{u \in V}c_ux_u : x \in \contrapolymatroid(p)\}$ can be computed in $\poly(|V|)$ time using $\poly(|V|)$ queries to \functionMaximizationEmptyOracle{p}, and
    \item there exists a unique supermodular function $\upperTruncation{p}:2^V \rightarrow\Z\cup\{-\infty\}$ such that $\contrapolymatroid(p) = \contrapolymatroid\left(\upperTruncation{p}\right)$.
    \end{enumerate}
\end{theorem}

A function $b'':2^V\rightarrow\Z$ is \emph{crossing-submodular} if $b''(A) + b''(B) \geq b''(A\cap B) + b''(A\cup B)$ for every $A,B \subseteq V$ for which the four sets $A - B, B - A, A\cap B$ and $V - (A\cup B)$ are non-empty. 
For convenience, we define the \emph{minimization oracle} \functionMinimizationOracle{f} for a set function $f: 2^V\rightarrow \R$ as follows: for given disjoint sets $S_0, T_0\subseteq V$ and vector $y_0\in \R^V$,  $\functionMinimizationOracle{f}(S_0, T_0, y_0)$ returns a set $Z$ and its value $f(Z)$ such that $S_0\subseteq Z\subseteq V-T_0$ and $f(Z)-y_0(Z)\le f(X)-y_0(X)$ for every $S_0\subseteq X\subseteq V-T_0$ (i.e., it returns the output of $\functionMaximizationEmptyOracle{(-f)}(S_0, T_0, -y_0)$).
The following result is implicit in \cite{frank-tardos-1988} and says that given access to \functionMinimizationOracle{b''} of a crossing-submodular function $b''$, we can optimize over the polyhedron $\basepolyhedron(b'')$ in strongly polynomial time.

\begin{theorem}[\hspace{-1sp}\cite{frank-tardos-1988} Section IV]\label{thm:FT:optimizing-over-base-polymatroid}
     Let $b'':2^V \rightarrow\Z \cup \{-\infty\}$ be a crossing-submodular function. Then, for a given vector $c \in \R^{V}$, $\arg\min\{\sum_{v \in V}c_vx_v : x \in \basepolyhedron(b'')\}$ can be computed in $\poly(|V|)$ time using $\poly(|V|)$ queries to $\functionMinimizationOracle{b''}$.
\end{theorem}

We recall that our goal is to show  \Cref{lem:optimizing-over-Q-polyhedron-intersection} which says that for a non-negative function $m:V\rightarrow\Z_{\geq 0}$ and a function $p:2^V\rightarrow\Z$ defined as the maximum of two skew-supermodular functions,
we can optimize over the polyhedron $Q(p, m)$
in strongly polynomial time using $\poly(|V|)$ queries to the function maximization oracles of the individual skew-supermodular functions. 
As a step towards showing \Cref{lem:optimizing-over-Q-polyhedron-intersection}, we first show that we can optimize over the polyhedron $Q(p, m)$ when the function $p$ is itself a skew-supermodular function. In fact, we will show the following stronger statement that we can optimize over the polyhedra $\QpmGeneralization(p, k, m)$ for every given positive integer $k \geq K_p$.

\begin{lemma}\label{lem:optimizing-over-Q-polyhedron}
    Let $p:2^V \rightarrow\Z \cup \{-\infty\}$ be a skew-supermodular function. 
    Then, the following optimization problem can be solved in $\poly(|V|)$ time using $\poly(|V|)$ queries to \functionMaximizationEmptyOracle{p}: for a given non-negative function $m:V\rightarrow\Z_{\ge 0}$ and a positive integer $k$ such that $\QpmGeneralization(p, k, m)\neq \emptyset$ and a given cost vector $c\in \R^V$, find an extreme point optimum solution to the following linear program: 
$$\max\left\{\sum_{u \in V}c_ux_u : x \in \QpmGeneralization(p,k,m) \right\}.$$

\end{lemma}
\begin{proof}
We consider the family $\calF_{p, k} := \{Z \subseteq V : p(Z) = k \text{ but } p(Z') < k \text{ for every } Z'\subsetneq Z\}$. We note that if $k = K_p$, then $\minimalMaximizerFamily{p, k} = \minimalMaximizerFamily{p}$, the family of minimal $p$-maximizers; furthermore, if $k > K_p$, then   $\minimalMaximizerFamily{p, k} = \emptyset$.
We let $R_{p, k} \coloneqq  V - \cup_{X \in \minimalMaximizerFamily{p, k}}X$ denote the set of elements not contained in any set of the family $\minimalMaximizerFamily{p, k}$. We define the family $\calL_{p, k} \coloneqq  \minimalMaximizerFamily{p, k}\cup \{R_{p, k}\}$ and let $L_{p, k} \coloneqq  \{Z \subseteq U : U \in \calL_{p, k}\}$ denote the family of all subsets of the members of $\calL_{p, k}$. We note that $\calL_{p, k}$ is the family of maximal sets in the family $L_{p, k}$.
For every $U \in \calL_{p, k}$, we let $p'_U:2^U\rightarrow\Z\cup \{-\infty\}$ be the function defined as $p'_U(Z) \coloneqq  p(Z) - 1$ for every $Z \subseteq U$. We note that the function $p'_U$ is skew-supermodular for every $U\in \calL_{p, k}$. Then, by \Cref{thm:Frank:optimizing-over-contrapolymatroid-skew-supermodular-function}, there exists a unique fully supermodular function $\upperTruncation{(p'_U)}:2^U \rightarrow\Z\cup\{-\infty\}$ such that $\contrapolymatroid\left(\upperTruncation{(p'_U)}\right) = \contrapolymatroid(p'_U)$. We will need the following claim which is implicit in \cite{Bernath-Kiraly}.

\begin{claim}[{\cite[Proof of Lemma 2.4]{Bernath-Kiraly}}]\label{claim:BK:evaluation-oracle-p'_U-upper-trancation} For every $U \in \calL_{p, k}$ and for every $Z\subseteq U$, we have that $$\upperTruncation{(p'_U)}(Z) = \min\left\{\sum_{u \in Z}x_u : x \in \contrapolymatroid\left(\upperTruncation{(p'_U)}\right)\right\}.$$
\end{claim}

The next claim says that an evaluation oracle for the function $\upperTruncation{(p'_U)}$ can be implemented in strongly polynomial time.

\begin{claim}\label{claim:evaluation-oracle-p'_U-upper-trancation}
    For every $U \in\calL_{p, k}$ and each $Z\subseteq U$, the value $\upperTruncation{(p'_U)}(Z)$ can be evaluated in $\poly(|V|)$ time using $\poly(|V|)$ queries to \functionMaximizationEmptyOracle{p}.
\end{claim}
\begin{proof}
 Let $U \in \calL_{p, k}$ and $Z\subseteq U$. By \Cref{claim:BK:evaluation-oracle-p'_U-upper-trancation} and \Cref{thm:Frank:optimizing-over-contrapolymatroid-skew-supermodular-function}, we have that $$\upperTruncation{(p'_U)}(Z) = \min\left\{\sum_{u \in Z}x_u : x \in \contrapolymatroid\left(\upperTruncation{(p'_U)}\right)\right\} = \min\left\{\sum_{u \in Z}x_u : x \in \contrapolymatroid\left(p'_U\right)\right\}.$$
Thus, it suffices to solve the RHS optimization problem. We recall that the function $p'_U$ is skew-supermodular. Furthermore, we note that a query to \functionMaximizationEmptyOracle{p'_U} can be implemented in $O(|V|)$ time using one query to $\functionMaximizationEmptyOracle{p}$. Then, by \Cref{thm:Frank:optimizing-over-contrapolymatroid-skew-supermodular-function} we can 
   solve the RHS optimization problem
   in $\poly(|V|)$ time and $\poly(|V|)$ queries to $\functionMaximizationEmptyOracle{p}$.
\end{proof}

We now define a central function that will allow us to optimize over the polyhedron $\QpmGeneralization(p, k, m)$. For convenience, we define this function in three stages using two intermediate functions. Let the first intermediate function $f_{p, k}:2^V \rightarrow\Z\cup\{-\infty\}$ be defined as follows: for every $Z \subseteq V$,
$$f_{p, k}(Z) \coloneqq  \begin{cases}
    \upperTruncation{(p'_U)}(Z)
    & \text{ for every $Z \in L_{p, k}$,}\\
    -\infty& \text{ otherwise}.
\end{cases}$$
For the second stage, we will require a larger ground set. Let $V' \coloneqq  V \uplus \{t\}$ be the ground set extended with a new element $t$. Let the second intermediate function $b'_{p, k, m}:2^{V'} \rightarrow \Z\cup \{+\infty\}$ be defined as follows: for every $Z \subseteq V'$,
$$b'_{p, k, m}(Z) \coloneqq  \begin{cases}
    m(Z) - 1& \text{ if $Z \in \minimalMaximizerFamily{p, k}$}\\
    +\infty& \text{ if $t \not \in Z$ and $Z \not \in \minimalMaximizerFamily{p, k}$}\\
    -f_{p, k}(V - Z)& \text{ if $t \in Z$}.
\end{cases}$$
We now define our central function of interest. Let $\calD_{m, k} \coloneqq  \{u \in V : m(u) = k\}$. Let  the function $b''_{p,k,m}:2^{V'} \rightarrow\Z\cup\{+\infty\}$  be defined as follows: for every $Z \subseteq V'$,

$$b''_{p,k,m}(Z) \coloneqq  \begin{cases}
    b'_{p,k,m}(Z)& \text{ if $Z$ is not a singleton or complement of a singleton}\\
    \min\left\{b'_{p,k,m}(u), m(u) - \indicator_{u \in \calD_{m, k}}\right\}& \text{ if $X = \{u\} \subseteq V$}\\
    \min\left\{b'_{p,k,m}(t), -m(V) + \ceil{m(V)/k}\right\}&\text{ if $X = \{t\}$}\\
    \min\left\{b'_{p,k,m}(Z), -m(u) + 1\right\}&\text{ if $X = V' - \{u\}$ for some $u\in V$}\\
    \min\left\{b'_{p,k,m}(V), m(V) - \floor{m(V)/k}\right\}& \text{ if $X = V$}.
\end{cases}$$

The next claim is also implicit in \cite{Bernath-Kiraly} and says that the projection of the base contrapolymatroid with border function $b''_{p,k,m}$ onto the ground set $V$ is a translated reflection of the polyhedron $\QpmGeneralization(p,k,m)$. We formally define the projection of the base contrapolymatroid onto the ground set $V$ as $\mathtt{proj}_V(\basepolyhedron(b''_{p,k,m})) := \{x \in \R^V : \exists x_t \in \R \text{ s.t. } (x, x_t) \in \basepolyhedron(b''_{p,k,m})\}$.

\begin{claim}[\hspace{-1sp}\cite{Bernath-Kiraly}]\label{claim:BK:projection-of-base-polyhedron} 
$\mathtt{proj}_V(\basepolyhedron(b''_{p,k,m})) = m - \QpmGeneralization(p,k,m)$.
\end{claim}

By \Cref{claim:BK:projection-of-base-polyhedron}, it follows that if we can optimize over $\basepolyhedron(b''_{p,k,m})$ in $\poly(|V|)$ time using $\poly(|V|)$ queries to \functionMinimizationOracle{p}, then we can optimize over $\QpmGeneralization(p, k,m)$ in $\poly(|V|)$ time using $\poly(|V|)$ queries to \functionMinimizationOracle{p}. 
We now show that we can optimize over the base-contrapolymatroid $\basepolyhedron(b''_{p,k,m})$ in two steps. 
Firstly, we show that the function $b''_{p,k,m}$ is crossing-submodular. Thus, by \Cref{thm:FT:optimizing-over-base-polymatroid}, we can optimize over the base-contrapolymatroid  $\basepolyhedron(b''_{p,k,m})$ in $\poly(|V|)$ time using $\poly(|V|)$ queries to \functionMinimizationOracle{b''_{p,k,m}}. 
Secondly, we show that \functionMinimizationOracle{b''_{p,k,m}} can be implemented in $\poly(|V|)$ time using $\poly(|V|)$ queries to 
\functionMaximizationEmptyOracle{p}. This will complete the proof of the lemma.

We now show the first step, i.e.  the function $b''_{p,k,m}$ is crossing-submodular. We prove this as \Cref{claim:b''_pm-crossing-submodular} after the next claim which is also implicit in \cite{Bernath-Kiraly}. 

\begin{claim}[{\cite[Proof of Lemma 2.4]{Bernath-Kiraly}}]\label{claim:BK:b'-crossing-submodular} The function $b'_{p,k,m}$ is crossing-submodular. 
\end{claim}

\begin{claim}\label{claim:b''_pm-crossing-submodular}
    The function $b''_{p,k,m}$ is crossing-submodular.
\end{claim}
\begin{proof}
    The function $b''_{p,k,m}$ is obtained from the function $b'_{p,k,m}$ by decreasing values on singletons and complement of singletons. This operation is known to preserve crossing submodularity (for example, see part (c) in Section 2 of \cite{Gabow-crossing-supermodularity}) and so the claim follows from \Cref{claim:BK:b'-crossing-submodular}.
\end{proof}

We now show the second part, i.e. \functionMinimizationOracle{b''_{p,k,m}} can be implemented in $\poly(|V|)$ time using $\poly(|V|)$ queries to 
\functionMaximizationEmptyOracle{p}. 
Let the input to $\functionMinimizationOracle{b''_{p,k,m}}$ be disjoint sets $A, B\subseteq V'$ and a vector $y\in \R^{V'}$. 
The goal is to answer $\functionMinimizationOracle{b''_{p,k,m}}(A, B, y)$, i.e., we would like to find an optimum solution to $\min\{b''_{p, k,m}(Z)-y(Z): A\subseteq Z\subseteq V'-B\}$. 
By definition of the function $b''_{p,k,m}$, it follows that if we can find an optimum solution to $\min\{-f_{p,k}(V-Z) + y(Z): A\cup\{t\}\subseteq Z\subseteq V'-B\}$ 
in $\poly(|V|)$ time using $\poly(|V|)$ queries to 
\functionMaximizationEmptyOracle{p}, 
then we can find an optimum solution to $\min\{b''_{p, k,m}(Z)-y(Z): A\subseteq Z\subseteq V'-B\}$
in $\poly(|V|)$ time using $\poly(|V|)$ queries to 
\functionMaximizationEmptyOracle{p}. 
We have that 
    \begin{align*}
        \min&\left\{-f_{p, k}(V - Z) + y(Z) : A\cup \{t\} \subseteq Z\subseteq V' - B\right\}\\
        & = \max\left\{f_{p, k}(Z) - y(V' - Z) : B \subseteq Z \subseteq  V' - (A\cup \{t\})\right\}.&
    \end{align*}
    By definition of the function $f_{p,k}$, in order to solve the RHS optimization problem in $\poly(|V|)$ time using $\poly(|V|)$ queries to \functionMaximizationEmptyOracle{p}, it suffices to solve the following problem in $\poly(|V|)$ time using $\poly(|V|)$ queries to \functionMaximizationEmptyOracle{p}: 
    \begin{align*}
        \max \left\{\max_{U \in \calL_{p, k}}\left\{\upperTruncation{(p'_U)}(Z) - y(U - Z) : B \subseteq Z \subseteq  U - A
        \right\}\right\}.&
    \end{align*}
    We note that the family $\calL_{p, k} = \minimalMaximizerFamily{p, k} \cup \{R_p\}$ can be computed in $\poly(|V|)$ time using $\poly(|V|)$ queries to $\functionMaximizationEmptyOracle{p}$: if $k = K_p$, then the family $\minimalMaximizerFamily{p,k} = \minimalMaximizerFamily{p}$ which can be computed  by \Cref{lem:SzigetiAlgorithm:computingCalFp}. Otherwise (i.e., if $k > K_p$) the family $\minimalMaximizerFamily{p, k} = \emptyset$, and consequently, $\calL_{p, k} = \{V\}$. Thus, in order to solve the above optimization problem 
    in $\poly(|V|)$ time using $\poly(|V|)$ queries to $\functionMaximizationEmptyOracle{p}$, it suffices to solve the following optimization problem 
    in $\poly(|V|)$ time using $\poly(|V|)$ queries to $\functionMaximizationEmptyOracle{p}$ 
    for every $U\in \calL_{p,k}$:    
    $$\max\left\{\upperTruncation{(p'_U)}(Z) - y(U - Z) : B \subseteq Z \subseteq  U - A\right\}.$$ 
    We recall from \Cref{claim:BK:evaluation-oracle-p'_U-upper-trancation} that the function $\upperTruncation{(p'_U)}$ is supermodular and we also have an evaluation oracle for this function that runs in $\poly(|V|)$ time using $\poly(|V|)$ queries to \functionMaximizationEmptyOracle{p} by \Cref{claim:evaluation-oracle-p'_U-upper-trancation}. Thus, we can solve the above optimization problem for every $U\in \calL_{p,k}$ using supermodular maximization (i.e., submodular minimization) in $\poly(|V|)$ time using $\poly(|V|)$ queries to the evaluation oracle of $\upperTruncation{p'_U}$. This completes the proof of the lemma. 
\end{proof}

Armed with the ability to optimize over $\QpmGeneralization(p, k,m)$ for skew-supermodular functions $p$, we now move on to show \Cref{lem:optimizing-over-Q-polyhedron-intersection}, i.e. optimizing over $Q(p,m)$ when $p$ is the maximum of two skew-supermodular functions. We will show that this reduces to optimizing over the intersection of two \QpmGeneralization-polyhedra which we already know how to optimize over by \Cref{lem:optimizing-over-Q-polyhedron}. We will need the following additional definitions and results from \cite{Frank-book}.

\begin{definition}
Let $p:2^V\rightarrow\Z\cup\{-\infty\}$ and $b:2^V\rightarrow\Z\cup\{+\infty\}$ be two set functions. The pair $(p,b)$ is \emph{paramodular} (or is a \emph{strong pair}) if $p(\emptyset) = b(\emptyset) = 0$, the function $p$ is supermodular, the function $b$ is submodular, and the functions $p$ and $b$ satisfy the \emph{cross-inequality}: $b(X) - p(Y) \geq b(X - Y) - p(Y - X)$ for every pair of subsets $X, Y \subseteq V$. If $(p,b)$ is paramodular, then the polyhedren $$\gpolymatroid(p, b) \coloneqq  \{x \in \R^V : p(Z) \leq x(Z) \leq b(Z) \text{ for every $Z\subseteq V$}\}$$ is said to be a \emph{generalized polymatroid} (or \emph{g-polymatroid}) with \emph{border pair} $(p,b)$. Here, $p$ is the lower border function and $b$ is the upper border function.     
\end{definition}

The following theorem says that a g-polymatroid uniquely determines its border paramodular pair. 

\begin{theorem}[{\cite[Theorem 14.2.8]{Frank-book}}]\label{thm:Frank:gpolymatroid-unique-paramodular-pair}
    Let $Q \subseteq \R^V$ be a g-polymatroid. Then, there exists a unique 
    paramodular pair $(p:2^V\rightarrow\Z\cup\{-\infty\},  b:2^V\rightarrow\Z\cup\{+\infty\})$ such that $Q = \gpolymatroid(p,b)$. Furthermore, for every $Z\subseteq V$, $b(Z) = \max\{x(Z) : z \in Q\}$ and $p(Z) = \min\{x(Z) : x \in Q\}$.
\end{theorem}

We will also need the following theorem by \cite{frank-tardos-1988} which says that for a crossing-submodular function, there exists a unique fully submodular function that realizes the same base polymatroid and can be evaluated in strongly polynomial time given access to a function minimization oracle.
\begin{theorem}[{\cite[Section IV.4]{frank-tardos-1988}}]\label{thm:ft:bitruncation}
    Let $b:2^V\rightarrow\Z$ be a crossing submodular function. We have the following:
    \begin{enumerate}[label=(\arabic*)]
        \item if $\basepolyhedron(b)$ is non-empty, then there exists a unique fully submodular function $\biTruncation{b}:2^V \rightarrow\Z\cup\{+\infty\}$ such that $\basepolyhedron(b) = \basepolyhedron\left(\biTruncation{b}\right)$, and
        \item $\biTruncation{b}(Z)$ can be computed in $\poly(|V|)$ time using $\poly(|V|)$ queries to \functionMinimizationOracle{b}.
    \end{enumerate}
\end{theorem}

We now restate and prove \Cref{lem:optimizing-over-Q-polyhedron-intersection}.
\lemQpmOracle*
\begin{proof} 
We note that $K_p = \max\{K_q, K_r\}$. Consequently, $Q(p,m) = \QpmGeneralization(p, K_p, m) = \QpmGeneralization(q, K_p, m) \cap \QpmGeneralization(r, K_p, m)$. Thus, it suffices to optimize over the intersection of the polyhedra $\QpmGeneralization(q,K_p, m)$ and $\QpmGeneralization(r,K_p,m)$ both of which are non-empty. 
The following claim from \cite{Bernath-Kiraly} says that the polyhedra $\QpmGeneralization(r,K_p, m)$ and $\QpmGeneralization(q,K_p,m)$ are g-polymatroids.

\begin{claim}[{\cite[Lemma 2.4]{Bernath-Kiraly}}]\label{claim:BK:Qpm-gpolymatroid}
    Let $p:2^V\rightarrow\Z\cup\{-\infty\}$ be a skew-supermodular function with $K_p>0$ and $m:V\rightarrow\Z_{\geq 0}$ be a non-negative function. Then, for every $k \in \Z_+$ such that $k\geq K_p$, the polyhedron $\QpmGeneralization(p,k,m)$ is a g-polymatroid.
\end{claim}

Thus, by \Cref{claim:BK:Qpm-gpolymatroid} and \Cref{thm:Frank:gpolymatroid-unique-paramodular-pair}, there exist functions $p_1, p_2:2^V \rightarrow\Z\cup\{-\infty\}$ and $b_1, b_2:2^V \rightarrow\Z\cup\{+\infty\}$ such that $\QpmGeneralization(q,K_p,m) = \gpolymatroid(p_1, b_1)$ and $\QpmGeneralization(r, K_p, m) = \gpolymatroid(p_2, b_2)$. By \Cref{lem:optimizing-over-Q-polyhedron}, we can optimize over the polyhedra $\QpmGeneralization(q,K_p, m)$ and $\QpmGeneralization(r,K_p, m)$ in $\poly(|V|)$ time using $\poly(|V|)$ queries to \functionMaximizationOracle{q} and \functionMaximizationOracle{r} respectively. Thus, by \Cref{thm:Frank:gpolymatroid-unique-paramodular-pair}, we can implement the evaluation oracle for the functions $p_1, p_2, b_1, b_2$ in $\poly(|V|)$ time using $\poly(|V|)$ queries to \functionMaximizationOracle{q} and \functionMaximizationOracle{r}. 

Let $V' \coloneqq  V^{1} \uplus V^2$ denote an extended ground set, where $V^1$ and $V^2$ are disjoint copies of the ground set $V$. We define the function $b:2^{V'}\rightarrow\Z\cup\{+\infty\}$ as follows: for every $Z \subseteq V$, 
$$b(Z) \coloneqq  \begin{cases}
    b_1(Z)  & \text{ if $Z \subseteq V^1$}\\
    -p_1(Z) & \text{ if $V' - Z \subseteq V^1$ }\\
    b_2(Z)  & \text{ if $Z \subseteq V^2$}\\
    -p_2(Z) & \text{ if $V' - Z \subseteq V^2$}\\
    +\infty & \text{ otherwise.}
\end{cases}$$
The function $b$ can be shown to be crossing submodular (implicit in \cite{Frank-book}). 
\begin{claim}[{\cite[Proof of Theorem 16.1.7]{Frank-book}}]\label{claim:frank:Q-polyhedron-intersection-is-subflow-polyhedron}
    The function $b$ is crossing submodular. 
\end{claim}

Thus, by \Cref{thm:ft:bitruncation}(1), there exists a unique fully-submodular function $\biTruncation{b}:2^{V'}\rightarrow\Z\cup\{+\infty\}$ such that $\basepolyhedron(b) = \basepolyhedron\left(\biTruncation{b}\right)$. Using the results in \cite{Frank-book}, it can be shown that $\QpmGeneralization(q,K_p, m) \cap \QpmGeneralization(r,K_p,m)$ is 
in fact a \emph{sub-flow} polyhedron \emph{constrained} by the function $b$. This allows us to optimize over $\QpmGeneralization(r,K_p, m) \cap \QpmGeneralization(q,K_p,m)$ in $\poly(|V|)$ time using $\poly(|V|)$ queries to \functionMinimizationOracle{\biTruncation{b}}. 
We note that that we will not require any technical details regarding subflow polyhedra beyond  this connection and so we refrain from defining subflow polyhedra. We refer the reader to \cite[Section 16.1]{Frank-book} for an extensive survey of the topic.

\begin{claim}[{\cite[Section 16.3.2]{Frank-book}}]\label{claim:optimizing-over-subflow-polyhedron-constrained-by-crossing-submod-function}
    For a given vector $c \in \R^V$, we can find an optimum solution to $\min\{\sum_{u \in V}c_ux_u : x \in \QpmGeneralization(q,K_p, m) \cap \QpmGeneralization(r,K_p,m)\}$ in $\poly(|V|)$ time using $\poly(|V|)$ queries to \functionMinimizationOracle{\biTruncation{b}}.
\end{claim}

Thus, by \Cref{claim:optimizing-over-subflow-polyhedron-constrained-by-crossing-submod-function}, it suffices to implement \functionMinimizationOracle{\biTruncation{b}} in $\poly(|V|)$ time using $\poly(|V|)$ queries to \functionMaximizationOracle{q} and \functionMaximizationOracle{r}.
We note that since the function $\biTruncation{b}$ is submodular, \functionMinimizationOracle{\biTruncation{b}} can be 
implemented in $\poly(|V|)$ time using $\poly(|V|)$ queries to the function evaluation oracle of \biTruncation{b} (via submodular minimization). 
By \Cref{thm:ft:bitruncation}(2), the function evaluation oracle of \biTruncation{b} can be implemented in $\poly(|V|)$ time using $\poly(|V|)$ queries to \functionMinimizationOracle{b}. 
By definition of the function $b$, \functionMinimizationOracle{b} can be implemented in $\poly(|V|)$ time using $\poly(|V|)$ queries to \functionMinimizationOracle{b_1}, \functionMinimizationOracle{b_2}, \functionMinimizationOracle{(-p_1)}, and \functionMinimizationOracle{(-p_2)}. We note that \functionMinimizationOracle{b_1}, \functionMinimizationOracle{b_2}, \functionMinimizationOracle{(-p_1)}, and \functionMinimizationOracle{(-p_2)} can be implemented in $\poly(|V|)$ time using $\poly(|V|)$ queries to respective function evaluation oracles via submodular minimization since $b_1, b_2$ are submodular functions and $p_1, p_2$ are supermodular functions. We recall that the evaluation oracles for $p_1, p_2, b_1, b_2$ can be implemented in 
$\poly(|V|)$ time using $\poly(|V|)$ queries to \functionMaximizationOracle{q} and \functionMaximizationOracle{r}. This completes the proof of the lemma. 
\end{proof}

\subsection{Computing $\alpha^{(4)}$: Ratio Maximization Oracle in Strongly Polynomial Time}\label{appendix:sec:Function-Maximization-Oracles:AlphaFour-oracle}
In this section, we prove \Cref{lem:alpha4-oracle}. We will need the following general result that is implicit in the work of Cunningham \cite{Cunningham-Optimal-Attack-Reinforcement}. We include a proof of the result for completeness.

\begin{proposition}[\hspace{-1sp}\cite{Cunningham-Optimal-Attack-Reinforcement}]\label{prop:cunningham}
Let $f: 2^{V} \rightarrow \Z$, $g: 2^V\rightarrow\Z_{+}$ be two functions given by evaluation oracles. Given an oracle that takes $\lambda \in \Q$ as input and returns $\arg\max\{f(X) - \lambda g(X) : X\subseteq V\}$, there exists an algorithm that runs in time $O(K_g)$ and returns $\arg\max\left\{\frac{f(X)}{g(X)} : X \subseteq V\right\}$, where $K_g \coloneqq  \max\{g(X) : X\subseteq V\}$.
\end{proposition}
\begin{proof}
We will show that \Cref{alg:Cunningham} 
(\Cref{alg:Cunningham} is a generalization of Algorithm 1 of \cite{Cunningham-Optimal-Attack-Reinforcement}) satisfies the claimed property.
Since the input function $g$ is a positive function, we have that $\max\{f(X)/g(X) : X\subseteq V\} = \lambda$ if and only if $\max\left\{f(X) - \lambda g(X) : X\subseteq V\right\} = 0$. If \Cref{alg:Cunningham} terminates, then for the set $Z$ returned by the algorithm and for $\lambda \coloneqq  f(Z)/g(Z)$, we have that $f(Z)-\lambda g(Z) = 0$ by definition of $\lambda$ and $\max\{f(X)-\lambda g(X):X\subseteq V\}\le 0$ by termination condition of the algorithm, which together imply that $\max\{f(X)-\lambda g(X):X\subseteq V\}=0$. Consequently, if \Cref{alg:Cunningham} terminates, then the set $Z$ returned by the algorithm  
indeed maximizes the ratio $f(X)/g(X)$. 
\begin{algorithm}[!htb]
\caption{Ratio Maximization}\label{alg:Cunningham}
$\mathtt{Algorithm}(f, g)$:
\begin{algorithmic}[1]
\State{$Z \coloneqq  V$, $\lambda \coloneqq  f(V)/g(V)$}
\While{$\max\left\{f(X) - \lambda g(X) : X\subseteq V\right\} > 0$}
\State{$Z \coloneqq  \arg\max\left\{f(X) - \lambda g(X) : X\subseteq V\right\}$}
    \State{$\lambda \coloneqq  f(Z) / g(Z)$}
\EndWhile
\State{\Return{Z}}
\end{algorithmic}
\end{algorithm}

Consider an arbitrary while-loop iteration of \Cref{alg:Cunningham} in which $\max\{f(X) - \lambda g(X) : X\subseteq V\} > 0$. Let $\lambda'$ and $Z'$ denote by the value $\lambda$ and the set $Z \subseteq V$ in the subsequent while-loop iteration of the algorithm. Furthermore, let  $\lambda''$ and $\Z''$ denote the value $\lambda$ and the set $Z \subseteq V$ after two while-loop iterations of the algorithm. In particular, we have that
\begin{align*}
    Z' &\coloneqq  \arg\max\{f(X) - \lambda g(X) : X\subseteq V\},&\\
    \lambda'& \coloneqq  f(Z')/g(Z'),&\\
    Z''&\coloneqq  \arg\max\{f(X) - \lambda' g(X) : X\subseteq V\}, &\\
    \lambda''& \coloneqq  f(Z'')/g(Z'').&
\end{align*}

Then the following claim says that either the algorithm returns the set $Z''$, or $g(Z') < g(Z'')$.
\begin{claim}\label{claim:cunningham-alg:g-strictly-increases}
    Either $f(Z'') - \lambda' g(Z'') = 0$ or $g(Z'') > g(Z')$. 
\end{claim}
\begin{proof}
    Suppose that $f(Z'') - \lambda' g(Z'') > 0$. Then, we have the following: 
    \begin{align*}
        0& < f(Z'') - \lambda' g(Z'')&\\
        & = f(Z'') - \lambda g(Z'') + \lambda g(Z'') - \lambda' g(Z'')&\\
        & \leq f(Z') - \lambda g(Z') + \lambda g(Z'') - \lambda' g(Z'')&\\
        & = \lambda' g(Z') - \lambda g(Z') + \lambda g(Z'') - \lambda' g(Z'')&\\
        &= (\lambda' - \lambda) (g(Z') - g(Z''))&\\
        &< g(Z') - g(Z'')&
    \end{align*}
    Here, the second inequality is by  definition of the set $Z'$. The second equality is by the definition of the value $\lambda'$. The final inequality is because $\lambda < f(Z')/g(Z') = \lambda'$.
\end{proof}

By \Cref{claim:cunningham-alg:g-strictly-increases}, we have that during every iteration of the while-loop, the value $g(Z') > g(Z)$. Since the range of the function $g$ is the set of positive integers, the number of while-loop iterations of the algorithm cannot exceed $K_g$.
\end{proof}

\lemAlphaFourOracle*
\begin{proof}  For every pair of vertices $u,v \in {A\choose 2}$, we consider $f_{uv}:2^{V - \{u, v\}}\rightarrow\Z$ and $g_{u,v}:2^{V - \{u, v\}} \rightarrow \Z_+$ defined as follows:
    \begin{align*}
        f_{uv}(Z)&\coloneqq  p\left(Z\cup \left\{u, v\right\}\right) - y\left(Z \cup \left\{u, v\right\}\right) \ \ \text{for every $Z\subseteq V - \{u,v\}$,}\\
         g_{uv}(Z)&\coloneqq  1 + \left|\left(A - \left\{u,v\right\}\right)\cap Z\right| \ \ \text{for every $Z\subseteq V - \{u,v\}$.}
    \end{align*}
    Then, we have that 
    \[
    \max_{u, v \in \binom{A}{2}}\left\{\max\left\{\frac{f_{uv}(Z)}{g_{uv}(Z)}:Z\subseteq V - \{u,v\}\right\}\right\} = \max\left\{\frac{p(Z) - y(Z)}{|A\cap Z| - 1} : |A\cap Z|\geq 2, Z\subseteq V\right\}.
    \]
    Moreover, for a pair $u, v\in \binom{A}{2}$ and a set $Z\subseteq V-\{u, v\}$ that achieves the maximum in the LHS, the set $Z\cup \{u, v\}$ achieves the maximum in the RHS. Hence, it suffices to solve the optimization problem
    $$\arg\max\left\{\frac{f_{uv}(Z)}{g_{uv}(Z)}:Z\subseteq V - \{u,v\}\right\}$$  for every pair $u,v \in {A\choose 2}$. Fix a pair $u, v\in {A\choose 2}$. By \Cref{prop:cunningham}, we can solve 
    the above optimization problem in strongly polynomial time given access to evaluation oracles for $f_{uv}$ and $g_{uv}$, and the maximization oracle $\arg\max\left\{f_{uv}(Z) - \lambda g_{uv}(Z) : Z \subseteq V - \{u,v\}\right\}$ for given $\lambda \in \Z$. We note that the value $g_{uv}(Z)$ can be computed explicitly in time $O(|V|)$. Furthermore, the value $f_{uv}(Z)$ can be computed by querying \functionMaximizationEmptyOracle{p} with inputs 
    $S_0 = Z \cup \{u,v\}$, $T_0 = V - (Z \cup \{u, v\})$, and using the same $y_0$. Finally, the maximization oracle $\arg\max\left\{f_{uv}(Z) - \lambda g_{uv}(Z) : Z \subseteq V - \{u,v\}\right\}$ for a given $\lambda$ can be implemented in $\poly(|V|)$ time by querying  \functionMaximizationEmptyOracle{p} with inputs $S_0 = \{u,v\}$, $T_0 = \emptyset$, and $y_0 = -\lambda\chi_{A-\{u, v\}} - y$. 
    
\end{proof}